\pdfoutput = 1
\documentclass[a4paper,11pt]{article}
\usepackage[margin=2cm]{geometry}

\usepackage{commands}

\definecolor[named]{urlblue}{cmyk}{1,0.58,0,0.21}

\hypersetup{
    breaklinks=true,
    colorlinks=true,
    citecolor=purple!70!blue!60!black,
    linkcolor=purple!70!blue!90!black,
    urlcolor=urlblue,
    pdflang={en},
    pdfborder={0 0 0},
    pdftitle={Optimally Repurposing Existing Algorithms to Obtain Exponential-Time Approximations},
    pdfauthor={Bar\i\c{s} Can Esmer, Ariel Kulik, D{\'{a}}niel Marx, Daniel Neuen, and Roohani Sharma},
    pdfcreator={},
}

\title{Optimally Repurposing Existing Algorithms to\\Obtain Exponential-Time Approximations}

\author[1]{Bar\i\c{s} Can Esmer\thanks{The author is part of Saarbrücken Graduate School of Computer Science, Germany.}}
\author[1]{Ariel Kulik}
\author[1]{D{\'{a}}niel Marx\thanks{Research supported by the European Research Council (ERC) consolidator grant No.~725978 SYSTEMATICGRAPH.}}
\author[2]{Daniel Neuen}
\author[3]{Roohani Sharma}
\affil[1]{CISPA Helmholtz Center for Information Security, Saarbr\"ucken, Germany. \{\texttt{baris-can.esmer|ariel.kulik|marx}\}\texttt{@cispa.de}}
\affil[2]{University of Bremen, Bremen, Germany. \texttt{dneuen@uni-bremen.de}}
\affil[3]{Max Planck Institute for Informatics, Saarland Informatics Campus, Saarbr\"ucken, Germany. \texttt{rsharma@mpi-inf.mpg.de}}

\begin{document}

\maketitle
\thispagestyle{empty}

\begin{abstract}
 The goal of this paper is to understand how exponential-time approximation algorithms can be obtained from existing polynomial-time approximation algorithms, existing parameterized exact algorithms, and existing parameterized approximation algorithms.
 More formally, we consider a monotone subset minimization problem over a universe of size $n$ (e.g., \textsc{Vertex Cover} or \textsc{Feedback Vertex Set}).
 We have access to an algorithm that finds an $\alpha$-approximate solution in time $c^k\cdot n^{\OO(1)}$ if a solution of size $k$ exists (and more generally, an extension algorithm that can approximate in a similar way if a set can be extended to a solution with $k$ further elements).
 Our goal is to obtain a $d^n \cdot n^{\OO(1)}$ time $\beta$-approximation algorithm for the problem with $d$ as small as possible.
 That is, for every fixed $\alpha,c,\beta \geq 1$, we would like to determine the smallest possible $d$ that can be achieved in a model where our problem-specific knowledge is limited to checking the feasibility of a solution and invoking the $\alpha$-approximate extension algorithm.
 Our results completely resolve this question:
 \begin{enumerate}
  \item For every fixed $\alpha,c,\beta \geq 1$, a simple algorithm (``approximate monotone local search'') achieves the optimum value of $d$.
  \item Given $\alpha,c,\beta \geq 1$, we can efficiently compute the optimum $d$ up to any precision $\eps > 0$.
 \end{enumerate}
 Earlier work presented algorithms (but no lower bounds) for the special case $\alpha = \beta = 1$ [Fomin et al., J.\ ACM 2019] and for the special case $\alpha = \beta > 1$ [Esmer et al., ESA 2022].
 Our work generalizes these results and in particular confirms that the earlier algorithms are optimal in these special cases.

 We compare the performance of the resulting algorithms to what is obtainable by brute force, that is, in a setting where we have no problem-specific knowledge beyond checking the feasibility of a solution. We show that, except in the case $\alpha>\beta=1$, the resulting $d$ is strictly better than what can be obtained by brute force.
 For example, somewhat counterintuitively, given access to a $1000$-approximate extension algorithm running in time $1000^k\cdot n^{\OO(1)}$ allows us to obtain a $1.001$-approximation algorithm with running time $d^n\cdot n^{\OO(1)}$ strictly better than what is possible by brute force.
 Our technique gives novel results for a wide range of problems including \textsc{Feedback Vertex Set}, \textsc{Directed Feedback Vertex Set}, \textsc{Odd Cycle Traversal}  and \textsc{Partial Vertex Cover}.
\end{abstract}

\newpage
\tableofcontents
\thispagestyle{empty}

\clearpage
\setcounter{page}{1}

\section{Introduction}
\label{sec:intro}
It is widely believed that NP-hard problems cannot be solved in polynomial time and any algorithm solving them has some form of exponential running time. During the past decades, there has been a great deal of interest in trying to obtain improved exponential-time algorithms for basic NP-hard problems, see for example the monograph of Fomin and Kratsch \cite{FominK10}.
Typically, for \emph{subset problems}, where the goal is to find a subset of a given $n$-sized universe $U$ that satisfies some property $\Pi$, a solution can be found by enumerating all $2^n$ subsets of $U$.
Therefore, the goal is to design algorithms that beat this exhaustive search and run in time $\OO^*\left(d^n\right)$\footnote{The $\OO^*$ notation hides polynomial factors in the expression.} for as small $1 < d < 2$ as possible.
More recently, there has been interest in \emph{exponential-time approximation algorithms} \cite{EsmerKMNS22,AroraBS15,BansalCLNN19,BourgeoisEP11,CyganKW09,EscoffierPT16,ManurangsiT18} to obtain approximation ratios that are better than what is considered possible in polynomial time.
In this paper, we analyze how the simple technique of monotone local search can be used to derive exponential-time approximation algorithms by repurposing existing exact parameterized algorithms, existing polynomial-time approximation algorithms, and existing parameterized approximation algorithms.
Furthermore, we show that monotone local search is the optimal way to convert between those types of algorithms.

Our setting is the following.
We consider \emph{subset minimization} problems where the goal is to find a subset of the $n$-sized universe $U$ of \emph{minimum cardinality} that satisfies some additional property $\Pi$.
To make approximation feasible, we consider only monotone properties, that is, if $S\subseteq U$ satisfies $\Pi$, then so does any superset of $S$.
For any approximation ratio $\beta \geq 1$, we say that a subset $S \subseteq U$ satisfying the property $\Pi$ is a \emph{$\beta$-approximate solution} if $|S| \leq \beta \cdot |\OPT|$, where $\OPT \subseteq U$ is an optimum solution.

An \emph{exponential $\beta$-approximation algorithm} for a subset minimization problem returns a $\beta$-ap\-pro\-xi\-ma\-te solution and runs in time $\OO^*(d^n)$ for some $1 < d < 2$.
We assume that we are given access to an algorithm with the following specification: given a problem instance and an integer $k$, if the optimum solution has size at most $k$, then the algorithm returns a solution of size at most $\alpha\cdot k$ in time $\OO^*(c^k)$.
Let us observe that in the special case of $c=1$, it is equivalent to the notion of polynomial-time constant-factor approximation algorithm \cite{Vazirani01} and in the special case of $\alpha=1$, it is equivalent to an exact fpt-algorithm \cite{CyganFKLMPPS15}.
In general, the definition covers constant-factor parameterized approximation algorithms, which have received increased attention recently \cite{BrankovicF13,BrankovicF11,FellowsKRS18,KulikS20,LokshtanovMRSZ21,Marx08,BhattacharyyaBE21,ChitnisFM21,FeldmannSLM20,ChalermsookCKLM20,SLM19,LinRSW23,LinRSW22,Lin21,KawarabayashiL20,ChenL19}.
For technical reasons, instead of an algorithm finding a small solution, we need an algorithm finding a small \emph{extension}: given a set $X$ that can be extended to a solution by $k$ further elements, it returns such an extension with at most $\beta\cdot k$ further elements.
For many problems that are defined in terms of deletions (e.g., \textsc{Vertex Cover}, \textsc{Feedback Vertex Set}, \textsc{Multicut} etc.), the two notions are equivalent via a simple reduction: the extension problem is equivalent to solving the problem on $G-X$.

Our main goal is to understand, for a given $\alpha$, $\beta$, and $c$, what is the best $\OO^*(d^n)$ time $\beta$-approximation algorithm we can obtain if we have access to a parameterized $\alpha$-approximate extension algorithm running in time $\OO^*(c^k)$.

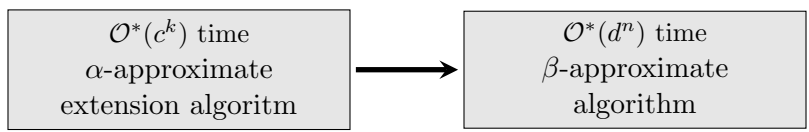
\medskip
\begin{center}
 \begin{tikzpicture}[>=stealth]
  \node[draw,rectangle,fill=gray!20,minimum width=4.5cm,minimum height=1.5cm,align=center] (left) at (0,0.75) {\small $\OO^*(c^k)$ time\\$\alpha$-approximate\\extension algoritm};
  \node[draw,rectangle,fill=gray!20,minimum width=4.5cm,minimum height=1.5cm,align=center] (right) at (6,0.75) {\small $\OO^*(d^n)$ time\\$\beta$-approximate\\algorithm};
  \draw[line width=2pt,->,shorten >=2pt,shorten <=2pt] (left.east) -- (right.west);
 \end{tikzpicture}
\end{center}
\medskip

The special case when $\alpha=\beta=1$, that is, using exact fpt-algorithms to obtain exact exponential-time algorithms, was treated by Fomin et al.~\cite{FominGLS19}: they give a very simple procedure, monotone local search, that repurposes an exact fpt-algorithm with running time $\OO^*(c^k)$ to obtain an exponential-time algorithm with $d=2-\frac{1}{c}$.
Monotone local search was extended to an approximate version by Esmer et al.~\cite{EsmerKMNS22} to handle the case $\alpha=\beta>1$, with a much more complicated (non-closed-form) expression for $d $, which we denote by $\esaamlsbound(\beta,c)$.
Some  simulated  results for the case $\alpha=1$ and $\beta>1$ were given in the thesis of Lee~\cite{Lee21}. 
Table~\ref{tab:intro} shows various special cases of our setting.

\begin{table}
 \newcommand{\polyword}{\textcolor{red}{\sf polytime}}
 \newcommand{\approxword}{\textcolor{black}{\rm approximation}}
 \newcommand{\exactword}{\textcolor{red}{\sf exact}}
 \newcommand{\fptword}{\textcolor{black}{\rm fpt}}
 \newcommand{\expword}{\textcolor{black}{\rm exptime}}
 \newcommand{\forword}{{\Large $\rightsquigarrow$}}

 \centering
 \begin{tabular}{|m{1cm}|m{2.5cm}|m{9cm}m{0.5cm}|}
  \hline
  $c$ & $\alpha,\beta$ & &\\
  \hline
  $c=1$ & $\alpha>\beta>1$ & \polyword\ \approxword\ \forword\ \expword\ \approxword\newline with \textit{better} ratio &\ \newline\ \\
  \hline
  $c>1$ & $\alpha=\beta=1$ & \fptword\ \exactword\ \forword\ \expword\ \exactword\ \cite{FominGLS19}&\ \newline\ \\
  \hline
  $c>1$ & $\alpha=1, \beta>1$ & \fptword\ \exactword\ \forword\ \expword\  \approxword\  &\ \newline\ \\
  \hline
  $c>1$ & $\alpha>1 , \beta=1$ & \fptword\ \approxword\ \forword\ \expword\ \exactword\newline useless, cannot improve $2^n$ brute force&\ \newline\ \\
  \hline
  $c>1$ & $\alpha=\beta>1$ & \fptword\ \approxword\ \forword\ \expword\ \approxword\newline  with the \textit{same} ratio \cite{EsmerKMNS22}&\ \newline\ \\
  \hline
  $c>1$ & $1< \alpha<\beta$ & \fptword\ \approxword\ \forword\ \expword\ \approxword\newline  with \textit{worse} ratio &\ \newline\ \\
  \hline
  $c>1$ & $\alpha>\beta>1$ & \fptword\ \approxword\ \forword\ \expword\ \approxword\newline  with \textit{better} ratio&\ \newline\ \\
  \hline
 \end{tabular}
 \caption{Special cases of our setting.}
 \label{tab:intro}
\end{table}

These previous results suggest two obvious further research goals.
First, one would like to extend the understanding to the $\alpha\neq \beta$ case.
For example, Esmer et al.~\cite{EsmerKMNS22} showed how to obtain an exponential 5-approximation algorithm if we are given an $\OO^*(2^k)$ time 5-approximate parameterized extension algorithm (i.e., $c=2$, $\alpha=\beta=5$).
We would like to understand whether we can obtain a \emph{faster} 5-approximation algorithm if the extension algorithm is 3-approximate ($c=2$, $\alpha=3$, $\beta=5$) and whether the 5-approximate extension algorithm is useful \emph{at all} for obtaining an exponential 3-approximation ($c=2$, $\alpha=5$, $\beta=3$).

Second, the previous results \cite{FominGLS19,EsmerKMNS22} did not provide any lower bounds.
Is the $\OO^*((2-\frac{1}{c})^n)$ algorithm obtained by Fomin~et al.~\cite{FominGLS19} really the best we can have without any problem-specific knowledge?
We can formalize this question in a model where all we can do is checking the validity of a solution in polynomial time and using an $\alpha$-approximate extension algorithm running in time $\OO^*(c^k)$.
If we have lower bounds in this model, then we can evaluate whether the previous results \cite{FominGLS19,EsmerKMNS22} really repurposed the extension algorithms in an optimal way and we can compare how the algorithms resulting from two sets of parameters $(\alpha,c,\beta)$ and $(\alpha',c',\beta')$ relate to each other.

Our main result fully achieves both of these goals: for every combination of parameters, we provide tight upper and (unconditional) lower bounds on the best possible exponential-time approximation algorithm. 

\medskip
\begin{mdframed}[backgroundcolor=gray!20]
 For every fixed $\alpha,\beta,c\ge 1$, we determine the best possible $d = \bestbound(\alpha,c,\beta)$ such that a $\OO^*(d^n)$ time $\beta$-approximation algorithm can be obtained from an $\alpha$-approximate extension algorithm running in time $\OO^*(c^k)$.
\end{mdframed}
\medskip

Similar to \cite{EsmerKMNS22}, we do not expect a simple closed-form expression for $\bestbound(\alpha,c,\beta)$.
Indeed, it may very well be that $\bestbound(\alpha,c,\beta)$ has no closed-form description similar to, for example, the running time of certain branching algorithms where the base corresponds to the root of a polynomial of degree at least five (see, e.g., \cite{CyganFKLMPPS15}).
This raises the philosophical question of when can we consider the problem of determining $\bestbound(\alpha,c,\beta)$ ``resolved.'' Our answer consists of two parts:
\begin{enumerate}[label = (\arabic*)]
 \item\label{item:goal-1} For every $\alpha,\beta,c \geq 1$, a simple approximate mononotone local search algorithm (which naturally extends existing algorithms \cite{FominGLS19,Lee21,EsmerKMNS22}) achieves the optimal running time (up to polynomial factors; \Cref{thm:upper_bound_random}).
 \item\label{item:goal-2} This algorithm runs precisely in time $\OO^*((\bestbound(\alpha,c,\beta))^n)$ and given $\alpha, \beta, c \geq 1$ and $\eps>0$, we can compute $\bestbound(\alpha,c,\beta)$ up to an additive error of $\eps>0$, in time polynomial in the total encoding length of the input (Theorem~\ref{thm:compute}).
\end{enumerate}
 
That is, we describe the optimal algorithm and show how to analyze its running time.
Arguably, these two results satisfy any intuitive expectation of resolving the problem. 
The basic approximate monotone local search algorithm in inherently randomized, but it can be derandomized at the cost of a subexponential factor in the running time (Theorem~\ref{thm:upper_bound_deter}).
To attain statement \ref{item:goal-1}, we show that the running time of approximate monotone local search is optimal, up to polynomial factors, independently of the computation of the running time itself.
This lower bound proof uses a simple combinatorial argument that lower bounds the running time of \emph{any} repurposing algorithm in terms of the (unknown) running time of approximate monotone local search.
 
To reach statement \ref{item:goal-2}, we describe $\bestbound(\alpha,c,\beta)$ as the solution of a continuous, convex optimization problem, which allows us to evaluate $\bestbound(\alpha,c,\beta)$ up to any precision $\eps>0$ in time polynomial in the encoding length of $\alpha$, $c$, $\beta$ and $\eps$ using standard tools from convex optimization (see, e.g., \cite{GrotschelLS88}).

We show that $\bestbound(1,c,1) = 2 - \frac{1}{c}$ and, more generally, $\bestbound(\beta,c,\beta) = \esaamlsbound(\beta,c)$ which implies that previous algorithms \cite{FominGLS19,EsmerKMNS22} already exploited existing algorithms in an optimal way in their respective restricted setting.
These lower bounds are unconditional and do not rely on any complexity assumption such as the (Strong) Exponential-Time Hypothesis: the lower bounds are proved in a formal setting where our only problem-specific knowledge is being able to test the feasibility of a solution and invoke the $\alpha$-approximation extension algorithm.

To further appreciate the running time of our algorithm, we mathematically compare $\bestbound(\alpha,c,\beta)$ to existing benchmarks.

\paragraph*{Benchmark 1: Brute-Force for Exponential Approximation.}
\label{benchmark-brute}

A key feature of the bound $d = 2 - \frac{1}{c}$ obtained by Fomin et al.~\cite{FominGLS19} is that it is always strictly better than the brute-force search running in time $\OO^*(2^n)$.
We extend this result to the approximate setting.

If our goal is to find a $\beta$-approximation for some $\beta>1$, then the  $\OO^*(2^n)$ brute force search is certainly not optimal: for example, if $\beta =2$ it suffices to only iterate over subsets of size at most $\frac{1}{3} n $ and at least $\frac{2}{3}n$, which only takes $\OO^*\left(1.8899^n\right)$.   This approach can be further optimized.
Indeed, Esmer et al.\ \cite{EsmerKMNS22} showed that for every monotone subset minimization problem, the classic brute-force approach can be generalized to a \emph{$\beta$-approximation brute-force algorithm} running in time $\OO^*(\brute(\beta)^n)$,
where $\brute(\beta) \coloneqq 1 + \exp\left(-\beta \cdot  \HH\left(\frac{1}{\beta}\right) \right)$ and $\HH(\beta) \coloneqq -\beta \ln \beta - (1-\beta) \ln (1-\beta)$ denotes the entropy function.
Moreover, this running time is optimal if the family of the solution sets can only be accessed via a membership oracle.
Note that $\brute(1) = 2$, i.e., in the exact setting, this recovers the standard brute-force algorithm running in time $\OO^*(2^n)$.

We compare approximate monotone local search to the $\beta$-approximation brute-force algorithm for every choice of $\alpha,c \geq 1$.

\medskip
\begin{mdframed}[backgroundcolor=gray!20]
 For every fixed $\alpha,c\ge 1$ and $\beta>1$, we have $\bestbound(\alpha,c,\beta) < \brute(\beta)$: \\approximate monotone local search is strictly faster than what can be obtained by brute force.
\end{mdframed}
\medskip

In other words, our main finding is that repurposing an $\alpha$-approximation algorithm always leads to a $\beta$-approximation algorithm strictly better than brute force except in the degenerate case $\alpha>\beta=1$: an approximation algorithm cannot be used to obtain an exact algorithm better than the $\OO^*(2^n)$ brute force.
That is, somewhat counterintuitively, even a $1000$-approximation algorithm running in time $\OO^*(1000^k)$ is actually useful for obtaining an exponential-time $1.001$-approximation algorithm better than brute force. Intuitively, brute force corresponds to the limit $c\to \infty$ and indeed approximate local search converges to brute force as $c$ goes to $\infty$. This also implies that even if the parameterized extension algorithm  is exact (i.e., $\alpha=1$), a running time $\OO^*(2^{\omega(k)})$ is not sufficient to obtain a $\beta$-approximation algorithm running in time $\OO^*((\brute(\beta) - \eps)^n)$ for any fixed $\eps > 0$.

\paragraph*{Benchmark 2: AMLS with Equal Approximation Ratios.}
\label{benchmark-amls-equal}

The results of \cite{EsmerKMNS22} can also be used to derive an exponential-time $\beta$-approximation algorithm from a $\alpha$-approximate parameterized extension algorithm with running time $\OO^*(c^k)$, in case $\alpha\leq \beta$.
This is done by interpreting the $\alpha$-approximate parameterized extension algorithm as a $\beta$-approximate parameterized extension algorithm (which is correct as $\alpha\leq \beta$), therefore leading to an exponential-time $\beta$-approximation algorithm which runs in time $\OO^*(d^n)$, where $d=\esaamlsbound(\beta,c)$.
Since $\esaamlsbound(\beta,c)<\brute(\beta)$ for all $\beta >1$ and $c\geq 1$ (see \cite{EsmerKMNS22}), this approach leads to a better than brute-force $\beta$-approximation for a wide range of problems for which there is an exact (i.e., $\alpha = 1$) parameterized algorithm with running time~$\OO^*(c^k)$.

For example, the best known exact parameterized algorithm for \textsc{Odd Cycle Traversal} runs in time $\OO^*(2.3146^k)$ \cite{LokshtanovNRRS14}.
In particular, this algorithm is a parameterized $1.5$-approximation algorithm for \textsc{Odd Cycle Traversal}.
Thus, using the result of \cite{EsmerKMNS22} the algorithm can be used to derive an exponential time $1.5$-approximation algorithm for \textsc{Odd Cycle Traversal} which runs in time $\OO^*(d^n)$ where $d=\esaamlsbound(1.5, 2.3146) \approx 1.340<\brute(1.5)\approx1.3849$.
Intuitively, using the result of \cite{EsmerKMNS22} in such a setting appears suboptimal.
We confirm this intuition.

\medskip
\begin{mdframed}[backgroundcolor=gray!20]
 For every $\beta>\alpha\geq 1$ and every $c > 1$ it holds that $\bestbound(\alpha,c,\beta) < \esaamlsbound(\beta,c)$.
\end{mdframed}

\paragraph*{Using Multiple Parameterized Approximation Algorithms.}

So far, all algorithms we described only use a single parameterized extension algorithm as a subroutine.
However, since with our new results any $\alpha$-approximate extension algorithm can be used to obtain a $\beta$-approximation algorithm, a natural extension is to use multiple $\alpha$-approximate extension algorithms for different values of $\alpha$ and $c$ at the same time.
For example, \textsc{Feedback Vertex Set} can be solved exactly in time $\OO^*(2.7^k)$ \cite{LiN22} (i.e., $\alpha_1 = 1$ and $c_1 = 2.7$) and admits a polynomial-time $2$-approximation algorithm \cite{BafnaBF99} (i.e., $\alpha_2 = 2$ and $c_2 = 1$).
Instead of using only one of these subroutines to design an exponential approximation, it seems much more natural to allow an algorithm to rely on both subroutines together.

We extend all of our results to the setting where any finite number of parameterized extension subroutines may be used by a single approximation algorithm.
Maybe surprisingly, this allows us to obtain further improvements over using only a single extension algorithm as a subroutine.
That is, there are parameter settings where given two extension algorithms with $(\alpha_1,c_1)$ and $(\alpha_2,c_2)$, we can obtain a $\OO^*(d^n)$ time $\beta$-approximation algorithm with $d$ being \emph{strictly smaller} than both $\bestbound(\alpha_1,c_1,\beta)$ and $\bestbound(\alpha_2,c_2,\beta)$.
Unfortunately, we observe that, for many concrete problems, these improvements are small and often restricted to only a small range of approximation ratios.

\paragraph*{Applications.}

Our results can be used to obtain  exponential approximation algorithms for a wide range of problems.
For many of these problems, there is no direct previous work on exponential-time approximations, thus our results serve as a baseline for future works.
For problems, such as \textsc{Vertex Cover} or \textsc{Feedback Vertex Set}, for which there are existing works on exponential approximations, our algorithms attain better running times than the state of art for the majority of approximation ratios.

The most natural application are deletion problems to hereditary graph classes, where the input is a graph $G$ (which may be undirected or directed, and may contain labeled vertices), and we wish to delete the minimum number of vertices to ensure a certain hereditary property (i.e., the family of solution sets is closed under supersets).
For example, we obtain exponential $\beta$-approximation algorithms for \textsc{FVS}, \textsc{Tournament FVS}, \textsc{Subset FVS}, \textsc{$d$-Hitting Set} , \textsc{Interval Vertex Deletion}, \textsc{Proper Interval Vertex Deletion}, \textsc{Block Graph Vertex Deletion}, \textsc{Cluster Graph Vertex Deletion}, \textsc{Split Vertex Deletion}, \textsc{Edge Multicut on Trees}, \textsc{Subset DFVS}, \textsc{DOCT} and \textsc{Multicut}.

To demonstrate the wide applicability, let us briefly discuss three illustrative examples here (a more thorough discussion of the applications to our results can be found in Section \ref{sec:application}; also running times for all problems listed above and various approximation ratios $\beta$ are listed in Appendix \ref{sec:running_times}).

\begin{itemize}
 \item \textsc{Odd Cycle Transversal} has no constant-factor polynomial-time approximation under UGC \cite{Khot02}, but it can be solved exactly in time $\OO^*(2.3146^k)$ \cite{LokshtanovNRRS14}.
  We obtain an exponential $\beta$-approximation algorithm for every $\beta > 1$, which significantly improves upon brute force.
  For example, we obtain a $1.1$-approximation running in time $\OO^*(1.3689^n)$ while $\brute(1.1) \approx 1.7153$, and the $1.1$-approximation obtained via \hyperref[benchmark-amls-equal]{Benchmark 2} runs in time $\OO^*(1.4223^n)$.

 \item \textsc{Directed Feedback Vertex Set (DFVS)} has no constant-factor polynomial-time approximation under UGC \cite{GuruswamiL16}, and it is an open question to determine if \textsc{DFVS} can be solved exactly in time $\OO^*(c^k)$ for any constant (see, e.g., \cite{ChenLLOR08}).
  But \textsc{DFVS} has a parameterized $2$-approximation algorithm running in time $\OO^*(c^k)$ for some constant $c$ \cite{LokshtanovMRSZ21}.
  Hence, we obtain an exponential $\beta$-approximation algorithm for every $\beta > 1$ that is faster than the $\beta$-approximation brute-force algorithm.
  For $1 < \beta < 2$, our algorithm is the first non-trivial $\beta$-approximation algorithm, and for $\beta > 2$ our algorithm improves over the previous best algorithm from \cite{EsmerKMNS22}.
  For $\beta = 2$, our running time matches that of \cite{EsmerKMNS22}.

 \item \textsc{Partial Vertex Cover} has a polynomial-time $2$-approximation \cite{BshoutyB98}, but is known to be ${\sf W[1]}$-hard \cite{GuoNW07}.
  We obtain the first non-trivial exponential $\beta$-approximation algorithm for every $1 < \beta < 2$.
  For example, for $\beta = 1.1$, our algorithm runs in time $\OO^*(1.6588^n)$.
\end{itemize}

\section{Our Results}
\label{sec:computational_model}
In this section, we discuss our results in five parts.
We define the computational model in \Cref{sec:comp_model}, and present the optimal algorithm as well as our the main results in \Cref{sec:amls_results}.
In \Cref{sec:optimality_overview} we describe the arguments to show that approximate monotone local search is the optimal way of repurposing existing parameterized approximation algorithms.
After that, \Cref{sec:evaluate_running_time} deals with the computation of the running time, and \Cref{sec:comparison_overview} compares the running times to the benchmarks.

\subsection{Computation Model}
\label{sec:comp_model}

We state our results in an oracle-based computation model that properly reflects the setting described in the introduction. 
Let $U$ be a universe of elements (i.e., a finite set).
A \emph{set system} of $U$ is a family $\CF \subseteq 2^{U}$ of subsets of $U$.
We say the set system $\CF$ is \emph{monotone} if (i) $U\in \CF $ and (ii) for every $S\subseteq T \subseteq U$, if $S\in \CF$ then $T\in \CF$.
We consider minimization problems in which the objective is to find $S\in \CF$ that minimizes $|S|$.

In the computation model, the universe $U$ is given as part of the input to the algorithm.
The set system $\CF$, however, is not part of the input.
Instead, the algorithm can implicitly access $\CF$ using \emph{extension oracles}.

\begin{definition}
	Let $U$ be a finite universe, $\CF$ be a set system of $U$ and $\ell \in \NN$.
	We say that $S \subseteq U$ is an $\ell$-extension of $X \subseteq U$ if $X \cup S \in \CF$ and $|S| \leq \ell$.
\end{definition}

Informally, a \emph{random $\alpha$-extension} oracle of a universe $U$ and a monotone set system $\CF$ gets $X\subseteq U$ and $\ell \in \NN$ as an input, and returns a set $Y\subseteq U$ such such $X\cup Y \in \CF$ and $Y$ satisfies the following property with probability at least $\frac{1}{2}$:
\begin{itemize}
	\item[] If there exists an $\ell$-extension of $X$ then $Y$ is an $(\alpha\cdot \ell)$-extension of $X$.
\end{itemize}

Though intuitive, this definition does not properly define what kind of an object an oracle is, and considers an undefined probability space.
These details will be important when proving lower bounds.
We provide a formal definition of an extension oracle as a function that, in addition to $X$ and $\ell$, also receives a bit-string $r$ as part of its argument.
The bit string serves as the source of randomness for the oracle, and we assume the algorithm provides a random bit-string alongside each query.

\begin{definition}[Random Extension Oracle]
	\label{def:random_extension_oracle}
	Let $U$ be a set and $\CF$ be a monotone set system of $U$.
	A \emph{random $\alpha$-extension oracle} for $U$ and $\CF$ is a function
	$\oracle\colon2^U \times \mathbb{N} \times \{0,1\}^m \to 2^U$ where $m \in \mathbb{N}$ that satisfies the following properties:
	\begin{enumerate}
		\item $\oracle(X,\ell,r)\cup  X \in \CF$ for every $(X,\ell, r)\in 2^U\times \mathbb{N}\times \{0,1\}^m$, and
		\item for every $(X,\ell) \in 2^{U} \times \NN$ such that $X$ has an $\ell$-extension it holds that
			\[\left|\Big\{r \in \{0,1\}^m \Bigmid \oracle(X,\ell,r) \text{ is an $\alpha \cdot \ell$-extension of } X \Big\}\right| \geq \frac{1}{2} \cdot \abs{\{0,1\}^m}.\]
	\end{enumerate}
	If $m=0$, then we say the oracle is deterministic.
\end{definition}

An algorithm may have access to several extension oracles.
We associate a \emph{cost} $c \geq 1$ with each oracle, representing the cost incurred by quering the oracle. 
The \emph{cost} of the oracle query $(X,\ell)$ is $c^{\ell}$.\footnote{We commonly omit the third argument to the oracle.}
Observe that invocations to an extension oracle with $\ell=0$ are equivalent to membership queries, and hence extension oracles can be viewed as generalizations of membership oracles.
The cost of a query $(X,\ell)$ represents the running time $\OO^*(c^\ell)$ of a parameterized algorithm which emulates the oracle in our applications.

An \emph{(oracle) specification list} is a non-empty and finite set $\CL=\{(\alpha_1, c_1),\ldots, (\alpha_{s}, c_{s})\}$ such that  $\alpha_j, c_j \geq 1$ for every $j\in [s]$.
We define a minimization problem for every oracle specification $\CL$.
An instance of  the \emph{$\CL$-subset minimization problem}  ($\LSUB$)
consists of a set $U$ and a monotone set system $\CF$ of $U$.
The objective is to find a set $S\in \CF$ such that $|S|$ is minimized.
In our computational model the set $U$ is given to the algorithm as part of the input.
Furthermore, the algorithm has access to an $\alpha$-extension oracle for $U$ and $\CF$, associated with cost $c$, for every $(\alpha,c)\in \CL$ (that is, the algorithm is given $|\CL|$ oracles for the same set system $\CF$). In particular, $\CF$ is part of the \emph{instance}, but is not part of the \emph{input}.

Let $\A$ be an algorithm for $\LSUB$.
The \emph{cost of an execution} of $\A$ is the sum of costs over all oracle queries initiated by the algorithm, plus the number of computational operations conducted throughout the execution.  That is, if $Q_{\alpha,c}\subseteq 2^U\times \NN$ is the set of queries the algorithm makes to the $\alpha$-extension oracle for every $(\alpha,c)\in \CL$ in a specific execution and $p$ is the number of computational operations, then the cost of the execution is $p+\sum_{(\alpha, c)\in \CL} \sum_{(X,\ell)\in Q_{\alpha,c}} c^{\ell}$. We define $\cost_{\A}(n)$ to be the maximal cost of an execution of $\A$ given an input which satisfies $\abs{U}\leq n$.
We say  $\A$ is of cost $f:\NN\rightarrow \NN$ if  $\cost_{\A}(n)\leq f(n)$ for all $n\in \NN$.

Following the standard notion of approximation algorithms, we say an algorithm $\A$ is a \emph{(randomized) $\beta$-approximation} for $\LSUB$ if for every universe $U$, monotone set system $\CF$ over $U$,  and
randomized extension oracles  $\oracle_{\alpha,c}$ for every $(\alpha,c)\in \CL$, the algorithm always returns $S\in \CF$ (with probability $1$) and it holds that $|S|\leq \beta\cdot\min_{T\in \CF} |T|$ with probability at least $\frac{1}{2}$.

We also consider deterministic algorithms for $\LSUB$.
In this case we restrict our attention to inputs with deterministic oracles.
Formally, we say an algorithm $\A$ is a \emph{deterministic $\beta$-approximation} for $\LSUB$ if for every universe $U$, monotone set system $\CF$ over $U$, and \emph{deterministic} extension oracles $\oracle_{\alpha,c}$ for every $(\alpha,c)\in \CL$, the algorithms returns $S\in \CF$ such that $|S|\leq \beta\cdot\min_{T\in \CF} |T|$.

For every specification list $\CL$ and $\beta\geq 1$, we define $\bestbound(\CL,\beta)$ to be the base of the best cost $\beta$-approximation algorithm for $\LSUB$. Formally,
\begin{equation}
	\label{eq:amls_star}
	\bestbound(\CL,\beta) = \inf\left\{ d \geq 1 ~\middle|~ \textnormal{there is a  $\beta$-approximation for $\LSUB$ with cost $d^n \cdot n^{\OO(1)}$}\right\}.
\end{equation}
The paper revolves around the value of $\bestbound(\CL,\beta)$.
Our primary objectives are to attain an algorithm with cost $\OO^*\left( \bestbound(\left(\CL,\beta\right)^n\right)$, derive a method to compute $\bestbound(\CL,\beta)$, and analytically compare it to the benchmarks.

\subsection{Approximate Monotone Local Search}
\label{sec:amls_results}

Our first main result is that a simple monotone local search algorithm, $\amlsalgo_{\CL,\beta}$ (see \Cref{algo:final}) is a $\beta$-approximation algorithm for $\LSUB$ with \emph{optimal} cost of $\OO^*\left( (\bestbound(\CL,\beta))^n\right)$.
The algorithm is a natural generalization of the monotone local search algorithms used in \cite{FominGLS19,EsmerKMNS22}. 

The algorithm is based on a simple sampling procedure (\Cref{algo:intermediary}).
Let $\OPT = \argmin_{S\in \CF} \abs{S}$ and assume $k=\abs{\OPT}$.
The sampling procedure sample a set $X\subseteq U$ of size $t$ uniformly at random, and then extends the set to a solution $Z=X\cup Y$, where $Y$ is attained via a query of the form $(X,\ell)$ to the $\alpha$-extension oracle $\oracle_{\alpha,c}$.
To keep the intuitive description simple we assume the oracle is deterministic.

The sampling produces a solution of size $\beta \cdot k$ assuming  $\abs{X\cap\OPT}\geq x$, for a carefully selected value $x$. Subject to this assumption, the set $\OPT\setminus X$ is a $(k-x)$-extension of $X$, thus, according to \Cref{def:random_extension_oracle}, $Y=\oracle_{\alpha,c}(X,k-x)$ is an $\alpha(k-x)$-extension of $X$.
It is therefore guaranteed that $X\cup Y\in \CF$ and $\abs{X\cup Y} \leq t+ \alpha(k-x)$.
As our objective is to find a set in $\CF$ of cardinality at most $\beta\cdot k$, the value of $x$ needs to satisfy $t+\alpha(k-x)\leq \beta k$.
Indeed, we set $x = \x(k,t) \coloneqq  \left( 1-\frac{\beta}{\alpha}\right) k +\frac{t}{\alpha}$, which is the minimal value which satisfies $t+\alpha(k-x)\leq \beta k$.

The distribution of $\abs{X\cap \OPT}$ is commonly referred as \emph{hyper-geometric}.
Define $\p(n,k,t,x)$ to be the probability that a uniformly random set $X$ of $t$ items out of $[n] \coloneqq \{1,\ldots,n\}$ satisfies $|X \cap [k]| \geq x$.
Note that
\begin{equation}
	\label{eq:hyper}
	\p(n,k,t,x) = \sum_{y = \lceil x\rceil}^{\min\{t,k\}} \frac{\binom{k}{y}\binom{n-k}{t-y}}{\binom{n}{t}} =\sum_{y = \lceil x\rceil}^{\min\{t,k\}} \frac{\binom{t}{y}\binom{n-t}{k-y}}{\binom{n}{k}} .
\end{equation}
It follows that the sampling procedure returns a solution of cardinality $\beta k$ or less with probability (at least) $\p\left(n,k,t, \x(k,t) \right)$.
Furthermore, the cost of the procedure is $\OO\left(c^{k-\ceil{\x(k,t)}} \right) = \OO\left(c^{\frac{\beta k - t}{\alpha}} \right)$.
Thus, to obtain a constant success probability the sampling procedure has to be executed $\approx \left(\p\left(n,k,t, \x(k,t) \right)\right)^{-1}$ times, leading to a total cost of
\begin{equation}
	\label{eq:sampling_cost}
	\OO\left(\frac{c^{\frac{\beta \cdot k - t}{\alpha}}}{\p\left(n,k,t,\x(k,t) \right)}\right).
\end{equation}

The sampling procedure is used by \Cref{algo:final}.
This algorithm iterates over all possible values of $k=\abs{\OPT}$.
For each value of $k$ the algorithm selects an oracle $(\alpha,c)\in \CL$ and $t\in \NN$ which minimizes~\eqref{eq:sampling_cost} (Line \ref{amls:select_t} of \Cref{algo:final}) and then invokes the sampling procedure sufficiently many times to attain success probability of $\frac{1}{2}$.
The range of values $t$ can take is restricted to $\left[ M^*_{\alpha,\beta}\cdot k,~\beta \cdot k\right]\cap \NN$, where
\begin{equation}
	\label{eq:Mstar_def}
	M^*_{\alpha, \beta} \coloneqq
	\begin{cases}
		0 &\text{if } \alpha \leq \beta\\
		\frac{\alpha - \beta}{\alpha - 1} &\text{if } \alpha > \beta
	\end{cases}
\end{equation}
for all $\alpha,\beta \geq 1$.
This restriction ensures the algorithm only considers values of $t$ for which  $\x(k,t)\leq t$.

\begin{algorithm}
	\begin{algorithmic}[1]
		\Input A universe $U$ , $k \in \mathbb{N}$, $t \in \mathbb{N}$, $\alpha,\beta \geq 1$ and an $\alpha$-extension oracle $\oracle_{\alpha,c}$.
		\State Sample a set $X$ of size $t$ from $U$ uniformly at random. \label{int:select}
		\State  $Y \gets \oracle_{\alpha,c}\left( X, k - \left\lceil\left( 1 - \frac{\beta}{\alpha} \right)\cdot k + \frac{t}{\alpha}\right\rceil \right) $. \label{int:Aext}
		\State Return $Z \gets X \cup Y$. \label{int:T}
	\end{algorithmic}
	\caption{$\sample(U,k,t,\alpha,\beta, \oracle_{\alpha,c})$}
	\label{algo:intermediary}
\end{algorithm}

\begin{algorithm}
	\begin{algorithmic}[1]
		\Input A universe $U$ and an extension oracle $\oracle_{\alpha,c}$ for every $(\alpha,c)\in \CL$
		\State $\sol\leftarrow\emptyset$, $n\leftarrow \abs{U}$.
		\For{\label{amls:loop}$k$ from $0$ to $\frac{n}{\beta}$}
		\State 
		Find $(\alpha,c)\in \CL$ and  $t\in\left[M^*_{\alpha, \beta}\cdot k,\beta \cdot k\right]\cap \NN$  which minimize 
		$\left(\frac{c^{\frac{\beta k - t}{\alpha}}}{\p\left(n,k,t,\left( 1-\frac{\beta}{\alpha}\right)\cdot k+\frac{t}{\alpha}\right)}\right)$. \label{amls:select_t}
		\State Run $\sol \leftarrow \sol \cup \left\{ \sample(U,k,t,\alpha,\beta, \oracle_{\alpha,c})\right\}$ for $2\cdot \left\lceil\left(\p\left( n, k, t, (1 - \frac{\beta}{\alpha} )\cdot k + \frac{t}{\alpha}\right)\right)^{-1}\right\rceil$ times. \label{amls:call_sample}
		\EndFor
		\State {\bf Return} a minimum-sized set in $\sol$.
	\end{algorithmic}
	\caption{$\amlsalgo_{\CL,\beta}$}
	\label{algo:final}
\end{algorithm}

Our main theorem asserts that $\amlsalgo_{\CL,\beta}$ has the best possible cost of a $\beta$-approximation for $\LSUB$.

\begin{theorem}[\bf Main result: randomized algorithm]
	\label{thm:upper_bound_random}
	For every specification list $\CL$ and $\beta\geq 1$, $\amlsalgo_{\CL,\beta}$ is a randomized $\beta$-approximation for $\LSUB$ of cost $n^{\OO(1)}\cdot \left(\bestbound(\CL,\beta)\right)^n$.
\end{theorem}

Similar to \cite{FominGLS19,EsmerKMNS22}, it is possible to derandomize Algorithm \ref{algo:final} with sub-exponential overhead in the running time.
The derandomized version of the algorithm, $\detamlsalgo_{\CL,\beta}$ (\Cref{algo:derandomized}), is given in \Cref{sec:analysis}.

\begin{theorem}[\bf Main result: deterministic algorithm]
	\label{thm:upper_bound_deter}
	For every specification list $\CL$ and $\beta\geq 1$, $\detamlsalgo_{\CL,\beta}$ (\Cref{algo:derandomized}) is a  $\beta$-approximation for $\LSUB$ with cost at most $\left(\bestbound(\CL,\beta)\right)^{n} \cdot 2^{o(n)}$.
\end{theorem}

Though \Cref{thm:upper_bound_random} states the cost of \Cref{algo:final} is the best possible, it does not provide any method by which this cost can be computed.
The next theorem addresses this issue.

\begin{theorem}[\bf Main result: computing $\bestbound$]
	\label{thm:compute}
	There is an algorithm  which given $\beta\geq 1$, a specification list $\CL$ and $\eps>0$ computes $\bestbound(\CL,\beta)$ up to additive precision of $\eps$, and runs in polynomial time in the encoding length of $\CL$,~$\beta$ and $\eps$.
\end{theorem}

\subsection{Optimality of the Algorithm}
\label{sec:optimality_overview}

A naive calculation reveals that the cost of \Cref{algo:final} can be bounded by the function $f_{\CL,\beta}$ defined by
\begin{equation}
	\label{eq:fdef_intro}
	f_{\CL,\beta}(n) \coloneqq \max_{~k\in \left[0,\frac{n}{\beta}\right] \cap \NN~}
	\min_{~(\alpha,c)\in \CL~}
	\min_{~t\in\left[M^*_{\alpha, \beta}\cdot k,\beta\cdot  k\right] \cap \NN~}
	\frac{\exp\left( \frac{\beta k -t} {\alpha }\cdot \ln c \right)} {\p\left( n,k, t, (1 - \frac{\beta}{\alpha}) \cdot k +\frac{t}{\alpha} \right)}.
\end{equation}

\begin{lemma}
	\label{lem:amls_upper_bound_via_f}
	For every $\beta \geq 1$ and specification list $\CL$, it holds that $\amlsalgo_{\CL,\beta}$ is a $\beta$-approximation algorithm for $\LSUB$ with cost at most $n^{\OO(1)} \cdot f_{\CL,\beta}(n)$.
\end{lemma}

The proof of \Cref{lem:amls_upper_bound_via_f} is given in \Cref{sec:analysis}.
The same section also proves a variant of \Cref{lem:amls_upper_bound_via_f} which refers to $\detamlsalgo_{\CL,\beta}$ (\Cref{algo:derandomized}).

\begin{lemma}
	\label{lem:deter_upper_bound_via_f}
	For every $\beta \geq 1$ and specification list $\CL$,  \Cref{algo:derandomized}  is a deterministic $\beta$-approximation for $\LSUB$ with cost at most $f_{\CL,\beta}(n) \cdot 2^{o(n)}$.
\end{lemma}

Maybe surprisingly, one of the main insights in this paper is that these simple algorithms are actually optimal in the oracle model defined above, i.e.,  we can also use $f_{\CL,\beta}$ as a lower bound on the cost of any algorithm for $\LSUB$.

\begin{lemma}
	\label{thm:minimization_lower_bound}
	For any $\beta\geq 1$ and specification list $\CL$, every $\beta$-approximation algorithm for $\LSUB$ has cost at least $n^{-\OO(1)} \cdot f_{\CL,\beta}(n)$.
\end{lemma}

The proof of \Cref{thm:minimization_lower_bound}, given in \Cref{sec:lower_bounds}, follows from the inability of an algorithm for $\LSUB$ to distinguish between instances in which $\CF$ contains all sets of size at least $\beta\cdot k+1$, versus instances in which $\CF$ contains a set $R$ of cardinality $k$, its supersets, and all sets of size at least $\beta\cdot k+1$.
Returning a valid solution for the later requires the algorithm to initiate an oracle query of the form $(X,\ell)$ to an $\alpha$-extension oracle such that $\abs{X}+\alpha\ell\leq \beta k$ and $X$ has an $\ell$-extension.
As we select $R$ to be a random set, we can use this property to lower bound the total cost of the queries the algorithm must initiate in order find an $\ell$-extension with a constant probability.
The value of $k$ used in the construction is the value which attains the maximum in \eqref{eq:fdef_intro}.
We note the oracles used in the proof of \Cref{thm:minimization_lower_bound} are deterministic.
Hence, the lower bound holds even if the algorithm is guaranteed the oracles are deterministic.
Together \Cref{lem:amls_upper_bound_via_f,thm:minimization_lower_bound} indicate that $\amlsalgo_{\CL,\beta}$ attains the best possible cost of a $\beta$-approximation algorithm for $\LSUB$, up to polynomial factors.
Similarly, \Cref{lem:deter_upper_bound_via_f,thm:minimization_lower_bound} imply that \Cref{algo:derandomized} is  optimal up to sub-exponential factors.

It follows from \Cref{lem:amls_upper_bound_via_f,thm:minimization_lower_bound} that
\begin{equation}
	\label{eq:best_to_f}
	\bestbound(\CL,\beta) = \lim_{n\rightarrow \infty} \left(f_{\CL,\beta}(n)\right)^{\frac{1}{n}}
\end{equation}
for all specification lists $\CL$ and $\beta\geq 1$.
We note that the above limit does not imply \Cref{thm:upper_bound_random}, though it can be used to establish as slightly weaker claim.
Hypothetically, it is possible that $f_{\CL,\beta}(n) = 2^{n +\sqrt{n}}$ and thus, $\bestbound(\CL,\beta)=2$ but $f_{\CL,\beta}(n)\neq \OO^*(2^n)$.  We will later rule out the existence of such cases.

The proof of \Cref{thm:minimization_lower_bound} can also be adapted to the exact setting of~\cite{FominGKLS10}.
Given a set $U$ and a subset family $\CF$ of $U$ (not necessarily monotone), an \emph{exact extension oracle} for $\CF$  takes as an input a set $X\subseteq U$ and $\ell \in \NN$. The oracle either returns $\yes$ or $\no$.
If $X$ has an $\ell$-extension then the oracle returns $\yes$ with probability at least $\frac{1}{2}$.
If $X$ does not have an $\ell$-extension then the oracle returns $\no$.
Similarly to the approximate case, we associate a number $c\geq 1$ with the oracle.
The \emph{cost} of an oracle query is $c^{\ell}$.

In the \emph{$c$-decision problem} ($\cDEC$) the input is a universe $U$ and an exact extension oracle for a set family $\CF$ of $U$. The objective is to determine if $\CF\neq \emptyset$ (in particular, the set system does not have to be monotone). The execution cost of an algorithm for $\cDEC$ is the sum of costs of all oracle queries plus the number of computational operations, where the cost of a query $(X,\ell)$ is $c^{\ell}$.  Similarly to $\LSUB$, we say that an algorithm $\A$ for $\cDEC$ is of cost $f:\NN\to \NN$ if every execution of $\A$ with an input for which $\abs{U}\leq n$ has cost at most $f(n)$.
In \cite{FominGLS19} it was shown that there is a randomized algorithm for $c$-DEC of cost $n^{\OO(1)}\cdot \left(2-\frac{1}{c}\right)^n$.

Using the same ideas as in the proof of \Cref{thm:minimization_lower_bound} we can show the following.

\begin{lemma}
	\label{lem:exact_lower_bound_via_f}
	For every $c>1$, every randomized algorithm for $\cDEC$ has cost of at least $n^{-\OO(1)}\cdot f_{\{(1,c)\},1}(n)$.
\end{lemma}

As we can also show that  $f_{\{(1,c)\},1}(n)\geq n^{-\OO(1)}\cdot \left(2-\frac{1}{c}\right)^n$, we obtain the following theorem.

\begin{theorem}
	\label{lem:exact_lower_bound}
	For all $c>1$, every randomized algorithm for $\cDEC$ has cost of at least $n^{-\OO(1)}\cdot \left( 2-\frac{1}{c}\right)^n$.
\end{theorem}

In particular, \Cref{lem:exact_lower_bound} indicates the result of \cite{FominGLS19} cannot be improved. The proofs of \Cref{lem:exact_lower_bound_via_f}  and \Cref{lem:exact_lower_bound} are given in \Cref{sec:lower_bounds}.

\subsection{Evaluating the Running Time of Approximate Monotone Local Search}
\label{sec:evaluate_running_time}

So far, we showed that $\amlsalgo_{\CL,\beta}$ attains the best possible cost of a $\beta$-approximation for $\LSUB$, up to polynomial factors.
However, the tools presented so far do not provide a method for evaluating the running time of the algorithm, and do not suffice to show \Cref{thm:upper_bound_random,thm:upper_bound_deter}.

The proof of \Cref{thm:compute}, which shows $\bestbound(\CL,\beta)$ can be computed efficiently,  consists of two main stages.
The first stage shows that $f_{\CL,\beta}\approx d^n$ (ignoring polynomial factors), where $d$ is a solution for a continuous max-min optimization problem.
As a by product, the stage provides the missing ingredient towards the proofs of \Cref{thm:upper_bound_random,thm:upper_bound_deter}.
The second stage shows the minimization part of the optimization problem is a minimization of a convex function, and the maximization part is a maximization of a concave function.
Hence, both parts of the optimization problem can be easily solved using known tools from convex optimization (see, e.g., \cite{GrotschelLS88}).

Using the standard $\binom{n}{k}\approx \exp\left(n\cdot \entropy\left( \frac{k}{n}\right) \right)$ estimation for binomial coefficients and basic analysis of the hyper-geometric distribution $\p$,
the discrete optimization problem defined in \eqref{eq:fdef_intro} can be converted to a continuous optimization problem. 	For every $\alpha,\beta,c\geq 1$ we define the following functions.
\begin{alignat}{2}
	&\delta_{\alpha,\beta}(\kappa, \tau) &&=~
	\begin{cases}
		\frac{\frac{\beta}{\alpha} \kappa -\frac{\tau}{\alpha}}{1-\tau} = \frac{\frac{\beta}{\alpha} \kappa -\frac{1}{\alpha}}{1-\tau}+\frac{1}{\alpha} &\text{if } \tau \neq 1\\
		\frac{1}{\alpha} &\text{if } \tau = 1
	\end{cases}
	\label{eq:delta_def}
	\\
	\label{eq:gamma_def}
	&\gamma_{\alpha,\beta}(\kappa, \tau) &&=~
	\begin{cases}
		\left(1-\frac{\beta}{\alpha}\right)\frac{\kappa}{\tau} +\frac{1}{\alpha} &\text{if } \tau \neq 0\\
		\frac{1}{\alpha}  &\text{if } \tau=0
	\end{cases}
	\\
	\label{eq:g_def}
	&g_{\alpha,\beta,c}(\kappa,\tau) &&= ~\frac{\beta \kappa - \tau}{\alpha} \ln c - \tau\cdot \entropy\left(\gamma_{\alpha,\beta}(\kappa,\tau)\right) -(1-\tau)\cdot \entropy\left(\delta_{\alpha,\beta} (\kappa,\tau)\right) + \entropy\left(\kappa\right)
	\\
	\label{eq:Mdef}
	&M_{\alpha, \beta}(\kappa) &&=~
	\begin{cases}
		\frac{\beta - \alpha}{1-\alpha \cdot \kappa} \cdot \kappa &\text{if } \alpha < \beta\\
		0 &\text{if } \alpha =\beta\\
		\frac{\alpha - \beta}{\alpha - 1}\cdot  \kappa &\text{if } \alpha > \beta
	\end{cases}
\end{alignat}

Observe that $M_{\alpha,\beta}(\kappa) = M^*_{\alpha,\beta}\cdot \kappa$ ($M^*_{\alpha,\beta}$ is defined in \eqref{eq:Mstar_def}) if $\alpha\geq\beta$, but $M_{\alpha,\beta}(\kappa) \neq M^*_{\alpha,\beta}\cdot \kappa$ if $\alpha<\beta$.
We follow the standard notation in which $0 \ln 0 = 0$ and $\entropy(0) = \entropy(1)=0$.
For every $\beta \geq 1$ and specification list $\CL$ we define
\begin{equation}
	\label{eq:amls_def}
	\amlsbound(\CL,\beta) = \exp\left(  \max_{~0\leq \kappa \leq \frac{1}{\beta} ~} ~\min_{(\alpha,c)\in \CL} ~\min_{~ M_{\alpha,\beta}(\kappa)  \leq  \tau \leq \beta \kappa~}~ g_{\alpha,\beta,c}(\kappa,\tau)\right).
\end{equation}
If $\CL = \{(\alpha,c)\}$, we also write $\amlsbound(\alpha,c,\beta)$ instead of $\amlsbound(\{(\alpha,c)\},\beta)$.

\begin{restatable}{lemma}{ftoamls}
	\label{lem:f_to_amls}
	For every $\beta \geq 1$ and specification list $\CL$ it holds that
	\[n^{-\OO(1)}\cdot \left(\amlsbound(\CL,\beta )\right)^n~\leq~ f_{\CL,\beta}(n)~\leq~ n^{\OO(1)}\cdot \left(\amlsbound(\CL,\beta)\right)^n.\]
\end{restatable}

Note that the constants represented by $\OO(1)$ in \Cref{lem:f_to_amls} may depend on $\CL$ and $\beta$.
The following corollary is an immediate consequence of \Cref{lem:f_to_amls} and \eqref{eq:best_to_f}.

\begin{corollary}
	\label{cor:amls_is_the_best}
	 For every specification list  $\CL$  and $\beta\geq 1$ it holds that $	\amlsbound(\CL,\beta) = \bestbound(\CL,\beta)$.
\end{corollary} 

\Cref{thm:upper_bound_random} (\Cref{thm:upper_bound_deter}) follows from \Cref{lem:amls_upper_bound_via_f} (\Cref{lem:deter_upper_bound_via_f}),  \Cref{lem:f_to_amls} and \Cref{cor:amls_is_the_best}.

Now, our next challenge is to show that $\amlsbound$ can be computed.
Our first observation towards this goal is that $g_{\alpha,\beta,c}(\kappa,\tau)$ is \emph{convex} as a function of $\tau$ for every fixed $\kappa$.

\begin{restatable}{lemma}{gconvexbytau}
	\label{lem:g_convex_by_tau}
	Let $\alpha,c\geq 1$, $\beta > 1$ and $0 < \kappa < \frac{1}{\beta}$. The function $h(\tau) =g_{\alpha,\beta,c} (\kappa,\tau)$ is convex in the domain $\left[M_{\alpha,\beta}(\kappa), \beta \cdot \kappa \right]$ and $\min_{\tau \in\left[M_{\alpha,\beta}(\kappa), \beta \cdot \kappa \right] }g_{\alpha,\beta,c}(\kappa,\tau)<h(\beta\kappa)=g_{\alpha,\beta,c} (\kappa,\beta \kappa)$.
\end{restatable}

The lemma follows from a standard calculus argument. 
In fact, we are able to show a slightly stronger claim, which states that, up to some corner cases, the minimum of $h(\tau)$ (as defined in \Cref{lem:g_convex_by_tau}) in the interval $\left[M_{\alpha,\beta}(\kappa), \beta \cdot \kappa \right]$ is an interior point (that is, not $M_{\alpha,\beta}(\kappa)$ or $\beta \kappa$).
One of the corner cases occurs when $\alpha=\beta$, in which the minimum may be at $\tau=0=M_{\alpha,\alpha}(\kappa)$.
This distinction provides some evidence that the analysis inevitably has to differ from the approaches taken in \cite{FominGLS19,EsmerKMNS22} which deal with the special case of $\alpha=\beta$. 

For any $\alpha,\beta,c\geq 1$ and $0\leq \kappa\leq \frac{1}{\beta}$ we define
\begin{equation}
	\label{eq:gstar_def}
	g^*_{\alpha,\beta,c}(\kappa) = ~\min_{~ M_{\alpha,\beta}(\kappa)  \leq  \tau \leq \beta \kappa~}~ g_{\alpha,\beta,c}(\kappa,\tau).
\end{equation}
Therefore,
\begin{equation}
	\label{eq:amls_via_gstar}
	\amlsbound(\CL,\beta)= \exp\left(  \max_{~0\leq \kappa \leq \frac{1}{\beta} ~} ~\min_{(\alpha,c)\in \CL} ~g^*_{\alpha,\beta,c}(\kappa)\right).
\end{equation}
By \Cref{lem:g_convex_by_tau}, $g^*_{\alpha,\beta,c}(\kappa)$ is the solution for a convex minimization over a closed interval. Observe that  $g_{\alpha,\beta,c}(\kappa,\tau)$ can be evaluated up to additive precision of $\eps$ in polynomial time in the encoding length of   $\alpha$, $\beta$, $\kappa$, $\tau$, $\kappa$ and  $\eps$. Thus, using standard  convex optimization tools (e.g., \cite[Theorem 4.3.13]{GrotschelLS88}) we attain the following result.

\begin{corollary}
	\label{lem:compute_gstar}
	There exists an algorithm which, given $\alpha,\beta,c\geq 1$, $0\leq \kappa\leq \frac{1}{\beta}$ and $\eps>0$, computes  $g^*_{\alpha,\beta,c}(\kappa)$ up to an additive precision of $\eps>0$, and runs in polynomial time in the encoding length of $\alpha$, $\beta$, $c$, $\kappa$ and $\eps$.
\end{corollary}

The main insight behind the proof of \Cref{thm:compute} is the following.

\begin{restatable}{lemma}{concave}
	\label{lem:concave}
	For all $\alpha,c\geq 1$ and $\beta > 1$, it holds that $g^*_{\alpha,\beta,c}(\kappa)$ is concave in the interval $\left[0,\frac{1}{\beta} \right]$.	
\end{restatable}

To prove \Cref{lem:concave} we show that if $(\kappa,\tau)$ is a critical point of $g_{\alpha,\beta,c}$ then the determinant of the Hessian matrix  of $g_{\alpha,\beta,c}$ at $(\kappa,\tau)$ is negative. Once this argument is established, the lemma follows quite easily. Since the minimum of concave functions is also a concave function, \Cref{lem:concave} immediately implies the following.

\begin{corollary}
	\label{cor:min_concave}
	For every specification list $\CL$ and $\beta \geq 1$, the function $h(\kappa)=\min_{ (\alpha,c)\in \CL} g^*_{\alpha,\beta,c}(\kappa)$ is concave on the interval $\left[0,\frac{1}{\beta}\right]$.
\end{corollary}

By \Cref{cor:min_concave} it follows that $\amlsbound(\CL,\beta)$ \eqref{eq:amls_via_gstar} is  the maximum of a concave function on a closed interval. Furthermore, by \Cref{lem:compute_gstar} it holds that the function  being maximized can be computed, up to an additive error of $\eps$, in polynomial time. Thus, using convex optimization once more (e.g.,  \cite[Theorem 4.3.13]{GrotschelLS88}), we get the following.

\begin{corollary}
	\label{cor:compute_amls}
	There is an algorithm which, given a specification list $\CL$, $\beta\geq 1$ and $\eps>0$, computes $\amlsbound(\CL,\beta)$ up to an additive error of $\eps$ in time polynomial in the encoding length of $\CL$, $\beta$ and $\eps$.
\end{corollary}

\Cref{thm:compute} immediately follows from \Cref{lem:f_to_amls} and \Cref{cor:compute_amls}.
In particular, to prove \Cref{thm:compute} we are left to provide proofs for \Cref{lem:f_to_amls,lem:g_convex_by_tau,lem:concave}.
The proof of \Cref{lem:f_to_amls} is given in \Cref{sec:f_to_amls}, and the proofs of \Cref{lem:g_convex_by_tau,lem:concave} are given in \Cref{sec:math_functions}.
We note that  our computations of specific values of $\bestbound(\CL,\beta)$ do not implement the theoretical algorithm from \cite{GrotschelLS88}.
Instead, we use a combination of Golden Section Search \cite[Section 10.2]{PressTVF07} and a simple binary search which finds the root of the derivative.
We use \texttt{mpmath} \cite{mpmath} for high-precision arithmetics.

\subsection{Comparisons}
\label{sec:comparison_overview}

Since $\bestbound(\CL,\beta)= \amlsbound(\CL,\beta)$, 
we can use the definition of $\amlsbound(\CL,\beta)$ as the optimum of an optimization problem \eqref{eq:amls_def} to compare its value to Benchmarks \hyperref[benchmark-brute]{1} and \hyperref[benchmark-amls-equal]{2}.
We first compare $\bestbound$ to $\brute$.

\begin{theorem}
	\label{thm:amls_smaller_brute}
	For every $\beta>1$ and specification list $\CL$ it holds that $\bestbound(\CL,\beta)<\brute(\beta)$.
	Moreover, $\lim_{c \to \infty} \bestbound(\alpha,c,\beta) = \brute(\beta)$ for every $\alpha\geq 1$ and $\beta > 1$. 
\end{theorem}

The proof of \Cref{thm:amls_smaller_brute} is given in \Cref{sec:better_than_brute}.

In \cite{EsmerKMNS22} the authors showed that $\amls_{\{(\beta,c)\},\beta}$ is a $\beta$-approximation algorithm for $\LSUB[\{(\beta,c)\}]$ of cost $\OO^*(\left(\esaamlsbound(\beta,c)\right)^n)$,
where $\esaamlsbound(\beta,c)$ is the unique value $d\in \left(1, 1+\frac{c-1}{\beta} \right)$ which satisfies $\D{\frac{1}{\beta}}{\frac{d-1}{c-1}}=\frac{\ln c}{\beta}$, for every $\beta,c>1$.\footnote{$\D{a}{b}  = a \ln \frac{a}{b} + (1-a)\ln \frac{1-a}{1-b}$ is the Kullback-Leibler divergence between two Bernoulli distributions with parameters $a$ and $b$.}
The next lemma, which we prove in Section \Cref{sec:math_functions}, also implies that the analysis in \cite{EsmerKMNS22} is tight.

\begin{restatable}{lemma}{coincidewithesa}
	\label{lem:coincide_with_esa}
	For every $\beta,c>1$ it holds that $\bestbound(\beta,c,\beta) = \esaamlsbound(\beta,c)$.
\end{restatable}
Let $\beta>\alpha\geq 1$ and $c\geq 1$. Since an $\alpha$-extension oracle is also a $\beta$-extension oracle, 
the result of \cite{EsmerKMNS22} can be used to obtain a $\beta$-approximation algorithm for $\LSUB[\{(\alpha,c)\}]$ by executing the $\beta$-approximation algorithm for $\LSUB[\{(\beta,c)\}]$ 
whose running time is $\esaamlsbound(\beta,c) =\bestbound(\beta,c,\beta)$. 
This approach, which views an $\alpha$-extension oracle as a special case of $\beta$-extension oracle intuitively seems suboptimal.
The next lemma, proven in \Cref{sec:monotonicity_amls}, confirms this intuition.

\begin{lemma}
	\label{lem:better_than_esa}
	For every $\beta> \alpha \geq 1$ and every $c > 1$ it holds that $\bestbound(\alpha,c,\beta) < \esaamlsbound(\beta,c)$.
\end{lemma}

\section{Applications}
\label{sec:application}
In this section, we demonstrate how our results can be used to obtain exponential approximation algorithms for a wide range of problems.
All the problems considered in this section are defined in Appendix \ref{sec:problem_definitions}.
Moreover, extensive data sets providing the running times of the obtained algorithms are provided in Appendix \ref{sec:running_times}.

\subsection{Combining Exact FPT and Polynomial-Time Approximation Algorithms}
\label{sec:application_main}

The most common application of our results is to problems that admit a single-exponential FPT algorithm and/or a constant-factor approximation algorithm.
Indeed, both types of algorithms have been intensively studied in the literature (see, e.g., \cite{CyganFKLMPPS15,Vazirani01}), and there is an abundance of problems admitting single-exponential fpt algorithms and/or constant-factor approximation algorithms which we can use to obtain exponential approximation algorithms.
Actually, from the view point of applications, this is a key advantage over the previous work \cite{EsmerKMNS22} that requires a parameterized $\beta$-approximation algorithm, since such algorithms are still somewhat rare.

A large class of problems, many of which fall into this category, are deletion problems to some graph property $\Pi$.

\medskip
\defproblem{{\sc $\Pi$ Vertex Deletion}}{An (undirected or directed) graph $G$.}{Find a minimum set $S$ of vertices of $G$ such that $G-S \in \Pi$.}
\medskip

This type problem can be translated into our framework by setting $U \coloneqq V(G)$ to be the set of vertices of $G$, and the task is to find a minimum set in the set system $\F \coloneqq \{S \subseteq U \mid G - S \in \Pi\}$.
If $\Pi$ is a hereditary graph property (i.e., it is closed under subgraphs), the set system $\F$ is monotone which allows us to apply the algorithmic tools described in \Cref{sec:computational_model}.
As two illustrative examples, let us consider the \textsc{Feedback Vertex Set (FVS)} problem (over undirected graphs), which corresponds to $\Pi$ being the class of forests, and the \textsc{Tourament Feedback Vertex Set (Tournament FVS)} where the input is a tournament graph, and $\Pi$ contains all acyclic tournaments.

\begin{table}
 \small
 \centering
 {\sc Feedback Vertex Set}
 \medskip

 \begin{tabular}{c|c|c|c|c|c|c|c|c|c|}
	$(\alpha,c)$ & $1.1$ & $1.2$ & $1.3$ & $1.4$ & $1.5$ & $1.6$ & $1.7$ & $1.8$ & $1.9$\\
	\hline
	$(\beta,2.69998)$ & $1.465$ & $1.3861$ & $1.3331$ & $1.294$ & $1.2637$ & $1.2393$ & $1.2193$ & $1.2024$ & $1.188$\\
	\hline
	$(1.0,2.69998)$ & $1.4156$ & $1.3289$ & $1.2753$ & $1.2378$ & $1.2099$ & $1.1881$ & $1.1706$ & $1.1561$ & $1.144$\\
	\hline
	$(2.0,1.0)$ & $1.6588$ & $1.4847$ & $1.3657$ & $1.2768$ & $1.2072$ & $1.1507$ & $1.1037$ & $1.064$ & $1.0298$\\
	\hline
	combined & $1.4156$ & $1.3289$ & $1.2753$ & $1.2378$ & $1.2068$ & $1.1507$ & $1.1037$ & $1.064$ & $1.0298$\\
	\hline
\end{tabular}

 \medskip
 \medskip
 {\sc Tournament Feedback Vertex Set}
 \medskip

 \begin{tabular}{c|c|c|c|c|c|c|c|c|c|}
	$(\alpha,c)$ & $1.1$ & $1.2$ & $1.3$ & $1.4$ & $1.5$ & $1.6$ & $1.7$ & $1.8$ & $1.9$\\
	\hline
	$(\beta,1.618)$ & $1.2912$ & $1.2463$ & $1.2152$ & $1.1918$ & $1.1734$ & $1.1583$ & $1.1458$ & $1.1352$ & $1.126$\\
	\hline
	$(1.0,1.618)$ & $1.2348$ & $1.1837$ & $1.1531$ & $1.132$ & $1.1164$ & $1.1042$ & $1.0945$ & $1.0865$ & $1.0798$\\
	\hline
	$(2.0,1.0)$ & $1.6588$ & $1.4847$ & $1.3657$ & $1.2768$ & $1.2072$ & $1.1507$ & $1.1037$ & $1.064$ & $1.0298$\\
	\hline
	combined & $1.2348$ & $1.1837$ & $1.1531$ & $1.132$ & $1.1164$ & $1.1042$ & $1.0945$ & $1.064$ & $1.0298$\\
	\hline
\end{tabular}

 \caption{Running times for {\sc Feedback Vertex Set} and {\sc Tournament Feedback Vertex Set}.
  An entry at in row $(\alpha,c)$ and column $\beta$ is $\bestbound(\alpha,c,\beta)$.
  The last row contains $\bestbound(\CL_{\textsc{FVS}},\beta)$ and $\bestbound(\CL_{\textsc{TFVS}},\beta)$, respectively.}
 \label{tab:runtimes_fvs_tfvs}
\end{table}

The best (randomized) parameterized algorithm for \textsc{FVS} has been obtained by Li and Nederlof \cite{LiN22} and runs in time $\OO^*(2.69998^k)$.
Moreover, \textsc{FVS} admits a polynomial-time $2$-approximation algorithm \cite{BafnaBF99}.
The first algorithm provides a $1$-extension oracle with cost $2.69998$, and second algorithm implements a $2$-extension oracle with cost $1$.
Together, we obtain an oracle specification list $\CL_{\textsc{FVS}} \coloneqq \{(1,2.69998),(2,1)\}$.
Using \Cref{thm:upper_bound_random}, we can use approximate monotone local search to obtain a $\beta$-approximation algorithm for \textsc{FVS} running in time $\OO^*(\left(\bestbound(\CL_{\textsc{FVS}},\beta)\right)^n)$ for every $\beta \geq 1$.

Similarly, \textsc{Tournament Feedback Vertex Set} can be solved in time $\OO^*(1.618^k)$ \cite{KumarL16} and admits polynomial-time $2$-approximation algorithm \cite{LokshtanovMMPPS18}, which gives rise to the oracle specification list $\CL_{\textsc{TFVS}} \coloneqq \{(1,1.618),(2,1)\}$.
Hence, we obtain a $\beta$-approximation algorithm for \textsc{Tournament FVS} running in time $\OO^*(\left(\bestbound(\CL_{\textsc{TFVS}},\beta)\right)^n)$.

We provide the values of  $\bestbound(\CL_{\textsc{FVS}},\beta)$ and $\bestbound(\CL_{\textsc{TFVS}},\beta)$ for selected approximation ratios $\beta$ in \Cref{tab:runtimes_fvs_tfvs}, and give a graphical visualization in \Cref{fig:runtimes_fvs_tfvs}.
We also compare $\bestbound(\CL_{\textsc{FVS}},\beta)$ and $\bestbound(\CL_{\textsc{TFVS}},\beta)$ with the running times of several other algorithms.
As the most basic benchmark, we compare the running times to the brute-force search as described in \cite{EsmerKMNS22} (see also \hyperref[benchmark-brute]{Benchmark 1}).
Also, by interpreting an exact single-exponential fpt algorithm as a $\beta$-approximation algorithm, we can use Approximate Monotone Local Search for $\alpha = \beta$ \cite{EsmerKMNS22} as a second benchmark (see also \hyperref[benchmark-amls-equal]{Benchmark 2}).
It can be observed that $\bestbound(\CL_{\textsc{FVS}},\beta)$ and $\bestbound(\CL_{\textsc{TFVS}},\beta)$ are strictly better than both of these algorithms for all $\beta > 1$ (see also \Cref{lem:better_than_esa} and \Cref{thm:amls_smaller_brute}).
We remark that another exponential $\beta$-approximation algorithm for \textsc{FVS} has been obtained in \cite{EscoffierPT16}.\footnote{The running time of this algorithm is not correctly stated in \cite[Theorem 3.1]{EscoffierPT16}; their $\beta$-approximation algorithm runs in time $\OO^*(d^n)$ where $d \geq 1$ is the unique solution to the equation $1 = d^{-1} + d^{-\beta}$.}
However, this algorithm is slower than the brute-force $\beta$-approximation algorithm described in \cite{EsmerKMNS22}, and thus, our algorithm is also significantly faster than the algorithm from \cite{EscoffierPT16}.

\begin{figure}
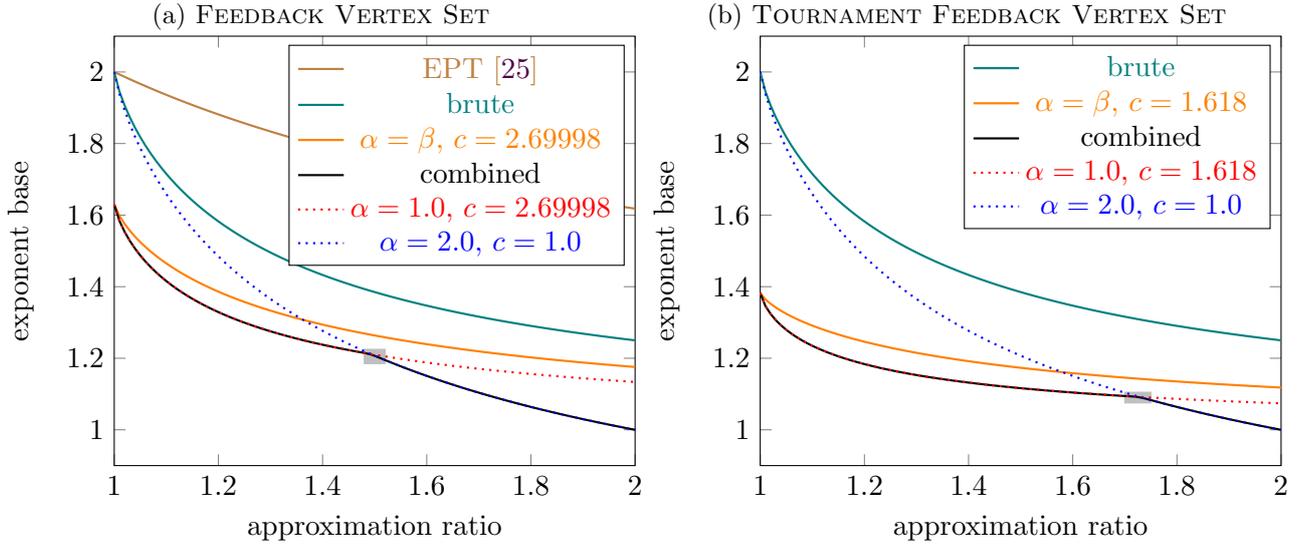

 \centering
 \begin{subfigure}{.5\textwidth}
  \centering
  \caption{{\sc Feedback Vertex Set}}
  \input{plots/figure_fvs_box.tex}
  \label{fig:fvs_results}
 \end{subfigure}%
 \begin{subfigure}{.5\textwidth}
  \centering
  \caption{{\sc Tournament Feedback Vertex Set}}
  \input{plots/figure_tfvs_box.tex}
  \label{fig:tfvs_results}
 \end{subfigure}%
 \caption{Results for {\sc Feedback Vertex Set} and {\sc Tournament Feedback Vertex Set}.
   A dot at $(\beta,d)$ means that the respective algorithm outputs an $\beta$-approximation in time $\OO^*(d^n)$.
   Figure~\ref{fig:runtimes_fvs_tfvs_zoom} zooms into the gray regions.}
 \label{fig:runtimes_fvs_tfvs}
\end{figure}

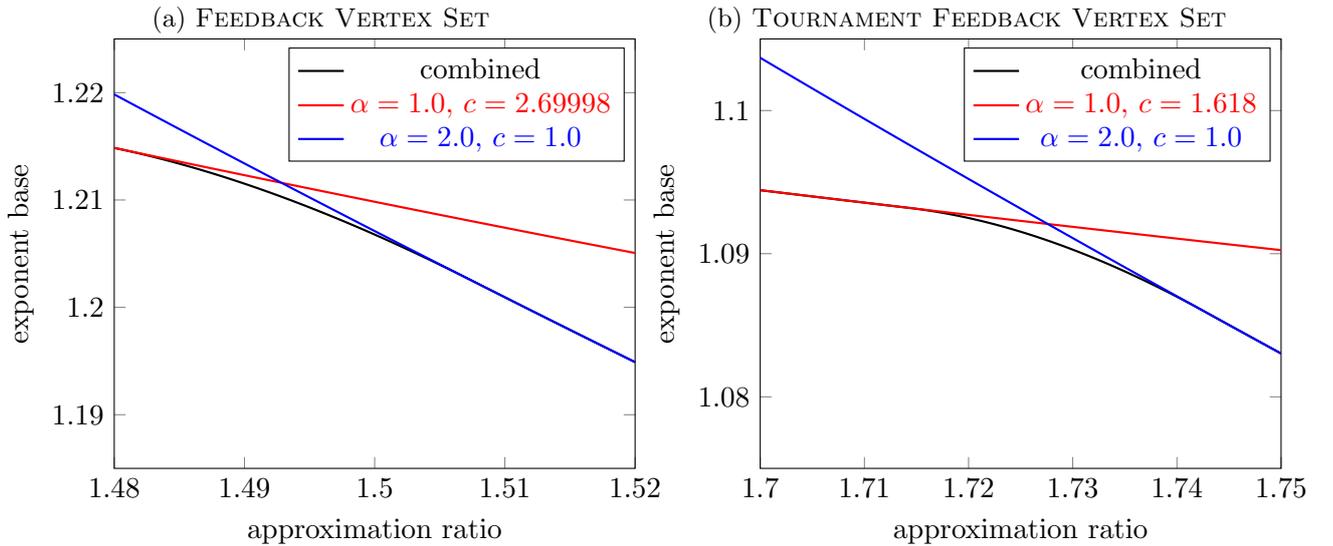
\begin{figure}
 \centering
 \begin{subfigure}{.5\textwidth}
  \centering
  \caption{{\sc Feedback Vertex Set}}
  \begin{tikzpicture}[scale = 1.0]
	\begin{axis}[xmin = 1.48, xmax = 1.52, ymin = 1.185, ymax = 1.225, xlabel = {approximation ratio}, ylabel = {exponent base}]

	\addplot[black, thick] coordinates {
		(1.48, 1.214852570165177)
		(1.4805, 1.21472361545438)
		(1.481, 1.214591080061285)
		(1.4815, 1.214454296156112)
		(1.482, 1.214313307865633)
		(1.4825, 1.214168158927233)
		(1.483, 1.214018892692388)
		(1.4835, 1.213865552130128)
		(1.484, 1.213708179830459)
		(1.4845, 1.213546818007764)
		(1.485, 1.213381508504176)
		(1.4855, 1.21321229279292)
		(1.486, 1.213039211981629)
		(1.4865, 1.212862306815632)
		(1.487, 1.21268161768122)
		(1.4875, 1.212497184608875)
		(1.488, 1.21230904727648)
		(1.4885, 1.2121172450125)
		(1.489, 1.211921816799138)
		(1.4895, 1.21172280127546)
		(1.49, 1.211520236740499)
		(1.4905, 1.211314161156331)
		(1.491, 1.211104612151127)
		(1.4915, 1.210891627022173)
		(1.492, 1.210675242738876)
		(1.4925, 1.21045549594573)
		(1.493, 1.210232422965275)
		(1.4935, 1.21000605980101)
		(1.494, 1.209776442140304)
		(1.4945, 1.209543605357261)
		(1.495, 1.209307584515576)
		(1.4955, 1.209068414371359)
		(1.496, 1.208826129375942)
		(1.4965, 1.208580763678649)
		(1.497, 1.208332351129558)
		(1.4975, 1.20808092528223)
		(1.498, 1.207826519396417)
		(1.4985, 1.207569166440745)
		(1.499, 1.20730889909538)
		(1.4995, 1.207045749754662)
		(1.5, 1.206779750529724)
		(1.5005, 1.206510933251087)
		(1.501, 1.206239329471228)
		(1.5015, 1.205964970467132)
		(1.502, 1.205687887242819)
		(1.5025, 1.205408110531847)
		(1.503, 1.205125670799799)
		(1.5035, 1.204840598246742)
		(1.504, 1.204552922809673)
		(1.5045, 1.204262674164931)
		(1.505, 1.203969881730603)
		(1.5055, 1.203674574668897)
		(1.506, 1.203376781888501)
		(1.5065, 1.203076532046919)
		(1.507, 1.202773853552788)
		(1.5075, 1.202468774568171)
		(1.508, 1.202162183770185)
		(1.5085, 1.201855856503398)
		(1.509, 1.201549844844549)
		(1.5095, 1.20124414826243)
		(1.51, 1.200938766227105)
		(1.5105, 1.20063369820991)
		(1.511, 1.200328943683445)
		(1.5115, 1.200024502121572)
		(1.512, 1.199720372999413)
		(1.5125, 1.199416555793341)
		(1.513, 1.19911304998098)
		(1.5135, 1.198809855041202)
		(1.514, 1.198506970454118)
		(1.5145, 1.198204395701079)
		(1.515, 1.19790213026467)
		(1.5155, 1.197600173628705)
		(1.516, 1.197298525278227)
		(1.5165, 1.196997184699499)
		(1.517, 1.196696151380004)
		(1.5175, 1.196395424808441)
		(1.518, 1.196095004474719)
		(1.5185, 1.195794889869953)
		(1.519, 1.195495080486464)
		(1.5195, 1.195195575817772)
		(1.52, 1.194896375358593)
	};

	\addplot[red, thick] coordinates {
		(1.48, 1.214852570165177)
		(1.4805, 1.214724040980207)
		(1.481, 1.214595675193105)
		(1.4815, 1.214467472474564)
		(1.482, 1.214339432496202)
		(1.4825, 1.214211554930566)
		(1.483, 1.214083839451123)
		(1.4835, 1.21395628573226)
		(1.484, 1.213828893449279)
		(1.4845, 1.213701662278392)
		(1.485, 1.213574591896723)
		(1.4855, 1.213447681982299)
		(1.486, 1.213320932214049)
		(1.4865, 1.213194342271801)
		(1.487, 1.21306791183628)
		(1.4875, 1.2129416405891)
		(1.488, 1.212815528212768)
		(1.4885, 1.212689574390673)
		(1.489, 1.212563778807088)
		(1.4895, 1.212438141147166)
		(1.49, 1.212312661096934)
		(1.4905, 1.212187338343294)
		(1.491, 1.212062172574017)
		(1.4915, 1.21193716347774)
		(1.492, 1.211812310743964)
		(1.4925, 1.211687614063051)
		(1.493, 1.211563073126219)
		(1.4935, 1.211438687625541)
		(1.494, 1.21131445725394)
		(1.4945, 1.21119038170519)
		(1.495, 1.211066460673905)
		(1.4955, 1.210942693855546)
		(1.496, 1.21081908094641)
		(1.4965, 1.21069562164363)
		(1.497, 1.210572315645173)
		(1.4975, 1.210449162649835)
		(1.498, 1.21032616235724)
		(1.4985, 1.210203314467835)
		(1.499, 1.210080618682889)
		(1.4995, 1.209958074704487)
		(1.5, 1.209835682235533)
		(1.5005, 1.20971344097974)
		(1.501, 1.209591350641632)
		(1.5015, 1.209469410926539)
		(1.502, 1.209347621540595)
		(1.5025, 1.209225982190734)
		(1.503, 1.20910449258469)
		(1.5035, 1.208983152430989)
		(1.504, 1.208861961438953)
		(1.5045, 1.20874091931869)
		(1.505, 1.208620025781099)
		(1.5055, 1.208499280537858)
		(1.506, 1.208378683301431)
		(1.5065, 1.208258233785056)
		(1.507, 1.208137931702751)
		(1.5075, 1.208017776769304)
		(1.508, 1.207897768700275)
		(1.5085, 1.207777907211989)
		(1.509, 1.20765819202154)
		(1.5095, 1.207538622846781)
		(1.51, 1.207419199406325)
		(1.5105, 1.207299921419543)
		(1.511, 1.207180788606559)
		(1.5115, 1.20706180068825)
		(1.512, 1.206942957386241)
		(1.5125, 1.206824258422903)
		(1.513, 1.206705703521351)
		(1.5135, 1.206587292405443)
		(1.514, 1.206469024799773)
		(1.5145, 1.206350900429672)
		(1.515, 1.206232919021205)
		(1.5155, 1.206115080301166)
		(1.516, 1.205997383997081)
		(1.5165, 1.205879829837199)
		(1.517, 1.205762417550492)
		(1.5175, 1.205645146866654)
		(1.518, 1.205528017516099)
		(1.5185, 1.205411029229953)
		(1.519, 1.205294181740059)
		(1.5195, 1.205177474778969)
		(1.52, 1.205060908079943)
	};

	\addplot[blue, thick] coordinates {
		(1.48, 1.219837214810051)
		(1.4805, 1.21951232444582)
		(1.481, 1.219187781603104)
		(1.4815, 1.218863585672358)
		(1.482, 1.218539736045568)
		(1.4825, 1.218216232116243)
		(1.483, 1.217893073279409)
		(1.4835, 1.217570258931607)
		(1.484, 1.217247788470889)
		(1.4845, 1.216925661296806)
		(1.485, 1.216603876810412)
		(1.4855, 1.216282434414253)
		(1.486, 1.215961333512363)
		(1.4865, 1.215640573510262)
		(1.487, 1.215320153814946)
		(1.4875, 1.215000073834888)
		(1.488, 1.21468033298003)
		(1.4885, 1.214360930661775)
		(1.489, 1.214041866292991)
		(1.4895, 1.213723139287995)
		(1.49, 1.213404749062559)
		(1.4905, 1.213086695033896)
		(1.491, 1.212768976620663)
		(1.4915, 1.212451593242952)
		(1.492, 1.212134544322283)
		(1.4925, 1.211817829281606)
		(1.493, 1.211501447545292)
		(1.4935, 1.211185398539127)
		(1.494, 1.210869681690312)
		(1.4945, 1.210554296427454)
		(1.495, 1.210239242180564)
		(1.4955, 1.209924518381052)
		(1.496, 1.209610124461722)
		(1.4965, 1.209296059856766)
		(1.497, 1.208982324001764)
		(1.4975, 1.208668916333674)
		(1.498, 1.208355836290831)
		(1.4985, 1.208043083312942)
		(1.499, 1.207730656841081)
		(1.4995, 1.207418556317684)
		(1.5, 1.207106781186548)
		(1.5005, 1.206795330892819)
		(1.501, 1.206484204882998)
		(1.5015, 1.206173402604928)
		(1.502, 1.205862923507795)
		(1.5025, 1.205552767042119)
		(1.503, 1.205242932659754)
		(1.5035, 1.204933419813883)
		(1.504, 1.204624227959011)
		(1.5045, 1.204315356550963)
		(1.505, 1.20400680504688)
		(1.5055, 1.203698572905213)
		(1.506, 1.203390659585722)
		(1.5065, 1.203083064549466)
		(1.507, 1.202775787258807)
		(1.5075, 1.202468827177399)
		(1.508, 1.202162183770185)
		(1.5085, 1.201855856503398)
		(1.509, 1.201549844844549)
		(1.5095, 1.20124414826243)
		(1.51, 1.200938766227105)
		(1.5105, 1.20063369820991)
		(1.511, 1.200328943683445)
		(1.5115, 1.200024502121572)
		(1.512, 1.199720372999413)
		(1.5125, 1.199416555793341)
		(1.513, 1.19911304998098)
		(1.5135, 1.198809855041202)
		(1.514, 1.198506970454118)
		(1.5145, 1.198204395701079)
		(1.515, 1.19790213026467)
		(1.5155, 1.197600173628705)
		(1.516, 1.197298525278227)
		(1.5165, 1.196997184699499)
		(1.517, 1.196696151380004)
		(1.5175, 1.196395424808441)
		(1.518, 1.196095004474719)
		(1.5185, 1.195794889869953)
		(1.519, 1.195495080486464)
		(1.5195, 1.195195575817772)
		(1.52, 1.194896375358593)
	};

	\addlegendentry[no markers, black]{combined}
	\addlegendentry[no markers, red]{$\alpha = 1.0$, $c = 2.69998$}
	\addlegendentry[no markers, blue]{$\alpha = 2.0$, $c = 1.0$}

	\end{axis}
\end{tikzpicture}
  \label{fig:fvs_zoom_results}
 \end{subfigure}%
 \begin{subfigure}{.5\textwidth}
  \centering
  \caption{{\sc Tournament Feedback Vertex Set}}
  \begin{tikzpicture}[scale = 1.0]
	\begin{axis}[xmin = 1.7, xmax = 1.75, ymin = 1.075, ymax = 1.105, xlabel = {approximation ratio}, ylabel = {exponent base}]

	\addplot[black, thick] coordinates {
		(1.7, 1.094435710503451)
		(1.7005, 1.09439180846425)
		(1.701, 1.094347949101091)
		(1.7015, 1.094304132349109)
		(1.702, 1.094260358143573)
		(1.7025, 1.094216626419892)
		(1.703, 1.094172937113605)
		(1.7035, 1.094129290160392)
		(1.704, 1.094085685496065)
		(1.7045, 1.09404212305657)
		(1.705, 1.09399860277799)
		(1.7055, 1.09395512459654)
		(1.706, 1.093911688448568)
		(1.7065, 1.093868294270557)
		(1.707, 1.093824941999121)
		(1.7075, 1.093781631571008)
		(1.708, 1.093738362923098)
		(1.7085, 1.0936951359924)
		(1.709, 1.093651950716059)
		(1.7095, 1.093608807031347)
		(1.71, 1.093565704875669)
		(1.7105, 1.093522644186559)
		(1.711, 1.093479624901682)
		(1.7115, 1.093436646958832)
		(1.712, 1.093393710295932)
		(1.7125, 1.093350814851035)
		(1.713, 1.093307960562322)
		(1.7135, 1.093265147368102)
		(1.714, 1.093222375206812)
		(1.7145, 1.093179329075699)
		(1.715, 1.093133427302266)
		(1.7155, 1.09308411927807)
		(1.716, 1.09303142985525)
		(1.7165, 1.092975383786692)
		(1.717, 1.092916005726858)
		(1.7175, 1.092853320232634)
		(1.718, 1.092787351764152)
		(1.7185, 1.092718124685626)
		(1.719, 1.09264566326617)
		(1.7195, 1.092569991680621)
		(1.72, 1.092491134010353)
		(1.7205, 1.092409114244087)
		(1.721, 1.092323956278701)
		(1.7215, 1.092235683920029)
		(1.722, 1.09214432088366)
		(1.7225, 1.092049890795733)
		(1.723, 1.091952417193723)
		(1.7235, 1.091851923527228)
		(1.724, 1.091748433158749)
		(1.7245, 1.091641969364464)
		(1.725, 1.091532555335)
		(1.7255, 1.0914202141762)
		(1.726, 1.091304968909883)
		(1.7265, 1.091186842474603)
		(1.727, 1.091065857726397)
		(1.7275, 1.090942037439541)
		(1.728, 1.090815404307283)
		(1.7285, 1.090685980942589)
		(1.729, 1.090553789878872)
		(1.7295, 1.090418853570722)
		(1.73, 1.090281194394629)
		(1.7305, 1.0901408346497)
		(1.731, 1.089997796558379)
		(1.7315, 1.089852102267146)
		(1.732, 1.089703773847232)
		(1.7325, 1.089552833295307)
		(1.733, 1.089399302534181)
		(1.7335, 1.089243203413492)
		(1.734, 1.089084557710386)
		(1.7345, 1.0889233871302)
		(1.735, 1.08875971330713)
		(1.7355, 1.088593557804905)
		(1.736, 1.088424942117448)
		(1.7365, 1.088253887669531)
		(1.737, 1.088080415817431)
		(1.7375, 1.087904547849576)
		(1.738, 1.087726304987187)
		(1.7385, 1.087545708384915)
		(1.739, 1.087362779131473)
		(1.7395, 1.08717753825026)
		(1.74, 1.086990006699982)
		(1.7405, 1.086800205375269)
		(1.741, 1.086608155107285)
		(1.7415, 1.086413876664328)
		(1.742, 1.086217390752432)
		(1.7425, 1.086018718015962)
		(1.743, 1.085818235937165)
		(1.7435, 1.085617684089784)
		(1.744, 1.085417291370921)
		(1.7445, 1.085217057578325)
		(1.745, 1.085016982510105)
		(1.7455, 1.084817065964728)
		(1.746, 1.084617307741019)
		(1.7465, 1.084417707638158)
		(1.747, 1.084218265455684)
		(1.7475, 1.084018980993486)
		(1.748, 1.083819854051813)
		(1.7485, 1.083620884431262)
		(1.749, 1.083422071932786)
		(1.7495, 1.083223416357689)
		(1.75, 1.083024917507624)
	};

	\addplot[red, thick] coordinates {
		(1.7, 1.094435710503451)
		(1.7005, 1.09439180846425)
		(1.701, 1.094347949101091)
		(1.7015, 1.094304132349109)
		(1.702, 1.094260358143573)
		(1.7025, 1.094216626419892)
		(1.703, 1.094172937113605)
		(1.7035, 1.094129290160392)
		(1.704, 1.094085685496065)
		(1.7045, 1.09404212305657)
		(1.705, 1.09399860277799)
		(1.7055, 1.09395512459654)
		(1.706, 1.093911688448568)
		(1.7065, 1.093868294270557)
		(1.707, 1.093824941999121)
		(1.7075, 1.093781631571008)
		(1.708, 1.093738362923098)
		(1.7085, 1.0936951359924)
		(1.709, 1.093651950716059)
		(1.7095, 1.093608807031347)
		(1.71, 1.093565704875669)
		(1.7105, 1.093522644186559)
		(1.711, 1.093479624901682)
		(1.7115, 1.093436646958832)
		(1.712, 1.093393710295932)
		(1.7125, 1.093350814851035)
		(1.713, 1.093307960562322)
		(1.7135, 1.093265147368102)
		(1.714, 1.093222375206812)
		(1.7145, 1.093179644017017)
		(1.715, 1.093136953737409)
		(1.7155, 1.093094304306805)
		(1.716, 1.093051695664153)
		(1.7165, 1.093009127748523)
		(1.717, 1.092966600499112)
		(1.7175, 1.092924113855244)
		(1.718, 1.092881667756367)
		(1.7185, 1.092839262142053)
		(1.719, 1.092796896952)
		(1.7195, 1.092754572126029)
		(1.72, 1.092712287604087)
		(1.7205, 1.092670043326242)
		(1.721, 1.092627839232686)
		(1.7215, 1.092585675263735)
		(1.722, 1.092543551359826)
		(1.7225, 1.09250146746152)
		(1.723, 1.092459423509498)
		(1.7235, 1.092417419444563)
		(1.724, 1.09237545520764)
		(1.7245, 1.092333530739775)
		(1.725, 1.092291645982135)
		(1.7255, 1.092249800876005)
		(1.726, 1.092207995362792)
		(1.7265, 1.092166229384023)
		(1.727, 1.092124502881343)
		(1.7275, 1.092082815796516)
		(1.728, 1.092041168071427)
		(1.7285, 1.091999559648077)
		(1.729, 1.091957990468585)
		(1.7295, 1.091916460475189)
		(1.73, 1.091874969610245)
		(1.7305, 1.091833517816225)
		(1.731, 1.091792105035718)
		(1.7315, 1.091750731211429)
		(1.732, 1.091709396286182)
		(1.7325, 1.091668100202913)
		(1.733, 1.091626842904677)
		(1.7335, 1.091585624334642)
		(1.734, 1.091544444436093)
		(1.7345, 1.091503303152427)
		(1.735, 1.09146220042716)
		(1.7355, 1.091421136203917)
		(1.736, 1.09138011042644)
		(1.7365, 1.091339123038584)
		(1.737, 1.091298173984316)
		(1.7375, 1.091257263207719)
		(1.738, 1.091216390652985)
		(1.7385, 1.091175556264421)
		(1.739, 1.091134759986445)
		(1.7395, 1.091094001763586)
		(1.74, 1.091053281540486)
		(1.7405, 1.091012599261898)
		(1.741, 1.090971954872686)
		(1.7415, 1.090931348317823)
		(1.742, 1.090890779542394)
		(1.7425, 1.090850248491594)
		(1.743, 1.090809755110727)
		(1.7435, 1.090769299345208)
		(1.744, 1.090728881140559)
		(1.7445, 1.090688500442412)
		(1.745, 1.090648157196509)
		(1.7455, 1.090607851348699)
		(1.746, 1.090567582844937)
		(1.7465, 1.090527351631291)
		(1.747, 1.090487157653932)
		(1.7475, 1.09044700085914)
		(1.748, 1.090406881193302)
		(1.7485, 1.090366798602912)
		(1.749, 1.09032675303457)
		(1.7495, 1.090286744434981)
		(1.75, 1.090246772750959)
	};

	\addplot[blue, thick] coordinates {
		(1.7, 1.103684103612658)
		(1.7005, 1.103469068554248)
		(1.701, 1.103254211458409)
		(1.7015, 1.103039532088625)
		(1.702, 1.102825030208821)
		(1.7025, 1.102610705583361)
		(1.703, 1.102396557977046)
		(1.7035, 1.102182587155116)
		(1.704, 1.101968792883245)
		(1.7045, 1.101755174927544)
		(1.705, 1.101541733054556)
		(1.7055, 1.101328467031259)
		(1.706, 1.101115376625059)
		(1.7065, 1.100902461603797)
		(1.707, 1.100689721735742)
		(1.7075, 1.100477156789591)
		(1.708, 1.10026476653447)
		(1.7085, 1.100052550739929)
		(1.709, 1.099840509175948)
		(1.7095, 1.099628641612927)
		(1.71, 1.099416947821692)
		(1.7105, 1.099205427573492)
		(1.711, 1.098994080639996)
		(1.7115, 1.098782906793293)
		(1.712, 1.098571905805895)
		(1.7125, 1.098361077450728)
		(1.713, 1.098150421501138)
		(1.7135, 1.097939937730887)
		(1.714, 1.097729625914153)
		(1.7145, 1.097519485825528)
		(1.715, 1.097309517240018)
		(1.7155, 1.097099719933039)
		(1.716, 1.096890093680423)
		(1.7165, 1.096680638258408)
		(1.717, 1.096471353443645)
		(1.7175, 1.096262239013192)
		(1.718, 1.096053294744515)
		(1.7185, 1.095844520415486)
		(1.719, 1.095635915804384)
		(1.7195, 1.095427480689892)
		(1.72, 1.095219214851096)
		(1.7205, 1.095011118067485)
		(1.721, 1.094803190118952)
		(1.7215, 1.094595430785788)
		(1.722, 1.094387839848686)
		(1.7225, 1.094180417088736)
		(1.723, 1.093973162287429)
		(1.7235, 1.093766075226649)
		(1.724, 1.093559155688681)
		(1.7245, 1.093352403456201)
		(1.725, 1.093145818312281)
		(1.7255, 1.092939400040388)
		(1.726, 1.092733148424378)
		(1.7265, 1.092527063248501)
		(1.727, 1.092321144297398)
		(1.7275, 1.092115391356097)
		(1.728, 1.091909804210016)
		(1.7285, 1.091704382644963)
		(1.729, 1.091499126447131)
		(1.7295, 1.091294035403097)
		(1.73, 1.091089109299828)
		(1.7305, 1.09088434792467)
		(1.731, 1.090679751065356)
		(1.7315, 1.090475318510001)
		(1.732, 1.090271050047098)
		(1.7325, 1.090066945465526)
		(1.733, 1.08986300455454)
		(1.7335, 1.089659227103774)
		(1.734, 1.089455612903242)
		(1.7345, 1.089252161743334)
		(1.735, 1.089048873414815)
		(1.7355, 1.088845747708827)
		(1.736, 1.088642784416885)
		(1.7365, 1.08843998333088)
		(1.737, 1.088237344243072)
		(1.7375, 1.088034866946096)
		(1.738, 1.087832551232958)
		(1.7385, 1.08763039689703)
		(1.739, 1.087428403732059)
		(1.7395, 1.087226571532156)
		(1.74, 1.087024900091801)
		(1.7405, 1.086823389205842)
		(1.741, 1.08662203866949)
		(1.7415, 1.086420848278323)
		(1.742, 1.086219817828283)
		(1.7425, 1.086018947115675)
		(1.743, 1.085818235937165)
		(1.7435, 1.085617684089784)
		(1.744, 1.085417291370921)
		(1.7445, 1.085217057578325)
		(1.745, 1.085016982510105)
		(1.7455, 1.084817065964728)
		(1.746, 1.084617307741019)
		(1.7465, 1.084417707638158)
		(1.747, 1.084218265455684)
		(1.7475, 1.084018980993486)
		(1.748, 1.083819854051813)
		(1.7485, 1.083620884431262)
		(1.749, 1.083422071932786)
		(1.7495, 1.083223416357689)
		(1.75, 1.083024917507624)
	};

	\addlegendentry[no markers, black]{combined}
	\addlegendentry[no markers, red]{$\alpha = 1.0$, $c = 1.618$}
	\addlegendentry[no markers, blue]{$\alpha = 2.0$, $c = 1.0$}

	\end{axis}
\end{tikzpicture}
  \label{fig:tfvs_zoom_results}
 \end{subfigure}%
 \caption{Results for {\sc Feedback Vertex Set} and {\sc Tournament Feedback Vertex Set}.
   A dot at $(\beta,d)$ means that the respective algorithm outputs an $\beta$-approximation in time $\OO^*(d^n)$.}
 \label{fig:runtimes_fvs_tfvs_zoom}
\end{figure}

As a further comparison, we also consider the running of our algorithm when only a single oracle is used.
More precisely, let us define $\CL_{\textsc{FVS}}' \coloneqq \{(1,2.69998)\}$ and $\CL_{\textsc{FVS}}'' \coloneqq \{(2,1)\}$.
Clearly, $\bestbound(\CL_{\textsc{FVS}},\beta) \leq \min(\bestbound(\CL_{\textsc{FVS}}',\beta),\bestbound(\CL_{\textsc{FVS}}'',\beta))$ for all $\beta \geq 1$.
Interestingly, this inequality is strict for some values of $\beta$.
Indeed, while this may not be visible from Figure \ref{fig:runtimes_fvs_tfvs}, one can observe from Figure \ref{fig:runtimes_fvs_tfvs_zoom} that using both oracles together leads to a better running for $\beta$ roughly in the range $[1.481,1.507]$.
However, it can also be observed that the improvement obtained this way is rather small.
For example, we have $\bestbound(\CL_{\textsc{FVS}},1.5) \approx 1.2068$ and $\min(\bestbound(\CL_{\textsc{FVS}}',1.5),\bestbound(\CL_{\textsc{FVS}}'',1.5)) = \bestbound(\CL_{\textsc{FVS}}'',1.5) \approx 1.2072$.
Similar observations can be made for \textsc{Tournament FVS}.
Note that the fact that we only obtain small improvements by using multiple oracles is not a shortcoming of the algorithms designed in this paper, but inherent to the problem by Theorem \ref{thm:minimization_lower_bound}.

\begin{table}
 \centering
 \begin{tabular}{l|lr|c|lr|c|}
  Problem & $c_1$ & & det. & $\alpha_2$ & & det. \\
  \hline
  \textsc{FVS} & $2.69998$ & \cite{LiN22} & \xmark & $2$ & \cite{BafnaBF99} & \cmark \\
  \hline
  \textsc{Tournament FVS} & $1.618$ & \cite{KumarL16} & \cmark & $2$ & \cite{LokshtanovMMPPS18} & \xmark \\
  \hline
  \textsc{Subset FVS} & $4.0$ & \cite{IwataWY16} & \cmark & $8$ & \cite{EvenNZ00} & \cmark \\
  \hline
  \textsc{$d$-Hitting Set} ($d \geq 3$) & $(d - 0.9245)$ & \cite{FominGKLS10} & \cmark & $d$ & \cite{Bar-YehudaE81} & \cmark \\
  \hline
  \textsc{Interval Vertex Deletion} & $8.0$ & \cite{Cao16} & \cmark & $8$ & \cite{Cao16} & \cmark \\
  \hline
  \textsc{Proper Interval Vertex Deletion} & $6.0$ & \cite{HofV13} & \cmark & $6$ & \cite{HofV13} & \cmark \\
  \hline
  \textsc{Block Graph Vertex Deletion} & $4.0$ & \cite{AgrawalKLS16} & \cmark & $4$ & \cite{AgrawalKLS16} & \cmark \\
  \hline
  \textsc{Cluster Graph Vertex Deletion} & $1.9102$ & \cite{BoralCKP16} & \cmark & $2$ & \cite{AprileDFH23} & \cmark \\
  \hline
  \textsc{Cograph Vertex Deletion} & $3.0755$ & \cite{FominGKLS10} & \cmark & $4$ & & \cmark \\
  \hline
  \textsc{Split Vertex Deletion} & $2.0$ & \cite{GhoshK0MPRR15} & \cmark & $2 + \epsilon$ & \cite{DrescherFH20} & \cmark\\
  \hline
  \textsc{Edge Multicut on Trees} & $1.5538$ & \cite{KanjLLTXXYZZZ15} & \cmark & $2$ & \cite{GargVY93} & \cmark \\
  \hline
 \end{tabular}
 \caption{List of deletion problems admitting an single-exponential parameterized algorithm running in time $\OO^*(c_1^k)$ and a polynomial-time $\alpha_2$-approximation algorithm.}
 \label{tab:overview_problems}
\end{table}

We stress that these results are not limited to \textsc{FVS} and \textsc{Tournament FVS}.
Indeed, there is wide range of vertex-deletion problems for which a single-exponential fpt algorithm running in time $\OO^*(c_1^k)$ as well as  a polynomial-time $\alpha_2$-approximation is known, for suitable constants $c_1,\alpha_2 > 1$.
A list of examples is given in Table \ref{tab:overview_problems} (running times for all problems can be found in Appendix \ref{sec:running_times}).

Notably, \textsc{Edge Multicut on Trees} is not a vertex-deletion problem, but an edge-deletion problem.
This means we set $U \coloneqq E(G)$ which implies that the running time in \Cref{thm:upper_bound_random} is measured with respect to the number of edges rather than the number of vertices.
However, for this particular problem, the input graph $G$ is a tree which implies that $|E(G)| \leq |V(G)|$, and hence we obtain the same runtime bound with respect to the number of vertices.

For all the problems listed in Table \ref{tab:overview_problems}, by \Cref{thm:upper_bound_random}, monotone local search results in a $\beta$-approximation algorithm running in time $\OO^*(\left(\bestbound(\CL_{c_1,\alpha_2},\beta)\right)^n)$, where $\CL_{c_1,\alpha_2} \coloneqq \{(1,c_1),(\alpha_2,1)\}$, for all $\beta \geq 1$ that outperforms all previously existing algorithms.
Note that the algorithms we obtain are randomized.
However, by Theorem \ref{thm:upper_bound_deter}, we can also obtain a deterministic $\beta$-approximation algorithm at the cost of an additional subexponential factor if all parameterized extension subroutine are deterministic.
Looking at Table \ref{tab:overview_problems}, this is true for all listed problems except \textsc{FVS} and \textsc{Tournament FVS}.

Note that our results are also applicable if only either a single-exponential fpt algorithm or a polynomial-time constant-factor approximation algorithm is available.
As a notable example, \textsc{Odd Cycle Transversal (OCT)} can be solved in time $\OO^*(2.3146^k)$ \cite{LokshtanovNRRS14}, and has no constant-factor approximation algorithm assuming the Unique Games Conjecture \cite{Khot02}.
On the other side, the \textsc{Partial Vertex Cover} problem has a polynomial-time $2$-approximation \cite{BshoutyB98}, and is known to be ${\sf W[1]}$-hard \cite{GuoNW07} which means that it cannot be solved in single-exponential fpt time assuming ${\sf FPT} \neq {\sf W[1]}$.
Still, for both problems, we obtain a $\beta$-approximation algorithm that is faster than the brute-force search (see \Cref{thm:amls_smaller_brute}) and, in the case of \textsc{Odd Cycle Transversal (OCT)}, than the algorithm obtained from \cite{EsmerKMNS22} (see \Cref{lem:better_than_esa}).
For both problems, the running times of the obtained algorithms can again be found in Appendix \ref{sec:running_times}.

\subsection{Exploiting Parameterized Approximation Algorithms}

We also obtain new algorithms for problems that are neither known to have a single-exponential fpt algorithm nor a polynomial-time approximation algorithm, but admit a single-exponential parameterized constant-factor approximation algorithm.
For \textsc{DFVS}, \textsc{Subset DFVS}, \textsc{DOCT} and \textsc{Multicut}, \cite{LokshtanovMRSZ21} provides a $2$-approximation algorithm that runs in time $\OO^*(c^k)$ for some constant $c$.
For example, one can easily observe from the description of the \textsc{DFVS} algorithm in \cite{LokshtanovMRSZ21} that it runs in time $\OO^*(1024^k)$.
Using \Cref{thm:upper_bound_random}, monotone local search results in an exponential $\beta$-approximation algorithm that runs in time $\OO^*(\bestbound(2, 1024,\beta))$ for all $\beta > 1$.
By \Cref{thm:amls_smaller_brute}, this algorithm is qualitatively better than the brute-force $\beta$-approximation algorithm running in time $\OO^*((\brute(\beta))^n)$.
For example, $\bestbound(2,1024,1.1) \approx 1.71520$ and $\brute(1.1) \approx 1.71527$.
Similar results can be obtained for the other problems.

Moreover, for the problem \textsc{Symmetric Directed  Multicut}, it is possible to adapt a parameterized $2$-approximation algorithm (which runs in time $k^{\OO(k)} \cdot n^{\OO(1)}$) \cite{EibenRW22} to obtain a parameterized $\alpha$-approximation algorithm running in time $\OO^*(c^k)$ for some constants $\alpha,c > 1$ \cite{Wahlstrom23}.
As a consequence, we also obtain an exponential $\beta$-approximation algorithm for this problem that beats the brute-force $\beta$-approximation algorithm for every $\beta > 1$. 

Similarly, using the $\OO^*(c^k)$-time $2$-approximation algorithm for \textsc{$d$-Steiner Multicut} in~\cite[Theorem 37]{OsipovW23} for some $c >1$, one can obtain a $\beta$-approximation algorithm that beats the brute-force $\beta$-approximation algorithm for all $\beta > 1$.

\subsection{Vertex Cover and 3-Hitting Set}
\label{sec:application_vc_3hs}

Finally, we consider the \textsc{Vertex Cover} and \textsc{$3$-Hitting Set} problem.
Both problems have not only been extensively studied for their exact parameterized complexity \cite{ChenKX10,Wahlstrom07}, but also received significant attention in the area of parameterized approximation algorithms \cite{BrankovicF13,FellowsKRS18,KulikS20}.
As a result, these two problems are the main applications considered in \cite{EsmerKMNS22} for transforming a parameterized $\beta$-approximation algorithm into an exponential $\beta$-approximation algorithm.
Despite a parameterized $\beta$-approximation being the natural oracle choice to obtain an exponential $\beta$-approximation, our algorithmic framework allows us to obtain further improvements for both problems compared to \cite{EsmerKMNS22}.

For every $\alpha \in [1,2]$ the best known running time of a parameterized randomized $\alpha$-approximation algorithm for \textsc{VC} is attained in \cite{KulikS20} if $\alpha \succsim 1.03$, and in \cite{BrankovicF13} if $\alpha \precsim 1.03$ (using the exact algorithm from \cite{ChenKX10} for $\alpha = 1$).
Let us denote by $c_{\textsc{vc}}(\alpha)$ the base of the currently fastest known parameterized $\alpha$-approximation algorithm for \textsc{VC} (i.e., an $\alpha$-approximation can be computed in time $\OO^*((c_{\textsc{vc}}(\alpha))^k)$).
Similarly, for $\alpha \in [1,3]$, we write $c_{\textsc{hs}}(\alpha)$ for the base of the currently fastest known parameterized $\alpha$-approximation algorithm for \textsc{$3$-HS}.
This best known base is attained by either \cite{FellowsKRS18} if $\alpha \precsim 1.08$ (using the exact algorithm from \cite{Wahlstrom07} for $\alpha = 1$), or \cite{KulikS20} if $\alpha \succsim 1.08$.
Note that $c_{\textsc{vc}}(2) = c_{\textsc{hs}}(3) = 1$.

As indicated above, the currently fastest (randomized) exponential $\beta$-approximation algorithm for \textsc{VC} (resp.\ \textsc{$3$-HS}) was obtained in \cite{EsmerKMNS22} and runs in time $\OO^*(d^n)$ where $d \coloneqq \bestbound(\beta,c_{\textsc{vc}}(\beta),\beta)$ (resp.\ $d \coloneqq \bestbound(\beta,c_{\textsc{hs}}(\beta),\beta)$) using \Cref{lem:coincide_with_esa}.

Now, to apply our algorithmic framework, we need to fix an oracle specification list $\CL_{\textsc{vc}}$ (resp.\ $\CL_{\textsc{hs}}$).
Since we can only provide a finite number of oracles (and it is a priori unclear how to choose those oracles optimally), we adopt the basic approach of equally discretizing the range for $\alpha$.
We set $A_{\textsc{vc}} \coloneqq \{1,1.01,1.02,1.03,\dots,1.99,2\}$ and $A_{\textsc{hs}} \coloneqq \{1,1.02,1.04,1.06,\dots,2.98,3\}$ (both sets contain $101$ elements). 
Then we define $\CL_{\textsc{vc}} \coloneqq \{(\alpha,c_{\textsc{vc}}(\alpha)) \mid \alpha \in A_{\textsc{vc}}\}$ and $\CL_{\textsc{hs}} \coloneqq \{(\alpha,c_{\textsc{hs}}(\alpha)) \mid \alpha \in A_{\textsc{hs}}\}$.

\begin{table}
 \small
 \centering
 {\sc Vertex Cover}
 \medskip

 \begin{tabular}{c|c|c|c|c|c|c|c|c|c|c|}
	$(\alpha,c)$ & $1.01$ & $1.02$ & $1.03$ & $1.04$ & $1.05$ & $1.06$ & $1.07$ & $1.08$ & $1.09$ & $1.1$\\
	\hline
	$(\beta,c_{\textsc{vc}}(\beta))$ & $1.2038$ & $1.1955$ & $1.183$ & $1.1697$ & $1.158$ & $1.1475$ & $1.138$ & $1.1294$ & $1.1214$ & $1.114$\\
	\hline
	$\CL_{\textsc{vc}}$ & $1.1891$ & $1.1752$ & $1.1649$ & $1.1566$ & $1.1496$ & $1.1433$ & $1.1358$ & $1.1275$ & $1.1197$ & $1.1125$\\
\end{tabular}

 \medskip
 \medskip
 {\sc $3$-Hitting Set}
 \medskip

 \begin{tabular}{c|c|c|c|c|c|c|c|c|c|c|}
	$(\alpha,c)$ & $1.02$ & $1.04$ & $1.06$ & $1.08$ & $1.1$ & $1.12$ & $1.14$ & $1.16$ & $1.18$ & $1.2$\\
	\hline
	$(\beta,c_{\textsc{hs}}(\beta))$ & $1.472$ & $1.441$ & $1.4157$ & $1.3863$ & $1.3543$ & $1.3262$ & $1.3013$ & $1.2787$ & $1.2584$ & $1.2399$\\
	\hline
	$\CL_{\textsc{hs}}$ & $1.4489$ & $1.4083$ & $1.3775$ & $1.3527$ & $1.3319$ & $1.314$ & $1.2984$ & $1.2783$ & $1.2579$ & $1.2393$\\
\end{tabular}

 \caption{Running times for {\sc Vertex Cover} and {\sc$3$-Hitting Set}. An entry in row $(\alpha,c)$ and column $\beta$ is $\bestbound(\{(\alpha,c)\}, \beta)$. The middle row in each table is the result from \cite{EsmerKMNS22}, the last row is the result attained in this paper}
 \label{tab:runtimes_vc_3hs_small}
\end{table}

For the \textsc{Vertex Cover} problem, it can be observed that $\bestbound(\CL_{\textsc{vc}},\beta) < \bestbound(\beta,c_{\textsc{vc}}(\beta),\beta)$ for all $\beta \in (1,2) \cap A_{\textsc{vc}}$ (by evaluating both functions up to a sufficiently large precision).
The most significant improvements occur for small values of $\beta$ (see \Cref{tab:runtimes_vc_3hs_small}), and it seems that this improvement can be mostly attributed to the possibility of using the exact fpt algorithm for \textsc{Vertex Cover} \cite{ChenKX10} as a subroutine.
For larger approximation ratios, we only obtain small improvements as can be observed from \Cref{tab:runtimes_vc}.

\begin{table}
 \small
 \centering
 {\sc Vertex Cover}
 \medskip

 \footnotesize
 \begin{tabular}{c|c|c|c|c|c|c|c|c|}
	$(\alpha,c)$ & $1.2$ & $1.3$ & $1.4$ & $1.5$ & $1.6$ & $1.7$ & $1.8$ & $1.9$\\
	\hline
	$(\beta,c_{\textsc{vc}}(\beta))$ & $1.063058$ & $1.036524$ & $1.020288$ & $1.0098549$ & $1.0043411$ & $1.0015504$ & $1.00039597$ & $1.000042813$\\
	\hline
	$\CL_{\textsc{vc}}$              & $1.061819$ & $1.035901$ & $1.019999$ & $1.0097939$ & $1.0042837$ & $1.0015355$ & $1.00039185$ & $1.000042504$\\
\end{tabular}

 \caption{Running times for {\sc Vertex Cover}. An entry in row $(\alpha,c)$ and column $\beta$ is $\bestbound(\{(\alpha,c)\}, \beta)$. The middle row in each table is the result from \cite{EsmerKMNS22}, the last row is the result attained in this paper}
 \label{tab:runtimes_vc}
\end{table}

Generally speaking, it is also noteworthy that, even if the algorithm has access to $|A_{\textsc{vc}}| = 101$ many different oracles, only $2$-$3$ oracles corresponding to tuples from the specification list $\CL_{\textsc{vc}}$ are actually used, and the corresponding approximation ratios $\alpha$ are close to $\beta$.
For example,
\begin{align*}
 \bestbound(\CL_{\textsc{vc}},1.5) &= \bestbound(\{(1.49,c_{\textsc{vc}}(1.49)),(1.5,c_{\textsc{vc}}(1.5))\},1.5)\\
                                   &< \min\{\bestbound(1.49,c_{\textsc{vc}}(1.49),1.5),\bestbound(1.5,c_{\textsc{vc}}(1.5),1.5)\}\\
                                   &= \bestbound(1.49,c_{\textsc{vc}}(1.49),1.5) \approx 1.0098063.
\end{align*}
Similar observations can be made for {\sc $3$-Hitting Set}.

\section{Approximate Monotone Local Search}
\label{sec:analysis}
In this section, we analyse Algorithm \ref{algo:final} and describe how it can be derandomized.
More precisely, we prove \Cref{lem:amls_upper_bound_via_f,lem:deter_upper_bound_via_f}

\subsection{Correctness and Basic Analysis}

We first analyse the randomized algorithm $\amlsalgo_{\CL,\beta}$ for the $\LSUB$ problem. This algorithm uses algorithm $\sample$ as a subroutine.
See Section~\ref{sec:computational_model} for the description of these algorithms.
Also recall the definitions of $\p(n,k,t,x)$ and $M^*_{\alpha, \beta}$ from Equations~\eqref{eq:hyper} and~\eqref{eq:Mstar_def}, respectively.

\begin{lemma}[Correctness]
 \label{lem:correctness}
 For every specification list $\CL$ and $\beta\geq 1$,
 $\amlsalgo_{\CL,\beta}$ (Algorithm~\ref{algo:final}) is a randomized $\beta$-approximation algorithm for $\LSUB$.
\end{lemma}

\begin{proof}
 Let $U$ be a finite set system and $\F$ be a monotone set system of  $U$.
 Also, for every $(\alpha,c)\in \CL$ let $\oracle_{\alpha,c}$ be an $\alpha$-extension oracle of $U$ and $\CF$.
 Let $\F$ be the implicit monotone subset family associated with $\CL$.
 Also, let $\OPT= \argmin_{S\in \CF} |S|$ be a minimum size solution of the $\LSUB$ instance $U$ and $\CF$.
 Consider an execution of $\sample(U,k,t,\alpha,\beta,\oracle_{\alpha,c})$ (Algorithm~\ref{algo:intermediary}) in which $k = |\OPT|$, $M^*_{\alpha,\beta} \cdot k \leq t \leq \beta \cdot k$ and $(\alpha,c) \in \CL$.
 
 If the algorithm selects a set $X$ in Step~\ref{int:select} such that
 \[|\OPT \cap X| \geq \left( 1 - \frac{\beta}{\alpha} \right)\cdot k + \frac{t}{\alpha}\]
 then
 \[|\OPT \setminus X| \leq k - \left\lceil\left( 1 - \frac{\beta}{\alpha} \right)\cdot k + \frac{t}{\alpha}\right\rceil.\]
 Moreover, $(\OPT \setminus X) \cup X \in \F$ since $(\OPT \setminus X) \cup X = \OPT \cup X \supseteq \OPT$ as $\F$ is monotone by assumption.
 
 Since $\oracle_{\alpha,c}$ is an $\alpha$-extension oracle for $(U,\F)$, given the input $\left( U, X, k - \left\lceil\left( 1 - \frac{\beta}{\alpha} \right)\cdot k + \frac{t}{\alpha}\right\rceil \right)$ it returns a set $Y$ such that $X\cup Y\in \CF$ and with probability at least $\frac{1}{2}$ it holds that
 \[|Y| \leq \alpha \cdot \left(k - \left\lceil\left( 1 - \frac{\beta}{\alpha} \right)\cdot k + \frac{t}{\alpha}\right\rceil \right) \leq \alpha \cdot \left(k - \left( 1 - \frac{\beta}{\alpha} \right)\cdot k - \frac{t}{\alpha} \right) = \beta k - t.\]

 Let $Z \coloneqq X \cup Y$ as in Step~\ref{int:T} of Algorithm~\ref{algo:intermediary}.
 Then $Z \in \F$ and $|Z| = |X| + |Y| \leq t + \beta k - t = \beta k$.

 It follows that
 \begin{equation}
  \label{eq:sample_prob}
  \begin{aligned}
   \Pr&\left(\sample(U, |\OPT|,t,\alpha,\beta,\oracle_{\alpha,c}) \textnormal{ returns a set of size at most } \beta \cdot |\OPT|\right)\\
      &\geq \frac{1}{2}\cdot  \Pr\left(|X \cap \OPT| \geq \left( 1 - \frac{\beta}{\alpha} \right)\cdot k + \frac{t}{\alpha} \right) = \frac{1}{2} \cdot \p\left(n,|\OPT|,t,\left( 1 - \frac{\beta}{\alpha} \right)\cdot k + \frac{t}{\alpha}\right)
  \end{aligned}
 \end{equation}
 where $\p$ is the function defined in Equation~\eqref{eq:hyper}.
 
 Now, consider the execution of Algorithm~\ref{algo:final} with $U$ as its input and let $S$ be the set returned by Algorithm~\ref{algo:final}.
 It is easy to see that $S \in \F$ since Algorithm~\ref{algo:intermediary} always returns a set from $\F$.
 If $|\OPT| \geq \frac{n}{\beta}$ then $|S| \leq |U| \leq \beta \cdot |\OPT|$ and the algorithm returns an $\beta$-approximate solution as desired.
 So we may assume that $|\OPT| < \frac{n}{\beta}$.
 Consider the iteration of the for-loop in Step~\ref{amls:loop} of Algorithm~\ref{algo:final} in which $k = |\OPT|$.
 Using Equation \eqref{eq:sample_prob}, at least one of the calls to Algorithm~\ref{algo:intermediary} in this iteration returns a set of size at most $\beta \cdot |\OPT|$ with probability at least
 \[1 - \left(1 - \frac{1}{2}\cdot \p\left(n,k,t,\left(1 - \frac{\beta}{\alpha} \right)\cdot k + \frac{t}{\alpha}\right)\right)^{2 / \p\left(n,k,t,\left( 1 - \frac{\beta}{\alpha} \right)\cdot k + \frac{t}{\alpha}\right)}\geq 1- \exp\left(-1\right)>\frac{1}{2}.\]
 
 So the minimum cardinality set in $\sol$ (at the end of Algorithm \ref{algo:final}) has size at most $\beta \cdot |\OPT|$ with probability at least $\frac{1}{2}$.
 Hence, the set $S$ returned by the algorithm satisfies $|S| \leq \beta \cdot |\OPT|$ with probability at least $\frac{1}{2}$.
\end{proof}

Recall the definition of $f_{\CL,\beta}(n)$ from Equation~\ref{eq:fdef_intro}.

\begin{lemma}[Running time]
 \label{lem:runtime_f}
 $\amlsalgo_{\CL,\beta}$ (Algorithm~\ref{algo:final}) has cost $f_{\CL,\beta}(n) \cdot n^{\OO(1)}$.
\end{lemma}

\begin{proof}
 First consider Algorithm \ref{algo:intermediary}.
 In Line \ref{int:Aext}, the algorithm calls $\oracle_{\alpha,c}$ with parameter
 \[k - \left\lceil\left( 1 - \frac{\beta}{\alpha} \right)\cdot k + \frac{t}{\alpha}\right\rceil \leq k - \left( 1 - \frac{\beta}{\alpha} \right)\cdot k - \frac{t}{\alpha} = \frac{\beta k - t}{\alpha}.\]
 So this step incurs a cost of at most   $c^{\frac{\beta k - t}{\alpha}}= \exp\left( \frac{\beta k -t} {\alpha }\cdot \ln c \right)$.
 This means that Step \ref{amls:call_sample} of \Cref{algo:final} incurs a total cost of
 \[2\cdot \frac{\exp\left( \frac{\beta k -t} {\alpha }\cdot \ln c \right)} {\p\left( n,k, t, (1 - \frac{\beta}{\alpha}) \cdot k +\frac{t}{\alpha} \right)} .\]
 Since $t, \alpha,c$ are chosen to minimize this cost, we obtain that one iteration of the for-loop takes incurs a cost of 
 \[\min_{(\alpha,c) \in \CL}
 \min_{~t\in\left[M^*_{\alpha, \beta}k,\beta k\right] \cap \NN~} 2\cdot\frac{\exp\left( \frac{\beta k -t} {\alpha }\cdot \ln c \right)} {\p\left( n,k, t, (1 - \frac{\beta}{\alpha}) \cdot k +\frac{t}{\alpha} \right)}.\]
 As a result, every single iteration costs at most
 \[\max_{~k\in \left[0,\frac{n}{\beta}\right] \cap \NN~ } 
 \min_{(\alpha,c) \in \CL}
 \min_{~t\in\left[M_{\alpha, \beta}k,\beta k\right] \cap \NN~} 2\cdot \frac{\exp\left( \frac{\beta k -t} {\alpha }\cdot \ln c \right)} {\p\left( n,k, t, (1 - \frac{\beta}{\alpha}) \cdot k +\frac{t}{\alpha} \right)}.\]
 Since there are at most $n$ iterations of the for-loop, the entire algorithm has cost  $f_{\CL,\beta}(n) \cdot n^{\OO(1)}$, as desired.
\end{proof}

Now, \Cref{lem:amls_upper_bound_via_f} immediately follows from \Cref{lem:correctness,lem:runtime_f}.

\subsection{Derandomization}

Next, we prove \Cref{lem:deter_upper_bound_via_f}, i.e., we argue how to derandomize $\amlsalgo_{\CL,\beta}$.
Towards this end, the key notion is that of a set-intersection-family.

\begin{definition}
 Let $U$ be a universe of size $n$ and let $p,q,r \geq 1$ such that $n \geq p \geq r$ and $n - p + r \geq q \geq r$.
 A family $\CC \subseteq \binom{U}{q}$ is a \emph{$(n,p,q,r)$-set-intersection-family} if for every $T \in \binom{U}{p}$ there is some $X \in \CC$ such that $|T \cap X| \geq r$.
\end{definition}

The basic idea of the derandomization is, instead of repeatedly sampling a random set $X$ in Algorithm \ref{algo:intermediary}, to compute a suitable set-intersection-family $\CC$ and iterate over all its elements $X$.
Towards this, let us define
\[\kappa(n,p,q,r) \coloneqq \frac{\binom{n}{q}}{\binom{p}{r} \cdot \binom{n - p}{q - r}}.\]
The following theorem computes the desired set-intersection-family of small size.

\begin{theorem}[{\cite[Theorem 4.2]{EsmerKMNS22}}]
\label{thm:family}
 There is an algorithm that, given a set $U$ of size $n$ and numbers $p,q,r \geq 1$ such that $n \geq p \geq r$ and $n - p + r \geq q \geq r$,
 computes an $(n,p,q,r)$-set-intersection-family of size $\kappa(n,p,q,r)\cdot2^{o(n)}$ in time $\kappa(n,p,q,r)\cdot2^{o(n)}$. 
\end{theorem}

With the last theorem in hand, we are ready to prove \Cref{lem:deter_upper_bound_via_f}.
The updated deterministic algorithm is given in Algorithm \ref{algo:det}.
Observe that it receives deterministic extension oracles.

\begin{algorithm}
	\begin{algorithmic}[1]
		\Input A universe $U$ and a deterministic extension oracle $\oracle_{\alpha,c}$ for every $(\alpha,c) \in \CL$
		\State $\sol\leftarrow\emptyset$.
		\For{$k$ from $0$ to $\frac{n}{\beta}$}
		\State Find $(\alpha,c)\in \CL$ and $t\in\left[M^*_{\alpha, \beta}k,\beta k\right]\cap \NN$ which minimize
		$\left(\frac{c^{\frac{\beta k - t}{\alpha}}}{\p\left(n,k,t,\left( 1-\frac{\beta}{\alpha}\right)\cdot k+\frac{t}{\alpha}\right)}\right)$.
		\State Set $x \coloneqq (1 - \frac{\beta}{\alpha} )\cdot k + \frac{t}{\alpha}$.
		\State Find $y \in \{\lceil x\rceil,\dots,\min\{t,k\}\}$ for which $\kappa(n,k,t,y)$ is minimized.
		\State Compute a $(n,k,t,y)$-set-intersection-family $\CC$.
		\For{$X \in \CC$}
		\State $Y \gets \oracle_{\alpha,c}\left( U, X, k - \left\lceil x\right\rceil \right) $.
		\State $\sol \leftarrow \sol \cup \{X \cup Y\}$.
		\EndFor
		\EndFor
		\State {\bf Return} a minimum-sized set in $\sol$.
	\end{algorithmic}
	\caption{$\detamlsalgo_{\CL,\beta}$}\label{algo:det}
	\label{algo:derandomized}
\end{algorithm}

\begin{proof}[Proof of \Cref{lem:deter_upper_bound_via_f}]
 Let $U$ be a finite set system of size $n$ and $\F$ be a monotone set system of $U$.
 Also, for every $(\alpha,c)\in \CL$ let $\oracle_{\alpha,c}$ be a deterministic $\alpha$-extension oracle of $U$ and $\CF$.
 Let $\F$ be the implicit monotone subset family associated with $\CL$.
 Also, let $\OPT= \argmin_{S\in \CF} \abs{S}$ be a minimum size solution of the $\LSUB$ instance $U$ and $\CF$.

 Consider Algorithm \ref{algo:det}.
 First, observe that $\emptyset \neq \sol \subseteq \F$ by \Cref{def:random_extension_oracle}.
 If $|\OPT| \geq \frac{n}{\beta}$ then every set in $\sol$ is a valid $\beta$-approximation.
 So suppose $|\OPT| \leq \frac{n}{\beta}$ and consider the iteration in which $k = |\OPT|$.

 By definition of a set-intersection-family, there is some $X \in \CC$ such that $|X| = t$ and
 \[|\OPT \cap X| \geq y \geq \lceil x\rceil.\]
 Then
 \[|\OPT \setminus X| \leq k - \lceil x\rceil.\]
 Moreover, $(\OPT \setminus X) \cup X \in \F$ since $(\OPT \setminus X) \cup X = \OPT \cup X \supseteq \OPT$ as $\F$ is monotone by assumption.

 Since $\oracle_{\alpha,c}$ is an $\alpha$-extension oracle for $(U,\F)$, given the input $\left( U, X, k - \lceil x\rceil\right)$ it returns a set $Y$ such that $X\cup Y\in \CF$ and
 \[|Y| \leq \alpha \cdot \left(k - \lceil x\rceil \right) \leq \alpha \cdot \left(k - \left( 1 - \frac{\beta}{\alpha} \right)\cdot k - \frac{t}{\alpha} \right) = \beta k - t.\]
 So $|X \cup Y| \leq t + \beta k - t = \beta k$ which means that $\sol$ contains a solution set of size at $\beta \cdot |\OPT|$ as desired.

 It remains to analyse the cost of Algorithm \ref{algo:det}.
 Suppose $k \in \{0,\dots,\lfloor\frac{n}{\beta}\rfloor\}$.
 The algorithm computes a number $y \in \{\lceil x,\dots,\min\{t,k\}\rceil\}$ for which $\kappa(n,k,t,y)$ is minimized, i.e., $1/\kappa(n,k,t,y)$ is maximized.
 By \Cref{thm:family} we get that
 \[|\CC| = \kappa(n,k,t,y)\cdot 2^{o(n)} \leq n \cdot \frac{1}{\p\left(n,k,t,x\right)} \cdot 2^{o(n)}.\]
 Also note that the family $\CC$ can be computed within the same time bound.

 It follows that the execution of the inner for-loop requires cost
 \[\frac{\exp\left( \frac{\beta k -t} {\alpha }\cdot \ln c \right)} {\p\left( n,k, t, (1 - \frac{\beta}{\alpha}) \cdot k +\frac{t}{\alpha} \right)} \cdot 2^{o(n)}.\]
 Since $t, \alpha,c$ are chosen to minimize this cost, we obtain that one iteration of the outer for-loop incurs a cost of
 \[\min_{(\alpha,c) \in \CL}
 \min_{~t\in\left[M^*_{\alpha, \beta}k,\beta k\right] \cap \NN~} \frac{\exp\left( \frac{\beta k -t} {\alpha }\cdot \ln c \right)} {\p\left( n,k, t, (1 - \frac{\beta}{\alpha}) \cdot k +\frac{t}{\alpha} \right)} \cdot 2^{o(n)}.\]
 Note that all other steps before the computation of the set-intersection family can be done using polynomially many computation steps.
 As a result, every single iteration costs at most
 \[\max_{~k\in \left[0,\frac{n}{\beta}\right] \cap \NN~ }
 \min_{(\alpha,c) \in \CL}
 \min_{~t\in\left[M_{\alpha, \beta}k,\beta k\right] \cap \NN~} \frac{\exp\left( \frac{\beta k -t} {\alpha }\cdot \ln c \right)} {\p\left( n,k, t, (1 - \frac{\beta}{\alpha}) \cdot k +\frac{t}{\alpha} \right)} \cdot 2^{o(n)}.\]
 Since there are at most $n$ iterations of the outer for-loop, the entire algorithm has cost $f_{\CL,\beta}(n) \cdot 2^{o(n)}$ as desired.
\end{proof}

\section{Lower Bounds}
\label{sec:lower_bounds}
In this section we prove \Cref{thm:minimization_lower_bound,lem:exact_lower_bound_via_f}.
We also argue how to derive \Cref{lem:exact_lower_bound} from \Cref{lem:exact_lower_bound_via_f}.

We begin with the proof of \Cref{lem:exact_lower_bound_via_f} as it is technically simpler.
We actually prove the following slightly stronger statement.
Recall the definition of $f_{\CL,\beta}(n)$ \eqref{eq:fdef_intro}.

\begin{lemma}
	\label{lem:exact_lower_bound_aux}
	Let $c>1$ and let $\A$ be an algorithm for $\cDEC$.
	Then $\cost_\A(n) \geq \frac{1}{2} \cdot f_{\{(1,c)\},1}(n)$ for every $n \geq 1$.
\end{lemma}

\begin{proof}
	Let $n \in \NN$.
	We assume $n$ is fixed throughout this proof.
	By \eqref{eq:fdef_intro} it holds that
	\begin{equation}
		\label{eq:f_simplified_emls}
		f_{\{(1,c)\},1}(n) = \max_{~k\in \left[0,n\right] \cap \NN~ }
		\min_{~t\in\left[0,k\right] \cap \NN~}
		\frac{\exp\left(  (k -t)\cdot \ln c \right)} {\p\left( n,k, t, t \right)} =\max_{~k\in \left[0,n\right] \cap \NN~ }
		\min_{~t\in\left[0,k\right] \cap \NN~} c^{k-t} \cdot \frac{ \binom{n}{k}}{\binom{n-t}{k-t}}.
	\end{equation}
	The last equality uses $\p\left(n,k,t,t\right) = \frac{\binom{n-t}{k-t}}{\binom{n}{k}}$ by \eqref{eq:hyper}.

	For all $0 \leq k \leq n$ and $0 \leq t \leq k$ we define
	$$
	\begin{aligned}
		&G(k,t) &\coloneqq& ~\frac{c^{k-t} \cdot \binom{n}{k}}{ \binom{n-t}{k-t}},\\
		&t^*(k) &\coloneqq& ~\argmin_{t\in [0,k]\cap \mathbb{N}} G(k,t),\\
		&k^*    &\coloneqq& ~\argmax_{k\in [0, n]\cap \mathbb{N}} G(k,t^*(k)).
	\end{aligned}
	$$
	By \eqref{eq:f_simplified_emls} it follows that $f_{\{(1,c)\},1}(n) = G(k^*,t^*(k^*))$.
	Furthermore, it holds that
	\begin{equation}
		\label{eq:exact_lb_tstar}
		t^*(k) = \argmin_{t\in [0,k]\cap \mathbb{N}} \frac{c^{k-t} \cdot \binom{n}{k}}{ \binom{n-t}{k-t}} = \argmin_{t\in [0,k]\cap \mathbb{N}} \frac{c^{k-t} }{ \binom{n-t}{k-t}}.
	\end{equation}

	We set $U \coloneqq [n]$.
	Our lower bound is based on the difficulty that algorithms have to distinguish between the set-systems $\F=\emptyset$ and $\F=\{R\}$ where $R$ is a uniformly sampled random subset of $U$ of size $k^*$.

	For a set system $\F$ of $U$ we define an oracle $\oracle_{\F}$ by setting
	$$
		\oracle_{\F}(X,\ell) \coloneqq
		\begin{cases}
			\yes & \textnormal {if there is an $\ell$-extension of $X$ with respect to $U$ and $\F$,}\\
			\no  & \textnormal {otherwise.}
		\end{cases}
	$$
	Clearly, $\oracle_{\F}$ is a exact extension oracle for $\F$.
	Note that $\oracle_{\emptyset}$ always returns $\no$.

	We assume the algorithm $\A$ gets a string of bits $b \in \{0,1\}^{q(n)}$ as its source of randomness, where $q$ is an arbitrary function.
	This means $\A$ is deterministic given the input set $U$, the oracle $\oracle$ and the random bits $b$.
	Let us denote by $\A(U,\oracle,b) \in \{\yes,\no\}$ the output of the algorithm $\A$.

	Let $Q(b) \subseteq U \times \mathbb{N}$ be the set of oracle queries the algorithm $\A$ makes on input $U$ with oracle $\oracle_{\emptyset}$ and random bits~$b$.
	Equivalently, $Q(b)$ is the set of queries $\A$ makes given the universe $U$ in case the oracle always returns $\no$ for an answer.
	Observe that, in general, if all the responses to the queries the algorithm makes are $\no$, then it has to return $\no$, because otherwise it violates the correctness requirement in case its given the oracle  $\oracle_{\emptyset}$ for the set system $\emptyset$.

	We define the \emph{coverage} of an oracle query $(X,\ell)$ by
	$$
	\begin{aligned}
		\coverage(X,\ell) ~&\coloneqq~ \Big\{S\subseteq U \Bigmid |S|=k^*,~X\subseteq S,~|S\setminus X|\leq \ell\Big\}\\
		&=~ \Big\{ S\subseteq U \Bigmid |S|=k^* \textnormal{ and $X$ has an $\ell$-extension w.r.t.\ the set system $\{S\}$}\Big\}\\
		&=~\Big\{ S\subseteq U \Bigmid \oracle_{\{S\}}(X,\ell)=\yes\Big\}.
	\end{aligned}
	$$
	Given a set $W\subseteq 2^{U}\times \mathbb{N}$ of queries we define $\coverage(W) \coloneqq \bigcup_{(X,\ell)\in W} \coverage(X,\ell)$.

	\begin{claim}
		\label{claim:exact_lb_cost_to_coverage}
		Let $b \in \{0,1\}^{q(n)}$. Then
		\[\cost_{\A} (n)\geq \frac{\abs{ \coverage(Q(b))} }{\binom{n}{k^*}} \cdot f_{\{(1,c)\},1}(n).\]
	\end{claim}

	\begin{claimproof}
		Consider the execution of $\A$ on input $U$ using the oracle $\oracle_{\emptyset}$ and random bits $b$.
		By definition, the cost of the execution is $\sum_{(X,\ell)\in Q(b)} c^\ell$ and thus, $\cost_\A(n) \geq \sum_{(X,\ell)\in Q(b)} c^\ell$.
		Therefore,
		$$
		\begin{aligned}
			\abs{ \coverage(Q(b))}&\leq
			\sum_{(X,\ell)\in Q(b)} \abs{\coverage(X,\ell)} \\
			&= \sum_{(X,\ell)\in Q(b)\textnormal{ s.t.\ } k^*-\ell\leq |X|\leq k^* }  \binom{n-|X|}{ k^*-|X|}\\
			&= \sum_{(X,\ell)\in Q(b)\textnormal{ s.t.\ } k^*-\ell\leq |X|\leq k^* }  c^{k^*-|X|} \cdot \frac{ \binom{n-|X|}{ k^*-|X|}}{ c^{k^*-|X|} }\\
			&\leq \sum_{(X,\ell)\in Q(b)\textnormal{ s.t.\ } k^*-\ell\leq |X|\leq k^* }  c^{k^*-|X|} \cdot \frac{ \binom{n-t^*(k^*)}{ k^*-t^*(k^*)}}{ c^{k^*-t^*(k^*)} }\\
			&= \frac{\binom{n}{k^*}}{G(k^*,t^*(k^*))} \cdot
			\sum_{(X,\ell)\in Q(b)\textnormal{ s.t.\ } k^*-\ell\leq |X|\leq k^*  }   c^{k^*-|X|} \\
			&\leq \frac{\binom{n}{k^*}}{G(k^*,t^*(k^*))} \cdot
			\sum_{(X,\ell)\in Q(b) }  c^{\ell}\\
			&\leq \frac{\binom{n}{k^*}}{G(k^*,t^*(k^*))} \cdot \cost_{\A} (n).
		\end{aligned}
		$$
		The second inequality follows from \eqref{eq:exact_lb_tstar}.
		Since $f_{\{(1,c)\},1}(n) = G(k^*, t^*(k^*))$, the assertion of the claim follows.
	\end{claimproof}

	Now let $b^* \in \{0,1\}^{q(n)}$ be the bit-string for which $\abs{\coverage(Q(b^*))}$ is maximal.
	In light of the last claim, in order to lower bound the cost of $\A$, it suffices to lower bound the cardinality of $\coverage(Q(b^*))$.
	We use the correctness properties of $\A$ to attain such a lower bound.

	\begin{claim}
		\label{claim:exact_lb_cover}
		It holds that
		\[\abs{\coverage(Q(b^*))} \geq \frac{1}{2} \cdot \binom{n}{k^*}.\]
	\end{claim}

	\begin{claimproof}
		Consider the execution of $\A$ on input $U$ using the oracle $\oracle_{\{S\}}$, where $S\subseteq U$ such that $|S|=k^*$, and a bit-string~$b$.
		If $S \not\in \coverage(Q(b))$ then the set of oracle queries the algorithm makes is exactly $Q(b)$ and all the queries return $\no$.
		So the algorithm also has to return $\no$.
		It follows that
		\[\A(U,\oracle_{\{S\}},b) = \yes \implies S \in \coverage(Q(b)).\]

		We define two independent random variable.
		Let $R \subseteq U$ be a uniformly random subset of $U$ of size $k^*$.
		Also, we define $r\in \{0,1\}^{q(n)}$ to be a uniformly random string of bits of length $q(n)$.
		Then
		\begin{equation}\label{eq:lower_exact_first}
		\begin{aligned}
			\Pr\left(\A\left(U,\oracle_{\{S\}},r\right) = \yes\right) &\leq~\Pr(R \in \coverage(Q(r))) \\
			&=~ \sum_{b\in \{0,1\}^{q(n)}} \Pr(r=b) \cdot \Pr(R \in \coverage(Q(r)) \mid r=b) \\
			&=~ \sum_{b\in \{0,1\}^{q(n)}} \Pr(r=b) \cdot \Pr(R \in \coverage(Q(b))) \\
			&=~ \sum_{b\in \{0,1\}^{q(n)}} \Pr(r=b) \cdot \frac{\abs{\coverage(Q(b))}}{\binom{n}{k^*}} \\
			&\leq~ \sum_{b\in \{0,1\}^{q(n)}} \Pr(r=b) \cdot \frac{\abs{\coverage(Q(b^*))}}{\binom{n}{k^*}}\\
			&=~ \frac{\abs{\coverage(Q(b^*))}}{\binom{n}{k^*}}.
		\end{aligned}
		\end{equation}
		The second equality holds since $r$ is independent of $R$.
		Furthermore, as $\A$ returns $\yes$ with probability at least $\frac{1}{2}$ for $\yes$ instances,
		\begin{equation}\label{eq:lower_exact_second}
		\begin{aligned}
			\Pr(\A\left(U,\oracle_{\{R\}},r\right) = \yes)
			&= \sum_{S\subseteq U} \Pr(R=S) \cdot \Pr(\A\left(U,\oracle_{\{R\}},r\right)) = \yes \mid R=S) \\
			&\geq \sum_{S\subseteq U} \Pr(S=R) \cdot \frac{1}{2} = \frac{1}{2}.
		\end{aligned}
		\end{equation}
		The assertion of the claim now follows by combining \eqref{eq:lower_exact_first} and \eqref{eq:lower_exact_second}.
	\end{claimproof}

	Combining \Cref{claim:exact_lb_cost_to_coverage,claim:exact_lb_cover} we get $\cost_{\A} (n) \geq \frac{1}{2} \cdot f_{\{(1,c)\},1}(n)$, which completes the proof.
\end{proof}

Note that \Cref{lem:exact_lower_bound_via_f} immediately follows from \Cref{lem:exact_lower_bound_aux}.
Next, we prove \Cref{lem:exact_lower_bound} using \Cref{lem:exact_lower_bound_via_f}.

\begin{proof}[Proof of \Cref{lem:exact_lower_bound}]
	Let $c>1$ and let $\A$ be a randomized algorithm for $\cDEC$.
	By \Cref{lem:exact_lower_bound_via_f} we have
	\begin{equation}
		\label{eq:exact_lower_bound_to_amlsbound}
		\cost_\A(n) \geq n^{-\OO(1)} \cdot f_{\{(1,c)\},1}(n) \geq n^{-\OO(1)} \cdot \left(\amlsbound(\{(1,c)\},1)\right)^n,
	\end{equation}
	where second equality follows from \Cref{lem:f_to_amls}.

	Observe that $2-\frac{1}{c} \in (1,{c+1})$ and
	\[\D{1}{\frac{(2-\frac{1}{c}) - 1}{c-1} } =  \D{1}{\frac{1-\frac{1}{c}}{c-1}} =\D{1}{\frac{1}{c}} = \ln c.\]
	Therefore, by \Cref{lem:coincide_with_esa}, it holds that $\amlsbound(\{(1,c)\},1) = 2-\frac{1}{c}$.
	In combination with \eqref{eq:exact_lower_bound_to_amlsbound} it follows that $\cost_\A(n) \geq n^{-\OO(1)} \cdot \left(2-\frac{1}{c}\right)^n$.
\end{proof} 

The proof of \Cref{thm:minimization_lower_bound} follows the same principles as the proof of \Cref{lem:exact_lower_bound_via_f}.
It defines a \emph{coverage} for each query, shows the cost of the algorithm is at least the cardinality of the coverage of all queries, and then provides a lower bound on the cardinality of the coverage.
The proof is slightly more complicated than the proof of \Cref{lem:exact_lower_bound_via_f} due to the involvement of multiple oracles and since the oracles only provide approximations.
As before, we actually prove a slightly stronger statement.

\begin{lemma}
	\label{lem:minimization_lower_bound_aux}
	Let $\beta \geq 1$.
	Also let $\CL$ be a specification list and $\A$ be a randomized $\beta$-approximation algorithm for $\LSUB$.
	Then
	\[\cost_\A(n+1) \geq \frac{f_{\CL, \beta}(n)}{2 \cdot (n+1) \cdot \max_{(\alpha,c)\in \CL } c}\]
	for every $n \geq 1$.
\end{lemma}

\begin{proof}
	Let $n \geq 1$.
	We assume $n$ is fixed for the remainder of the proof.
	Also suppose that $\CL = \{(\alpha_1,c_1),\ldots,(\alpha_s,c_s)\}$ and define $U \coloneqq [n+1]$.
	For every $j\in [s]$ we define the functions
	\begin{alignat}{3}
		&G_j(k,t) &&\coloneqq ~\frac{\left(c_j\right)^{\frac{\beta k - t}{\alpha_j}}}
		{\p\left( n,k,t,k - \frac{\beta k  - t}{\alpha_j} \right)} \label{eq:lb_Gj_def}, \\
		\nonumber\\
		&t_j^*(k) &&\coloneqq \argmin_{t\in \left[M^*_{\alpha_j, \beta} \cdot k ,  ~\beta \cdot k\right] \cap \mathbb{N}} G_j(k,t),
		\label{eq:lb_tj_def}
	\end{alignat}
	where $M^*_{\alpha_j,\beta}$ is as defined in \eqref{eq:Mstar_def}.
	Furthermore, we define
	$$j^*(k) = \argmin_{j\in [s]} G_j(k,t^*_j(k))$$
	and
	$$k^* = \argmax_{k\in [0, \frac{n}{\beta}]\cap \mathbb{N}} G_{j^*(k)}(k,t^*_{j^*(k)}(k)).$$
	With a slight abuse of notation we write $j^* \coloneqq j^*(k^*)$ and $t^* \coloneqq t^*_{j^*}(k^*)$.
	By \eqref{eq:fdef_intro} we have
	$$f_{\CL,\beta}(n) =G_{j^*}(k^*,t^*).$$

	We define the set system $\Fadv \coloneqq \{ S\subseteq U \mid |S| \geq \left\lfloor \beta \cdot k^{*} \right\rfloor + 1\}$.
	Since $\beta \cdot k^* \leq n$ and $U = [n+1]$ it holds that $\Fadv \neq \emptyset$.
	For every $T \subseteq U$ we define the set system  $\F_T = \{ S\subseteq U \mid T \subseteq S \} \cup \Fadv$, i.e., $\F_T$ is the set system containing all supersets of $T$ and all sets of cardinality at least $\left\lfloor \beta \cdot k^{*} \right\rfloor + 1$.
	Clearly, $\Fadv$ and $\F_{T}$ are monotone set systems of $U$ for every $T \subseteq U$.
	Our lower bound is based on the fact that $\A$ requires queries of high total cost to distinguish between $\Fadv$ and $\F_{T}$ when $T$ is a random set.

	For every $X \subseteq U$ we fix an arbitrary set $Q_X \subseteq U\setminus X$ such that $\abs{Q_X}= \max\left\{\floor{\beta \cdot k^*} +1 - \abs{X},~0\right\}$.
	Consequently, it holds that $X\cup Q_X \in \Fadv$.
	We define an extension oracle $\oracle_{\adv}$ via $\oracle_{\adv}(X,\ell) \coloneqq Q_X$ for all $X \subseteq U$ and $\ell \geq 0$.
	\begin{claim}
		For every $j\in [s]$ it holds that $\oracle_{\adv}$ is an $\alpha_j$-extension oracle for $U$ and $\Fadv$.
	\end{claim}
	\begin{claimproof}
		Let $X\subseteq U$ and $\ell\in \mathbb{N}$.
		It holds that $X\cup \oracle_{\adv}(X) =X\cup Q_X \in \Fadv$.
		Furthermore, if there is an $\ell$-extension $S$ of $X$, then $\abs{X}\geq \floor{\beta k^*}+1-\ell$.
		Hence, $\abs{Q_X} \leq \floor{\beta \cdot k^*} +1-\abs{X} \leq \ell \leq \alpha_j \cdot\ell$.
		That is, $Q_X$ is an $(\alpha_j\cdot \ell)$ extension of $X$.
		So $\oracle_{\adv}$ is an $\alpha_j$-extension oracle.
	\end{claimproof}

	We define the {\em coverage} of an oracle query $(X,\ell)$ to the $j$-th oracle (i.e., to the $\alpha_j$-extension oracle) by
	\begin{equation}
		\label{eq:cover_def}
		\coverage_j(X,\ell) \coloneqq
		\begin{cases}
			\emptyset & |X|+\alpha_j \cdot \ell \geq \left\lfloor \beta \cdot k^{*} \right\rfloor + 1,\\
			\{S\subseteq U \mid |S|=k^*, |S\setminus X| \leq \ell\} &\textnormal{otherwise.}
		\end{cases}
	\end{equation}
	Intuitively speaking, $\coverage_j(X,\ell)$ contains all subsets $\subseteq U$ of cardinality $k^*$ such that a deterministic $\alpha_j$-extension oracle for $\F_S$ cannot return $\oracle_\adv(X,\ell)$ to the query $(X,\ell)$.
	Given a set $W\subseteq 2^{U}\times \mathbb{N}$ of queries we define $\coverage_j(W) \coloneqq \bigcup_{(X,\ell)\in W} \coverage_j(X,\ell)$.

	\begin{claim}
		\label{claim:coverage_basic}
		Let $j\in [s]$ and suppose $(X,\ell)\in 2^U \times \NN$ such that $\coverage_j(X,\ell) \neq \emptyset$.
		Then
		$$n-\ell \geq \floor{\beta\cdot k^* - \alpha_j \cdot \ell} \geq \abs{X} \geq M^*_{\alpha_j,\beta} \cdot k^*.$$
	\end{claim}
	\begin{claimproof}
		Since  $\coverage_j(X,\ell) \neq \emptyset$ it holds that $|X|+\alpha_j \cdot \ell <  \floor{ \beta \cdot k^{*}}+ 1$.
		Therefore
		$$|X| < \floor{ \beta \cdot k^{*}}+ 1 - \alpha_j \cdot \ell \leq \beta \cdot k^* -\alpha_j\cdot \ell \leq n-\ell,$$
		which also implies $|X| \leq \floor{ \beta \cdot k^* -\alpha_j\cdot \ell}$.
		If $\alpha_j \leq \beta$ then $M^*_{\alpha_j,\beta} = 0$, and the trivial inequality $\abs{X} \geq 0$ completes the proof of the claim.
		So we assume that $\alpha_j > \beta \geq 1$ for the remainder of the proof.

		Using $\coverage_j(X,\ell) \neq \emptyset$ once more, by \eqref{eq:cover_def} there is some $S\subseteq U$ such that $\abs{S}=k^*$ and $\abs{S\setminus X} \leq \ell$.
		So $\abs{X} \geq k^*-\ell$ and we get that
		\begin{align*}
			\floor{ \beta \cdot k^* -\alpha_j\cdot \ell} &\geq~ \abs{X}\\
			&\geq~ \frac{\alpha_j}{\alpha_j -1}\cdot \abs{X} - \frac{1}{\alpha_j-1}\cdot \abs{X} \\
			&\geq~  \frac{\alpha_j}{\alpha_j -1} \cdot \left( k^* -  \ell\right)  - \frac{1}{\alpha_j-1} \cdot \left( \beta \cdot k^* -\alpha_j\cdot \ell\right) \\
			&=~ \frac{\alpha_j-\beta}{\alpha_j-1}\cdot k^*\\
			&=~ M_{\alpha_j,\beta}\cdot k^*.
		\end{align*}
	\end{claimproof}

	We assume the input for $\A$ consists of the set $U$, a collection of $s$ extension oracles $\oracle_1,\ldots,\oracle_{s}$, where $\oracle_j$ is an $\alpha_j$-extension oracle, and a bit-string $b\in \{0,1\}^{q(n+1)}$ where $q$ is an arbitrary function.
	The bit-string $b$ serves as the source of randomness, and we assume the algorithm is deterministic given $U$, the oracles $\oracle_1,\ldots,\oracle_{s}$ and $b$.
	We write $\A(U,\oracle_1,\ldots, \oracle_s, b) \subseteq U$ to denote the output of $\A$ given $U$, the oracles $\oracle_1,\ldots,\oracle_{s}$ and $b$.

	For every $b\in \{0,1\}^{q(n+1)}$ let $Q_j(b)\subseteq 2^U \times \NN$ be the set of queries $\A$ makes to the $j$-th oracle given the input $U$, $\oracle_{\adv}$ as the $j$-th oracle for every $j\in [s]$, and $b$.
	\begin{claim}
		\label{claim:cost_to_coverage}
		Let $b\in \{0,1\}^{q(n+1)}$.
		Then
		$$\cost_\A(n+1) \geq f_{\CL,\beta}(n) \cdot \frac{\sum_{j\in [s]} \abs {\coverage_j(Q_j(b))}}{(n+1) \cdot \binom{n}{k^*} \cdot \max_{j \in [s]} c_j}.$$
	\end{claim}
	\begin{claimproof}
		For every $j\in [s]$ it holds that
		\begin{equation}
		\label{eq:coverage_one}
		\begin{aligned}
			\big| & \coverage_j(Q_j(b)) \big| \leq \sum_{(X,\ell) \in Q_j(b)} |\coverage_j(X,\ell)|  \\
			&= \sum_{\substack{(X,\ell) \in Q_j(b) \text{ s.t.} \\ \coverage_j(X,\ell) \neq \emptyset}}\big|  \coverage_j(X,\ell) \big|\\
			&= \sum_{\substack{(X,\ell) \in Q_j(b) \text{ s.t.} \\ \coverage_j(X,\ell) \neq \emptyset }}~\sum_{y= k^* - \ell}^{k^*} \binom{|X|}{y} \cdot \binom{n+1-|X|}{k^*-y}\\
			&\leq (n+1)\cdot \sum_{\substack{(X,\ell) \in Q_j(b) \text{ s.t.} \\ \coverage_j(X,\ell) \neq \emptyset }}~\sum_{y= k^* - \ell}^{k^*} \binom{|X|}{y} \cdot \binom{n-|X|}{k^*-y}.\\
		\end{aligned}
		\end{equation}
		The second equality follows from a simple counting argument and the definition of $\coverage_j$ in \eqref{eq:cover_def}.
		The second inequality holds since $\binom{m+1}{r}\leq (m+1)\cdot \binom{m}{r}$ if $m\geq r\geq 0$.
		Observe that $n-|X| \geq \ell \geq k^* - y$ for every $(X,\ell) \in Q_j(b)$ such that $\coverage_j(X,\ell) \neq \emptyset$, and every $k^*\leq y \leq k^*-\ell$ using \Cref{claim:coverage_basic}.

		Plugging the formula for $\p$ from~\eqref{eq:hyper} into \eqref{eq:coverage_one} we obtain
		\begin{equation}
		\label{eq:coverage_two}
		\begin{aligned}
			\big| & \coverage_j(Q_j(b)) \big| \leq (n+1)\cdot\binom{n}{k^*} \sum_{\substack{(X,\ell) \in Q_j(b) \text{ s.t.} \\ \coverage_j(X,\ell) \neq \emptyset}} \p(n, k^*,|X|, k^*-\ell)\\
			&\leq ~(n+1)\cdot \binom{n}{k^*} \sum_{\substack{(X,\ell) \in Q_j(b) \text{ s.t.} \\ \coverage_j(X,\ell) \neq \emptyset}} \p\left(n,  k^*, \floor{\beta k^{*} - \alpha_j \ell }, k^*-\ell\right) \\
			&\leq ~(n+1)\cdot \binom{n}{k^*} \sum_{\substack{(X,\ell) \in Q_j(b) \text{ s.t.} \\ \coverage_j(X,\ell) \neq \emptyset}} \p\left(n,  k^*, \floor{\beta k^{*} - \alpha_j \ell }, k^*- \frac{\beta k^* - \floor{\beta k^{*} - \alpha_j \ell }}{\alpha_j}\right) \\
			&= ~(n+1)\cdot \binom{n}{k^*} \sum_{\substack{(X,\ell) \in Q_j(b) \text{ s.t.} \\ \coverage_j(X,\ell) \neq \emptyset}} \left( c_j\right)^{\frac{\beta k^* - \floor{\beta k^{*} - \alpha_j \ell }}{\alpha_j}} \cdot \frac{ 1 }{  G_j\left(k^*, \floor{\beta k^{*} - \alpha_j \ell } \right)}.
		\end{aligned}
		\end{equation}
		The second inequality holds since $\abs{X} \geq \floor{ \beta \cdot k^* -\alpha_j\ell}$ by \Cref{claim:coverage_basic}.
		The third inequality follows from $-\ell \geq -\frac{\beta k^* - \floor{\beta k^* -\alpha_j\cdot \ell }}{\alpha_j}$.
		The last equality follows from the definition of $G_j$ in \eqref{eq:lb_Gj_def}.

		By \Cref{claim:coverage_basic} it holds that $M^*_{\alpha_j,\beta}\cdot k^*\leq \floor{\beta \cdot k^*-\alpha_j\cdot \ell} \leq \beta\cdot k^*$ for every $(X,\ell)\in 2^U\times \NN$ for which $\coverage_j(X,\ell) \neq \emptyset$.
		Hence,
		\begin{equation}
			\label{eq:Gj_to_f}
			G_j(k^*, \floor{\beta \cdot k^*-\alpha_j\cdot \ell}) ~\geq~ G_j(k^*, t^*_j(k^*)) ~\geq~ G_{j^*}(k^*, t^*) ~=~ f_{\CL,\beta}(n)
		\end{equation}
		for every $(X,\ell) \in 2^U\times \NN$ such that $\coverage_j(X,\ell) \neq \emptyset$.
		Combining \eqref{eq:coverage_two} with \eqref{eq:Gj_to_f} we have
		$$
		\begin{aligned}
				\big| \coverage_j(Q_j(b)) \big|
				&\leq ~(n+1) \cdot \binom{n}{k^*} \sum_{\substack{(X,\ell) \in Q_j(b) \text{ s.t.} \\ \coverage_j(X,\ell) \neq \emptyset}} \left( c_j\right)^{\frac{\beta k^* - \floor{\beta k^{*} - \alpha_j \ell }}{\alpha_j}} \cdot \frac{ 1 }{ f_{\CL,\beta}(n)} \\
				&\leq ~(n+1) \cdot \binom{n}{k^*} \sum_{(X,\ell) \in Q_j(b)} \left(c_j\right)^{\ell+1} \cdot \frac{1}{f_{\CL,\beta}(n)}.
		\end{aligned}
		$$
		Since the last inequality holds for every $j\in [s]$ we get
		$$\sum_{j\in [s] } \abs {\coverage_j(Q_j(b))} \leq ~ (n+1)\cdot  \binom{n}{k^*}\cdot \left( \max_{j\in [s]} c_j \right)\cdot \sum_{j\in [s]}\sum_{(X,\ell) \in Q_j(b) } \left( c_j\right)^{\ell} \cdot \frac{ 1 }{ f_{\CL,\beta}(n)}$$
		and thus,
		$$\cost_\A(n+1) ~\geq ~\sum_{j\in [s]}\sum_{(X,\ell) \in Q_j(b) } \left( c_j\right)^{\ell}~ \geq ~ f_{\CL, \beta} (n )\cdot \frac{\sum_{j\in [s]} \abs {\coverage_j(Q_j(b))} }{ \left( \max_{j\in [s]} c_j \right)\cdot (n+1)\cdot \binom{n}{k*}}.\qedhere$$
	\end{claimproof}

	By \Cref{claim:cost_to_coverage}, in order to lower bound the cost of $\A$, we only need to provide a lower bound on~$\sum_{j\in [s]} \abs{\coverage_j(Q_j(b))}$ for some $b \in \{0,1\}^{q(n+1)}$.
	For every $T\subseteq U$ such that $\abs{T}=k^*$ and every $j \in [s]$ we define an $\alpha_j$-extension oracle $\oracle_{T,j}$ for $U$ and $\F_{T}$ by
	$$\oracle_{T,j}(X,\ell) \coloneqq
	\begin{cases}
		T \setminus X & T \in \coverage_j(X,\ell),\\
		\oracle_{\adv}(X,\ell) & \textnormal{otherwise.}
	\end{cases}$$

	\begin{claim}
		\label{claim:Tj_oracle_correct}
		For every $T\subseteq U$ such that $\abs{T} = k^*$ and $j \in [s]$ it holds that $\oracle_{T,j}$ is an $\alpha_j$-extension oracle for $U$ and the set system $\F_T$.
	\end{claim}
	\begin{claimproof}
		Let $(X,\ell) \in 2^U \times \NN$.
		If $T\in \coverage_j(X,\ell)$ then $X\cup \oracle_{T,j}(X,\ell) = X \cup (T\setminus X) = X\cup T\in \F_T$.
		Otherwise $X\cup \oracle_{T,j}(X,\ell) = X\cup \oracle_{\adv}(X,\ell) \in \Fadv \subseteq \F_T$.
		That is, $X \cup \oracle_{T,j}(X,\ell) \in \F_T$ in all cases.

		Suppose $X$ has an $\ell$-extension $S$ with respect to $U$ and $\F_T$.
		To complete the proof we need to show that  $\abs{\oracle_{T,j}(X,\ell)} \leq \alpha_j\cdot \ell$.
		We distinguish the following two cases.
		\begin{itemize}
			\item If $\abs{X}+\alpha_j\cdot \ell \geq \floor{\beta k^*}+1$ then $\coverage_j(X,\ell) =\emptyset$ by \eqref{eq:cover_def}, and thus
			$$\abs{\oracle_{T,j}(X,\ell)} ~=~ \abs{\oracle_{\adv}(X,\ell)} ~=~ \max\left\{ \floor{\beta\cdot k^*} +1-\abs{X}, 0 \right\} ~\leq~ \max\{\alpha_j\cdot \ell,0\} ~\leq~ \alpha_j\cdot \ell.$$
			\item Otherwise $\abs{X}+\alpha_j\cdot \ell < \floor{\beta k^*}+1$ and we have that $\abs{S\cup X} \leq \abs{X} + \ell < \floor{\beta k^*} +1$.
			This means $S \cup X \notin \Fadv$.
			Since $S\cup X \in \F_T$, we conclude that $T \subseteq S\cup X$.
			So $T\setminus X\subseteq S$ and $\abs{T\setminus X} \leq \abs{S}\leq \ell$.
			Since $\abs{T}=k^*$ we conclude that $T \in \coverage_j(X,\ell)$.
			It follows that $\oracle_{T,j}(X,\ell)=T\setminus X$, and
			$$\abs{\oracle_{T,j}(X,\ell)} = \abs{T\setminus X} \leq \abs{S} \leq \ell \leq \alpha_j \cdot \ell.\qedhere$$
		\end{itemize}
	\end{claimproof}

	Now, let $b^*\in \{0,1\}^{q(n+1)}$ be the bit-string for which $\sum_{j\in [s]}\abs{\coverage_j(Q_j(b^*))}$ is maximal.
	\begin{claim}
		\label{claim:approximate_coverage_lb}
		It holds that
		\begin{equation*}
			\sum_{j\in [s]}\abs{\coverage_j(Q_j(b^*))} \geq \frac{1}{2} \cdot \binom{n+1}{k^*}.
		\end{equation*}
	\end{claim}
	\begin{claimproof}
		Consider the execution of $\A$ with the universe $U$, the oracles $\oracle_{T,1},\ldots, \oracle_{T,s}$ and the bit-string $b$, where $T \subseteq U$ and $|T| =k^*$.
		Unless $T \in \bigcup_{j\in [s]} \coverage_j(Q_j(b))$ the execution is identical to the execution of $\A$ with the universe $U$, the oracles $\oracle_{\adv},\ldots,\oracle_{\adv}$ and $b$.
		Hence, $\A$ has to return a set $S \in \Fadv$ (otherwise it violates the correctness requirement for the latter execution), and thus $|S| \geq \floor{\beta k^*}+1$.
		It follows that
		\begin{equation}
			\label{eq:lb_success_cond}
			\abs{\A(U,\oracle_{T,1},\ldots, \oracle_{T,s},b)} \leq  \beta k^* \implies T \in \bigcup_{j\in [s]}\coverage_j(Q_j(b)).
		\end{equation}

		We define two independent random variables.
		Let $ R\subseteq U$ be a uniformly distributed random set of cardinality $k^*$, and let $r \in \{0,1\}^{q(n+1)}$ be a uniformly distributed bit-string.
		Using \eqref{eq:lb_success_cond} we get
		\begin{equation}
		\label{eq:lb_prob_upper}
		\begin{aligned}
			\Pr&\left(\abs{\A(U,\oracle_{R,1},\ldots, \oracle_{R,s},r)}\leq  \beta k^*\right) ~\leq~ \Pr\left(R \in \bigcup_{j\in [s]}\coverage_j(Q_j(r))\right) \\
			&=~ \sum_{b\in \{0,1\}^{q(n+1)}} \Pr(r=b) \cdot \Pr\left(R \in \bigcup_{j\in [s]}\coverage_j(Q_j(r)) ~\middle|~b=r \right)\\
			&=~ \sum_{b\in \{0,1\}^{q(n+1)}} \Pr(r=b) \cdot \Pr\left(R\in \bigcup_{j\in [s]}\coverage_j(Q_j(b))\right)\\
			&\leq~ \sum_{b\in \{0,1\}^{q(n+1)}} \Pr(r=b) \cdot \frac{\sum_{j\in [s]} \abs{\coverage_j(Q_j(b))}}{\binom{n+1}{k^*}}\\
			&\leq~ \sum_{b\in \{0,1\}^{q(n+1)}} \Pr(r=b) \cdot \frac{\sum_{j\in [s]} \abs{\coverage_j(Q_j(b^*))}}{\binom{n+1}{k^*}}\\
			&=~ \frac{\sum_{j\in [s]} \abs{\coverage_j(Q_j(b^*))}}{\binom{n+1}{k^*}}.
		\end{aligned}
		\end{equation}
		The second equality holds since $r$ is independent of $R$.
		Furthermore,
		\begin{equation}
		\label{eq:lb_prob_lower}
		\begin{aligned}
			\Pr&\left(\abs{\A(U,\oracle_{R,1},\ldots, \oracle_{R,s},r)}\leq  \beta k^*\right)  \\
			&= \sum_{\substack{T\subseteq U \textnormal{ s.t.\ } |T|=k^*}} \Pr(R=T)\cdot\Pr\left(\abs{\A(U,\oracle_{R,1},\ldots, \oracle_{R,s},r)}\leq \beta k^*~\middle|~R=T\right)\\
			&= \sum_{\substack{T\subseteq U\textnormal{ s.t.\ } |T|=k^*}}\Pr(R=T) \cdot\Pr\left(\abs{\A(U,\oracle_{T,1},\ldots, \oracle_{T,s},r)}\leq \beta k^*\right)\\*
			&\geq \sum_{\substack{T\subseteq U\textnormal{ s.t.\ } |T|=k^*}}\Pr(R=T)\cdot \frac{1}{2} ~\geq~ \frac{1}{2}
		\end{aligned}
		\end{equation}
		By \eqref{eq:lb_prob_upper} and \eqref{eq:lb_prob_lower} it holds that $\frac{\sum_{j\in [s]} \abs{\coverage_j(Q_j(b^*))}}{\binom{n+1}{k^*}} \geq \frac{1}{2}$ and the claim immediately follows.
	\end{claimproof}

	By \Cref{claim:cost_to_coverage,claim:approximate_coverage_lb} it holds that
	\begin{align*}
		\cost_\A(n+1)
		&\geq ~f_{\CL,\beta}(n) \cdot \frac{\sum_{j\in [s]} \abs {\coverage_j(Q_j(b^*))} }{(n+1) \cdot \binom{n}{k*} \cdot \max_{j\in [s]} c_j} \\
		&\geq ~f_{\CL,\beta}(n) \cdot \frac{\frac{1}{2} \cdot \binom{n+1}{k^*}}{(n+1) \cdot \binom{n}{k^*} \cdot \max_{j\in [s]} c_j}\\
		&\geq ~\frac{f_{\CL,\beta}(n)}{2 \cdot (n+1)\cdot \max_{j\in [s]} c_j}.
	\end{align*}
\end{proof}

\begin{proof}[Proof of \Cref{thm:minimization_lower_bound}]
	Let $\A$ be a $\beta$-approximation algorithm for $\LSUB$.
	By \Cref{lem:minimization_lower_bound_aux,lem:f_to_amls} it holds that
	\begin{align*}
		\cost_\A(n+1) &\geq \frac{f_{\CL,\beta}(n)}{2 \cdot (n+1) \cdot \max_{(\alpha,c) \in \CL} c} \geq n^{-\OO(1)} \cdot (\amlsbound(\CL,\beta))^n \geq n^{-\OO(1)} \cdot (\amlsbound(\CL,\beta))^{n+1}\\
		&\geq n^{-\OO(1)} \cdot f_{\CL,\beta}(n+1)
	\end{align*}
	for every $n\geq 1$.
\end{proof}

\section{From Discrete to Continuous Optimization}
\label{sec:f_to_amls}
\Cref{lem:f_to_amls} shows that$f_{\CL,\beta} (n)\approx \left(\amlsbound(\CL,\beta)\right)^n$ up to polynomial factors.
While $f_{\CL,\beta}(n)$ \eqref{eq:fdef_intro} is defined via maximum and minimum operations over a discrete set of values, the value of $\amlsbound(\CL,\beta)$~\eqref{eq:amls_def} is the outcome of continuous maximization and minimization.
The proof utilizes basic estimation of binomial coefficient using entropy and bounded-difference properties of the entropy function.

The value of $\tau$ in the definition of $\amlsbound$ \eqref{eq:amls_def} corresponds to $\frac{t}{n}$ in the formula of $f$ \eqref{eq:fdef_intro}.
We note that the range of $\frac{t}{n}$ in \eqref{eq:fdef_intro} may differ from the range of $\tau$ in \eqref{eq:amls_def}.
Part of the proof is dedicated for showing this difference is insignificant.

We first prove that $f_{\CL,\beta}(n)\lesssim \left( \amlsbound(\CL,\beta)\right)^n$ in \Cref{lem:f_atmost_amls}, and subsequently show that $f_{\CL,\beta}(n)\gtrsim \left( \amlsbound(\CL,\beta)\right)^n$ in \Cref{lem:f_atleast_amls}.
The proof of \Cref{lem:f_atmost_amls} is technically easier.
This stems from the fact that restricting the range of $t$, extending the range of $k$ and lower-bounding the value of $\p$ (as it appears in \eqref{eq:fdef_intro}) are trivial in this direction of the inequality, but not in the other.
We also note that special cases of the inequality $f_{\CL,\beta}(n)\lesssim \left( \amlsbound(\CL,\beta)\right)^n$ for $\CL=\{(\beta,c)\}$ implicitly appear in previous works on (Approximate) Monotone Local Search \cite{EsmerKMNS22,FominGLS19} as part of the analysis of the algorithm.
The opposite direction, $f_{\CL,\beta}(n)\gtrsim \left( \amlsbound(\CL,\beta)\right)^n$,
is central for the correctness of the lower bounds in \Cref{thm:minimization_lower_bound,lem:exact_lower_bound}, but has no algorithmic implications.
As such, this direction of the inequality was irrelevant to the previous works which only provided algorithmic results.

Recall $\entropy\left( x\right)= -x\ln (x)-(1-x)\ln(1-x)$.
With slight abuse of notation we define $0 \cdot \entropy\left(\frac{a}{0}\right) = 0$.
Our proofs utilize the following bound on binomial coefficients (see, e.g.,~\cite[Example~11.1.3]{CoverT06}):
\begin{equation}\label{eq:binom}
	\frac{1}{n+1} \cdot \exp\left( n\cdot \entropy\left(\frac{k}{n}\right)\right)\leq \binom{n}{k} \leq \exp\left( n\cdot \entropy\left(\frac{k}{n}\right)\right)
\end{equation}
for all $n,k\in \mathbb{N}$ such that $0 \leq k \leq n$.
Furthermore, we utilize the following technical lemma which follows from~\cite{EsmerKMNS22}. 

\begin{lemma}
	\label{lem:perturb}
	For all $0 \leq b \leq a \leq n$, $d>1$ and $\varepsilon,\delta \in [-d,d]$ such that $0\leq b+\delta \leq a+\varepsilon$, we have
	\begin{align*}
		\abs{a\cdot \entropy\left( \frac{b}{a} \right) - (a + \varepsilon) \cdot \entropy\left( \frac{b + \delta}{a + \varepsilon} \right) } = \OO(d\cdot \log(n)).
	\end{align*}
\end{lemma}

For every $\alpha,\beta\geq 1$ we define $\x(k,t) = \left(1-\frac{\beta}{\alpha}\right) \cdot k + \frac{t}{\alpha}$.
We utilize the following technical lemmas as part of the proofs of \Cref{lem:f_atleast_amls,lem:f_atmost_amls}.

\begin{lemma}
	\label{lem:ky_vs_nt}
	Let $\alpha,\beta\geq 1$, $n\in \NN$, $k\in \left[0,\frac{n}{\beta}\right]$, $t \geq 0$ and $y \geq \x(k,t)$.
	Then $k-y\leq n-t$.
\end{lemma}
\begin{proof}
	We have
	\[n-t\geq \beta k -t \geq \frac{\beta k -t}{\alpha}= k  -k +\frac{\beta k}{\alpha} - \frac{t}{\alpha} = k-\x(k,t) \geq k-y.\qedhere\]
\end{proof}

\begin{lemma}
	\label{lem:x_vs_minkt}
	Let $\alpha,\beta\geq 1$, $n\in \NN$,  $k\in \left[0,\frac{n}{\beta}\right]$, and $t\in \left[M^*_{\alpha,\beta}k, \beta k\right]$.
	Then $\x(k,t) \leq \min\{k,t\}$.
\end{lemma}
\begin{proof}
	Since $t \geq 0$ we have
	\[\x(k,t) = \left(1-\frac{\beta}{\alpha} \right) k +\frac{t}{\alpha} \leq  \left(1-\frac{\beta}{\alpha} \right) k \leq k.\]
	For the second part, we consider the following two cases.
	\begin{itemize}
		\item If $\beta \geq \alpha $ then
		\[
			\x(k,t) = \left(1-\frac{\beta}{\alpha} \right) k +\frac{t}{\alpha}\leq \frac{t}{\alpha}\leq t.
		\]
		\item Otherwise $\beta < \alpha$ and we have $t\geq M^*_{\alpha,\beta}\cdot k =\frac{\alpha-\beta}{ \alpha-1} \cdot k$.
		Thus,
		\[
		\x(k,t) = \left( 1-\frac{\beta}{\alpha}\right)k + \frac{t}{\alpha} = \frac{\alpha-1}{\alpha} \cdot \frac{\alpha-\beta}{\alpha-1} \cdot k + \frac{t}{\alpha}\leq \frac{\alpha-1}{\alpha }\cdot t + \frac{t}{\alpha}=t.
		\]
	\end{itemize}
	In both cases $\x(k,t) \leq t$ which completes the proof.
\end{proof}

\begin{lemma}
	\label{lem:x_vs_kt}
	Let $\alpha,\beta\geq 1$, $n\in \NN$, $k\in \left[0,\frac{n}{\beta}\right]$, and $t\in \left[M^*_{\alpha,\beta} k, \beta k\right]$.
	Then $\x(k,t) \geq \frac{kt}{n}$ if and only if $t \geq M_{\alpha,\beta}\left(\frac{k}{n}\right) \cdot n$.
\end{lemma}
\begin{proof}
	We consider the following two cases.
	\begin{itemize}
		\item If $\alpha <\beta$ it holds that
		\[\x(k,t) \geq \frac{kt}{n} \quad\iff\quad \left( 1-\frac{\beta}{\alpha}\right)\cdot k+\frac{t}{\alpha} \geq \frac{kt}{n} \quad\iff\quad \left( 1-\frac{\beta}{\alpha}\right)\cdot k \geq t\cdot \left(\frac{k}{n}-\frac{1}{\alpha}\right).\]
		Since $k\leq \frac{n}{\beta}$ we conclude that $\frac{k}{n}-\frac{1}{\alpha}<0$.
		So
		\[\x(k,t) \geq \frac{kt}{n} \quad\iff\quad t\geq \frac{ \left(1-\frac{\beta}{\alpha}\right)}{\frac{k}{n}-\frac{1}{\alpha}} \cdot k = \frac{\beta -\alpha}{1-\alpha\cdot \frac{k}{n}}\cdot \frac{k}{n}\cdot n = M_{\alpha,\beta}\left( \frac{k}{n}\right)\cdot n.\]
		\item Otherwise $\alpha \geq \beta$ and we have $M_{\alpha,\beta}\left( \frac{k}{n}\right) \cdot n = M^*_{\alpha,\beta} \cdot k$.
		Thus, we need to prove that $\x(k,t)\geq \frac{kt}{n}$ holds unconditionally.
		Indeed,
		\[
		\begin{aligned}
			\x(k,t) ~&=~ \left(1-\frac{\beta}{\alpha}\right)\cdot k +\frac{t}{\alpha} \\
			&=~ \left(1-\frac{\beta}{\alpha}\right)\cdot k +t \left(-\frac{k}{n}+ \frac{1}{\alpha} \right)+\frac{kt}{n}\\
			&\geq ~\left(1-\frac{\beta}{\alpha}\right)\cdot k +t \left(-\frac{1}{\beta}+ \frac{1}{\alpha} \right)+\frac{kt}{n}\\
			&\geq ~\left(1-\frac{\beta}{\alpha}\right)\cdot k +\beta \cdot k  \left(-\frac{1}{\beta}+ \frac{1}{\alpha} \right)+\frac{kt}{n} ~=~ \frac{kt}{n},
		\end{aligned}
		\]
		where the first inequality follows from $k\leq \frac{n}{\beta}$, and the second inequality holds since $t\leq \beta k$ and $-\frac{1}{\beta}+\frac{1}{\alpha} \leq 0$.\qedhere
	\end{itemize}
\end{proof}

Finally, we use the following relation between $M^*_{\alpha,\beta}$ and $M_{\alpha,\beta}$. 

\begin{lemma}
	\label{lem:Mstar_vs_nostar}
	Let $\alpha,\beta\geq 1$, $n\in \NN$, $k\in \left[0,\frac{n}{\beta}\right]$.
	Then $M^*_{\alpha,\beta}\cdot k \leq M_{\alpha,\beta}\left( \frac{k}{n}\right) \cdot n\leq \beta k$.
\end{lemma}
\begin{proof}
	We first show that $M_{\alpha,\beta}^*\cdot k \leq M_{\alpha,\beta}\left(\frac{k}{n}\right)\cdot n$.
	If $\alpha\leq \beta$ it holds that $M^*_{\alpha,\beta }=0$.
	Since $M_{\alpha,\beta}\left(\frac{k}{n}\right)\geq 0$ it follows that $M^*_{\alpha,\beta }\cdot k = 0 \leq M_{\alpha,\beta}\left(\frac{k}{n}\right) \cdot n$.
	Otherwise $\alpha>\beta$ and we have
	\[M_{\alpha,\beta}\left(\frac{k}{n}\right)\cdot n =\frac{\alpha-\beta}{\alpha-1} \cdot \frac{k}{n}\cdot n= M^*_{\alpha,\beta} \cdot k,\]
	where the last equality follows from the definition of $M^*_{\alpha,\beta}$ \eqref{eq:Mstar_def}.

	Next, we show that $M_{\alpha,\beta}\left(\frac{k}{n}\right)\cdot n\leq \beta k$.
	If $\alpha<\beta$ it holds that
	\[M_{\alpha,\beta}\left(\frac{k}{n}\right)\cdot n = \frac{\beta-\alpha}{1-\alpha\cdot\frac{k}{n}} \cdot \frac{k}{n} \cdot n\leq\frac{\beta-\alpha}{1-\alpha\cdot\frac{1}{\beta}} \cdot k = \beta \cdot k,\]
	where the inequality follows from $k\leq \frac{n}{\beta}$.
	Otherwise $\beta<\alpha$ and we have
	\[
	M_{\alpha,\beta}\left(\frac{k}{n}\right)\cdot n = \frac{\alpha-\beta}{\alpha-1}\cdot \frac{k}{n}\cdot n \leq\frac{\alpha-1}{\alpha-1}\cdot \frac{k}{n} \cdot n = k \leq \beta k.
	\]
	Finally, if $\alpha=\beta$, we have $M_{\alpha,\beta}\left(\frac{k}{n}\right)\cdot n =0\cdot n  \leq \beta k$.
\end{proof}

The next lemma show the second inequality of \Cref{lem:f_to_amls}.

\begin{lemma}
	\label{lem:f_atmost_amls}
	For every $\beta \geq 1$ and specification list $\CL$ it holds that
	\[f_{\CL,\beta}(n)\leq n^{\OO(1)}\cdot \left(\amlsbound(\CL,\beta)\right)^n.\]
\end{lemma}

\begin{proof}
	Let $n\geq 1$.
	For every $k \in \left[0,\frac{n}{\beta}\right] \cap \NN$ and $(\alpha,c)\in \CL$ it holds that
	\begin{equation}
		\label{eq:atmost_first}
		\begin{aligned}
			\min_{~t\in \left[M^*_{\alpha,\beta}\cdot k , \beta k\right]\cap \NN~} &\frac{\exp\left( \frac{\beta k -t}{\alpha} \cdot c\right)}{\p(n,k,t,\x(k,t))} ~=~	
			\min_{~t\in \left[M^*_{\alpha,\beta}\cdot k , \beta k\right]\cap \NN~} \frac{\exp\left( \frac{\beta k -t}{\alpha} \cdot c\right)}{\sum_{y=\lceil\x(k,t)\rceil}^{\min\{k,t\}} \frac{\binom{t}{y}\cdot \binom{n-t}{k-y}}{\binom{n}{k}}}\\
			&\leq~	\min_{~t\in \left[M^*_{\alpha,\beta}\cdot k , \beta k\right]\cap \NN~} \frac{\exp\left( \frac{\beta k -t}{\alpha} \cdot c\right)} {\left(\frac{\binom{t}{\max\{\lceil\x(k,t)\rceil,0\}}\cdot \binom{n-t}{k-\max\{\lceil\x(k,t)\rceil,0\}}}{\binom{n}{k}}\right)}\\
			&\leq~(n+1)^2 \cdot \min_{~t\in \left[M^*_{\alpha,\beta}\cdot k , \beta k\right]\cap \NN~} \exp\Biggr( \frac{\beta k - t}{\alpha} \cdot c- t\cdot \entropy\left( \frac{\max\{\lceil\x(k,t)\rceil,0\}}{t}\right)\\
			&\quad\quad\quad\quad - (n-t) \cdot \entropy\left(\frac{k-\max\{\lceil\x(k,t)\rceil,0\}}{n-t} \right)+n\cdot \entropy\left(\frac{k}{n}\right) \Biggr)
		\end{aligned}
	\end{equation}
	The first equality follows from the definition of $\p$ \eqref{eq:hyper}.
	The first inequality follows from selecting $y = \max\{\lceil\x(k,t)\rceil,0\}$.
	Note that the resulting expression is well defined by \Cref{lem:ky_vs_nt,lem:x_vs_minkt}.
	The last inequality follows from \eqref{eq:binom}.
	
	We can use \Cref{lem:perturb} to avoid the rounding of the values of $\x$ as well as extending the range of $t$ in \eqref{eq:atmost_first}.
	That is, for every $k\in \left[ 0,\frac{n}{\beta}\right]\cap \NN$ and $(\alpha,c)\in \CL$ it holds that
	\begin{equation}
		\label{eq:atmost_second}
		\begin{aligned}
			\min_{~t\in \left[M^*_{\alpha,\beta}\cdot k , \beta k\right]\cap \NN~} &\frac{\exp\left( \frac{\beta k -t}{\alpha} \cdot c\right)}{\p(n,k,t,\x(k,t))} \\
			&\leq ~n^{\OO(1)}\cdot \min_{~t\in \left[M^*_{\alpha,\beta}\cdot k , \beta k\right]\cap \NN~} \exp\Biggr( \frac{\beta k -t}{\alpha} \cdot c- t\cdot \entropy\left( \frac{\max\{{\x(k,t)},~0\} }{t }\right)\\
			&\quad\quad\quad\quad-(n-t) \cdot \entropy \left(\frac{k-\max\{{\x(k,t)},~0\}}{n-t} \right)+n\cdot \entropy\left(\frac{k}{n}\right) \Biggr)\\
			&\leq~n^{\OO(1)} \cdot \min_{~t\in \left[M^*_{\alpha,\beta}\cdot k , \beta k\right]~} \exp\Biggr( \frac{\beta k -t}{\alpha} \cdot c- t\cdot \entropy\left( \frac{\max\{{\x(k,t)},~0\} }{t }\right)\\
			&\quad\quad\quad\quad-(n-t) \cdot \entropy \left(\frac{k-\max\{{\x(k,t)},~0\}}{n-t} \right)+n\cdot \entropy\left(\frac{k}{n}\right) \Biggr)\\
			&\leq~n^{\OO(1)} \cdot \min_{~t\in \left[M_{\alpha,\beta}\left(\frac{k}{n}\right)\cdot n , \beta k\right]~} \exp\Biggr( \frac{\beta k -t}{\alpha} \cdot c- t\cdot \entropy\left( \frac{\max\{{\x(k,t)},~0\} }{t }\right)\\
			&\quad\quad\quad\quad-(n-t) \cdot \entropy \left(\frac{k-\max\{{\x(k,t)},~0\}}{n-t} \right)+n\cdot \entropy\left(\frac{k}{n}\right) \Biggr)\\
			&\leq~n^{\OO(1)} \cdot \min_{~t\in \left[M_{\alpha,\beta}\left(\frac{k}{n}\right)\cdot n , \beta k\right]~} \exp\Biggr( \frac{\beta k -t}{\alpha} \cdot c- t\cdot \entropy\left(\frac{ \x(k,t) }{t }\right)\\
			&\quad\quad\quad\quad-(n-t) \cdot \entropy \left(\frac{k-\x(k,t)}{n-t} \right)+n\cdot \entropy\left(\frac{k}{n}\right) \Biggr)
		\end{aligned}
	\end{equation}
	Observe that in the third expression the range of $t$ is not restricted to integers.
	The range of $t$ was further changed in the forth expression using \Cref{lem:Mstar_vs_nostar}.
	The last inequality follow from \Cref{lem:x_vs_kt} (trivially, $\frac{kt}{n}\geq 0$).  
	
	Observe that  
	\begin{equation}
		\label{eq:x_to_gamma}
		\frac{\x(k,t)}{t} =\left( 1-\frac{\beta}{\alpha}\right)\frac{k}{t} +\frac{1}{\alpha} = \gamma_{\alpha,\beta}\left( \frac{k}{n}, \frac{t}{n}\right)
	\end{equation}
	unless $t=0$, and 
	\begin{equation}
		\label{eq:x_to_delta}
		\frac{k-\x(k,t)}{n-t} = 
		\frac{k-\left( 1-\frac{\beta}{\alpha}\right)k-\frac{t}{\alpha}}{n-t} = 
		\frac{\frac{\beta}{\alpha}\cdot k-\frac{t}{\alpha}}{n-t} =
		\delta_{\alpha,\beta}\left(\frac{k}{n},\frac{t}{n}\right)
	\end{equation}
	unless $t=1$.
	
	Using \eqref{eq:x_to_gamma} and \eqref{eq:x_to_delta} we can simplify the expression in \eqref{eq:atmost_second} and obtain
	\begin{equation}
		\label{eq:atmost_third}
		\begin{aligned}
			\min_{~t\in \left[M^*_{\alpha,\beta}\cdot k , \beta k\right]\cap \NN~} &\frac{\exp\left( \frac{\beta k -t}{\alpha} \cdot c\right)}{\p(n,k,t,\x(k,t))} \\
			&\leq~n^{\OO(1)} \cdot \min_{~t\in \left[M_{\alpha,\beta}\left(\frac{k}{n}\right)\cdot n , \beta k\right]~} \exp\Biggr( \frac{\beta \frac{k}{n} -\frac{t}{n}}{\alpha} \cdot c- \frac{t}{n}\cdot \entropy\left(\gamma_{\alpha,\beta}\left(\frac{k}{n},\frac{t}{n}\right) \right)\\
			&\quad\quad\quad\quad-\left(1-\frac{t}{n}\right) \cdot \entropy \left(\delta_{\alpha,\beta}\left(\frac{k}{n},\frac{t}{n} \right)\right) + \entropy\left(\frac{k}{n}\right) \Biggr) ^n\\
			&=~n^{\OO(1)} \cdot \min_{~\tau\in \left[M_{\alpha,\beta}\left(\frac{k}{n}\right) , \beta \frac{k}{n}\right]~}  \left( g_{\alpha,\beta, c } \left(\frac{k}{n},\tau\right)\right)^n.	
		\end{aligned}
	\end{equation}
	By incorporating \eqref{eq:atmost_third} into the formula of $f_{\CL,\beta}$ \eqref{eq:fdef_intro} we get 
	\begin{equation*}
		\begin{aligned}
			f_{\CL,\beta}(n) &=~ \max_{~k\in \left[0,\frac{n}{\beta}\right] \cap \NN~}
				\min_{~(\alpha,c)\in \CL~}
				\min_{~t\in\left[M^*_{\alpha, \beta}\cdot k,\beta\cdot  k\right] \cap \NN~}
				\frac{\exp\left( \frac{\beta k -t} {\alpha }\cdot \ln c \right)} {\p\left( n,k, t, (1 - \frac{\beta}{\alpha}) \cdot k +\frac{t}{\alpha} \right)}\\
			&\leq~n^{\OO(1)} \cdot  \max_{~k\in \left[0,\frac{n}{\beta}\right] \cap \NN~ } \min_{(\alpha,c )\in \CL}
				\min_{~\tau\in \left[M_{\alpha,\beta}\left(\frac{k}{n}\right), \beta \frac{k}{n}\right]~} \exp\Biggr(g_{\alpha,\beta,c}\left(\frac{k}{n},\tau\right) \Biggr)^n\\
			&\leq~n^{\OO(1)} \cdot  \max_{~k\in \left[0,\frac{n}{\beta}\right] ~ } \min_{(\alpha,c )\in \CL}
				\min_{~\tau\in \left[M_{\alpha,\beta}\left(\frac{k}{n}\right), \beta \frac{k}{n}\right]~} \exp\Biggr(g_{\alpha,\beta,c}\left(\frac{k}{n},\tau\right) \Biggr)^n\\
			&=~n^{\OO(1)} \cdot  \max_{~\kappa\in \left[0,\frac{1}{\beta}\right] ~ } \min_{(\alpha,c )\in \CL}
				\min_{~\tau\in \left[M_{\alpha,\beta}\left(\kappa \right), \beta \kappa\right]~} \exp\Biggr(g_{\alpha,\beta,c}\left(\kappa,\tau\right) \Biggr)^n\\
			&=~n^{\OO(1)} \cdot \left( \amlsbound(\CL,\beta)\right)^n.
		\end{aligned}
	\end{equation*}
	The second inequality simply extended the range of values $k$ can takes, and the second equality substituted $k$ with $\kappa\cdot n$.
\end{proof}

\begin{lemma}
	\label{lem:f_atleast_amls}
	For every $\beta \geq 1$ and specification list $\CL$ it holds that 
	\[f_{\CL,\beta}(n)\geq n^{\OO(1)}\cdot \left(\amlsbound(\CL,\beta)\right)^n.\]
\end{lemma}
\begin{proof}
	Let $n\in \NN$.
	For every $k\in \left[0,\frac{n}{\beta}\right]\cap \NN$ and $(\alpha,c)\in \CL$ it holds that
	\begin{equation}
		\label{eq:atleast_first}
		\begin{aligned}
			\min_{~t\in \left[M^*_{\alpha,\beta}\cdot k , \beta k\right]\cap \NN~} &\frac{\exp\left( \frac{\beta k -t}{\alpha} \cdot \ln c\right)}{\p(n,k,t,\x(k,t))} ~=~
			\min_{~t\in \left[M^*_{\alpha,\beta}\cdot k , \beta k\right]\cap \NN~} \frac{\exp\left( \frac{\beta k -t}{\alpha} \cdot \ln c\right)}{\sum_{y=\lceil\x(k,t)\rceil}^{\min\{k,t\}} \frac{\binom{t}{y}\cdot \binom{n-t}{k-y}}{\binom{n}{k}}}\\
			&\geq~\min_{t\in \left[M^*_{\alpha,\beta}\cdot k , \beta k\right]\cap\NN} \frac{\exp\left( \frac{\beta k -t}{\alpha} \cdot \ln c\right)}{n\cdot \max_{y \in \left[\max\{\lceil\x(k,t)\rceil,0\},~\min\{k,t\}\right]\cap \NN} \frac{\binom{t}{y}\cdot \binom{n-t}{k-y}}{\binom{n}{k}}}\\
			&=~\frac{1}{n}\cdot \min_{~t\in \left[M^*_{\alpha,\beta}\cdot k , \beta k\right]\cap\NN~}\min_{~y\in \left[\max\{\lceil\x(k,t)\rceil,0\},~\min\{k,t\}\right]\cap \NN~} \frac{\exp\left( \frac{\beta k -t}{\alpha} \cdot \ln c\right)}{  \frac{\binom{t}{y}\cdot \binom{n-t}{k-y}}{\binom{n}{k}}}.
		\end{aligned}
	\end{equation}
	As in the proof of  \Cref{lem:f_atmost_amls}, we  use \eqref{eq:binom} to estimate the binomial coefficients in \eqref{eq:atleast_first}.
	Thus, for all $k\in \left[0,\frac{n}{\beta}\right]\cap \NN$ and $(\alpha,c)\in \CL$ we have
	\begin{equation}
		\label{eq:atleast_second}
		\begin{aligned}
			\min_{~t\in \left[M^*_{\alpha,\beta}\cdot k , \beta k\right]\cap \NN~} &\frac{\exp\left( \frac{\beta k -t}{\alpha} \cdot \ln c\right)}{\p(n,k,t,\x(k,t))} \\
			&\geq ~n^{-\OO(1)}\cdot \min_{~t\in \left[M^*_{\alpha,\beta}\cdot k , \beta k\right]\cap\NN~}\min_{~y\in \left[\max\{\lceil\x(k,t)\rceil,0\},~\min\{k,t\}\right]\cap \NN~}\\
			&\quad\quad\quad \exp\Biggr( \frac{\beta k -t}{\alpha} \cdot\ln c -t\cdot \entropy\left( \frac{y}{t}\right) - (n-t)\cdot \entropy\left( \frac{k-y}{n-t} \right)+ n\cdot \entropy\left(\frac{k}{n}\right)  \Biggr)\\
			&\geq ~n^{-\OO(1)}\cdot \min_{~t\in \left[M^*_{\alpha,\beta}\cdot k , \beta k\right]\cap\NN~}\min_{~y\in \left[ \max\{{\x(k,t)},0\},~\min\{k,t\}\right]}\\
			&\quad\quad\quad \exp\Biggr( \frac{\beta k -t}{\alpha} \cdot \ln c
			-t\cdot \entropy\left( \frac{y}{t}\right) 
			-(n-t)\cdot\entropy\left( \frac{k-y}{n-t} \right)+ n\cdot \entropy\left(\frac{k}{n}\right)  \Biggr).
		\end{aligned}
	\end{equation}
	For every $k\in \left[0,\frac{n}{\beta}\right]$ and $t\in \left[M^*_{\alpha,\beta}\cdot k , \beta k\right]$ we define 
	\begin{equation}
		\label{eq:atleast_hdef}
		h_{k,t}(y) = -t\cdot \entropy\left( \frac{y}{t}\right) - (n-t)\cdot\entropy\left( \frac{k-y}{n-t} \right)+ n\cdot \entropy\left(\frac{k}{n}\right).
	\end{equation}
	We can use $h_{k,t}$ to rewrite \eqref{eq:atleast_second} as 
	\begin{equation}
		\label{eq:atleast_third}
		\begin{aligned}
			&\min_{~t\in \left[M^*_{\alpha,\beta}\cdot k , \beta k\right]\cap \NN~} \frac{\exp\left( \frac{\beta k -t}{\alpha} \cdot \ln c\right)}{\p(n,k,t,\x(k,t))} \\
			&\geq ~n^{-\OO(1)} \cdot \min_{~t\in \left[M^*_{\alpha,\beta}\cdot k , \beta k\right]\cap\NN~}\min_{~y\in \left[\max\{\x(k,t),0\},~\min\{k,t\}\right]}
			\exp\Biggr( \frac{\beta k -t}{\alpha} \cdot \ln c + h_{k,t}(y) \Biggr)
		\end{aligned}
	\end{equation}
	for all $k\in \left[0,\frac{n}{\beta}\right]\cap \NN$ and $(\alpha,c)\in \CL$.
	\begin{claim}
	\label{claim:atleast_h}
		For all $k\in \left[0,\frac{n}{\beta}\right]\cap \NN$,  $(\alpha,c)\in \CL$ and  $t\in \left[M^*_{\alpha,\beta}\cdot k , \beta k\right]$ it holds that
		\[\min_{~y\in \left[ \max\{\x(k,t),0\},~\min\{k,t\}\right]~} h_{k,t}(y) = 
			\begin{cases}
				h_{k,t}(\x(k,t)) & \textnormal{if } \x(k,t) \geq \frac{kt}{n},\\ 
				0 & \textnormal{otherwise.}
			\end{cases}
		\]
		Furthermore, $h_{k,t}\left(\frac{kt}{n}\right) = 0$.
	\end{claim}
	\begin{claimproof}
		We have $\entropy'(x) = \ln\frac{1-x}{x}$ where $\entropy'$ is the first derivative of $\entropy$.
		So the first derivative of $h_{k,t}$ is
		\begin{align*}
		h'_{k,t}(y) &= -t\cdot \frac{1}{t}\cdot \ln \left( \frac{1-\frac{y}{t}}{\frac{y}{t}}\right)-(n-t)\cdot \frac{-1}{n-t} \cdot \ln\left( \frac{1-\frac{k-y}{n-t}}{\frac{k-y}{n-t}}\right)\\
		&= - \ln \left( \frac{t}{y}-1\right) + \ln \left(\frac{n-t}{k-y} -1 \right).
		\end{align*}
		It can  be easily observed that $h'_{k,t}$ is a monotonically increasing function and thus, $h_{k,t}$ is convex.
		Furthermore,
		\[h'_{k,t}\left(\frac{kt}{n}\right) = - \ln \left( t\cdot \frac{n}{kt}-1\right) + \ln \left(\frac{n-t}{k-\frac{kt}{n}} -1 \right) = -\ln\left(\frac{n}{k} -1\right) + \ln\left(\frac{n}{k} -1\right) =0\]
		and 
		\begin{align*}
			h_{k,t}\left(\frac{kt}{n}\right) &= -t\cdot \entropy\left( \frac{\frac{kt}n}{t}\right) - (n-t)\cdot\entropy\left( \frac{k-\frac{kt}{n}}{n-t} \right)+ n\cdot \entropy\left(\frac{k}{n}\right)\\
			&=-t\cdot \entropy\left( \frac{k}{n}\right) - (n-t)\cdot\entropy\left( \frac{k}{n} \right)+ n\cdot \entropy\left(\frac{k}{n}\right)=0.
		\end{align*}
		So overall, $h_{k,t}(y)$ is a convex function with a global minimum of value $0$ at $y = \frac{kt}{n}$.
		Since $k,t\leq n$ it also holds that $\frac{kt}{n}\leq \min\{k,t\}$.
		Hence, 
		\[\min_{~y\in \left[ \max\{\x(k,t),0\},~\min\{k,t\}\right]~} h_{k,t}(y) =
			\begin{cases}
				h_{k,t}(\x(k,t)) & \textnormal{if } \x(k,t) \geq \frac{kt}{n},\\ 
				0 & \textnormal{otherwise.}
			\end{cases}
		\]
	\end{claimproof}
	By \Cref{claim:atleast_h} and \eqref{eq:atleast_third}, for all $k\in \left[0,\frac{n}{\beta}\right]\cap \NN$ and $(\alpha,c)\in \CL$, we have 
	\begin{equation}
		\label{eq:atleast_forth}
		\begin{aligned}
			&\min_{~t\in \left[M^*_{\alpha,\beta}\cdot k , \beta k\right]\cap \NN~} \frac{\exp\left( \frac{\beta k -t}{\alpha} \cdot \ln c\right)}{\p(n,k,t,\x(k,t))} \\
			&\geq ~n^{-\OO(1)}\cdot \min_{~t\in \left[M^*_{\alpha,\beta}\cdot k , \beta k\right]\cap\NN~}
				\exp\Biggr( \frac{\beta k -t}{\alpha} \cdot \ln c + 
				\begin{cases}
					h_{k,t}(\x(k,t)) & \textnormal{ if } \x(k,t) \geq \frac{kt}{n} \\ 
					0 & \textnormal{ otherwise}
				\end{cases} \Biggr) \\
			&\geq~ n^{-\OO(1)}\cdot \min_{~t\in \left[M^*_{\alpha,\beta}\cdot k , \beta k\right]~}
				\exp\Biggr( \frac{\beta k -t}{\alpha} \cdot \ln c + 
				\begin{cases}
					h_{k,t}(\x(k,t))& \textnormal{if } \x(k,t) \geq \frac{kt}{n} \\ 
					0 & \textnormal{otherwise}
				\end{cases} \Biggr).\\
		\end{aligned}
	\end{equation}
	Observe the range of $t$ in the last  expression is not restricted to integral values.

	\begin{claim} 
		\label{claim:atleast_elimination}
		For all $k\in \left[ 0,\frac{n}{\beta}\right]\cap \NN$ and $(\alpha,c)\in\CL$ it holds that
		\begin{align*}
			&\min_{~t\in \left[M^*_{\alpha,\beta}\cdot k , \beta k\right]~}
				\exp\Biggr( \frac{\beta k -t}{\alpha} \cdot \ln c + 
				\begin{cases}
					h_{k,t}(\x(k,t))& \textnormal{ if } \x(k,t) \geq \frac{kt}{n} \\ 
					0 & \textnormal{ otherwise}
				\end{cases} \Biggr)\\
			=~&\min_{~t\in \left[M_{\alpha,\beta}\left(\frac{k}{n}\right)\cdot n , \beta k\right]~}
				\exp\Biggr( \frac{\beta k -t}{\alpha} \cdot \ln c + h_{k,t}(\x(k,t)) \Biggr).
		\end{align*}
	\end{claim}
	\begin{claimproof}
		First suppose $\alpha \geq \beta$. 
		Then $M^*_{\alpha,\beta} \cdot k = M_{\alpha,\beta}\left(\frac{k}{n}\right)\cdot n$.
		Furthermore, by \Cref{lem:x_vs_kt}, it holds that $\x(k,t)\geq \frac{kt}{n}$ for all $t\in \left[ M^*_{\alpha,\beta}\cdot k ,\beta k\right]$, and the statement of the claim immediately follows. 
	
		We are left to handle the case $\alpha< \beta$. By \Cref{lem:x_vs_kt} we have
		\begin{align*}
			&\min_{~t\in \left[M^*_{\alpha,\beta}\cdot k, M_{\alpha,\beta}\left(\frac{k}{n}\right)\cdot n\right]~}
				\exp\left( \frac{\beta k -t}{\alpha} \cdot \ln c + 
				\begin{cases}
					h_{k,t}(\x(k,t))& \textnormal{if } \x(k,t) \geq \frac{kt}{n} \\ 
					0 & \textnormal{otherwise}
				\end{cases} \right)\\
			=~&\min_{~t\in \left[M^*_{\alpha,\beta}\cdot k, M_{\alpha,\beta}\left(\frac{k}{n}\right)\cdot n\right]~}
				\exp\left( \frac{\beta k -t}{\alpha} \cdot \ln c + 
				\begin{cases}
					h_{k,t}(\x(k,t))& \textnormal{if } t=M_{\alpha,\beta}\left(\frac{k}{n}\right)\cdot n \\ 
					0 & \textnormal{otherwise}
				\end{cases} \right)\\
			=~&\min_{~t\in \left[M^*_{\alpha,\beta}\cdot k, M_{\alpha,\beta}\left(\frac{k}{n}\right)\cdot n\right]~}
				\exp\left( \frac{\beta k -t}{\alpha} \cdot \ln c \right)\\
			=~&\exp\left( \frac{\beta k -t}{\alpha} \cdot \ln c \right) \Biggr|_{t=M_{\alpha,\beta}\left(\frac{k}{n}\right)\cdot n }\\
			=~&\exp\left( \frac{\beta k -t}{\alpha} \cdot \ln c+h_{k,t}(\x(k,t))\right) \Biggr|_{t=M_{\alpha,\beta}\left(\frac{k}{n}\right)\cdot n },
		\end{align*}
		where the second and forth equalities follow from  $\x\left(k,t\right) = \frac{kt}{n}$ for $t= M_{\alpha,\beta}\left( \frac{k}{n}\right)\cdot n$ and $h_{k,t}\left( \frac{kt}{n}\right)=0$ (\Cref{claim:atleast_h}).
		The statement of the claim follows from the last series of equalities. 
	\end{claimproof}

	By~\eqref{eq:atleast_forth} and \Cref{claim:atleast_elimination} it holds that 
	\begin{equation}
		\label{eq:atmost_fifth}
		\begin{aligned}
			&\min_{~t\in \left[M^*_{\alpha,\beta}\cdot k , \beta k\right]\cap \NN~} \frac{\exp\left( \frac{\beta k -t}{\alpha} \cdot \ln c\right)}{\p(n,k,t,\x(k,t))} \\
			&\geq~ n^{-\OO(1)}\cdot \min_{~t\in \left[M_{\alpha,\beta}\left( \frac{k}{n}\right)\cdot n , \beta k\right]~}
				\exp\left( \frac{\beta k -t}{\alpha} \cdot \ln c + h_{k,t}(\x(k,t)) \right)\\
		\end{aligned}
	\end{equation}
	for all $k \in \left[ 0,\frac{n}{\beta}\right]\cap\NN$ and $(\alpha,c)\in \CL$.
	By the definition of $f_{\CL,\beta}$ \eqref{eq:fdef_intro} we have
	\begin{equation}
		\label{eq:atleast_sixth}
		\begin{aligned}
			f_{\CL,\beta}(n) &=~ \max_{~k\in \left[0,\frac{n}{\beta}\right] \cap \NN~ }
				\min_{~(\alpha,c)\in \CL~}
				\min_{~t\in\left[M^*_{\alpha, \beta}\cdot k,\beta\cdot  k\right] \cap \NN~}
				\frac{\exp\left( \frac{\beta k -t} {\alpha }\cdot \ln c \right)} {\p\left( n,k, t, (1 - \frac{\beta}{\alpha}) \cdot k +\frac{t}{\alpha} \right)}\\
			&\geq~n^{-\OO(1)} \max_{~k\in \left[0,\frac{n}{\beta}\right] \cap \NN}	\min_{~(\alpha,c)\in \CL~}\min_{~t\in \left[M_{\alpha,\beta}\left( \frac{k}{n}\right)\cdot n , \beta k\right]~}
				\exp\left( \frac{\beta k -t}{\alpha} \cdot \ln c + h_{k,t}(\x(k,t)) \right)
		\end{aligned}
	\end{equation}
	where the inequality is by \eqref{eq:atmost_fifth}.
	To complete the proof we need to change the range of $k$ in \eqref{eq:atleast_sixth} to a continuous range.
	The following claims are used to this end. 
	
	\begin{claim}
		\label{claim:Mprime_bound}
		Let $\beta, \alpha\geq 1$ and $\kappa\in \left[0,\frac{1}{\beta} \right]$.
		Then
		\[0 \leq M'_{\alpha,\beta}\left(\kappa\right) \leq  
			\begin{cases}
				\frac{\beta^2}{\beta-\alpha} & \textnormal{if } \beta > \alpha\\
				0 & \textnormal{if } \alpha = \beta\\
				\frac{\alpha-\beta}{\alpha-1} & \textnormal{if } \beta<\alpha
			\end{cases},\]
		where $M'_{\alpha,\beta}(\kappa)$ is the derivative of $M_{\alpha,\beta}(\kappa)$.
	\end{claim}
	\begin{claimproof}
		Consider the following cases.
		\begin{itemize}
			\item If $\beta >\alpha$ it holds that 
				\[M'_{\alpha,\beta}(\kappa) = \frac{(1-\alpha \kappa)\cdot (\beta -\alpha)  + \alpha(\beta-\alpha)\cdot \kappa}{(1-\alpha\kappa)^2} =\frac{\beta-\alpha}{(1-\alpha\kappa)^2}.\]
				Observe the value is   well define since $\alpha\cdot \kappa \leq \alpha \cdot \frac{1}{\beta} <1$. 
				Thus, $M'_{\alpha,\beta}(\kappa)  =\frac{\beta-\alpha}{(1-\alpha\kappa)^2}\geq 0$. Similarly, since $\kappa\leq \frac{1}{\beta}$, 
				\[M'_{\alpha,\beta}(\kappa)  =\frac{\beta-\alpha}{(1-\alpha\kappa)^2}\leq \frac{\beta -\alpha}{\left(1-\alpha\cdot \frac{1}{\beta}\right)^2} = \frac{\beta^2}{\beta-\alpha}.\]
			\item If $\alpha=\beta$ it holds that $M'_{\alpha,\beta}(\kappa) =0$.
			\item If $\beta<\alpha$ it holds that $M'_{\alpha,\beta}(\kappa) = \frac{\alpha-\beta}{\alpha-1} > 0$.\qedhere
		\end{itemize}
	\end{claimproof}

	\begin{claim}
		\label{claim:cont_k}
		Let $k\in \left[0,\frac{n}{\beta}\right]$ and $k'=\floor{k}$. Then
		\begin{align*}
			&\min_{(\alpha,c)\in \CL}\min_{~t\in \left[M_{\alpha,\beta}\left( \frac{k}{n}\right)\cdot n, \beta k\right]~}
				\exp\left( \frac{\beta k -t}{\alpha} \cdot \ln c + h_{k,t}(\x(k,t)) \right)\\
			\leq~& n^{\OO(1)} \cdot \min_{(\alpha,c)\in \CL}\min_{~t\in \left[M_{\alpha,\beta}\left( \frac{k'}{n}\right)\cdot n , \beta k' \right]~}
				\exp\left( \frac{\beta k' -t}{\alpha} \cdot \ln c + h_{k',t}(\x(k',t)) \right).
		\end{align*}
	\end{claim}
	\begin{claimproof}
		Pick $({\alpha}',{c'})\in \CL$ and $t' \in \left[M_{\alpha',\beta}\left( \frac{k'}{n}\right)\cdot n , \beta k'\right]$ such that
		\begin{align*}
			&\exp\left( \frac{\beta k' -t'}{\alpha'}  \ln{c'} + h_{k',t'}(\x[\alpha',\beta](k',t')) \right) \\
			=~ &\min_{(\alpha,c)\in \CL}\min_{t\in \left[M_{\alpha,\beta}\left( \frac{k'}{n}\right)\cdot n , \beta k'\right]}
				\exp\left( \frac{\beta k' -t}{\alpha}  \ln c + h_{k',t}(\x(k',t)) \right).
		\end{align*}
		We set $t'' \coloneqq \max\{t', M_{\alpha,\beta}\left( \frac{k}{n}\right)\cdot n\}$.
		Since $t' \leq \beta k'\leq \beta k$ and $M_{\alpha,\beta}\left(\frac{k}{n}\right)\cdot n \leq \beta k$ (\Cref{lem:Mstar_vs_nostar}) it follows that $t''\leq \beta k$.
		In order to bound $\abs{t'-t''}$ we consider the following cases.
		\begin{itemize} 
			\item If $t'=t''$ we have $\abs{t'-t''}=0$.
			\item If $t''=M_{\alpha,\beta}\left(\frac{k}{n}\right) \cdot n $ it holds that $M_{\alpha,\beta}\left( \frac{k'}{n}\right)\cdot n \leq t' \leq M_{\alpha,\beta}\left( \frac{k}{n}\right)\cdot n=t''$.
				Then 
				\[\abs{t''-t'} \leq n\cdot \abs{ M_{\alpha,\beta}\left(\frac{k'}{n}\right)- M_{\alpha,\beta}\left(\frac{k}{n}\right)} \leq n\cdot \OO(1)\cdot \frac{k'-k}{n} = \OO(1),\]
			where the second inequality follows from \Cref{claim:Mprime_bound}. 
		\end{itemize}
		So overall $\abs{t''-t'} = \OO(1)$.
		Hence,
		\begin{align*}
			& \min_{(\alpha,c)\in \CL}\min_{~t\in \left[M_{\alpha,\beta}\left( \frac{k}{n}\right)\cdot n , \beta k \right]~}
				\exp\left( \frac{\beta k -t}{\alpha} \cdot \ln c + h_{k,t}(\x(k,t)) \right)\\
			\leq ~& \exp\left( \frac{\beta k -t''}{\alpha'} \cdot \ln c' + h_{k,t''}(\x[\alpha',\beta](k,t'')) \right)\\
			= ~&\exp \left( \frac{\beta k - t''}{\alpha'}\cdot \ln c'  -t''\cdot \entropy\left( \frac{\x[\alpha',\beta](k,t'')}{t''}\right) 
				-(n-t'')\cdot\entropy\left( \frac{k-\x[\alpha',\beta](k,t'')}{n-t''} \right)+ n\cdot \entropy\left(\frac{k}{n}\right) \right) \\
			= ~&\exp \left( \frac{\beta k - t''}{\alpha'}\cdot \ln c'  -t''\cdot \entropy\left( \frac{\left( 1-\frac{\beta}{\alpha}\right)\cdot k +\frac{t''}{\alpha}}{t''}\right) 
				-(n-t'')\cdot\entropy\left( \frac{\frac{\beta}{\alpha} \cdot{k}- \frac{t''}{\alpha}}{n-t''} \right)+ n\cdot \entropy\left(\frac{k}{n}\right) \right) \\
			\leq ~& n^{\OO(1)} \cdot \exp \left( \frac{\beta k' - t'}{\alpha'}\cdot \ln c'  -t'\cdot \entropy\left( \frac{\left( 1-\frac{\beta}{\alpha}\right) k' +\frac{t'}{\alpha}}{t'}\right) 
				-(n-t')\cdot\entropy\left( \frac{\frac{\beta}{\alpha} {k'}- \frac{t'}{\alpha}}{n-t'} \right)+ n\cdot \entropy\left(\frac{k'}{n}\right) \right)\\
			= ~& n^{\OO(1)} \cdot \exp \left( \frac{\beta k' -t'}{\alpha'}\cdot \ln c' +h_{k',t'}(\x[\alpha',\beta](k',t')\right)\\
			= ~& n^{\OO(1)} \cdot \min_{(\alpha,c)\in \CL}\min_{t\in \left[M_{\alpha,\beta}\left( \frac{k'}{n}\right)\cdot n , \beta k'\right]}
				\exp\left( \frac{\beta k' -t}{\alpha}  \ln c + h_{k',t}(\x(k',t)) \right).
		\end{align*}
		The second inequality follows from \Cref{lem:perturb}.
	\end{claimproof}

	By \eqref{eq:atleast_sixth} and \Cref{claim:cont_k} it holds that
	\begin{align*}
		f_{\CL,\beta}(n) 
		&\geq~n^{-\OO(1)} \max_{~k\in \left[0,\frac{n}{\beta}\right] \cap \NN}	\min_{~(\alpha,c)\in \CL~}\min_{~t\in \left[M_{\alpha,\beta}\left( \frac{k}{n}\right)\cdot k , \beta k\right]~}
			\exp\left( \frac{\beta k -t}{\alpha} \cdot \ln c + h_{k,t}(\x(k,t)) \right)\\
		&\geq ~n^{-\OO(1)} \max_{~k\in \left[0,\frac{n}{\beta}\right] } \min_{~(\alpha,c)\in \CL~}\min_{~t\in \left[M_{\alpha,\beta}\left( \frac{k}{n}\right)\cdot n, \beta k\right]~}
			\exp\left( \frac{\beta k -t}{\alpha} \cdot \ln c + h_{k,t}(\x(k,t)) \right)\\
		&= ~n^{-\OO(1)} \max_{~\kappa\in \left[0,\frac{1}{\beta}\right] } \min_{~(\alpha,c)\in \CL~}\min_{~\tau\in \left[M_{\alpha,\beta}\left( \kappa\right), \beta \kappa\right]~}
			\exp\left( \frac{\beta \kappa -\tau}{\alpha} \cdot \ln c + \frac{1}{n }\cdot h_{\kappa\cdot n,\tau \cdot n}(\x(\kappa\cdot n ,\tau\cdot n  )) \right)^n\\
		&=~n^{-\OO(1)} \max_{~\kappa\in \left[0,\frac{1}{\beta}\right] } \min_{~(\alpha,c)\in \CL~}\min_{~\tau\in \left[M_{\alpha,\beta}\left( \kappa\right), \beta \kappa\right]~}
			\exp\left(g_{\alpha,\beta,c}(\kappa,\tau) \right)^n\\
		&=~n^{-\OO(1)} \left( \amlsbound(\CL,\beta)\right)^n.
	\end{align*}
	The first equality simply replaces $k$ and $t$ with $\kappa \cdot n $ and $\tau \cdot n$.
	The second equality follows from
	\begin{align*}
		&\frac{\beta \kappa -\tau}{\alpha} \cdot \ln c + \frac{1}{n}\cdot h_{\kappa\cdot n,\tau \cdot n}(\x(\kappa\cdot n ,\tau \cdot n )) \\
		=~& \frac{\beta \kappa -\tau}{\alpha} \cdot \ln c - \tau \cdot \entropy\left( \frac{\left(1-\frac{\beta}{\alpha}\right)\cdot \kappa +\frac{\tau}{\alpha}}{\tau}\right) - (1-\tau )\cdot \entropy\left( \frac{\frac{\beta}{\alpha}\cdot \kappa -\frac{\tau}{\alpha} }{1-\tau}\right) + \entropy(\kappa )\\
		=~&g_{\alpha,\beta,c}(\kappa,\tau).
	\end{align*}
\end{proof}

\begin{proof}[Proof of \Cref{lem:f_to_amls}]
	The lemma follows immediately from \Cref{lem:f_atleast_amls,lem:f_atmost_amls}.
\end{proof}

\section{Evaluating the Running Time: Convexity and Concavity}
\label{sec:math_functions}
In this section we prove \Cref{lem:g_convex_by_tau,lem:concave} which provide the mathematical properties required for the evaluation of $\amlsbound$, as well as \Cref{lem:coincide_with_esa}.
In \Cref{sec:math_functions_basic} we state basic properties of the functions $\delta$, $\gamma$ and $g$ which we need later on.
In \Cref{sec:convexity} we show that $g_{\abc}(\kappa,\tau)$ is a convex function of the variable $\tau$ in the interval $\tau \in [\Mab(\kappa), \beta\cdot \kappa]$, for all $\kappa \in \left( 0, \frac{1}{\beta} \right)$.
This fact is used to show the value $\tau \in [\Mab(\kappa), \beta\cdot \kappa]$ that minimizes $g_{\abc}(\kappa,\tau)$ actually belongs to $\left[ \Mab(\kappa), \beta\cdot \kappa \right)$.
In \Cref{sec:gstar_is_concave} we prove that $\gst(\kappa)$ is a concave function of $\kappa$ in the interval $\kappa \in\left[0, \frac{1}{\beta}\right]$.
The proof uses properties of the Hessian which are shown in \Cref{sec:computing_det_hes}, and relies on technical computations from \Cref{sec:g_par_der} and \Cref{sec:hes_formula}.
\Cref{sec:gstar_is_concave} also contains the proof of \Cref{lem:coincide_with_esa} which follows from the technical lemmas proved in the same section.

Recall the definitions of the functions used in the definition of $\amlsbound$:
\begin{alignat*}{2}
	&\delta_{\alpha,\beta}(\kappa, \tau) &&=~
	\begin{cases}
		\frac{\frac{\beta}{\alpha} \kappa -\frac{\tau}{\alpha}}{1-\tau} = \frac{\frac{\beta}{\alpha} \kappa -\frac{1}{\alpha}}{1-\tau}+\frac{1}{\alpha} &\text{if } \tau \neq 1\\
		\frac{1}{\alpha} &\text{if } \tau = 1
	\end{cases}
	\\
	&\gamma_{\alpha,\beta}(\kappa, \tau) &&=~
	\begin{cases}
		\left(1-\frac{\beta}{\alpha}\right)\frac{\kappa}{\tau} +\frac{1}{\alpha} &\text{if } \tau \neq 0\\
		\frac{1}{\alpha}  &\text{if } \tau=0
	\end{cases}
	\\
	&g_{\alpha,\beta,c}(\kappa,\tau) &&= ~\frac{\beta \kappa - \tau}{\alpha} \ln c - \tau\cdot \entropy\left(\gamma_{\alpha,\beta}(\kappa,\tau)\right) -(1-\tau)\cdot \entropy\left(\delta_{\alpha,\beta} (\kappa,\tau)\right) + \entropy\left(\kappa\right)
	\\
	&M_{\alpha, \beta}(\kappa) &&=~
	\begin{cases}
		\frac{\beta - \alpha}{1-\alpha \cdot \kappa} \cdot \kappa &\text{if } \alpha < \beta\\
		0 &\text{if } \alpha =\beta\\
		\frac{\alpha - \beta}{\alpha - 1}\cdot  \kappa &\text{if } \alpha > \beta
	\end{cases}
\end{alignat*}
With a slight abuse of notation, we sometimes omit the subscript $(\alpha, \beta, c)$ from $g_{\alpha,\beta,c}$ or $(\alpha, \beta)$ from $\delta_{\alpha, \beta}$ and $\gamma_{\alpha, \beta}$, whenever it is clear from the context.

\subsection{Basic Properties}
\label{sec:math_functions_basic}

We start by discussing basic properties of the functions $\gamma$, $\delta$ and $g$.
We commonly rely on monotonicity properties of $\delta$ and $\gamma$, as well as their possible range of values.

\begin{lemma}\label{lem:prop_delta}
	For all $\alpha \geq 1$, $\beta > 1$ and $\kappa \in \left( 0, \frac{1}{\beta} \right) $, it holds that $\delta_\ab(\kappa, \tau)$ is strictly decreasing with $\tau$ in the range $\tau \in [\Mab(\kappa), \beta \cdot \kappa]$. Furthermore, $ 0 \leq \delta_\ab(\kappa, \tau) < \frac{1}{\alpha}$ for all $\tau \in [\Mab(\kappa), \beta \cdot \kappa]$. 
\end{lemma}

\begin{proof}
	Since $\kappa < \frac{1}{\beta}$ and $\frac{\beta}{\alpha} \cdot \kappa - \frac{1}{\alpha} < 0$, the term $\frac{\frac{\beta}{\alpha} \cdot \kappa - \frac{1}{\alpha}}{1 - \tau}$ is a strictly decreasing function of $\tau$. Therefore $\delta(\kappa, \tau)$ is strictly decreasing with $\tau$ because $\delta(\kappa, \tau) = \frac{\frac{\beta}{\alpha} \cdot \kappa - \frac{1}{\alpha}}{1 - \tau} + \frac{1}{\alpha}$. Therefore, for all $\tau \in [\Mab(\kappa), \beta \cdot \kappa]$, we have $0 = \delta(\kappa, \beta \cdot \kappa) \leq \delta(\kappa, \tau)$ and
	\begin{align*}
		\delta(\kappa, \tau) \leq \delta(\kappa, \Mab(\kappa)) = \frac{\frac{\beta}{\alpha} \cdot \kappa - \frac{1}{\alpha}}{1 - \Mab(\kappa)} + \frac{1}{\alpha} < \frac{1}{\alpha},
	\end{align*}
	where the last step holds because $\Mab(\kappa) \leq \beta \cdot \kappa < 1$ and $\frac{\beta}{\alpha} \cdot \kappa - \frac{1}{\alpha} < 0$.
\end{proof}

While the function $\delta$ is decreasing regardless of the values of $\alpha$ and $\beta$, the direction of monotonicity of~$\gamma$ does depend on the values of $\alpha$ and $\beta$.

\begin{lemma}\label{lem:prop_gamma}
	For all $\alpha \geq 1$, $\beta > 1$ and $\kappa \in \left( 0, \frac{1}{\beta} \right) $, the function $\gamma_\ab(\kappa, \tau)$ satisfies the following properties depending on the values of $\alpha$ and $\beta$:
	\begin{itemize}
		\item If $\alpha > \beta$, the function $\gamma_\ab(\kappa, \tau)$ is strictly decreasing with $\tau \in [\Mab(\kappa), \beta \cdot \kappa]$, and $\frac{1}{\beta} \leq \gamma_\ab(\kappa, \tau) \leq 1$.
		\item If $\alpha = \beta$, it holds that  $\gamma_\ab(\kappa, \tau) = \frac{1}{\alpha}$ for all $\tau \in [\Mab(\kappa), \beta \cdot \kappa]$.
		\item If $\alpha < \beta$, the function $\gamma_\ab(\kappa, \tau)$ is strictly increasing with $\tau \in [\Mab(\kappa), \beta \cdot \kappa]$, and $\kappa \leq \gamma_\ab(\kappa, \tau) \leq \frac{1}{\beta}$.
	\end{itemize}
\end{lemma}

\begin{proof}
	Let us fix a $\kappa \in \left( 0, \frac{1}{\beta} \right) $ and consider the different cases where $\alpha > \beta$, $\alpha = \beta$ and $\alpha < \beta$.

	\begin{itemize}
		\item If $\alpha > \beta$, the term $\left( 1 - \frac{\beta}{\alpha} \right)$ is strictly positive, therefore $\left( 1 - \frac{\beta}{\alpha} \right) \cdot \frac{\kappa}{\tau}$ is a strictly decreasing function of $\tau$.
			It follows that $\gamma(\kappa, \tau) = \left( 1 - \frac{\beta}{\alpha} \right) \cdot \frac{\kappa}{\tau} + \frac{1}{\alpha}$ is also a strictly decreasing function of $\tau$.
			Moreover, for all $\tau \in [\Mab(\kappa), \beta \cdot \kappa]$, we have
			\begin{align*}
				\gamma(\kappa,\tau) & \geq \gamma\left( \kappa, \beta \cdot \kappa \right)  = \left( 1 - \frac{\beta}{\alpha} \right)\cdot \frac{\kappa}{\beta \cdot \kappa} + \frac{1}{\alpha} = \frac{1}{\beta} \qquad \text{ and}\\
				\gamma(\kappa,\tau) & \leq \gamma\Big(\kappa, \Mab(\kappa)\Big) = \gamma\left( \kappa, \frac{\alpha - \beta}{\alpha - 1}\cdot \kappa \right)  = \left( 1 - \frac{\beta}{\alpha} \right)\cdot \frac{\alpha - 1}{\alpha - \beta} + \frac{1}{\alpha} = 1.
			\end{align*}

		\item If $\beta = \alpha$, it is easy to see that $\gamma(\kappa, \tau) = \frac{1}{\alpha}$ for all $\tau \in [\Mab(\kappa), \beta \cdot \kappa]$ because $(1 - \frac{\beta}{\alpha}) = 0$.
		\item If $\alpha < \beta$, the term $\left( 1 - \frac{\beta}{\alpha} \right) $ is strictly negative and $\left( 1 - \frac{\beta}{\alpha} \right) \cdot \frac{\kappa}{\tau}$ is a strictly increasing function of $\tau$.
			Therefore $\gamma(\kappa, \tau) = \left( 1 - \frac{\beta}{\alpha} \right) \cdot \frac{\kappa}{\tau} + \frac{1}{\alpha}$ is also a strictly increasing function of $\tau$.
			Moreover, for all $\tau \in [\Mab(\kappa), \beta \cdot \kappa]$ it holds that
			\begin{align*}
				\gamma(\kappa,\tau) &\leq \gamma(\kappa, \beta \cdot \kappa) = \frac{1}{\beta} \qquad \text{ and} \\
				\gamma(\kappa,\tau) &\geq \gamma\Big(\kappa, M_{\alpha,\beta}(\kappa)\Big) = \gamma\Big( \kappa, \frac{(\beta - \alpha)}{1 - \alpha\cdot\kappa} \cdot \kappa \Big)  = \left( 1 - \frac{\beta}{\alpha} \right)\cdot \frac{1- \alpha\cdot\kappa}{\beta - \alpha} + \frac{1}{\alpha}\\
				&= \frac{\alpha \kappa - 1 + 1}{\alpha}\\
			    &=\kappa.	
			\end{align*}
	\end{itemize}
\end{proof}

The next lemma is also used a few times.
\begin{lemma}
	\label{lem:delta_gamma_extreme_tau}
	Let $\beta>\alpha >1$ and $\kappa\in \left(0,\frac{1}{\beta}\right)$.
	Then
	$$\delta_{\alpha,\beta} \left( \kappa, M_{\alpha,\beta}(\kappa)\right) = \gamma_{\alpha,\beta}(\kappa, M_{\alpha,\beta}(\kappa)) =\kappa.$$
\end{lemma}
\begin{proof}
	By simple calculation we have
	\begin{align*}
		\delta(\kappa,\Mab(\kappa)) &= \frac{1}{\alpha} \cdot \left( \frac{\beta \cdot \kappa - 1}{1 - \frac{(\beta - \alpha) \cdot \kappa}{1 - \alpha \cdot \kappa}} + 1\right)\\
		&= \frac{1}{\alpha} \cdot \left( \frac{\beta \cdot \kappa - 1}{\frac{1 - \alpha \cdot \kappa  - \beta \cdot \kappa + \alpha \cdot \kappa}{1 - \alpha \cdot \kappa}}  +1 \right)\\
		&= \frac{1}{\alpha} \cdot \left( (1 - \alpha \cdot \kappa) \cdot \frac{\beta \cdot \kappa - 1}{1 - \beta\cdot \kappa} + 1\right)\\
		&= \frac{1}{\alpha} \cdot (\alpha \cdot \kappa )\\
		&= \kappa.
	\end{align*}
	Similarly, 
	\begin{align*}
		\gamma(\kappa,\Mab(\kappa)) &= \left( 1 - \frac{\beta}{\alpha} \right) \cdot \frac{\kappa}{\frac{(\beta - \alpha) \cdot \kappa}{1 - \alpha \cdot \kappa}} + \frac{1}{\alpha}\\
		&= \frac{\alpha - \beta}{\alpha} \cdot \frac{\kappa \cdot \left( 1 - \alpha \cdot \kappa \right) }{(\beta - \alpha) \cdot \kappa} + \frac{1}{\alpha}\\
		&= \frac{\alpha \cdot \kappa - 1}{\alpha} + \frac{1}{\alpha}\\
		&= \kappa.
	\end{align*}
\end{proof}

Finally, we show $g$ is non-negative.

\begin{lemma}
	\label{lem:g_nonneg}
	For all $\alpha,\beta,c \geq 1$, $\kappa \in \left[ 0,\frac{1}{\beta}\right]$ and $\tau \in \left[M_{\alpha,\beta}(\kappa),\beta\kappa\right]$ it holds that $g_{\abc}(\kappa,\tau)\geq 0$.
\end{lemma}

\begin{proof}
	It holds that 
	$$
	\begin{aligned}
		g_{\alpha,\beta,c}(\kappa,\tau) ~&=~ \frac{\beta \kappa-\tau}{\alpha}\cdot\ln c - \tau \cdot \entropy\left(\gamma_{\alpha,\beta} (\kappa,\tau)\right) - (1-\tau)\cdot \entropy\left(\delta_{\alpha,\beta}(\kappa,\tau) \right) +\entropy(\kappa)\\
		~&\geq~0 - \tau \cdot \entropy\left(\gamma_{\alpha,\beta} (\kappa,\tau)\right) - (1-\tau)\cdot \entropy\left(\delta_{\alpha,\beta}(\kappa,\tau) \right) +\entropy(\kappa)\\
		&\geq~ - \entropy\left(\tau\cdot  \gamma_{\alpha,\beta} (\kappa,\tau) + (1-\tau)\cdot \delta_{\alpha,\beta}(\kappa,\tau) \right) +\entropy(\kappa)\\
		&=~ - \entropy(\kappa) + \entropy(\kappa) ~= 0,
	\end{aligned}
	$$
	where the first inequality holds since $\tau \leq \beta \cdot \kappa$ and $c\geq 1$.
	The second inequality holds as $\entropy$ is concave, and the second equality follows from $\tau \cdot \gamma(\kappa,\tau) +(1-\tau)\cdot\delta(\kappa,\tau)  = \kappa$.
\end{proof}

\subsection{Convexity}
\label{sec:convexity}

In this section we fix some value for $\kappa$ and analyze the function $g_{\abc}(\kappa, \tau)$ as a function of $\tau$.
For all $\kappa \in \left[0, \frac{1}{\beta}\right]$, we define $\gk[\abc] \colon (\Mab(\kappa), \beta \cdot \kappa) \to \mathbb{R}$ via
\begin{align*}
	\gk[\abc](\tau) \coloneqq g_{\abc}(\kappa, \tau).
\end{align*}
Recall that $\D{a}{b}  = a \ln \frac{a}{b} + (1-a)\ln \frac{1-a}{1-b}$ is the Kullback-Leibler divergence between two Bernoulli distributions with parameters $a$ and $b$. 

Let $\gpart[\abc](\kappa_0, \tau_0)$ and $\gpartt[\abc](\kappa_0, \tau_0)$ denote the first and second order partial derivatives of the function $g_{\abc}(\kappa,\tau)$ with respect to the variable $\tau$, evaluated at $(\kappa_0, \tau_0)$.
\Cref{lem:deriv_g_by_tau,lem:deriv_g_by_tau_tau} provide formulas for $\gpart[\abc](\kappa_0, \tau_0)$ and $\gpartt[\abc](\kappa_0, \tau_0)$.
The lemmas follow from a simple calculation and we defer the proofs to \Cref{sec:g_par_der}.

\begin{restatable}{lemma}{derivgbytau}
	\label{lem:deriv_g_by_tau}
	For all $\alpha,c \geq 1$ and  $\beta > 1$ it holds that
	\begin{align*}
		\gpart[\abc](\kappa,\tau) = -\frac{\ln(c)}{\alpha} - \D{\frac{1}{\alpha}}{\gamma(\kappa,\tau)} + \D{\frac{1}{\alpha}}{\delta(\kappa,\tau)}.
	\end{align*}
\end{restatable}

We use the following functions to simplify the formula for $\gpartt[\abc](\kappa,\tau)$:
\begin{align}
	\label{eq:Gamma_def}
	\Gamma_{\alpha,\beta}(\kappa,\tau) \,&= \, \frac{1}{\tau} \cdot \frac{1}{\gamma_{\alpha,\beta} (\kappa,\tau) \cdot \left( 1- \gamma_{\alpha,\beta}(\kappa, \tau)\right)} \\ 
	\label{eq:Delta_def}
	\Delta_{\alpha,\beta}(\kappa,\tau) \,&= \, \frac{1}{(1-\tau)} \cdot \frac{1}{\delta_{\alpha,\beta} (\kappa,\tau) \cdot \left( 1- \delta_{\alpha,\beta}(\kappa, \tau)\right)} 
\end{align}

\begin{restatable}{lemma}{derivgbytautau}
	\label{lem:deriv_g_by_tau_tau}
	For all $\alpha,c \geq 1$ and  $\beta > 1$, the second order partial derivative of $g(\kappa,\tau)$ by $\tau$, i.e., $\gpartt[\abc](\kappa,\tau)$ is given by
	\begin{align*}
		\gpartt[\abc](\kappa,\tau) &=
		\left( \gamma(\kappa,\tau) - \frac{1}{\alpha} \right)^{2} \cdot \Gamma_{\alpha,\beta}(\kappa,\tau) +\left( \delta(\kappa,\tau) - \frac{1}{\alpha} \right)^{2} \cdot \Delta_{\alpha,\beta}(\kappa,\tau).
	\end{align*}
\end{restatable}

We sometimes omit the subscript $\abc$ from $\gk[\abc]$, $\gpart[\abc]$ and $\gpartt[\abc]$ unless it causes confusion.

\begin{lemma}
	\label{lem:gk_convex_in_tau}
	Let $\alpha,c \geq 1$, $\beta > 1$ and $0 < \kappa < \frac{1}{\beta}$.
	The function $\gk[\abc](\tau)$ is strictly convex in the open interval $\left( M_{\alpha,\beta}(\kappa), \beta \cdot \kappa  \right)$.
	In particular, the second order partial derivative of $g_{\abc}(\kappa,\tau)$ with respect to $\tau$ is strictly positive, i.e., $\gpartt(\kappa,\tau) > 0$ for all $\tau \in \left(M_{\alpha,\beta}(\kappa), \beta \cdot \kappa\right)$.
\end{lemma}

\begin{proof}
	Let $\alpha,c\geq 1$, $\beta > 1$, by \Cref{lem:deriv_g_by_tau_tau} and \eqref{eq:Delta_def} and \eqref{eq:Gamma_def} we have
	\begin{align*}
		\gpartt(\kappa,\tau) &= \frac{1}{\tau} \cdot\frac{\left( \gamma(\kappa,\tau) - \frac{1}{\alpha} \right)^{2}}{\gamma(\kappa,\tau) \cdot \Bigl(1 - \gamma(\kappa,\tau)\Bigr)} + \frac{1}{1-\tau}\cdot \frac{\left( \delta(\kappa,\tau) - \frac{1}{\alpha} \right)^{2}}{\delta(\kappa,\tau) \cdot \Bigl(1 - \delta(\kappa,\tau)\Bigr)}.
	\end{align*}
	By \Cref{lem:prop_delta,lem:prop_gamma}, we have that $0 < \delta(\kappa,\tau) < \frac{1}{\alpha} < 1$ and $0 < \gamma(\kappa,\tau) < 1$, since $\beta > 1$ and $\tau \in \left(M_{\alpha,\beta}(\kappa), \beta \cdot \kappa\right)\subseteq (0,1)$.
	Therefore, it follows that $\gpartt(\kappa,\tau) > 0$.
	Thus $\gk(\tau)$ is a strictly convex function.
\end{proof}

Since $\gk[\abc](\tau)$ is a continuous function of the variable $\tau$, it attains its minimum over the closed interval $[\Mab(\kappa), \beta\cdot \kappa]$.
We show that up to some corner cases, the optimal value of $\tau$ lies within the open interval $\left(\Mab(\kappa), \beta\cdot\kappa\right)$.
We use the following definition to easily exclude the corner cases.

\begin{definition}
	\label{def:simple}
	We say $\alpha,\beta,c\geq 1$ are \emph{simple} if $\beta>1$ and none of the following conditions hold:
	\begin{itemize}
		\item  $\alpha =\beta$, or
		\item  $c=1$  and  $\alpha<\beta$.
	\end{itemize}
\end{definition}

\begin{lemma}
	\label{lem:minimizer_tau}
	For all $\alpha \geq 1$, $c \geq 1$, $\beta > 1$ and $\kappa \in \left( 0, \frac{1}{\beta} \right) $, there is a unique value of $\tau \in [\Mab(\kappa), \beta \cdot \kappa]$, denoted by $\taust[\abc](\kappa)$, that minimizes $g_{\abc}(\kappa,\tau)$.
	Furthermore, it holds that $\taust[\abc](\kappa)<\beta \kappa$, and if $\abc$ are simple then $M_{\alpha,\beta}(\kappa)<\taust[\abc](\kappa)<\beta \kappa$ and $\gpart\left( \kappa, \taust[\abc](\kappa) \right) = 0$.
\end{lemma}
	
\begin{proof}
	Since the function $\gk(\tau)$ is strictly convex by \Cref{lem:gk_convex_in_tau}, the value
	$$\taust[\abc] (\kappa)=\argmin_{\tau \in  [\Mab(\kappa), \beta \cdot \kappa]}  g_{\abc}(\kappa,\tau)$$
	is uniquely defined.
	Also note that $\gpart(\kappa,\tau)$ is an increasing function of $\tau$ by \Cref{lem:gk_convex_in_tau}.

	By \Cref{lem:deriv_g_by_tau} we have
	\begin{equation}
		\label{eq:deriv_tau_to_bk}
		\begin{aligned}
			\lim_{\tau \to \beta\cdot\kappa} \gpart(\kappa,\tau) &= \lim_{\tau \to \beta\cdot\kappa} \Biggl(-\frac{\ln(c)}{\alpha} - \D{\frac{1}{\alpha}}{\gamma(\kappa,\tau)} + \D{\frac{1}{\alpha}}{\delta(\kappa,\tau)}\Biggr)\\
			&= -\frac{\ln(c)}{\alpha} - \D{\frac{1}{\alpha}}{\frac{1}{\beta}} + \D{\frac{1}{\alpha}}{\lim_{\tau \to \beta\cdot\kappa} \delta(\kappa,\tau)} = \infty,
		\end{aligned}
	 \end{equation}
	which follows from the fact that $\D{\frac{1}{\alpha}}{x}$ is a continuous function of $x$ and $\lim_{\tau \to \beta\cdot\kappa} \delta(\kappa,\tau) = 0$.
	Thus, there is $\eps>0$ such that $\gpart(\kappa,\tau)>0$ for all $\tau\in \left(\beta \kappa- \eps, \beta\kappa\right)$, and hence $g(\kappa,\tau)$ is strictly increasing in $\left(\beta \kappa- \eps, \beta\kappa\right)$, and by continuity in $\left(\beta \kappa- \eps, \beta\kappa\right]$.
	This implies that $\taust[\abc](\kappa)< \beta \kappa$ by its definition.

	This above completes the proof for general values of $\abc$, and thus we can assume $\abc$ are simple from this point onward.
	Using \Cref{lem:deriv_g_by_tau} once  more we get
	\begin{equation}\label{eq:gpart_lim}
		\lim_{\tau \to \Mab(\kappa)} \gpart(\kappa,\tau) = -\frac{\ln(c)}{\alpha} -\lim_{\tau \to \Mab(\kappa)} \D{\frac{1}{\alpha}}{ \gamma(\kappa,\tau)} + \lim_{\tau \to \Mab(\kappa)}\D{\frac{1}{\alpha}}{\delta\left( \kappa, \tau \right) }.
	\end{equation}
	Consider the following cases.
	\begin{itemize}
		\item  If $\alpha > \beta > 1$ and $c \geq 1$, we have
		 	\begin{align*}
		 		\lim_{\tau \to \Mab(\kappa)} \gamma(\kappa,\tau) = \gamma(\kappa,\Mab(\kappa)) = \gamma\left( \kappa, \frac{\alpha - \beta}{\alpha - 1}\cdot \kappa \right) =1,
		 	\end{align*}
		 	therefore $\lim_{\tau \to \Mab(\kappa)} \D{\frac{1}{\alpha}}{ \gamma(\kappa,\tau)} = \infty$.
		 	Similarly, by \Cref{lem:prop_delta}, we have
		 	\begin{align*}
		 		\lim_{\tau \to \Mab(\kappa)} \delta(\kappa,\tau) = \delta\left( \kappa, \Mab(\kappa)  \right) < \frac{1}{\alpha} < 1,
		 	\end{align*}
		 	and 
		 	\begin{equation}
		 		\lim_{\tau \to \Mab(\kappa)} \delta(\kappa,\tau) = \delta\left( \kappa, \Mab(\kappa)  \right) > 0,
		 	\end{equation}
		 	since $\Mab(\kappa) > 0$.
		 	Therefore it holds that $\lim_{\tau \to \Mab(\kappa)} \D{\frac{1}{\alpha}}{ \delta(\kappa,\tau)} < \infty$ and by \eqref{eq:gpart_lim} we get
		 	\begin{align*}
		 		\lim_{\tau \to M_{\alpha,\beta}(\kappa)} \gpart(\kappa,\tau) &= \lim_{\tau \to \beta\cdot\kappa} \Biggl(-\frac{\ln(c)}{\alpha} - \D{\frac{1}{\alpha}}{\gamma(\kappa,\tau)} + \D{\frac{1}{\alpha}}{\delta(\kappa,\tau)}\Biggr) = -\infty.
		 	\end{align*}
		 	
		\item If $\beta > \alpha \geq 1$, $c > 1$, by \Cref{lem:delta_gamma_extreme_tau} we have
		 	\begin{align*}
		 		\lim_{\tau \to \Mab(\kappa)} \delta(\kappa,\tau) = \delta(\kappa,\Mab(\kappa)) = \kappa
		 	\end{align*}
		 	and
		 	\begin{align*}
		 		\lim_{\tau \to \Mab(\kappa)} \gamma(\kappa,\tau) = \gamma(\kappa,\Mab(\kappa)) = \kappa.
		 	\end{align*}
		 	Thus, by \eqref{eq:gpart_lim}, we have
		 	\begin{align*}
		 		\lim_{\tau \to \Mab(\kappa)} \gpart(\kappa,\tau) = -\frac{\ln(c)}{\alpha} - \D{\frac{1}{\alpha}}{\kappa} + \D{\frac{1}{\alpha}}{\kappa} = -\frac{\ln(c)}{\alpha} < 0,
		 	\end{align*}
		 	since $c > 1$.
	\end{itemize}
	So in both cases $\lim_{\tau \to \Mab(\kappa)} \gpart(\kappa,\tau)<0$.
	Thus, by \eqref{eq:deriv_tau_to_bk} there is $\tilde{\tau} \in \left( M_{\alpha,\beta}(\kappa),\beta \cdot \kappa\right)$ such that $\gpart(\kappa,\tilde{\tau})=0$, and since $g_{\abc}(\kappa,\tau)$ is convex as a function of $\tau$ this implies $\taust[\abc](\kappa)=\tilde{\tau}$.
	Hence $\gpart(\kappa,\taust[\abc](\kappa))=0$ and $\taust[\abc](\kappa)\in \left( M_{\alpha,\beta}(\kappa),\beta \cdot \kappa\right)$.
\end{proof}
	
\gconvexbytau*

\begin{proof}
	The first part of the claim follows from the fact that $\gk(\tau)$ is a continuous function on the closed interval $[\Mab(\kappa), \beta \cdot \kappa ]$ and convex on the open interval $\left( \Mab(\kappa), \beta \cdot \kappa  \right) $ by \Cref{lem:gk_convex_in_tau}.
	The second part of the claim simply follows from \Cref{lem:gk_convex_in_tau,lem:minimizer_tau}.
\end{proof}

\subsection{Concavity}
\label{sec:gstar_is_concave}

In this section we prove \Cref{lem:concave}, that is, we show $\gst[\abc](\kappa)$ is concave.
The proof relies on properties of the Hessian of $g_{\abc}$ when $\abc$ are simple (see \Cref{def:simple}).
The excluded corner cases, in which $\abc$ are not simple, are handled separately.

As in previous sections, we use $g$ and $g^*$ instead of $g_{\abc}$ and $g^*_{\abc}$ when the values of $\alpha$, $\beta$ and~$c$ are known by context.
Recall that the Hessian matrix of~$g$ at $(\kappa,\tau)$, denoted $H_{g}(\kappa,\tau)$, is defined by
\[
\begingroup
\renewcommand*{\arraystretch}{1.5}
H_g(\kappa,\tau)=
\begin{pmatrix}
	\frac{\partial^2 g(\kappa,\tau)}{\partial \kappa^2} & \frac{\partial^2 g(\kappa,\tau)}{\partial \kappa \partial \tau} \\
	\frac{\partial^2 g(\kappa,\tau)}{\partial \kappa \partial \tau} & \frac{\partial^2 g(\kappa,\tau)}{\partial \tau^2}
\end{pmatrix}
\endgroup.
\]
For every $(\kappa,\tau)$ in the domain of $g$, $\dhes{g}{\kappa,\tau}$ denotes the determinant of the Hessian of the function $g$ evaluated at $(\kappa,\tau)$.
Specifically, we have
\begin{equation*}
	\dhes{g}{\kappa,\tau} = \frac{\partial^2 g(\kappa,\tau)}{\partial \kappa^2} \cdot \frac{\partial^2 g(\kappa,\tau)}{\partial \tau^2}- \left(  \frac{\partial^2 g(\kappa,\tau)}{\partial \kappa \partial \tau}\right)^{2} .
\end{equation*}
Our proof is motivated by the second partial derivative test for multivariate functions, which uses the Hessian to classify critical points to maximum, minimum and saddle points.
Technically, we use the Hessian directly and do explicitly rely on the second derivative test.
 
Recall that $\taust[\abc](\kappa)$ is the unique value of the $\tau \in [\Mab(\kappa), \beta \cdot \kappa]$ that minimizes $\gk[\abc](\tau)$.

\begin{restatable}{lemma}{hessiannegative}
	\label{lem:hessian_negative}
	For all  simple $\abc\geq 1$ and $\kappa \in \left(0, \frac{1}{\beta}\right)$, the determinant of the Hessian of $g$ at $(\kappa, \tau^*(\kappa))$ is negative, i.e., $\abs{H_g(\kappa, \tau^*(\kappa) )} < 0$.
\end{restatable}

The proof of \Cref{lem:hessian_negative} is given in \Cref{sec:computing_det_hes}.
We also use the next theorem from \cite{Oliveira13} (see also \cite{KrantzP02}) to show that $\tau^*_{\abc}$ is continuously differentiable and to calculate its derivative.

\begin{theorem}[Implicit Function Theorem for $\mathbb{R}^{2}$, {\cite[Theorem 4]{Oliveira13}}]
	\label{thm:implicit_function_theorem}
	Let $\G(x,y)$ be a real-valued continuously differentiable function defined in a neighbourhood of $\left( x_0, y_0 \right) \in \mathbb{R}^{2}$. Suppose that $\G(x,y)$ satisfies the two conditions
	\begin{align*}
		\G(x_0, y_0) &= 0,\\
		\frac{\partial \G(x,y)}{\partial y}\Bigg|_{(x,y) = (x_0, y_0)} &> 0.
	\end{align*}
	Then there exist open intervals $U\subseteq \mathbb{R}$ and $V\subseteq \mathbb{R}$, with $x_0 \in U, y_0 \in V$, and a  function $G \colon U \to V$ satisfying
	\begin{equation*}
		\G(x, G(x)) = 0, \qquad \text{ for all } x \in U.
	\end{equation*}
	Furthermore, this function $G$ is continuously differentiable with
	\begin{equation*}
		 G'(x_0) = \frac{\partial G(x)}{\partial x}\Bigg|_{x = x_0} = -\frac{\frac{\partial \G(x,y)}{\partial x}\Bigg|_{(x,y) = (x_0, y_0)}}{\frac{\partial \G(x,y)}{\partial y}\Bigg|_{(x,y) = (x_0, y_0)}}.
	\end{equation*}
\end{theorem}

We use \Cref{thm:implicit_function_theorem} in the proof of the following lemma.

\begin{lemma}
	\label{lem:tstar_differentiable}
	 For every simple $\alpha,\beta,c\geq1 $,  the function $\taust[\abc](\kappa)$ is continuously differentiable on $\left(0,\frac{1}{\beta}\right)$. Moreover, for all $\kappa_0\in \left(0,\frac{1}{\beta}\right)$  it holds that
	\begin{equation*}
		\frac{\partial \taust[\abc](\kappa)}{\partial \kappa}\Bigg|_{\kappa = \kappa_0} = - \frac{\gparkt[\abc](\kappa_0, \tau^*_{\abc}(\kappa_0))}{\gpartt[\abc](\kappa_0, \tau^*_{\abc}(\kappa_0))}.
	\end{equation*}
\end{lemma}

\begin{proof}
	Let $\kappa_0 \in \left( 0, \frac{1}{\beta} \right)$ and $\tau_0 = \taust(\kappa_0)$.
	Consider a function $\G$ defined on a neighborhood $E$ of $\left( \kappa_0 ,\tau_0 \right)$ by
	\begin{equation*}
		\G\left( {\kappa}, {\tau} \right) = \gpart[\abc]\left( {\kappa}, {\tau} \right),
	\end{equation*}
	for every $\left(\kappa, \tau\right) \in E$.
	\Cref{lem:deriv_g_by_tau_tau} implies that the partial derivative $\frac{\partial g(\kappa,\tau)}{\partial \tau}$ of $g(\kappa,\tau)$ is continuously differentiable.
	By \Cref{lem:minimizer_tau} we have that $\G\left( \kappa_0, \tau_0 \right) = 0$.
	Furthermore, by \Cref{lem:gk_convex_in_tau} we also have
	\begin{equation*}
		\frac{\partial \G}{\partial \tau} (\kappa_0, \tau_0) = \frac{\partial^{2} g(\kappa,\tau)}{\partial \tau^2}\Bigg|_{(\kappa,\tau) = (\kappa_0, \tau_0)} > 0,
	\end{equation*}
	therefore \Cref{thm:implicit_function_theorem} implies that exists open intervals $U,V$ with $\kappa_0 \in U, \tau_0 \in V$ and a  continuously differentiable function $G \colon U \to V$ such that
	\begin{equation*}
		\G \bigl({\kappa}, G(\overline{\kappa}) \bigr)  = \frac{\partial g(\kappa,\tau)}{\partial \tau}\Bigg|_{(\kappa,\tau) = (\overline{\kappa}, G(\overline{\kappa}))} = 0 \qquad \text {for all } \overline{\kappa} \in U .
	\end{equation*}
	By \Cref{lem:minimizer_tau} it holds that $\tau_0\in\left( M_{\alpha,\beta}(\kappa_0), \beta \kappa_0\right)$.
	So there is an environment $U'\subseteq U$ of $\kappa_0$ such that $G({\kappa}) \in \left(M_{\alpha,\beta}({\kappa}), \beta {\kappa}\right)$ for all ${\kappa}\in U'$.
	By \Cref{lem:gk_convex_in_tau} it also holds that $\gk$ is strictly convex for every $\kappa \in U'$.
	Thus, by the definition of $\taust$, we have $G(\kappa) = \taust(\kappa)$ for every $\kappa \in U'$.

	This implies that~$\taust(\kappa)$ is continuously differentiable in a neighborhood of $\kappa_0$.
	Since this holds for all $\kappa_0 \in \left( 0, \frac{1}{\beta} \right) $, it follows that $\taust(\kappa)$ is continuously differentiable on $\left(0, \frac{1}{\beta}\right)$.
	Moreover, \Cref{thm:implicit_function_theorem} further implies
	\begin{equation*}
		G'(\kappa_0) = \frac{\partial \taust(k)}{\partial \kappa}\Bigg|_{\kappa = \kappa_0} = - \frac{\gparkt[\abc](\kappa_0, \tau_0)}{\gpartt[\abc](\kappa_0, \tau_0)}.\qedhere
	\end{equation*}
\end{proof}

\begin{lemma}
	\label{lem:gstar_second_deriv_negative}
	For every simple $\alpha,\beta,c\geq 1$ and $\kappa\in \left(0,\frac{1}{\beta}\right)$ It holds that
	\begin{equation*}
		\frac{\partial^2 \gst[\abc](\kappa)}{\partial \kappa^{2}}\Bigg|_{\kappa = \kappa_0} < 0.
	\end{equation*}
\end{lemma}

\begin{proof}
	Let $\kappa_0\in \left(0,\frac{1}{\beta}\right)$ and $\tau_0=\tau^*_{\abc}(\kappa_0)$. 
	Using the chain rule for differentiation we get
	\begin{equation}\label{eq:gstar_deriv}
		\begin{aligned}
		\frac{\partial \gst(\kappa)}{\partial \kappa}\Bigg|_{\kappa = \kappa_0} &~=~ \frac{\partial g(\kappa,\tau)}{\partial \kappa}\Bigg|_{(\kappa,\tau) = (\kappa_0, \tau_0)} + \frac{\partial g(\kappa,\tau)}{\partial \tau}\Bigg|_{(\kappa,\tau) = (\kappa_0, \tau_0)} \cdot \frac{\partial \taust(\kappa)}{\partial \kappa}\Bigg|_{\kappa = \kappa_0} \\
		&~=~ \frac{\partial g(\kappa,\tau)}{\partial \kappa}\Bigg|_{(\kappa,\tau) = (\kappa_0, \tau_0)} \\
		&~=~\gpark (\kappa_0,\tau^*(\kappa_0))
		\end{aligned}
	\end{equation}
	where the second equality follows from $\frac{\partial g(\kappa,\tau)}{\partial \tau}\Bigg|_{(\kappa,\tau) = (\kappa_0, \tau_0)} = 0$ by \Cref{lem:minimizer_tau}. By \eqref{eq:gstar_deriv} and using \Cref{lem:tstar_differentiable} we get
	\begin{align*}
		\frac{\partial^2 \gst(\kappa)}{\partial^2 \kappa}\Bigg|_{\kappa = \kappa_0} &=~
		\gparkk(\kappa_0, \tau_0) + \frac{\partial \tau^*(\kappa) }{\partial \kappa} \bigg|_{\kappa=\kappa_0} \cdot \gparkt(\kappa_0,\tau_0)\\
		&=\gparkk(\kappa_0, \tau_0)  - \frac{\gparkt(\kappa_0, \tau_0)}{\gpartt(\kappa_0, \tau_0)} \cdot \gparkt(\kappa_0,\tau_0)
		\\
		&=  ~\frac{1}{\gpartt(\kappa_0, \tau_0)} \cdot \Biggl(\gparkk(\kappa_0, \tau_0) \cdot \gpartt(\kappa_0, \tau_0) - \Bigl(\gparkt(\kappa_0, \tau_0)\Bigr)^2 \Biggr)\\
		&=~ \frac{1}{\gpartt(\kappa_0, \tau_0)} \cdot \dhes{g}{\kappa_0, \tau_0}\\
		 &<0,
	\end{align*}
	where the last inequality follows from \Cref{lem:hessian_negative,lem:gk_convex_in_tau}.
\end{proof}

It can also be easily shown using standard calculus arguments that $\gk[\abc](\kappa)$ is continuous in the closed interval $\left[ 0,\frac{1}{\beta}\right]$ (see, e.g., \cite{Still18}).
Thus, the following corollary is an immediate consequence of \Cref{lem:gstar_second_deriv_negative}.

\begin{corollary}
	\label{cor:simple_concave}
	For every simple $\alpha,\beta,c\geq 1$ it holds that $g^*_{\abc}$ is concave on the interval $\left[0,\frac{1}{\beta}\right]$.
\end{corollary}

To complete the proof \Cref{lem:concave} we need to handle the corner cases excluded from \Cref{cor:simple_concave}.

\begin{lemma}\label{lem:gst_conc_first_corner}
	For all $\beta > \alpha \geq 1$ and $c = 1$, it holds that $\gst[\abc](\kappa)$ is concave in the interval $\left[0, \frac{1}{\beta}\right]$.
\end{lemma}

\begin{proof}
	For every $\kappa \in \left[ 0, \frac{1}{\beta} \right]$, by the definition of $\taust$, it holds that
	\begin{align*}
		g_{\alpha, \beta ,1}\Bigl( \kappa, \taust[\alpha, \beta, 1](\kappa) \Bigr) &\leq g\left( \kappa, \Mab(\kappa) \right) \\
		&= \frac{\beta \cdot \kappa - \Mab(\kappa)}{\alpha} \cdot \ln(1) - \Mab(\kappa) \cdot \entropy \Bigl(\gamma\left( \kappa, \Mab(\kappa) \right) \Bigr) \\ &\quad- \left( 1 - \Mab(\kappa)  \right) \cdot \entropy \Bigl(\delta\left( \kappa, \Mab(\kappa) \right) \Bigr) + \entropy(\kappa)\\
		&= \entropy(\kappa) \cdot \biggl( -\Mab(\kappa) - \Bigl(1 - \Mab(\kappa)\Bigr) + 1 \biggr) \\
		&= 0,
	\end{align*}
	where the second equality follows from \Cref{lem:delta_gamma_extreme_tau}.
	By \Cref{lem:g_nonneg} we also have $g_{\alpha, \beta ,1}\Bigl( \kappa, \taust[\alpha, \beta, 1](\kappa) \Bigr) \geq 0$.
	So $\gst[\alpha, \beta ,1]\Bigl( \kappa) = g_{\alpha, \beta ,1}\Bigl( \kappa, \taust[\alpha, \beta, 1](\kappa) \Bigr)=0$.
	Thus, the function is trivially concave.
\end{proof}
	
Another easy to handle corner case occurs when $\alpha = \beta$ and $c=1$.

\begin{lemma}
	\label{lem:a_is_b_and_c_is_one}
	Let $\alpha > 1$ and $c=1$, then  $g^*_{\alpha,\alpha,c}(\kappa)=0$ for all $\kappa \in \left[0,\frac{1}{\beta}\right]$.
\end{lemma}

\begin{proof}
	Let $\kappa\in\left[0,\frac{1}{\beta}\right]$.
	By \Cref{lem:g_nonneg}, for every $\tau \in \left[M_{\alpha,\alpha}(\kappa),\alpha\cdot \kappa \right] = [0,\alpha\kappa]$, it holds that
	$g_{\alpha,\alpha,c}(\kappa,\tau) \geq 0$.
	Furthermore,
	$$
	\begin{aligned}
		g_{\alpha,\alpha,c}(\kappa,0) ~&=~\left(\kappa-\frac{0}{\alpha}\right)\cdot\ln c - 0 \cdot \entropy\left(\gamma_{\alpha,\alpha} (\kappa,0)\right) - (1-0)\cdot \entropy\left(\delta_{\alpha,\alpha}(\kappa,0) \right) +\entropy(\kappa) \\
		~&=~ - \left(\kappa-\frac{0}{\alpha}\right)\cdot0 -\entropy(\kappa) +\entropy(\kappa )~=~0.
	\end{aligned}
	$$
	So we have
	$$g^*_{\alpha,\alpha,c}(\kappa) = \min_{\tau \in \left[ 0,\alpha\cdot \kappa\right]}g_{\alpha,\alpha,c}(\kappa,\tau) = 0.$$
\end{proof}

We are left to handle the case in which $\alpha = \beta>1$ and $c>1$.
The analysis for this case is based on ideas from \cite{EsmerKMNS22}.
The analysis is also used as part of the proof of \Cref{lem:coincide_with_esa}.
We first provide an explicit formula for $g^*(\kappa)$ in this case.

\begin{lemma}
	\label{lem:gst_alpha_eq_beta_formula}
	Let $\alpha > 1$ and $c > 1$.
	Then
	$$g^*_{\alpha,\alpha,c}(\kappa) = \kappa\cdot \ln c -
		\begin{cases}
			0 & \kappa < \delta^* \\
			\D{\kappa}{\delta^*} &\kappa \geq \delta^*
		\end{cases}~.
	$$
	for all $\kappa\in \left(0,\frac{1}{\alpha}\right)$ where $\delta^* \in \left(0,\frac{1}{\alpha}\right)$ is the unique value which satisfies $\D{\frac{1}{\alpha} }{\delta^*}  = \frac{\ln c}{\alpha}$. 
\end{lemma}

In the proof of \Cref{lem:gst_alpha_eq_beta_formula} we use the following identity.

\begin{lemma}
	\label{lem:H_to_D_eq}
	For all $a,b\in (0,1)$ it holds that 
	$$\entropy(a)-(a-b)\cdot \ln\left( \frac{1-a}{a}\right) = \D{b}{a} +\entropy(b).$$
\end{lemma}

\begin{proof}
	By expanding the term $\entropy(a)$ we get
	\begin{align*}
		\entropy(a) -(a-b)\cdot \ln\left( \frac{1-a}{a}\right) ~&=~ -a \ln a -(1-a) \ln (1-a) -(a-b) \ln(1-a) + (a-b) \ln a\\
		&=~ -b\ln a -(1-b)\ln(1-a)\\
		&=~ b\ln \left(\frac{b}{a}\right) +(1-b)\ln \left(\frac{1-b}{1-a}\right) -b\ln b -(1-b)\ln(1-b) \\
		&=~ \D{b}{a} +\entropy(b).
	\end{align*}
\end{proof}

\begin{proof}[Proof of \Cref{lem:gst_alpha_eq_beta_formula}]
	Define $\tilde{\tau}(\kappa) =\frac{\kappa-\delta^*}{\frac{1}{\alpha}-\delta^*}$.
	It can be easily verified that $\delta_{\alpha,\alpha}(\kappa,\tilde{\tau}(\kappa) ) =\delta^*$ for all $\kappa\in \left[0,\frac{1}{\alpha}\right]$.
	By \Cref{lem:deriv_g_by_tau} we have
	$$
	\begin{aligned}
		\gpart[\alpha,\alpha,c](\kappa,\tilde{\tau}(\kappa)) &=~ -\frac{\ln(c)}{\alpha} - \D{\frac{1}{\alpha}}{\gamma_{\alpha,\alpha}(\kappa,\tilde{\tau}(\kappa))} + \D{\frac{1}{\alpha}}{\delta_{\alpha,\alpha}(\kappa,\tilde{\tau}(\kappa))}\\
		&=~-\frac{\ln(c)}{\alpha} - \D{\frac{1}{\alpha}}{\frac{1}{\alpha}} + \D{\frac{1}{\alpha}}{\delta^*}\\
		&=~-\frac{\ln(c)}{\alpha} +-\frac{\ln(c)}{\alpha} ~=~0,
	\end{aligned}
	$$
	where the third equality follows from the definition of $\delta^*$.
	By \Cref{lem:g_convex_by_tau} it holds that $g_{\alpha,\alpha,c}(\kappa,\tau)$ is convex as a function of $\tau$, and we can conclude that the function has a global minimum at $\tilde{\tau}(\kappa)$.

	It also holds that
	$$\tilde{\tau}(\kappa ) =\frac{\kappa-\delta^*}{\frac{1}{\alpha} -\delta^*} =\frac{ \alpha\cdot \kappa\cdot \left(\frac{1}{\alpha}-\delta^*\right) +\alpha\kappa \delta^* -\delta^*}{\frac{1}{\alpha} -\delta^*} \leq \alpha \kappa$$
	for all $\kappa\leq \frac{1}{\alpha}$.

	Thus, for every $\delta^*\leq\kappa \leq\frac{1}{\beta}$ it holds that $\tilde{\tau}(\kappa) \geq 0$, and hence,
	\begin{equation*}
		\begin{aligned}
			g^*_{\alpha,\alpha,c}(\kappa)~&=~\min_{0\leq \tau \leq \alpha\kappa} g_{\alpha,\alpha,c}(\kappa,\tau)\\
			&=~ g_{\alpha,\alpha,c}(\kappa,\tilde{\tau}(\kappa))\\
			&=~ \left( \kappa-\frac{\tilde{\tau}(\kappa)}{\alpha}\right)\cdot \ln c -\tilde{\tau}(\kappa)\cdot \entropy\left(\frac{1}{\alpha}\right)-(1-\tilde{\tau}(\kappa) ) \cdot \entropy\left( \delta_{\alpha,\alpha}(\kappa,\tilde{\tau}(\kappa))\right) +\entropy(\kappa)\\
			&=~ \kappa \cdot \ln c - \tilde{\tau}(\kappa) \cdot \D{\frac{1}{\alpha}}{\delta^*}  -\tilde{\tau}(\kappa)\cdot \entropy\left(\frac{1}{\alpha}\right)-(1-\tilde{\tau}(\kappa) ) \cdot \entropy(\delta^*) +\entropy(\kappa)\\
			&=~ \kappa \cdot \ln c - \tilde{\tau}(\kappa )\cdot \left( \entropy(\delta^*)  -\left(\delta^* -\frac{1}{\alpha}\right) \cdot \ln \left(\frac{1-\delta^*}{\delta^*}\right)\right) -(1-\tilde{\tau}(\kappa) ) \cdot \entropy(\delta^*) +\entropy(\kappa)\\
			&=~ \kappa \cdot \ln c + \tilde{\tau}(\kappa )\cdot \left(\delta^* -\frac{1}{\alpha}\right) \cdot \ln \left(\frac{1-\delta^*}{\delta^*}\right) - \entropy(\delta^*) +\entropy(\kappa)\\
		\end{aligned}
	\end{equation*}
	The forth equality holds as $\frac{\ln c}{\alpha} = \D{\frac{1}{\alpha}}{\delta^*}$, and the fifth equality follows from \Cref{lem:H_to_D_eq}. By the definition of $\tilde{\tau}(\kappa)$ we have $\tilde{\tau}(\kappa) \left(\delta -\frac{1}{\alpha}\right) = (\delta^*-\kappa)$, thus for every $\delta^*\leq\kappa<\frac{1}{\beta}$ we have
	\begin{equation}
		\label{eq:aisb_greater}
		\begin{aligned}
			g^*_{\alpha,\alpha,c}(\kappa)~
			&=\kappa\cdot \ln(c) + \tilde{\tau}(\kappa )\cdot \left(\delta^* -\frac{1}{\alpha}\right) \cdot \ln \left(\frac{1-\delta^*}{\delta^*}\right) - \entropy(\delta^*) +\entropy(\kappa)\\
			&= \kappa\cdot \ln(c) + \left(\delta^* -\kappa^*\right) \cdot \ln \left(\frac{1-\delta^*}{\delta^*}\right) - \entropy(\delta^*) +\entropy(\kappa)\\
			&=\kappa \cdot \ln(c)  -\D{\kappa}{\delta^*},
		\end{aligned}
	\end{equation}
	where the last equality follows from \Cref{lem:H_to_D_eq}.

	Also, for every $0\leq \kappa< \delta^*$ it holds that $\tilde{\tau}(\kappa)<0$.
	Hence, since $g_{\alpha,\alpha,c}(\kappa, \tau)$ is convex as a function of $\tau$ (\Cref{lem:gk_convex_in_tau}), it holds that
	\begin{equation}
		\label{eq:aisb_smaller}
		\begin{aligned}
			g^*_{\alpha,\alpha,c}(\kappa)~&=~\min_{0\leq \tau \leq \alpha\kappa} g_{\alpha,\alpha,c}(\kappa,\tau)\\
			&=~ g_{\alpha,\alpha,c}(\kappa,0)\\
			&=~ \left( \kappa-\frac{0}{\alpha}\right)\cdot \ln c -0\cdot \entropy\left(\frac{1}{\alpha} \right) - (1-0) \cdot \entropy\left( \delta_{\alpha,\alpha}(\kappa,0)\right) + \entropy(\kappa)\\
			&=~ \kappa\cdot \ln c -   \entropy\left(\kappa\right) + \entropy(\kappa)\\
			&=~ \kappa\cdot \ln c.
		\end{aligned}
	\end{equation}
	By \eqref{eq:aisb_greater} and \eqref{eq:aisb_smaller} it holds that
	$$g^*_{\alpha,\alpha,c}(\kappa) = \kappa\cdot \ln c -
	\begin{cases}
		0 & \kappa < \delta^* \\
		\D{\kappa}{\delta^*} &\kappa \geq\delta^*
	\end{cases}~.$$
\end{proof}

Observe the function
$$\zeta(x) =
\begin{cases}
	0 & x < \delta^* \\
	\D{x}{\delta^*} &x \geq \delta^*
\end{cases}
$$ is convex.
Thus the following is a corollary of \Cref{lem:gst_alpha_eq_beta_formula} and the continuity of $g^*_{\alpha,\alpha,c}(\kappa)$.

\begin{corollary}
	\label{cor:alpha_eq_beta_conv}
	For every $\alpha>1$ and $c>1$ it holds that $g^*_{\alpha,\alpha,c}(\kappa)$ is concave on $\left[0,\frac{1}{\alpha} \right]$.
\end{corollary}

We can now proceed to the proof of \Cref{lem:concave}.

\concave*

\begin{proof}
	The lemma follows immediately from \Cref{cor:simple_concave}, \Cref{lem:gst_conc_first_corner,lem:a_is_b_and_c_is_one}, and \Cref{cor:alpha_eq_beta_conv}.
\end{proof}

We also use the formula in \Cref{lem:gst_alpha_eq_beta_formula} to prove \Cref{lem:coincide_with_esa}.
\coincidewithesa*
\begin{proof}
	Let $\beta,c > 1$ and let $\delta^*\in \left(0,\frac{1}{\beta}\right)$ be the unique value such that $\D{\frac{1}{\beta}}{\delta^*} =\frac{\ln c}{\beta}$.
	Then by \Cref{lem:gst_alpha_eq_beta_formula} it holds that
	$$
		\max_{0\leq \kappa\leq \frac{1}{\beta} } g^*_{\beta,\beta,c} (\kappa)  = \max_{0\leq \kappa \leq \frac{1}{\beta}}  \left(\kappa\cdot \ln c -
		\begin{cases}
			0 & \kappa < \delta^* \\
			\D{\kappa}{\delta^*} &\kappa \geq \delta^*
		\end{cases}
		~\right) = \max_{\delta^*\leq \kappa \leq \frac{1}{\beta}} \left( \kappa\cdot \ln c - \D{\kappa}{\delta^*}\right).
	$$
	Define $h(\kappa) = \kappa \ln c - \D{\kappa}{\delta^*}$.
	Then, by \Cref{cor:amls_is_the_best} and \eqref{eq:amls_def}, we get
	\begin{equation}
		\label{eq:best_to_h}
		\bestbound(\beta,c,\beta) = \amlsbound(\beta,c,\beta) = \exp\left(\max_{0\leq \kappa\leq \frac{1}{\beta} } g^*_{\beta,\beta,c} (\kappa)  \right) = \exp\left(\max_{\delta^*\leq \kappa \leq \frac{1}{\beta} } h(\kappa)\right).
	\end{equation}

	Let $h'$ be the derivative of $h$, and observe that $\frac{\partial \D{a}{b}}{\partial a} = \ln  \left(\frac{a}{1-a}\cdot \frac{1-b}{b}\right)$.
	So
	\begin{equation}
		\label{eq:hprime_formula}
		h'(\kappa) = \ln c - \ln \left( \frac{\kappa}{1-\kappa} \cdot \frac{1-\delta^*}{\delta^*}\right).
	\end{equation}

	Since $h$ is concave it can be trivially deduced that $h'(\kappa)$ is increasing in $\kappa \in \left[\delta^*,\kappa^*\right]$.
	Furthermore,
	\begin{equation}
		\label{eq:hprime_positive}h'(\delta^*) = \ln(c) - \ln \left(\frac{\delta^*}{1-\delta^*}\cdot \frac{1-\delta^*}{\delta^*}\right) =  \ln (c) >1
	\end{equation}
	and
	\begin{equation}
		\label{eq:hprime_negative}
		\begin{aligned}
			h'\left( \frac{1}{\beta}\right) &= \ln(c) - \ln \left(\frac{\frac{1}{\beta}}{1-\frac{1}{\beta}}\cdot \frac{1-\delta^*}{\delta^*}\right) \\
			&= \beta \cdot \D{\frac{1}{\beta}}{\delta^*} -  \ln \left(\frac{\frac{1}{\beta}}{1-\frac{1}{\beta}}\cdot \frac{1-\delta^*}{\delta^*}\right)\\
			&= \beta \cdot \frac{1}{\beta} \cdot \ln \left(\frac{\left( \frac{1}{\beta}\right)}{\delta^*}\right) + \beta\cdot \left(1-\frac{1}{\beta}\right) \cdot \ln \left( \frac{1-\frac{1}{\beta}}{1-\delta^*}\right)-  \ln \left(\frac{\frac{1}{\beta}}{1-\frac{1}{\beta}}\cdot \frac{1-\delta^*}{\delta}\right)\\
			&= \beta \cdot \ln \left( \frac{1-\frac{1}{\beta}}{1-\delta^*}\right)<0,
		\end{aligned}
	\end{equation}
	where the second equality uses $\frac{\ln c}{\beta}= \D{\frac{1}{\beta}}{\delta^*}$ and the last inequality holds since $\delta<\frac{1}{\beta}$.
	By \eqref{eq:hprime_positive} and \eqref{eq:hprime_negative}, there is $k^* \in \left( \delta^*, \frac{1}{\beta}\right)$ such that $h'\left(\kappa^*\right) = 0$.
	Furthermore, by \eqref{eq:hprime_formula} and simple algebraic manipulation, we get $\kappa^* = \frac{c\cdot \delta^*}{1+\delta^*(c-1)}$.
	Therefore,
	\begin{equation}
		\label{eq:hmax}\max_{\delta^*\leq \kappa \leq \frac{1}{\beta} } h(\kappa) = h(\kappa^*).
	\end{equation}
	It also holds that
	$$
	\begin{aligned}
		\D{\kappa^*}{\delta^*} &= \kappa^* \cdot \ln\left( \frac{\kappa^*}{\delta^*}\right)+(1-\kappa^*)\cdot \ln\left( \frac{1-\kappa^*}{1-\delta^*}\right) \\
		&= \kappa^* \cdot \left(\frac{\kappa^*}{1-\kappa^*}\cdot \frac{1-\delta^*}{\delta^*} \right) + \ln\left( \frac{1-\kappa^*}{1-\delta^*}\right)\\
		&= \kappa^*\ln(c)  + \ln\left( \frac{1-\kappa^*}{1-\delta^*}\right),
 	\end{aligned}
	$$
	where the last equality follows from $h'(\kappa^*)=0$ and \eqref{eq:hprime_formula}. Thus,
	$$h(\kappa^*) = \kappa^* \ln(c) -\D{\kappa^*}{\delta^*} = - \ln\left( \frac{1-\kappa^*}{1-\delta^*}\right) =-\ln \left( \frac{1- \frac{c\cdot \delta^*}{1+\delta^*(c-1)} }{1-\delta^*}\right) = \ln\left( 1+(c-1)\delta^*\right).$$
	By the above equitation, \eqref{eq:best_to_h} and \eqref{eq:hmax} we have $\bestbound(\beta,c,\beta) = 1+(c-1)\cdot \delta^*$.
	Thus, $1\leq\bestbound(\beta,c,\beta)\leq 1+\frac{c-1}{\beta} $ and
	$$\D{ \frac{1}{\beta}}{\frac{\bestbound(\beta,c,\beta)-1}{c-1}} =\D{\frac{1}{\beta}}{\delta^*} = \frac{\ln (c)}{\beta}.$$
	So $\bestbound(\beta,c,\beta) =\esaamlsbound(\beta,c) $ by the definition of $\esaamlsbound$.
\end{proof}

\subsection{The Determinant of the Hessian is Negative}
\label{sec:computing_det_hes}

In this section we prove \Cref{lem:hessian_negative}, that is, we show the determinant of the Hessian of $g$ is negative.
To do so we first obtain an explicit formula for the Hessian.
Recall $\Gamma$ and $\Delta$ are defined in~\eqref{eq:Gamma_def} and \eqref{eq:Delta_def}.

\begin{restatable}{lemma}{hessformula}
	\label{lem:hessian_formula}
	Let $\alpha\geq 1$, $\beta>1$ such that $\alpha \neq \beta $, $\kappa \in \left(0, \frac{1}{\beta}\right)$ and $\tau \in \left( \Mab(\kappa), \beta \cdot \kappa \right)$. Then
	$$
	\begin{aligned}
		\dhes{g}{\kappa,\tau}&= \frac{\Gamma_{\alpha,\beta}(\kappa,\tau)\cdot \Delta_{\alpha,\beta}(\kappa,\tau)}{\alpha^2\cdot (1-\kappa)} \cdot \frac{ \left(1-\frac{\beta}{\alpha}\right)\cdot \left(\gamma_{\alpha,\beta}(\kappa,\tau) - \delta_{\alpha,\beta}(\kappa,\tau)\right)}{\gamma_{\alpha,\beta}(\kappa,\tau)-\frac{1}{\alpha}} \cdot\\
		&~~~~~~~~~~~~~ \left( A_{\alpha,\beta} \left(\gamma_{\alpha,\beta}(\kappa,\tau) \right)  + \delta_{\alpha,\beta}(\kappa,\tau)\cdot B_{\alpha,\beta} \left( \gamma_{\alpha,\beta}(\kappa,\tau) \right) \right),
	\end{aligned}
	$$
	where
	\begin{align*}
		A_{\alpha, \beta}(x) &\coloneqq  -2 + x \left(1+\alpha+\beta \right) -\alpha\cdot \beta \cdot x^2\qquad \text{ and}\\
		B_{\alpha, \beta}(x) &\coloneqq x \cdot \alpha  \cdot \left(\beta-2\right) + 1 + \alpha - \beta.
	\end{align*}
\end{restatable}

The formula in \Cref{lem:hessian_formula} is derived from a technical computation of $\dhes{g}{\kappa,\tau}$ followed by re-arrangement of the terms. 
We defer the proof of \Cref{lem:hessian_formula} to \Cref{sec:hes_formula}. We use the notation $A_{\alpha,\beta}$ and $B_{\alpha,\beta}$ (or just $A$ and $B$) to refer to the functions defined in \Cref{lem:hessian_formula}.

As before, we often omit the subscripts $\alpha,\beta$ from functions (e.g., $\delta(\kappa,\tau)$ instead of $\delta_{\alpha,\beta}(\kappa,\tau)$) when $\alpha,\beta$ are known by context.
Furthermore, we often also omit the function parameters $(\kappa,\tau)$  when known by context (e.g., $\delta$ instead of $\delta_{\alpha,\beta}(\kappa,\tau)$).
 
By \Cref{lem:prop_gamma,lem:prop_delta}, for every $\alpha\geq 1$, $\beta>1$ such that $\alpha \neq \beta $ $\kappa \in (0, \frac{1}{\beta})$ and $\tau \in \left( \Mab(\kappa), \beta \cdot \kappa \right)$ it holds that $\delta_{\alpha,\beta}(\kappa,\tau),\gamma_{\alpha,\beta}(\kappa,\tau) \in (0,1)$.
So $\Gamma_{\alpha,\beta}(\kappa,\tau),~\Delta_{\alpha,\beta}(\kappa,\tau) > 0$.
Also, by \Cref{lem:prop_gamma}, if $\beta<\alpha$ it holds that $\left(1-\frac{\beta}{\alpha}\right)>0$ and $\gamma - \frac{1}{\alpha} >0$, and if $\beta>\alpha$ it holds that $\left(1-\frac{\beta}{\alpha}\right)<0$ and $\gamma - \frac{1}{\alpha} <0$.
Hence,
\[\frac{\left(1-\frac{\beta}{\alpha}\right)}{\gamma_{\alpha,\beta}(\kappa,\tau)-\frac{1}{\alpha}}>0\]
in both cases.
Following the above argument and \Cref{lem:hessian_formula} we attain the next corollary.

\begin{corollary}
	\label{cor:hes_sign_first}
	Let $\alpha\geq 1$, $\beta>1$ such that $\alpha \neq \beta $, $\kappa \in (0, \frac{1}{\beta})$ and $\tau \in \left( \Mab(\kappa), \beta \cdot \kappa \right)$.
	Then $\dhes{g}{\kappa,\tau} < 0$ if and only if
	$$\left(A_{\alpha,\beta} \left(\gamma_{\alpha,\beta}(\kappa,\tau) \right)  + \delta_{\alpha,\beta}(\kappa,\tau)\cdot B_{\alpha,\beta} \left( \gamma_{\alpha,\beta}(\kappa,\tau) \right)\right) \cdot \left( \gamma_{\alpha,\beta}(\kappa,\tau) - \delta_{\alpha,\beta}(\kappa,\tau)\right)<0.$$
\end{corollary}

The following lemma allows us to determine the sign of $\dhes{g}{\kappa,\tau}$ using an even simpler expression.

\begin{lemma}\label{lem:gamma_>_delta}
	Let $\alpha\geq 1$, $\beta>1$ such that $\alpha\neq \beta$, $\kappa \in \left( 0, \frac{1}{\beta} \right) $ and $\tau \in \left( \Mab(\kappa), \beta \cdot \kappa \right)$.
	Then $\gamma_\ab(\kappa,\tau) > \delta_\ab(\kappa,\tau)$.
\end{lemma}

\begin{proof}
	Consider the following cases.
	\begin{itemize}
		\item If $\alpha > \beta > 1$, by \Cref{lem:prop_gamma,lem:prop_delta}, we immediately have that
			$\delta(\kappa,\tau) < \frac{1}{\alpha} < \frac{1}{\beta} \leq \gamma(\kappa,\tau)$.
		\item If $\beta > \alpha\geq 1$, then by \Cref{lem:prop_gamma} it holds that $\gamma> \kappa$.
			Furthermore, by \Cref{lem:prop_delta}
			\begin{align*}
				\delta(\kappa,\tau) ~&\leq~ \delta(\kappa,M_{\alpha,\beta} (\kappa,\tau))  ~=~ \frac{\frac{\beta}{\alpha} \kappa-\frac{1}{\alpha} }{1-M_{\alpha,\beta}(\kappa)} +\frac{1}{\alpha} ~=~ \frac{\frac{\beta}{\alpha} \kappa -\frac{1}{\alpha} }{1-\frac{\beta-\alpha}{1-\alpha\kappa}\cdot \kappa} +\frac{1}{\alpha}  ~=~\frac{\frac{\beta}{\alpha} \kappa -\frac{1}{\alpha} }{\frac{ 1-\alpha\kappa  -(\beta-\alpha)\kappa}{1-\alpha\kappa}} +\frac{1}{\alpha}\\
				&=~ \frac{\frac{1}{\alpha} \left( \beta \kappa -1\right)  }{{ 1-\beta \kappa }}\cdot (1-\alpha\kappa) +\frac{1}{\alpha} ~=~\kappa.
			\end{align*}
			Thus, $\delta \leq \kappa < \gamma$.\qedhere
	\end{itemize}
\end{proof}

By \Cref{cor:hes_sign_second} and \Cref{lem:gamma_>_delta} we obtain the following. 

\begin{corollary}
	\label{cor:hes_sign_second}
	Let $\alpha\geq 1$, $\beta>1$ such that $\alpha \neq \beta $, $\kappa \in \left(0, \frac{1}{\beta}\right)$ and $\tau \in \left( \Mab(\kappa), \beta \cdot \kappa \right)$.
	Then $\dhes{g}{\kappa,\tau} < 0$ if and only if $$A_{\alpha,\beta} \left(\gamma_{\alpha,\beta}(\kappa,\tau) \right)  + \delta_{\alpha,\beta}(\kappa,\tau)\cdot B_{\alpha,\beta} \left( \gamma_{\alpha,\beta}(\kappa,\tau) \right)<0.$$
\end{corollary}

We proceed to analyze the functions $A_{\alpha,\beta}$ and $B_{\alpha,\beta}$ towards our goal of showing that $A_{\alpha,\beta}(\gamma) +\delta \cdot B_{\alpha,\beta}(\gamma)$ is negative.

\begin{lemma}\label{lem:sign_A_negative}
	For every $\alpha,\beta\geq 1$ such that $\alpha\neq \beta$, $\kappa \in \left( 0,\frac{1}{\beta} \right) $ and $\tau \in \Bigl( \Mab(\kappa), \beta\cdot \kappa \Bigr) $ it holds that  $A_{\alpha, \beta} \bigl( \gamma(\kappa,\tau) \bigr)<0$. 
\end{lemma}

\begin{proof}
	We can re-write $A_{\alpha,\beta}$ (as defined in \Cref{lem:hessian_formula}) as
		\begin{equation*}
		A_{\alpha,\beta}(x)= -\alpha\cdot\beta\cdot\Bigl(x - \frac{1 + \alpha + \beta}{2\cdot \alpha\cdot \beta}\Bigr)^{2} - 2 + \frac{\left( \alpha + \beta + 1 \right)^{2}}{4\cdot \alpha\cdot \beta}.
	\end{equation*}
	Note that $A_{\alpha,\beta}(x)$ is a quadratic polynomial in $x$, which reaches its maximum value at $x_0 \coloneqq \frac{1 + \alpha + \beta}{2 \cdot \alpha \cdot \beta}$.
	Consider the following cases.
	\begin{itemize}
		\item If $x_0\leq 1$, then for every $x\in (0,1)$ we have
			$$A_{\alpha,\beta}(x) ~\leq~ A_{\alpha,\beta}(x_0)  ~=~ -\alpha\cdot \beta (x_0-x_0)^2 -2 + x_0^2 ~ =~ -2 + x_0^2 ~<~0,$$
			where the last inequality holds as $0\leq x_0\leq 1$.
		\item If $x_0>1$, then $A_{\alpha,\beta}(x)$ is monotonically increasing in $[0,1]$.
			Thus, for every $x\in (0,1)$, it holds that
			$$A_{\alpha,\beta} (x) ~<~ A_{\alpha,\beta}(1) ~=~ -2 + \left(1+\alpha+\beta \right) -\alpha\cdot \beta ~=~ -(\alpha-1)(\beta-1)~\leq 0.$$
	\end{itemize}
	So overall $A_{\alpha,\beta}(x)<0$ for all $x\in (0,1)$.
	By \Cref{lem:prop_gamma} it holds that $\gamma_{\alpha,\beta}(\kappa,\tau)\in (0,1)$.
	So $A_{\alpha,\beta}\left(\gamma_{\alpha,\beta} (\kappa,\tau)\right) < 0$.
\end{proof}

The tools attained so far suffice to show that $A(\gamma)+\delta \cdot B(\gamma)$ is negative in case $\alpha<\beta$.

\begin{lemma}\label{lem:sign_A_delta_B_negative_1}
	For all $\beta > \alpha \geq 1$, $\kappa \in \left( 0,\frac{1}{\beta} \right) $ and $\tau \in \Bigl( \Mab(\kappa), \beta\cdot \kappa \Bigr) $ it holds that
	\begin{align*}
		A\Bigl( \gamma(\kappa,\tau) \Bigr) + \delta(\kappa,\tau) \cdot B\Bigl( \gamma(\kappa,\tau) \Bigr) < 0.
	\end{align*}
\end{lemma}

\begin{proof}
	Consider the following cases.
	\begin{itemize}
		\item If $B(\gamma(\kappa,\tau ))\leq 0 $, then $\delta\geq 0$ by \Cref{lem:prop_delta}. So
			$$A\Bigl( \gamma(\kappa,\tau) \Bigr) + \delta(\kappa,\tau) \cdot B\Bigl( \gamma(\kappa,\tau) \Bigr) ~\leq~  A\Bigl( \gamma(\kappa,\tau) \Bigr)  <0,$$
			where the last inequality follows from \Cref{lem:sign_A_negative}.
		\item If $B(\gamma(\kappa,\tau ))> 0 $, we can use $\gamma_{\alpha,\beta}(\kappa,\tau) > \delta_{\alpha,\beta}(\kappa,\tau)$ (\Cref{lem:gamma_>_delta}) to obtain
			\begin{align*}
				A(\gamma) + \delta\cdot B(\gamma) &<  A(\gamma) + \gamma\cdot B(\gamma)\\
				&= -2+ \gamma \left(1+\beta+\alpha\right) -\beta\cdot \alpha \cdot\gamma^2 + \alpha\cdot \left(\beta - 2\right) \cdot\gamma^2 + \gamma \cdot \left( 1+\alpha-\beta\right) \\
				&= {\gamma^2 \cdot\left( -2\alpha \right) + \gamma \cdot \left( 2( \alpha + 1) \right) -2 }\\
				&= -2 \cdot {(1 - \gamma) \cdot (1 - \alpha \cdot \gamma)} ~\leq~ 0,
			\end{align*}
		where the last inequality follows from $\gamma \leq \frac{1}{\alpha}$ by \Cref{lem:prop_gamma}.\qedhere
	\end{itemize}
\end{proof}

\Cref{lem:sign_A_delta_B_negative_1} suffices to show \Cref{lem:hessian_negative} for the case $\beta>\alpha \geq 1$.

\begin{lemma}
	\label{lem:hess_neg_beta_>_alpha}
	For all $\beta > \alpha \geq 1$, $c>1$ and $\kappa \in \left( 0,\frac{1}{\beta} \right)$ it holds that $\dhes{g}{\kappa,\taust(\kappa)}<0$.
\end{lemma}

\begin{proof}
	By \Cref{lem:minimizer_tau} it holds that $\taust(\kappa) \in \left(M_{\alpha,\beta}(\kappa) , \beta \cdot \kappa \right)$ and $\gpart(\kappa,\tau^*(\kappa)) = 0$.
	By \Cref{lem:sign_A_delta_B_negative_1} we have
	\begin{align*}
		A\Bigl( \gamma(\kappa,\taust(\kappa)) \Bigr) + \delta(\kappa,\taust(\kappa)) \cdot B\Bigl( \gamma(\kappa,\taust(\kappa)) \Bigr) < 0.
	\end{align*} 
	Thus, by \Cref{cor:hes_sign_second}, we get $\dhes{g}{\kappa,\taust(\kappa)}<0$.
\end{proof}

Note that \Cref{lem:sign_A_delta_B_negative_1} implies a stronger claim than the one stated in \Cref{lem:hess_neg_beta_>_alpha}: for all $\beta > \alpha \geq 1$ and $(\kappa,\tau)$  the determinant Hessian evaluated at $(\kappa,\tau)$ is negative.
This property, however, does not hold if $\alpha > \beta > 1$, i.e., in this case the determinant of the Hessian may be positive for some values of $(\kappa,\tau)$.
We use the following lemma to restrict the possible values $\kappa$ and $\tau$ may take.

\begin{lemma}\label{lem:critical_point_condition}
	For all $\alpha > \beta > 1$, $c\geq 1$, $\kappa \in (0, \frac{1}{\beta})$ and $\tau \in (\Mab(\kappa), \beta \cdot \kappa)$ such that $\gpart(\kappa,\tau) = 0$, it holds that 
	$\D{\frac{1}{\alpha}}{\gamma_{\alpha,\beta} (\kappa,\tau)} \leq \D{\frac{1}{\alpha}}{\delta_{\alpha,\beta} (\kappa,\tau)} $. 
\end{lemma}

\begin{proof}
	By \Cref{lem:deriv_g_by_tau}, the condition $\gpart(\kappa, \tau) = 0$ is equivalent to
	\begin{equation*}
		\label{eq:cricical_point_cond}
		-\D{\frac{1}{\alpha}}{\gamma(\kappa,\tau)} + \D{\frac{1}{\alpha}}{\delta(\kappa,\tau)} = \frac{\ln(c)}{\alpha} \geq 0
	\end{equation*}
	since $c \geq 1$.
	Therefore, $\D{\frac{1}{\alpha}}{\gamma_{\alpha,\beta} (\kappa,\tau)} \leq \D{\frac{1}{\alpha}}{\delta_{\alpha,\beta} (\kappa,\tau)}$.
\end{proof}

The next lemmas enables us to further simplify the criteria in \Cref{cor:hes_sign_second}.

\begin{lemma}\label{lem:sign_B_positive}
	Let $\alpha > \beta > 1$, $\kappa \in \left( 0,\frac{1}{\beta} \right) $ and $\tau \in \Bigl( \Mab(\kappa), \beta\cdot \kappa \Bigr) $, then
	$B_{\alpha,\beta}\bigl( \gamma_{\alpha,\beta}(\kappa,\tau)\bigr)>0$. 
\end{lemma}

\begin{proof}
	Consider the following cases.
	\begin{itemize}
		\item If $\beta \geq 2$, it holds that
			\begin{align*}
				B(\gamma) ~=~ \gamma \cdot \alpha \cdot (\beta - 2) + 1 + \alpha - \beta ~\geq ~0 \cdot \alpha (\beta - 2) +1 +\alpha - \beta~ =~ 1 + \alpha - \beta~>~ 0,
			\end{align*}
			where the first inequality follows from $\gamma(\kappa,\tau) \geq 0$ (\Cref{lem:prop_gamma}) and the last inequality uses~$\beta < \alpha$.
		\item If $\beta < 2$, since $\gamma <1$ (\Cref{lem:prop_gamma}) we have that
			\begin{align*}
				B(\gamma) ~=~ \gamma \cdot \left(\beta - 2\right) \cdot \alpha + 1 + \alpha- \beta
				~>~ 1 \cdot\alpha \left(\beta - 2\right)  + 1+\alpha-\beta
				~=~ (\beta - 1) \cdot (\alpha - 1) ~\geq~ 0.
			\end{align*}
			\qedhere
	\end{itemize}
\end{proof}

For every $\alpha>\beta >1$ we define $\hb_{\alpha, \beta}:(0,1)\rightarrow \mathbb{R}$ via
\begin{align}
	\label{eq:C_def}
	\hb_{\alpha,\beta}(x) \coloneqq -\frac{A_{\alpha, \beta}(x)}{B_{\alpha, \beta}(x)}
\end{align}
for all $x \in (0,1)$.
\Cref{lem:sign_B_positive,cor:hes_sign_second} imply the following.

\begin{corollary}
	\label{cor:hes_third_cond}
	Let  $\alpha > \beta > 1$, $\kappa \in \left( 0,\frac{1}{\beta} \right) $ and $\tau \in \Bigl( \Mab(\kappa), \beta\cdot \kappa \Bigr)$.
	Then $\dhes{g}{\kappa,\tau} <0$ if and only if $\hb(\gamma_{\alpha,\beta}(\kappa,\tau))>\delta_{\alpha,\beta}(\kappa,\tau)$.
\end{corollary}

The following lemma utilizes \Cref{lem:critical_point_condition} to show that the condition in \Cref{cor:hes_third_cond} holds on critical points of $g_{\alpha,\beta,c}$.

\begin{lemma}\label{lem:hb_bound_alpha_>_beta}
	Let $\alpha > \beta > 1$, $c\geq 1$, $\kappa \in (0, \frac{1}{\beta})$ and $\tau \in (\Mab(\kappa), \beta \cdot \kappa)$, such that  $\gpart(\kappa,\tau) = 0$.
	Then $C_{\alpha, \beta}\Bigl(\gamma(\kappa,\tau)\Bigr) > \delta_{\alpha,\beta}(\kappa,\tau)$.
\end{lemma}

\begin{proof}
	The next claim allows us to eliminate the dependency on $\beta$. 

	\begin{claim}
		\label{claim:beta_eliminator}
		$C_{\alpha,\beta}\left(\gamma(\kappa,\tau) \right) > \frac{\gamma_{\alpha,\beta}(\kappa,\tau)}{ 2\cdot \alpha \cdot \gamma_{\alpha,\beta}(\kappa,\tau) -1}$.
	\end{claim}
	\begin{claimproof}
		For every $\tilde{\beta}\in (1,\alpha)$   can rewrite
		\begin{equation}
			\label{eq:exp_A}
			A_{\alpha, \tilde{\beta}}(x) ~=~  -2 + x \left(1+\alpha+\tilde{\beta} \right) -\alpha \tilde{\beta}  x^2 ~=~ -2 + x(1+\alpha) + \tilde{\beta}  x \cdot \left( 1-\alpha x \right)  ~=~ a(\alpha,x) -\tilde{\beta} x \left(\alpha x-1\right),
		\end{equation}
		where $a(\alpha,x) \coloneqq -2 + x \left(1+\alpha \right)$.  Similarly, for every $\tilde{\beta}\in (1,\alpha)$    it holds that
		\begin{equation}
			\label{eq:exp_B}
			B_{\alpha, \tilde{\beta} }(x) ~=~ x \cdot \alpha  \cdot \left(\tilde{\beta}-2\right)  +1+\alpha-\tilde{\beta} ~=~ \tilde{\beta}\left( \alpha x - 1\right) -2\cdot x \cdot \alpha +1 +\alpha ~=~ \tilde{\beta} (\alpha x-1) +b(\alpha,x),
		\end{equation}
		where $b(\alpha,x) \coloneqq -2\cdot x \cdot \alpha +1 +\alpha$.

		Using \eqref{eq:exp_A} and \eqref{eq:exp_B}, for every $\tilde{\beta}\in (1,\alpha)$ and $x\in \left(\frac{1}{\alpha},1\right)$ it holds that
		\begin{equation}
			\label{eq:A_plus_xB}
			\begin{aligned}
				A_{\alpha,\tilde{\beta}} (x) + x\cdot  B_{\alpha,\tilde{\beta}}(x) ~&=~  a(\alpha,x) + x\cdot b(\alpha,x)\\
				&=~ -2 + x(1+\alpha) -2\cdot x^2 \cdot \alpha +x + \alpha x \\
				&=~ 2 \cdot \left(  -1 +x(1+\alpha)-x^2\cdot \alpha \right) ~>~ 0.
			\end{aligned}
		\end{equation}
		The last inequality holds as $\zeta(x) = -1+x(1+\alpha) -\alpha x^2$ is concave and $\zeta\left(\frac{1}{\alpha}\right)= \zeta(1) =0$.  Observe the sum $A_{\alpha,\tilde{\beta}}(x)+x\cdot B_{\alpha,\tilde{\beta}} (x)$ does not depend on $\tilde{\beta}$.

		By \eqref{eq:A_plus_xB} we have
		\begin{equation}
			\label{eq:C_mon}
			C_{\alpha,\tilde{\beta}}(x)  ~=~ -\frac{A_{\alpha,\tilde{\beta}(x)} }{B_{\alpha,\tilde{\beta}}(x)}  ~=~ -\frac{A_{\alpha,\tilde{\beta}}(x)+x\cdot  B_{\alpha,\tilde{\beta}}(x) -x\cdot B_{\alpha,\tilde{\beta}} (x)  }{B_{\alpha,\tilde{\beta}}(x)}  ~=~  x - \frac{a(\alpha,x) +x\cdot b(\alpha,x) }{B_{\alpha,\tilde{\beta}} (x) }.
		\end{equation}
		For a fixed $x>\frac{1}{\alpha}$, by \eqref{eq:exp_B}, we have $B_{\alpha,\tilde{\beta}} (x)$ is increasing as a function of $\tilde{\beta}$.
		Hence, by \eqref{eq:A_plus_xB} and \eqref{eq:C_mon}, the expression $C_{\alpha,\tilde{\beta}}(x)$ is increasing as a function of $\tilde{\beta}$.
		By \Cref{lem:prop_gamma} it holds that $\frac{1}{\alpha} < \frac{1}{\beta} < \gamma(\kappa,\tau) < 1$.
		Using the monotonicity property of $C_{\alpha,\tilde{\beta}} (x)$ we get
		\begin{align*}
			C_{\alpha, \beta}\left( \gamma \right) ~&>~ C_{\alpha, \frac{1}{\gamma}}\left( \gamma \right)\\
			&=~ -\frac{A_{\alpha, \frac{1}{\gamma}}(\gamma)}{B_{\alpha, \frac{1}{\gamma}}(\gamma)} \\
			&=~ -\frac{-2 + \gamma (1+\alpha) + \frac{1}{\gamma}  \cdot \gamma  \cdot \left( 1-\alpha \cdot \gamma  \right)}{\frac{1}{\gamma }\left( \alpha \gamma  - 1\right) -2\cdot \gamma \cdot \alpha +1 +\alpha}\\
			&=~ -\frac{ -1 +\gamma}{ 2\alpha -2 \gamma \alpha +1 -\frac{1}{\gamma}}\\
			&=~ \frac{\gamma} {2\cdot \alpha \cdot \gamma -1}
		\end{align*}
		where the last equality follows from a re-arrangement of terms.
	\end{claimproof}

	We combine \Cref{claim:beta_eliminator} with the following inequality.

	\begin{claim}
		\label{claim:KL_prop}
		For every $x\in \left[\frac{1}{\alpha},1\right)$ it holds that
		$\D{\frac{1}{\alpha} }{\frac{x}{2\alpha x-1}}\leq \D{\frac{1}{\alpha} }{x}$.
	\end{claim}

	\begin{claimproof}
		Define $h(x) = \D{\frac{1}{\alpha} }{x} - \D{\frac{1}{\alpha} }{\frac{x}{2\alpha x-1}}$.
		The statement of the claim is equivalent to $h(x) \geq 0$ for all $x\in \left[\frac{1}{\alpha},0\right)$.
		Recall $\frac{\partial \D{a}{b}}{\partial b} = \frac{b-a}{b(1-b)}$.
		Therefore,
		$$
		\begin{aligned}
			\frac{\partial }{\partial  x }\left( \D{\frac{1}{\alpha} }{\frac{x}{2\alpha x-1}} \right) ~&=~
			\frac{2\alpha x -1 - x \cdot 2\alpha}{\left(2\alpha x -1\right)^2} \cdot \frac{\frac{x}{2\alpha x-1}-\frac{1}{\alpha}}{ \frac{x}{2\alpha x-1}\cdot \left( 1-\frac{x}{2\alpha x-1}\right)} \\
			&=~ -\frac{\left( \frac{x-\frac{1}{\alpha} \left( 2\alpha x -1\right)}{2\alpha x -1} \right) }{x(2\alpha x -1 -x) } \\
			&=~ \frac{ x-\frac{1}{\alpha}}{x(2\alpha x -1 -x)(2\alpha x -1)}.
		\end{aligned}
		$$
		Let $h'$ be the derivative of $h$.
		Thus,
		$$h'(x)~=~ \frac{x-\frac{1}{\alpha}}{x(1-x)} - \frac{ x-\frac{1}{\alpha}}{x(2\alpha x -1 -x)(2\alpha x -1)} ~ =~ \frac{x-\frac{1}{\alpha}}{x } \left( \frac{1}{1-x} - \frac{1}{ (2\alpha x -1 -x )(2\alpha x -1)}\right).$$
		For every $x\in \left[ \frac{1}{\alpha},1\right)$ it holds that $1-x,~x-\frac{1}{\alpha },~ 2\alpha x -1-x , ~2\alpha x -1~\geq~ 0$.
		So
		$$
		\begin{aligned}
			&h'(x)~\geq~ 0 &\iff\\
			&1-x ~\leq~ (2\cdot \alpha x -x-1) \cdot (2\alpha x -1)  &\iff\\
			&1-x~\leq~ 4\alpha^2 x^2 -2\alpha x -2\alpha x^2 + x -2\alpha x +1&\iff \\
			&0\leq ~ 2x \cdot (x\alpha -1 )(2\alpha -1).
		\end{aligned}
		$$
		As the last condition is true for all $x\in \left[\frac{1}{\alpha},1\right)$, it follows that $h'(x)\geq 0$.
		So $h$ is (weakly) increasing in $\left[\frac{1}{\alpha},1\right)$.
		Hence,
		$$h(x)\geq h\left(\frac{1}{\alpha }\right) =\D{\frac{1}{\alpha} }{\frac{1}{\alpha}} - \D{\frac{1}{\alpha} }{\frac{\frac{1}{\alpha}}{2\alpha\cdot  \frac{1}{\alpha}-1}}  ~=~  \D{\frac{1}{\alpha} }{\frac{1}{\alpha}} - \D{\frac{1}{\alpha} }{\frac{1}{\alpha}}~=0$$
		for all $x\in \left[\frac{1}{\alpha},1\right)$.
	\end{claimproof}

	To wrap-up the proof, consider the following cases.
	\begin{itemize}
		\item If $C_{\alpha,\beta} \left( \gamma_{\alpha,\beta} (\kappa,\tau )\right)~>~  \frac{1}{\alpha}$, then $\delta_{\alpha,\beta}(\kappa,\tau) <\frac{1}{\alpha}$ by \Cref{lem:prop_delta}.
			So $\delta_{\alpha,\beta}(\kappa,\tau) < C_{\alpha,\beta} \left( \gamma_{\alpha,\beta} (\kappa,\tau )\right)$.
		\item If $C_{\alpha,\beta}(\gamma) \leq \frac{1}{\alpha}$, then $\D{\frac{1}{\alpha}}{x}$ is decreasing in the interval $\left[0,\frac{1}{\alpha}\right]$.
			Thus
			$$\D{\frac{1}{\alpha}}{C_{\alpha,\beta}(\gamma)} ~<~ \D{\frac{1}{\alpha}}{ \frac{\gamma}{2\cdot \alpha \cdot \gamma-1} } ~\leq~ \D{\frac{1}{\alpha}}{\gamma} ~\leq~ \D{\frac{1}{\alpha}}{\delta}.$$
			The first inequality follows from \Cref{claim:beta_eliminator}, the second follows from \Cref{claim:KL_prop} (recall $\gamma\in \left(\frac{1}{\alpha},1\right)$ by \Cref{lem:prop_gamma}) and the last inequality follows from \Cref{lem:critical_point_condition}.
			Since $C_{\alpha,\beta}(\gamma),\delta \leq \frac{1}{\alpha}$ (see \Cref{lem:prop_delta}), it follows that $C_{\alpha,\beta}(\gamma) <\delta$.\qedhere
	\end{itemize}
\end{proof}

The next lemma follows from \Cref{lem:hb_bound_alpha_>_beta}.

\begin{lemma}
	\label{lem:hess_neg_alpha_>_beta}
	For all $\alpha>\beta>1$, $c\geq 1$ and $\kappa\in \left(0,\frac{1}{\beta}\right)$ it holds that $\dhes{g}{\kappa,\taust(\kappa)}<0$. 
\end{lemma}

\begin{proof}
	By \Cref{lem:minimizer_tau} it holds that $\taust(\kappa) \in \left(M_{\alpha,\beta}(\kappa) , \beta \cdot \kappa \right)$ and $\gpart(\kappa,\tau^*(\kappa)) = 0$.
	Thus, by \Cref{lem:hb_bound_alpha_>_beta}, we have $C_{\alpha,\beta}\left( \gamma(\kappa,\taust(\kappa))\right) > \delta\left(\kappa,\taust(\kappa) \right)$.
	So $\dhes{g}{\kappa,\taust(\kappa)}<0$ by \Cref{cor:hes_third_cond}.
\end{proof}

We can now proceed to the proof of \Cref{lem:hessian_negative}.

\hessiannegative*

\begin{proof}
	The lemma follows from \Cref{lem:hess_neg_alpha_>_beta,lem:hess_neg_beta_>_alpha}.
\end{proof}

\subsection{Partial Derivatives}
\label{sec:g_par_der}

In this section we calculate the partial derivatives of the function $g_{\abc}(\kappa,\tau)$.
For notational brevity, we use the following naming scheme for the partial derivatives:

\begin{align*}
	\gpark[\abc](\kappa_0, \tau_0) &\coloneqq \frac{\partial g_{\abc} (\kappa, \tau)}{\partial \kappa}\Bigr|_{(\kappa,\tau) = (\kappa_0, \tau_0)}\\
	\gpart[\abc](\kappa_0, \tau_0) &\coloneqq \frac{\partial g_{\abc} (\kappa, \tau)}{\partial \tau}\Bigr|_{(\kappa,\tau) = (\kappa_0, \tau_0)}\\
	\gparkt[\abc](\kappa_0, \tau_0) &\coloneqq \frac{\partial^{2} g_{\abc} (\kappa, \tau)}{\partial \kappa \partial \tau}\Bigr|_{(\kappa,\tau) = (\kappa_0, \tau_0)}\\
	\gparkk[\abc](\kappa_0, \tau_0) &\coloneqq \frac{\partial^{2} g_{\abc} (\kappa, \tau)}{\partial \kappa^2}\Bigr|_{(\kappa,\tau) = (\kappa_0, \tau_0)}\\
	\gpartt[\abc](\kappa_0, \tau_0) &\coloneqq \frac{\partial^{2} g_{\abc} (\kappa, \tau)}{\partial \tau^2}\Bigr|_{(\kappa,\tau) = (\kappa_0, \tau_0)}
\end{align*}

We sometimes omit the subscript $(\alpha, \beta, c)$ from $\gpark[\abc]$, $\gpart[\abc]$, $\gparkk[\abc]$, $\gpartt[\abc]$ and $\gparkt[\abc]$.

Recall that $\D{a}{b}  = a \ln \frac{a}{b} + (1-a)\ln \frac{1-a}{1-b}$ is the Kullback-Leibler divergence between two Bernoulli distributions with parameters $a$ and $b$.
In the next lemmas, we use algebraic properties of the KL divergence to calculate the partial derivatives of $g(\kappa,\tau)$.

It can  be easily verified that the partial derivatives of $\D{a}{b}$ and $\entropy(x)$ are
\begin{align}
	\label{eq:KL_by_b}
	\frac{\partial \D{a}{b}}{\partial b} &= \frac{b-a}{b(1-b)}\\ \nonumber \\
	\frac{\partial \entropy(x)}{\partial x} &= \ln\left( \frac{1-x}{x} \right) \label{eq:entropy_by_x}.
\end{align}
Moreover, the partial derivatives of $\gamma(\kappa,\tau)$ and $\delta(\kappa,\tau)$ are
\begin{align}
	\label{eq:gammapark}
	\frac{\partial \gamma(\kappa, \tau)}{\partial \kappa } &= \left( 1 - \frac{\beta}{\alpha} \right)\cdot \frac{1}{\tau}\\
	\label{eq:gammapart}
	\frac{\partial \gamma(\kappa, \tau)}{\partial \tau } &= -\left(1-\frac{\beta}{\alpha} \right)\frac{\kappa}{\tau^2} = -\frac{\gamma(\kappa,\tau)-\frac{1}{\alpha}}{\tau}\\
	\label{eq:deltapark}
	\frac{\partial \delta(\kappa, \tau)}{\partial \kappa } &=  \frac{\beta}{\alpha \cdot (1 - \tau)} \\
	\label{eq:deltapart}
	\frac{\partial \delta(\kappa, \tau)}{\partial \tau } &= \frac{(\beta \cdot \kappa - 1)}{\alpha \cdot (1-\tau)^{2}} =\frac{\delta(\kappa,\tau)-\frac{1}{\alpha}}{1-\tau}
\end{align}
For notational brevity, we sometimes omit the arguments  $(\kappa, \tau)$ from $\gamma(\kappa,\tau)$ and $\delta(\kappa, \tau)$, and simply use $\gamma$ and $\delta$ instead.

\derivgbytau*

\begin{proof}
	By \eqref{eq:entropy_by_x} and standard derivation rules we have
	\begin{align*}
		\gpart(\kappa,\tau) &= \frac{\partial}{\partial \tau} \Bigg( \frac{\beta \kappa - \tau}{\alpha} \ln c - \tau\cdot \entropy\left(\gamma_{\alpha,\beta}(\kappa,\tau)\right) -(1-\tau)\cdot \entropy\left(\delta_{\alpha,\beta} (\kappa,\tau)\right) + \entropy\left(\kappa\right)\Bigg)\\
		&= -\frac{\ln(c)}{\alpha} -\entropy(\gamma)+\entropy(\delta)  -\tau\cdot \ln\left( \frac{1 - \gamma}{\gamma} \right) \cdot \frac{\partial \gamma(\kappa,\tau)}{\partial \tau} -(1-\tau)\cdot\ln\left( \frac{1 - \delta}{\delta} \right) \cdot \frac{\partial \delta(\kappa,\tau)}{\partial \tau}\\
		&= -\frac{\ln(c)}{\alpha} -\entropy(\gamma) +\entropy(\delta )+\left(\gamma -\frac{1}{\alpha }\right) \cdot \ln \left( \frac{1-\gamma}{\gamma} \right) - \left(\delta-\frac{1}{\alpha} \right) \cdot   \ln \left( \frac{1-\delta}{\delta} \right)  \\
		&= -\frac{\ln(c)}{\alpha} -\D{\frac{1}{\alpha}}{\gamma} +\entropy\left( \frac{1}{\alpha}\right)+\D{\frac{1}{\alpha}}{\delta} -\entropy\left( \frac{1}{\alpha}\right)\\
		&= -\frac{\ln(c)}{\alpha} -\D{\frac{1}{\alpha}}{\gamma} +\D{\frac{1}{\alpha}}{\delta} 
	\end{align*}
	where the third equality uses \eqref{eq:gammapart} and \eqref{eq:deltapart}, and the forth follows from \Cref{lem:H_to_D_eq}.
\end{proof}

\begin{lemma}\label{lem:deriv_g_by_kappa}
	For all $\alpha, c \geq 1$ and $\beta > 1$ it holds that
	$$
	\begin{aligned}
		\gpark[\abc](\kappa,\tau) = \frac{\beta}{\alpha} \cdot \ln(c) &+ \left( \beta - \alpha \right) \cdot \Biggl( \D{\frac{1}{\alpha}}{\gamma(\kappa,\tau)} + \ln\Bigl( 1 - \gamma(\kappa,\tau) \Bigr) \Biggr)\\
		&- \beta \cdot \Biggl(\D{\frac{1}{\alpha}}{\delta(\kappa,\tau)} + \ln\Bigl( 1 - \delta(\kappa,\tau) \Bigr)\Biggr)\\
		&+ \alpha \cdot \Biggl( \D{\frac{1}{\alpha}}{\kappa} + \ln\Bigl( 1 - \kappa \Bigr) \Biggr).
	\end{aligned}
	$$
\end{lemma}

The following identify is used in the proof of \Cref{lem:deriv_g_by_kappa}.

\begin{lemma}\label{lem:ln_kl_equality}
	For all $x \in [0,1]$ and $\alpha \geq 1$ it holds that
	\begin{equation*}
		\ln\left( \frac{1- x}{x} \right) = \alpha\cdot \D{\frac{1}{\alpha}}{x} - \ln\left( \frac{\frac{1}{\alpha}}{1 - \frac{1}{\alpha}} \right) - \alpha\cdot\ln\left( \frac{1 - \frac{1}{\alpha}}{1 - x} \right) 
	\end{equation*}
\end{lemma}

\begin{proof}
	By a sequence of algebraic manipulations we get
	\begin{align*}
		\frac{1}{\alpha}\cdot \ln\left( \frac{1- x}{x} \right) &= \frac{1}{\alpha}\cdot \ln\left( \frac{1}{x} \right) - \frac{1}{\alpha}\cdot \ln\left( \frac{1}{1-x} \right) \\
		&= \frac{1}{\alpha}\cdot \ln\left( \frac{\frac{1}{\alpha}}{x} \right) - \frac{1}{\alpha}\cdot \ln\left( \frac{1}{\alpha} \right) -\frac{1}{\alpha}\cdot \ln\left( \frac{1 - \frac{1}{\alpha}}{1-x} \right) + \frac{1}{\alpha}\cdot \ln\left( 1 - \frac{1}{\alpha} \right)\\
		&= \frac{1}{\alpha}\cdot \ln\left( \frac{\frac{1}{\alpha}}{x} \right) - \frac{1}{\alpha}\cdot \ln\left( \frac{1}{\alpha} \right) + \left( 1 - \frac{1}{\alpha} \right) \cdot \ln\left( \frac{1 - \frac{1}{\alpha}}{1-x} \right) + \frac{1}{\alpha}\cdot \ln\left( 1 - \frac{1}{\alpha} \right) -\ln\left( \frac{1 - \frac{1}{\alpha}}{1-x} \right)\\
		&= \D{\frac{1}{\alpha}}{x} - \frac{1}{\alpha}\cdot \ln\left( \frac{1}{\alpha} \right) +  \frac{1}{\alpha}\cdot \ln\left( 1 - \frac{1}{\alpha} \right) -\ln\left( \frac{1 - \frac{1}{\alpha}}{1-x} \right) \\
		&= \D{\frac{1}{\alpha}}{x} -\frac{1}{\alpha}\cdot \ln\left( \frac{\frac{1}{\alpha}}{1 - \frac{1}{\alpha}} \right)  -\ln\left( \frac{1 - \frac{1}{\alpha}}{1-x} \right) 
	\end{align*}
\end{proof}

\begin{proof}[Proof of \Cref{lem:deriv_g_by_kappa}]
	By $\eqref{eq:g_def}$ we have
	\begin{align*}
		\gpark(\kappa,\tau) = \frac{\beta}{\alpha}\ln(c)
		+ \underbrace{\frac{\partial}{\partial \kappa}\Biggl( -\tau\cdot \entropy\left(\gamma_{\alpha,\beta}(\kappa,\tau)\right)\Biggr)}_{\Pi_1}
		+  \underbrace{\frac{\partial}{\partial \kappa} \Biggl(-(1-\tau)\cdot \entropy\left(\delta_{\alpha,\beta} (\kappa,\tau)\right)\Biggr)}_{\Pi_2}
		+\underbrace{\frac{\partial}{\partial \kappa}\Biggl(\entropy\left(\kappa\right)\Biggr)}_{\Pi_3}.
	\end{align*}
	It holds that
	\begin{align*}
		\Pi_1 &= (-\tau)\cdot\ln\left( \frac{1 - \gamma(\kappa,\tau)}{\gamma(\kappa,\tau)} \right) \cdot \frac{\partial \gamma(\kappa,\tau)}{\partial \kappa} \\
		&= \frac{\beta - \alpha}{\alpha} \cdot  \ln\left( \frac{1 - \gamma(\kappa,\tau)}{\gamma(\kappa,\tau)} \right)\\
		&= \frac{\beta - \alpha}{\alpha} \cdot \Biggl(\alpha \cdot \D{\frac{1}{\alpha}}{\gamma(\kappa,\tau)} - \ln\left( \frac{\frac{1}{\alpha}}{1 - \frac{1}{\alpha}} \right) - \alpha \cdot \ln\left( \frac{1 - \frac{1}{\alpha}}{1 - \gamma(\kappa,\tau)} \right)  \Biggr),
	\end{align*}
	where the second equality follows from \eqref{eq:gammapark} and the last equality follows from \Cref{lem:ln_kl_equality}.
	By \Cref{lem:ln_kl_equality} and \eqref{eq:deltapark} we have
	\begin{align*}
		\Pi_2 &= (-1 + \tau)\cdot \ln\left( \frac{1 - \delta(\kappa, \tau)}{\delta(\kappa,\tau)} \right) \cdot \frac{\partial \delta(\kappa,\tau)}{\partial \kappa} \\
		&= -\frac{\beta}{\alpha} \cdot \ln\left( \frac{1 - \delta(\kappa,\tau)}{\delta(\kappa,\tau)} \right)\\
		&= -\frac{\beta}{\alpha} \cdot \Biggl( \alpha \cdot \D{\frac{1}{\alpha}}{\delta(\kappa,\tau)} - \ln\left( \frac{\frac{1}{\alpha}}{1 - \frac{1}{\alpha}}\right) - \alpha\cdot \ln\left( \frac{1 - \frac{1}{\alpha}}{1 - \delta(\kappa,\tau)} \right)    \Biggr),
	\end{align*}
	and
	\begin{align*}
		\Pi_3 &= \ln\left( \frac{1 - \kappa}{\kappa} \right)  = \alpha \cdot \D{\frac{1}{\alpha}}{\kappa}  - \ln\left( \frac{\frac{1}{\alpha}}{1 - \frac{1}{\alpha}} \right) - \alpha\cdot \ln\left( \frac{1 - \frac{1}{\alpha}}{1 - \kappa} \right).
	\end{align*}
	Overall, we get
	\begin{align*}
		\gpark(\kappa,\tau) &= \frac{\beta}{\alpha}\ln(c) + \Pi_1 +\Pi_2 + \Pi_3\\
	    &= \frac{\beta}{\alpha} \cdot \ln(c) + \left( \beta - \alpha \right) \cdot \Biggl( \D{\frac{1}{\alpha}}{\gamma(\kappa,\tau)} + \ln\Bigl( 1 - \gamma(\kappa,\tau) \Bigr) \Biggr)\\
		&\quad- \beta \cdot \Biggl(\D{\frac{1}{\alpha}}{\delta(\kappa,\tau)} + \ln\Bigl( 1 - \delta(\kappa,\tau) \Bigr)\Biggr)\\
		&\quad+ \alpha \cdot \Biggl( \D{\frac{1}{\alpha}}{\kappa} + \ln\Bigl( 1 - \kappa \Bigr) \Biggr).
	\end{align*}
\end{proof}

Recall the functions $\Gamma_{\alpha,\beta}$ and $\Delta_{\alpha,\beta}$ from \eqref{eq:Gamma_def} and \eqref{eq:Delta_def}.

\begin{lemma}\label{lem:deriv_g_by_kappa_kappa}
	For all $\alpha,c \geq 1$and  $\beta > 1$ it holds that
	\begin{align*}
		\gparkk[\abc](\kappa,\tau) \,&=\, \left( 1 - \frac{\beta}{\alpha} \right)^{2}\cdot \Gamma_{\alpha,\beta}(\kappa,\tau)+ \left( \frac{\beta}{\alpha} \right)^{2} \cdot \Delta_{\alpha,\beta}(\kappa,\tau) -\frac{1}{\kappa \cdot (1 - \kappa)}.
	\end{align*}
\end{lemma}

\begin{proof}
	Using \Cref{lem:deriv_g_by_kappa} we have
	\begin{align*}
		\gparkk(\kappa,\tau) &= \underbrace{(\beta - \alpha) \cdot \frac{\partial}{\partial \kappa} \Biggl( \D{\frac{1}{\alpha}}{\gamma(\kappa,\tau)}  + \ln\left( 1 - \gamma(\kappa,\tau) \right) \Biggr)}_{\Pi_1}\\
		&\quad \underbrace{-\beta\cdot \frac{\partial}{\partial \kappa} \Biggl( \D{\frac{1}{\alpha}}{\delta(\kappa,\tau)} + \ln\left( 1 - \delta(\kappa,\tau) \right) \Biggr) }_{\Pi_2}\\
		&\quad \underbrace{+\alpha \cdot \frac{\partial}{\partial \kappa} \Biggl( \D{\frac{1}{\alpha}}{\kappa} + \ln\left( 1 - \kappa \right)  \Biggr)}_{\Pi_3}.
	\end{align*}
	By \eqref{eq:KL_by_b}, \eqref{eq:gammapark} and \eqref{eq:deltapark} we obtain
	\begin{align*}
		\Pi_1 &= (\beta - \alpha)\cdot \Biggl( \frac{\gamma(\kappa,\tau) - \frac{1}{\alpha}}{\gamma(\kappa,\tau) \cdot \left( 1 - \gamma(\kappa,\tau) \right) } - \frac{1}{1 - \gamma(\kappa,\tau)}\Biggr) \cdot \frac{\partial \gamma(\kappa,\tau)}{\partial \kappa}\\
		&= (\beta - \alpha)\cdot \Biggl( \frac{\gamma(\kappa,\tau) - \frac{1}{\alpha}}{\gamma(\kappa,\tau) \cdot \left( 1 - \gamma(\kappa,\tau) \right) } - \frac{1}{1 - \gamma(\kappa,\tau)}\Biggr) \cdot \left( 1 - \frac{\beta}{\alpha} \right)\cdot \frac{1}{\tau}\\
		&= \left( 1 - \frac{\beta}{\alpha} \right)^{2} \cdot \Biggl(\frac{1}{\tau \cdot \gamma(\kappa,\tau) \cdot \bigl(1 - \gamma(\kappa,\tau)\bigr)}\Biggr)\\
		&= \left(1-\frac{\beta}{\alpha}\right)^2 \cdot \Gamma_{\alpha,\beta}(\kappa,\tau),\\
		\\
		\Pi_2 &= -\beta \cdot \Biggl( \frac{\delta(\kappa,\tau) - \frac{1}{\alpha}}{\delta(\kappa,\tau) \cdot \bigl(1 - \delta(\kappa,\tau)\bigr)} - \frac{1}{1 - \delta(\kappa,\tau)} \Biggr) \cdot \frac{\partial \delta(\kappa,\tau)}{\partial \kappa}\\
		&= -\beta \cdot \Biggl( \frac{\delta(\kappa,\tau) - \frac{1}{\alpha}}{\delta(\kappa,\tau) \cdot \bigl(1 - \delta(\kappa,\tau)\bigr)} - \frac{1}{1 - \delta(\kappa,\tau)} \Biggr) \cdot \left(  \frac{\beta}{\alpha \cdot (1 - \tau)} \right) \\
		&=\left( \frac{\beta}{\alpha} \right)^{2} \cdot \Biggl( \frac{1}{ (1 - \tau) \cdot \delta(\kappa,\tau) \cdot \bigl( 1 - \delta(\kappa,\tau)\bigr) } \Biggr)\\
		&=\left(\frac{\beta}{\alpha}\right)^2\cdot\Delta_{\alpha,\beta}(\kappa,\tau),\\
		\\
		\Pi_3 &= \alpha\cdot\Biggl( \frac{\kappa - \frac{1}{\alpha}}{\kappa \cdot (1 - \kappa)} - \frac{1}{1-\kappa} \Biggr) = \frac{-1}{\kappa \cdot \left( 1 - \kappa \right) }.
	\end{align*}
	Finally, we have
	\begin{align*}
		\gparkk(\kappa,\tau) = \Pi_1 + \Pi_2 + \Pi_3 &= \left( 1 - \frac{\beta}{\alpha} \right)^{2} \cdot \Gamma_{\alpha,\beta}(\kappa,\tau) ++ \left( \frac{\beta}{\alpha} \right)^{2} \cdot  \Delta_{\alpha,\beta}(\kappa,\tau) -\frac{1}{\kappa \cdot (1 - \kappa)}.
	\end{align*}
\end{proof}

\begin{lemma}
	\label{lem:deriv_g_by_kappa_tau}
	For all $\alpha,c \geq 1$ and $\beta > 1$ it holds that
	\begin{align*}
		\gparkt[\abc](\kappa,\tau) &= -\left( 1 - \frac{\beta}{\alpha} \right)\cdot  \left(\gamma(\kappa,\tau) - \frac{1}{\alpha}\right)\cdot \Gamma_{\alpha,\beta}(\kappa,\tau) +
		\frac{\beta}{\alpha} \cdot \left(\delta(\kappa,\tau) - \frac{1}{\alpha} \right)\cdot \Delta_{\alpha,\beta}(\kappa,\tau).
	\end{align*}
\end{lemma}

\begin{proof}
	By \cref{lem:deriv_g_by_kappa} we get
	\begin{align*}
		\gparkt{(\kappa,\tau)} &= \underbrace{(\beta - \alpha) \cdot \frac{\partial}{\partial \tau} \Biggl(\D{\frac{1}{\alpha}}{\gamma(\kappa,\tau)} + \ln\Bigl( 1 - \gamma(\kappa,\tau) \Bigr) \Biggr)}_{\Pi_1}\\
		&\quad \underbrace{- \beta \cdot \frac{\partial}{\partial \tau}\Biggl(\D{\frac{1}{\alpha}}{\delta(\kappa,\tau)} + \ln\Bigl( 1 - \delta(\kappa,\tau) \Bigr)\Biggr)}_{\Pi_2}.
	\end{align*}
	By \eqref{eq:KL_by_b}, \eqref{eq:gammapart} and \eqref{eq:deltapart} we have
	\begin{align*}
		\Pi_1 &= (\beta- \alpha) \cdot \Biggl( \frac{\gamma(\kappa,\tau) - \frac{1}{\alpha}}{\gamma(\kappa,\tau)\cdot\Bigl(1 - \gamma(\kappa,\tau)\Bigr)} - \frac{1}{1 - \gamma(\kappa,\tau)} \Biggr) \cdot \frac{\partial \gamma(\kappa,\tau)}{\partial \tau}\\
		&= (\beta - \alpha) \cdot \left( \frac{-1}{\alpha \cdot \gamma(\kappa,\tau)\cdot\Bigl(1 - \gamma(\kappa,\tau)\Bigr)} \right) \cdot \Biggl(\frac{-\gamma(\kappa,\tau) + \frac{1}{\alpha}}{\tau}\Biggr)\\
		&= -\left(1-\frac{\beta}{\alpha}\right) \cdot \left(\gamma(\kappa,\tau) -\frac{1}{\alpha}\right) \cdot \Gamma_{\ab}(\kappa,\tau)
	\end{align*}
	and
	\begin{align*}
		\Pi_2 &= -\beta \cdot \Biggl(\frac{\delta(\kappa,\tau) - \frac{1}{\alpha}}{\delta(\kappa,\tau) \cdot \Bigl(1 - \delta(\kappa,\tau)\Bigr)} - \frac{1}{1-\delta(\kappa,\tau)}\Biggr)\cdot\frac{\partial \delta(\kappa,\tau)}{\partial \tau}\\
		&= -\beta \cdot \Biggl(\frac{-1}{\alpha \cdot \delta(\kappa,\tau) \cdot \Bigl(1 - \delta(\kappa,\tau)\Bigr)}\Biggr)\cdot\left( \frac{\delta(\kappa,\tau) - \frac{1}{\alpha}}{1 - \tau} \right)\\
		&= \frac{\beta}{\alpha} \cdot \left(\delta(\kappa,\tau) - \frac{1}{\alpha} \right) \cdot \Delta_{\ab}(\kappa,\tau).
	\end{align*}
	So
	\begin{align*}
		\gparkt(\kappa,\tau) = \Pi_1 + \Pi_2 = -\left(1-\frac{\beta}{\alpha}\right) \cdot \left(\gamma(\kappa,\tau) -\frac{1}{\alpha}\right) \cdot \Gamma_{\ab}(\kappa,\tau) + \frac{\beta}{\alpha} \cdot \left(\delta(\kappa,\tau) - \frac{1}{\alpha} \right) \cdot \Delta_{\ab}(\kappa,\tau).
	\end{align*}
\end{proof}

\derivgbytautau*

\begin{proof}
	By \Cref{lem:deriv_g_by_tau} we have
	\begin{align*}
		\gpartt = \frac{\partial}{\partial \tau} \Biggl( \D{\frac{1}{\alpha}}{\delta(\kappa,\tau)} \Biggr) - \frac{\partial}{\partial \tau} \Biggl( \D{\frac{1}{\alpha}}{\gamma(\kappa,\tau)} \Biggr).
	\end{align*}
	By \eqref{eq:KL_by_b} and \eqref{eq:deltapart} we have
	\begin{align*}
		\frac{\partial}{\partial \tau} \Biggl( \D{\frac{1}{\alpha}}{\delta(\kappa,\tau)} \Biggr) &= \frac{\delta - \frac{1}{\alpha}}{\delta \cdot (1 - \delta)} \cdot \frac{\partial \delta(\kappa,\tau)}{\partial \tau}
		= \frac{\delta - \frac{1}{\alpha}}{\delta\cdot (1 - \delta)} \cdot \frac{\delta - \frac{1}{\alpha}}{1 - \tau} = \left( \delta-\frac{1}{\alpha}\right)^2 \cdot \Delta_{\ab}(\kappa,\tau).
	\end{align*}
	Similarly, by \eqref{eq:KL_by_b} and \eqref{eq:gammapart} we have
	\begin{align*}
		\frac{\partial}{\partial \tau} \Biggl(\D{\frac{1}{\alpha}}{\gamma(\kappa,\tau)} \Biggr)= \frac{\gamma- \frac{1}{\alpha}}{\gamma\cdot \Bigl(1 - \gamma\Bigr)}\cdot \frac{\partial \gamma(\kappa,\tau)}{\partial \tau}
		= \frac{\gamma - \frac{1}{\alpha}}{\gamma \cdot (1 - \gamma)} \cdot \frac{ -\left( \gamma  -\frac{1}{\alpha} \right)}{\tau}=
		-\left(\gamma-\frac{1}{\alpha}\right)^2 \cdot \Gamma_{\ab}(\kappa,\tau).
	\end{align*}
	Together
	\begin{align*}
		\gpartt &= \frac{\partial}{\partial \tau} \Biggl( \D{\frac{1}{\alpha}}{\delta(\kappa,\tau)} \Biggr) - \frac{\partial}{\partial \tau} \Biggl( \D{\frac{1}{\alpha}}{\gamma(\kappa,\tau)} \Biggr)\\
		&= \left( \delta-\frac{1}{\alpha}\right)^2 \cdot \Delta_{\ab}(\kappa,\tau) + 	  \left(\gamma-\frac{1}{\alpha}\right)^2 \cdot \Gamma_{\ab}(\kappa,\tau).
	\end{align*}
\end{proof}

\subsection{A Formula for the Determinant of the Hessian}
\label{sec:hes_formula}

In this section we prove \Cref{lem:hessian_formula}, that is, we provide a formula for $\dhes{g}{\kappa,\tau}$.

\hessformula*

\begin{proof}
	By \Cref{lem:deriv_g_by_kappa_kappa,lem:deriv_g_by_tau_tau} we have
	\begin{equation}
		\label{eq:hess_first_part}
		\begin{aligned}
			&\gparkk(\kappa,\tau) \cdot \gpartt(\kappa,\tau) \\
			=~& \left( \left( 1 - \frac{\beta}{\alpha} \right)^{2}\cdot \Gamma+ \left( \frac{\beta}{\alpha} \right)^{2} \cdot \Delta -\frac{1}{\kappa \cdot (1 - \kappa)}\right) \cdot \left( 	\left( \gamma - \frac{1}{\alpha} \right)^{2} \cdot \Gamma +\left( \delta - \frac{1}{\alpha} \right)^{2} \cdot \Delta\right)\\
			=~& \left( \gamma - \frac{1}{\alpha} \right)^{2}  \cdot \left( 1-\frac{\beta}{\alpha}\right)^2 \cdot \Gamma^2 + \left(\delta-\frac{1}{\alpha}\right)^2 \cdot \left(\frac{\beta}{\alpha}\right)^2 \cdot \Delta^2 - \frac{\gpartt(\kappa,\tau)}{\kappa(1-\kappa)} \\
			&\quad + \Delta \cdot \Gamma \left(\left( \gamma - \frac{1}{\alpha} \right)^{2} \cdot \left(\frac{\beta}{\alpha}\right)^2  +   \left( \delta - \frac{1}{\alpha} \right)^{2} \cdot \left(1-\frac{\beta}{\alpha}\right)^2  \right).
		\end{aligned}
	\end{equation}
	Similarly, by \Cref{lem:deriv_g_by_kappa_tau} it holds that 
	\begin{equation}
		\label{eq:hess_second_part}
		\begin{aligned}
			\left(\gparkt(\kappa,\tau) \right)^2 ~&=~ \left(-\left( 1 - \frac{\beta}{\alpha} \right)\cdot \left(\gamma - \frac{1}{\alpha}\right) \cdot \Gamma+ \frac{\beta}{\alpha} \cdot \left(\delta - \frac{1}{\alpha} \right)\cdot \Delta \right)^2\\
			&=~ \left( \gamma - \frac{1}{\alpha} \right)^{2}  \cdot \left( 1-\frac{\beta}{\alpha}\right)^2 \cdot \Gamma^2 + \left(\delta-\frac{1}{\alpha}\right)^2 \cdot \left(\frac{\beta}{\alpha}\right)^2 \cdot \Delta^2\\
			&\quad -2 \cdot \left( \gamma - \frac{1}{\alpha} \right) \cdot \left( 1-\frac{\beta}{\alpha}\right)\cdot \left(\delta-\frac{1}{\alpha}\right) \cdot \left(\frac{\beta}{\alpha}\right)\cdot \Gamma\cdot \Delta.
		\end{aligned}
	\end{equation}
	By \eqref{eq:hess_first_part} and \eqref{eq:hess_second_part} we have
	\begin{equation}
		\label{eq:hess_initial_exp}
		\begin{aligned}
			\dhes{g}{\kappa,\tau} ~&=~ \gparkk(\kappa,\tau) \cdot \gpartt(\kappa,\tau) - \left( \gparkt(\kappa,\tau) \right)^2 \\
			&=~ -\frac{\gpartt(\kappa,\tau)}{\kappa(1-\kappa)} + \Delta \cdot \Gamma \left(
			\left( \gamma - \frac{1}{\alpha} \right)^{2} \cdot \left(\frac{\beta}{\alpha}\right)^2
			+\left( \delta - \frac{1}{\alpha} \right)^{2} \cdot \left(1-\frac{\beta}{\alpha}\right)^2  \right)\\
			&\quad+2\cdot \left( \gamma - \frac{1}{\alpha} \right)  \cdot \left( 1-\frac{\beta}{\alpha}\right)\cdot \left(\delta-\frac{1}{\alpha}\right) \cdot \left(\frac{\beta}{\alpha}\right)\cdot \Gamma\cdot \Delta\\
			&=~ -\frac{\gpartt(\kappa,\tau)}{\kappa(1-\kappa)} + \Delta \cdot \Gamma \left(
			\left( \gamma - \frac{1}{\alpha} \right) \cdot \left(\frac{\beta}{\alpha}\right)
			+\left( \delta - \frac{1}{\alpha} \right) \cdot \left(1-\frac{\beta}{\alpha}\right)  \right)^2\\
			&=~ \frac{\Delta \cdot \Gamma}{1-\kappa}\cdot \left(- \frac{\gpartt}{\kappa\cdot \Delta \cdot \Gamma} + (1-\kappa ) \cdot
			\left(\left( \gamma - \frac{1}{\alpha} \right) \cdot \left(\frac{\beta}{\alpha}\right)
			+\left( \delta - \frac{1}{\alpha} \right) \cdot \left(1-\frac{\beta}{\alpha}\right) \right)^2\right).
		\end{aligned}
	\end{equation}
	We define
	\begin{equation}
		\label{eq:psi_def}
		\psi_\ab(\kappa,\tau) = \frac{\beta}{\alpha} \cdot \left( \gamma_{\alpha,\beta}(\kappa,\tau) -\frac{1}{\alpha} \right) +\left(1-\frac{\beta}{\alpha} \right)\cdot \left( \delta_{\alpha,\beta}(\kappa,\tau) -\frac{1}{\alpha}\right).
	\end{equation}
	As before, we use the shorthand $\psi=\psi_\ab(\kappa,\tau)$.
	Thus,
	\begin{equation}
		\label{eq:hess_second_exp}
		\dhes{g}{\kappa,\tau} = \frac{\Delta \cdot \Gamma}{1-\kappa}\cdot \left(- \frac{\gpartt}{\kappa\cdot \Delta \cdot \Gamma} + (1-\kappa )\cdot \psi^2 \right).
	\end{equation}
	We use the following algebraic identity to simplify \eqref{eq:hess_second_exp}.

	\begin{claim}
		\label{claim:k_psi}
		$\kappa \cdot \psi = \delta\left(\gamma-\frac{1}{\alpha}\right)$.
	\end{claim}

	\begin{claimproof}
		The statement of the claim immediately follows from the following equation.
		\begin{align*}
			\kappa\cdot \psi -\delta\left(\gamma-\frac{1}{\alpha}\right) ~&=~
			\kappa\cdot \left(1-\frac{\beta}{\alpha}\right) \cdot \left(\delta-\frac{1}{\alpha}\right) +
			\kappa \cdot \frac{\beta}{\alpha}\left(\gamma -\frac{1}{\alpha}\right)  -\delta\cdot \left(\gamma-\frac{1}{\alpha}\right)\\
			&=~ \left(\delta-\frac{1}{\alpha}\right)\cdot \left(\gamma-\frac{1}{\alpha }\right)\cdot \tau +\left(\gamma -\frac{1}{\alpha}\right) \cdot \left(\kappa \cdot \frac{\beta}{\alpha} - \delta\right)\\
			&=~ \left(\gamma -\frac{1}{\alpha}\right)\cdot \left( \left(\delta-\frac{1}{\alpha}\right)\cdot \tau +\frac{\beta}{\alpha}\cdot \kappa -\delta \right)\\
			&=~ \left(\gamma -\frac{1}{\alpha}\right)\cdot \left(-\delta \cdot (1-\tau) +\frac{\beta}{\alpha }\kappa-\frac{\tau}{\alpha} \right) ~=~ 0.
		\end{align*}
		The first equality follows from an expansion of $\psi$ by \eqref{eq:psi_def}.
		The second equality uses the identity $(\gamma -\frac{1}{\alpha}) \cdot \tau = \left( 1-\frac{\beta}{\alpha}\right)\cdot \kappa$ by \eqref{eq:gamma_def}.
		The last equality holds as $\delta\cdot(1-\tau)= \frac{\beta}{\alpha} \cdot \kappa -\frac{\tau}{\alpha}$.
	\end{claimproof}

	By \Cref{lem:deriv_g_by_tau_tau} and the definition of $\Delta$ \eqref{eq:Delta_def} and $\Gamma $ \eqref{eq:Gamma_def} it holds that
	\begin{equation}
		\label{eq:gtt_expansion}
		\begin{aligned}
			\frac{\gpartt}{\kappa \cdot \Delta\cdot \Gamma } ~&=~ \frac{1}{\kappa}\cdot \left(\gamma -\frac{1}{\alpha} \right)^2\cdot (1-\tau) \cdot \delta \cdot (1-\delta) +
			\frac{1}{\kappa} \cdot \left(\delta-\frac{1}{\alpha}\right)^2 \cdot \tau \cdot (1-\gamma)\cdot \gamma\\
			&=~ \frac{1}{\kappa}\cdot \left(\gamma -\frac{1}{\alpha} \right)^2 \cdot \delta  \cdot (1-\delta)  +
			\frac{\tau}{\kappa}\left( \left(\delta-\frac{1}{\alpha}\right)^2 \cdot (1-\gamma)\cdot  \gamma -\left(\gamma -\frac{1}{\alpha} \right)^2 \cdot \delta \cdot(1-\delta) \right)\\
			&=~  \psi \cdot \left(\gamma -\frac{1}{\alpha} \right)   \cdot (1-\delta)  +\Pi_1,
		\end{aligned}
	\end{equation}
	where the last equality uses $\frac{1}{\kappa}\cdot\delta \left(\gamma-\frac{1}{\alpha}\right)=\psi$ by \Cref{claim:k_psi} and
	$$\Pi_1 \coloneqq \frac{\tau}{\kappa}\left( \left(\delta-\frac{1}{\alpha}\right)^2 \cdot (1-\gamma)\cdot  \gamma -\left(\gamma -\frac{1}{\alpha} \right)^2 \cdot \delta \cdot(1-\delta) \right).$$
	Furthermore,
	\begin{equation}
		\label{eq:hess_psi_expansion}
		(1-\kappa) \cdot\psi^2 = \psi^2 -\kappa\cdot \psi^2 = \psi^2 - \psi \cdot \delta \cdot \left(\gamma-\frac{1}{\alpha}\right),
	\end{equation}
	where the second equality follows from \Cref{claim:k_psi}.
	By \eqref{eq:hess_second_exp}, \eqref{eq:gtt_expansion} and \eqref{eq:hess_psi_expansion} we have
	\begin{equation}
		\label{eq:hess_third_exp}
		\begin{aligned}
			\dhes{g}{\kappa,\tau} ~&=~ \frac{\Delta \cdot \Gamma}{1-\kappa}\cdot \left(- \psi\left(\gamma-\frac{1}{\alpha}\right)\cdot(1-\delta)-\Pi_1 + \psi^2 -\psi\cdot \delta \cdot \left(\gamma-\frac{1}{\alpha}\right)\right)\\
			&=~ \frac{\Delta \cdot \Gamma}{1-\kappa}\cdot \left(- \psi\left(\gamma-\frac{1}{\alpha}\right)+\psi^2-\Pi_1 \right)\\
			&=~ \frac{\Delta \cdot \Gamma}{1-\kappa}\cdot \left(\Pi_2-\Pi_1 \right),
		\end{aligned}
	\end{equation}
	where
	$$\Pi_2 \coloneqq - \psi\left(\gamma-\frac{1}{\alpha}\right)+\psi^2.$$
	Observe that
	\begin{equation*}
		\begin{aligned}
		&~ \gamma \cdot(1-\gamma ) \cdot \left(\delta -\frac{1}{\alpha}\right)^2-\delta  \cdot (1-\delta )\cdot  \left(\gamma -\frac{1}{\alpha}\right)^2 \\
		=&~ \gamma \cdot (1-\gamma) \cdot \left( \delta^2 -\frac{2}{\alpha} \delta +\frac{1}{\alpha^2} \right) -\delta \cdot (1-\delta) \cdot \left( \gamma^2 -\frac{2}{\alpha} \gamma +\frac{1}{\alpha^2} \right) \\
		=&~ \delta^2\gamma -\frac{2}{\alpha} \delta\gamma +\frac{\gamma}{\alpha^2} - \delta^2\cdot \gamma^2 +\frac{2}{\alpha} \delta\cdot \gamma^2 -\frac{\gamma^2}{\alpha^2} \\
		 &\quad -\gamma^2\delta +\frac{2}{\alpha} \gamma\delta -\frac{\delta}{\alpha^2} + \gamma^2\cdot \delta^2 -\frac{2}{\alpha} \gamma\cdot \delta^2 +\frac{\delta^2}{\alpha^2} \\
		=&~ (\delta-\gamma)\cdot \gamma \delta +(\gamma-\delta )\cdot \frac{1}{\alpha^2} +(\gamma -\delta)\cdot \frac{2}{\alpha}\gamma \delta -\frac{1}{\alpha^2}(\gamma+\delta)(\gamma -\delta)\\
		=&~ (\gamma -\delta)\cdot \frac{1}{\alpha^2}\left(-\alpha^2\cdot \gamma \delta +1 +2\alpha\cdot\gamma \delta -\gamma-\delta \right) \\
		\end{aligned}
	\end{equation*}
	and $\frac{\tau}{\kappa} = \frac{1-\frac{\beta}{\alpha}}{\gamma-\frac{1}{\alpha}}$ by \eqref{eq:gamma_def}.
	Thus
	\begin{equation}
		\label{eq:Pi1_proc}
		\begin{aligned}
			\Pi_1 ~&=~ \frac{\left(1-\frac{\beta}{\alpha}\right)\cdot \left( \gamma - \delta\right)}{\left(\gamma-\frac{1}{\alpha}\right) \cdot \alpha^2 } \cdot\left(  -\alpha^2\cdot \gamma \delta +1 +2\alpha\cdot\gamma \delta -\gamma-\delta \right) \\
			&=~ \frac{\left(1-\frac{\beta}{\alpha}\right)\cdot \left( \gamma - \delta\right)}{\left(\gamma-\frac{1}{\alpha}\right) \cdot \alpha^2 } \cdot \left( \delta \cdot \left( -\alpha^2 \cdot \gamma +2\alpha\gamma -1\right)+1-\gamma \right)
		\end{aligned}
	\end{equation}
	Additionally,
	$$
	\begin{aligned}
		\Pi_2 ~&=~ -\psi\left(\gamma-\frac{1}{\alpha}\right)+\psi^2\\
		&=~ \psi\left( - \left(\gamma-\frac{1}{\alpha}\right) +\left(\delta-\frac{1}{\alpha}\right) \cdot \left(1-\frac{\beta}{\alpha}\right) + \frac{\beta}{\alpha} \cdot\left(\gamma-\frac{1}{\alpha}\right) \right)\\
		&=~ \psi\left(1-\frac{\beta}{\alpha}\right) \cdot \left( \delta-\gamma\right).
	\end{aligned}
	$$
	By further expanding the expression for $\psi$ and dividing and multiplying by $\alpha^2\cdot \left(\gamma-\frac{1}{\alpha}\right)$ we get
	\begin{equation}
		\label{eq:Pi2_proc}
		\begin{aligned}
			\Pi_2 ~&=~ -\frac{\left(1-\frac{\beta}{\alpha}\right)\cdot \left( \gamma - \delta\right)}{\left(\gamma-\frac{1}{\alpha}\right) \cdot \alpha^2 } \cdot \left(\gamma -\frac{1}{\alpha} \right)\cdot\alpha^2 \left( \left(\gamma-\frac{1}{\alpha}\right)\cdot \frac{\beta}{\alpha} +\left(\delta-\frac{1}{\alpha}\right)\cdot \left(1-\frac{\beta}{\alpha}\right)\right)\\
			&=~ -\frac{\left(1-\frac{\beta}{\alpha}\right) \left( \gamma - \delta\right)}{\left(\gamma-\frac{1}{\alpha}\right) \cdot \alpha^2 }\cdot \left( \left(\gamma-\frac{1}{\alpha}\right)\cdot \left( \left(\gamma-\frac{1}{\alpha} \right)\alpha \beta - (\alpha-\beta)\right) +\delta  \alpha \left( \gamma -\frac{1}{\alpha}\right)  (\alpha -\beta)\right)\\
			&=~ -\frac{\left(1-\frac{\beta}{\alpha}\right)\cdot \left( \gamma - \delta\right)}{\left(\gamma-\frac{1}{\alpha}\right) \cdot \alpha^2 }\cdot \left( \left(\gamma-\frac{1}{\alpha}\right)\cdot \left(  \alpha \cdot \beta\cdot \gamma -\alpha\right) +\delta \cdot\left( \gamma \alpha\cdot  (\alpha -\beta) -\alpha +\beta\right)\right)\\
			&=~ -\frac{\left(1-\frac{\beta}{\alpha}\right)\cdot \left( \gamma - \delta\right)}{\left(\gamma-\frac{1}{\alpha}\right) \cdot \alpha^2 }\cdot \left(\alpha \beta \gamma^2 -\gamma(\beta +\alpha) +1 +\delta \cdot\left( \gamma \alpha\cdot  (\alpha -\beta) -\alpha +\beta\right)\right).
		\end{aligned}
	\end{equation}
	Define $\xi = \frac{\Delta \cdot \Gamma}{1-\kappa} \cdot \frac{\left(1-\frac{\beta}{\alpha}\right)\cdot \left( \gamma - \delta\right)}{\left(\gamma-\frac{1}{\alpha}\right) \cdot \alpha^2 }$.
	By \eqref{eq:hess_third_exp}, \eqref{eq:Pi1_proc} and \eqref{eq:Pi2_proc} we have
	$$
	\begin{aligned}
		\dhes{g}{\kappa,\tau} ~&=~ \xi\cdot \Bigg(
		-\left(  \alpha \beta \gamma^2 -\gamma(\beta +\alpha) + 1 +\delta \cdot\left( \gamma \alpha\cdot  (\alpha -\beta) -\alpha +\beta\right)\right) \\
		&\quad\quad\quad\quad -\delta \cdot \left( -\alpha^2 \cdot \gamma +2\alpha\gamma -1\right)-1+\gamma
		\Bigg)\\
		&=~ \xi \cdot \left(-\alpha \beta \gamma^2+\gamma(\alpha+\beta+1) -2 + \delta \left(\gamma \alpha \cdot (\beta -2  )+ \alpha-\beta +1 \right) \right)\\
		&=~ \xi \cdot \left(A_{\alpha,\beta} (\gamma) +\delta\cdot B_{\alpha,\beta}(\gamma )\right)\\
		&=~ \frac{\Delta \cdot \Gamma}{1-\kappa} \cdot \frac{\left(1-\frac{\beta}{\alpha}\right)\cdot \left( \gamma - \delta\right)}{\left(\gamma-\frac{1}{\alpha}\right) \cdot \alpha^2 } \cdot \left(A_{\alpha,\beta} (\gamma) +\delta\cdot B_{\alpha,\beta}(\gamma )\right),
	\end{aligned}
	$$
	which completes the proof of the lemma.
\end{proof}

\section{Better than Brute Force}
\label{sec:better_than_brute}
In this section we prove \Cref{thm:amls_smaller_brute}.
We first use \Cref{lem:g_convex_by_tau} to show that $\amlsbound(\alpha,c,\beta)<\brute(\beta)$ for all $\alpha,c\geq 1 $ and $\beta>1$.
Broadly speaking, the brute force algorithm presented in \cite{EsmerKMNS22} works as follows.
The algorithm iterates over $k$ from $0$ to $\frac{n}{\beta}$ (where $n$ is the size of the universe $U$), and the analysis focuses on the iteration in which $k$ is the minimum cardinality of a set in $\CF$.
For each value of $k$ the algorithm samples random subsets of the universe $U$ of size $\beta \cdot k$ and checks if each set is in the set system $\CF$.
The number of sampled sets is selected to be sufficiently large to ensure a constant success probability.
It can be shown (though not formally used by our proofs) that the number of sampled sets of size $\beta k$ should be $\approx \exp\left(n\cdot \xi_{\beta}\left(\frac{k}{n}\right) \right)$ where $\xi_{\beta}$ is defined by
\begin{equation}
	\label{eq:xi_def}
	\xi_{\beta}(\kappa) \coloneqq -\beta\cdot \entropy\left(\frac{1}{\beta}\right) \cdot \kappa +\entropy(\kappa).
\end{equation}
for all $0\leq \kappa\leq 1$.
It can also be easily verified that
\begin{equation}
	\label{eq:xi_to_g}
	\xi_{\beta}(\kappa) = g_{\alpha,\beta,c}(\kappa,\beta\cdot \kappa)
\end{equation}
for all $\alpha,\beta,c\geq 1$ and $0\leq \kappa\leq \frac{1}{\beta}$.
We use the following property of $\xi_{\beta}$.

\begin{lemma}
	\label{lem:xi_max}
	For all $\beta > 1$ it holds that $\max_{0\leq\kappa\leq \frac{1}{\beta}} \xi_{\beta}(\kappa) = \ln\left(\brute(\beta)\right)$.
\end{lemma}

\begin{proof}
	The expression $-\beta \cdot \entropy\left(\frac{1}{\beta}\right)\cdot \kappa$ is a linear function of $\kappa$, and $\entropy(\kappa)$ is a concave function of $\kappa$.
	So $\xi_{\beta}$ is a concave function.
	Also,
	$$\xi_{\beta}(0)= -\beta\cdot \entropy\left(\frac{1}{\beta}\right) \cdot 0 +\entropy(0)= 0$$
	and
	$$\xi_{\beta}\left(\frac{1}{\beta}\right)= -\beta\cdot \entropy\left(\frac{1}{\beta}\right) \cdot \frac{1}{\beta} +\entropy\left(\frac{1}{\beta}\right)= 0.$$
	Thus, $\xi_{\beta}$ has a maximum in $\left(0,\frac{1}{\beta} \right)$.
	Let $\kappa^*\in\left(0,\frac{1}{\beta} \right)$  be such a maximum and let $\xi'_{\beta}$ be the derivative of $\xi_{\beta}$.
	Then $\xi'_{\beta}(\kappa^*) = 0$.

	Using basic differentiation rules we have
	$$\xi'_{\beta}(\kappa) ~=~ -\beta\cdot \entropy\left(\frac{1}{\beta}\right)  + \ln\left(\frac{1-\kappa}{\kappa}\right).$$
	It follows that
	\begin{equation}
		\label{eq:kappa_star}
		\ln(1-\kappa^*) - \ln(\kappa^*)~=~\ln\left(\frac{1-\kappa^*}{\kappa^*}\right) ~=~ \beta\cdot \entropy\left(\frac{1}{\beta}\right).
	\end{equation}
	Hence,
	$$
	\begin{aligned}
		\max_{0\leq\kappa\leq \frac{1}{\beta}} \xi_{\beta}(\kappa) ~&=~ \xi_{\beta}(\kappa^*)\\
		&=~ -\beta\cdot \entropy\left(\frac{1}{\beta}\right) \cdot \kappa^* +\entropy(\kappa^*)\\
		&=~ -\kappa^*\cdot \left( \ln(1-\kappa^*) - \ln(\kappa^*)\right)  - \kappa^*\cdot \ln\left(\kappa^*\right) -(1-\kappa^*)\ln\left(1-\kappa^*\right) \\
		&=~ -\ln\left(1-\kappa^*\right) \\
		&=~ \ln\left( \frac{1-\kappa^*+\kappa^*}{1-\kappa^*}\right)\\
		&=~ \ln\left( 1+\frac{\kappa^*}{1-\kappa^*}\right)\\
		&=~ \ln\left( 1+\exp\left( - \beta\cdot \entropy\left(\frac{1}{\beta}\right)\right)\right).
	\end{aligned}
	$$
	The third equality follows from \eqref{eq:kappa_star} and the definition of $\entropy$.
	The forth, fifth and sixth equalities are simple re-arrangements of the terms.
	The seventh equality uses \eqref{eq:kappa_star} once more.
	Recall that $\brute(\beta)= 1+\exp\left(-\beta\cdot \entropy\left( \frac{1}{\beta}\right) \right)$, so $\max_{0\leq\kappa\leq \frac{1}{\beta}} \xi_{\beta}(\kappa) = \ln\left(\brute(\beta)\right)$.
\end{proof}

We use \Cref{lem:xi_max} in the proof of the following lemma.

\begin{lemma}
	\label{lem:amls_vs_brute}
	Let $\alpha,c\geq 1$ and $\beta>1$. Then $\amlsbound(\alpha,c,\beta) <\brute(\beta)$.
\end{lemma}

\begin{proof}
	There exists $\kappa \in \left[0,\frac{1}{\beta}\right]$ such that $\exp\left( g^*_{\alpha,\beta,c}(\kappa) \right)= \amlsbound(\alpha,c,\beta)$ (see \eqref{eq:amls_via_gstar}).
	Consider the following cases.
	\begin{itemize}
		\item If $\kappa=0$ or $\kappa=\frac{1}{\beta}$, it can be easily verified that $g^*_{\alpha,\beta,c}(\kappa)=0$.
			So
			$$\amlsbound(\alpha,c,\beta) ~=~ \exp\left( g^*_{\alpha,\beta,c}(\kappa) \right) ~=~ 1~ <~ \brute(\beta).$$
		\item If $0<\kappa<\frac{1}{\beta}$, then, by \Cref{lem:g_convex_by_tau}, we have
			$$
			\begin{aligned}
				\amlsbound(\alpha,c,\beta)  ~&=~ \exp\left( g^*_{\alpha,\beta,c}(\kappa) \right) \\
				&<~ \exp\left( g_{\alpha,\beta,c}(\kappa,\beta\cdot \kappa)\right) \\
				&=~ \exp\left( \xi_{\beta}(\kappa)\right)\\
				&\leq~ \exp\left(\max_{0\leq\kappa'\leq \frac{1}{\beta}} \xi_{\beta}(\kappa') \right) \\
				&=~ \exp\left( \ln\left(\brute(\beta)\right)\right) \\
				&=~ \brute(\beta).
			\end{aligned}
			$$
			The second equality follows from \eqref{eq:xi_to_g} and the forth equality follows from \Cref{lem:xi_max}.\qedhere
	\end{itemize}
\end{proof}
 
The next lemma provides the missing ingredient towards the proof of \Cref{thm:amls_smaller_brute}.

\begin{lemma}
	\label{lem:amls_limit}
	Let $\alpha\geq 1$ and $\beta>1$.
	Then $\lim_{c\rightarrow \infty} \amlsbound(\alpha,c,\beta) = \brute(\beta)$.
\end{lemma}

\begin{proof}
	By \Cref{lem:xi_max} there is some $0 \leq \kappa \leq \frac{1}{\beta}$ such that $\xi_{\beta}(\kappa) =\ln\left(\brute(\beta)\right)$.
	We define $L \coloneqq \liminf_{c\rightarrow \infty } g^*_{\alpha,c,\beta}(\kappa)$.
	By \Cref{lem:amls_vs_brute}, for every $c\geq 1$, it holds that $g^*_{\alpha,c,\beta} (\kappa)\leq \ln(\amlsbound(\alpha,c,\beta))\leq \ln\left(\brute(\beta)\right))$.
	So $L \leq \ln\left(\brute(\beta)\right))$.

	There exists  a strictly increasing sequence $(c_i)_{i=1}^{\infty}$ such that $L = \lim_{i\rightarrow \infty } g_{\alpha,c_i,\beta}(\kappa)$.
	For every $i\in \mathbb{N}$ define $\tau_i \coloneqq \argmin_{M^*_{\alpha,\beta}(\kappa)\leq \tau \leq \beta\kappa} g_{\alpha,\beta,c_i} (\kappa,\tau)$.
	Recall that the Bolzano-Weierstrass Theorem asserts that every bounded sequence has a convergent subsequence (see, e.g., in \cite[Theorem 2.4.1]{LafferriereLM22}).
	By the Bolzano-Weierstrass Theorem, as $M^*_{\alpha,\beta}(\kappa)\leq \tau_i \leq \beta\kappa$ for all $i\in \NN$, there exists a monotone sequence of indices $\left(i_j\right)_{j=0}^{\infty}$ such that $\tau_{i_j} \xrightarrow[j\rightarrow \infty]{} \tau^*$ for some $M^*_{\alpha,\beta}(\kappa)\leq \tau^* \leq \beta\kappa$.
	Thus, we have
	\begin{equation}
		\label{eq:liminf_L}
		\begin{aligned}
			L ~&=~ \lim_{j\rightarrow \infty} g^*_{\alpha,c_{i_j}, \beta}(\kappa )\\
			&=~ \lim_{j\rightarrow \infty} g_{\alpha,c_{i_j}, \beta}(\kappa,\tau_{i_j})\\
			&=~ \lim_{j\rightarrow \infty}  \left( \frac{\beta \kappa - \tau_{i_j}}{\alpha} \ln c_{i_j} - \tau_{i_j}\cdot \entropy\left(\gamma_{\alpha,\beta}(\kappa,\tau_{i_j})\right) -(1-\tau_{i_j})\cdot \entropy\left(\delta_{\alpha,\beta} (\kappa,\tau_{i_j})\right) + \entropy\left(\kappa\right)\right),
		\end{aligned}
	\end{equation}
	where second equality uses the definition of $g^*_{\abc}$~\eqref{eq:gstar_def} and the third equality uses the definition of $g_{\abc}$~\eqref{eq:g_def}.
 
	Assume towards contradiction that $\tau^* < \beta\kappa$.
	Then $ \frac{\beta \kappa - \tau_{i_j}}{\alpha} \ln c_{i_j} \xrightarrow[j\rightarrow \infty]{} \infty$.
	Also, since the entropy function and $\tau_{i_j}$ are both bounded, it follows that the expression $-(1-\tau_{i_j})\cdot \entropy\left(\delta_{\alpha,\beta} (\kappa,\tau_{i_j})\right) + \entropy\left(\kappa\right)$ is bounded.
	Thus, by \eqref{eq:liminf_L}, we have $L=\infty$, contradicting $L\leq \ln \left(\brute(\beta)\right)$.
	So $\tau^*=\beta \kappa$.
 
	Using $\tau^*=\beta\cdot\kappa$ and $\tau_{i_j} \xrightarrow[j\rightarrow \infty]{} \tau^*$, we can simplify the limit in \eqref{eq:liminf_L} and obtain
	\[
 	\begin{aligned}
 		L ~&=~ \lim_{j\rightarrow \infty}  \left( \frac{\beta \kappa - \tau_{i_j}}{\alpha} \ln c_{i_j} - \tau_{i_j}\cdot \entropy\left(\gamma_{\alpha,\beta}(\kappa,\tau_{i_j})\right) -(1-\tau_{i_j})\cdot \entropy\left(\delta_{\alpha,\beta} (\kappa,\tau_{i_j})\right) + \entropy\left(\kappa\right)\right).\\
 		&\geq~ -\beta \cdot \kappa\cdot \entropy\left(\gamma_{\alpha,\beta}(\kappa,\beta \cdot \kappa)\right) -(1-\beta \cdot \kappa)\cdot \entropy\left(\delta_{\alpha,\beta} (\kappa,\beta \cdot \kappa)\right) + \entropy\left(\kappa\right)\\
 		&=~ -\kappa \cdot \beta\cdot  \entropy\left(\frac{1}{\beta}\right)-(1-\beta \cdot \kappa)\cdot \entropy\left(0\right) + \entropy\left(\kappa\right)\\
 		&=~ \xi_{\beta}(\kappa) ~=~ \ln(\brute(\beta)).
 	\end{aligned} 
	\]
	Therefore,
	$$\liminf_{c\rightarrow \infty } \ln\left( \amlsbound(\alpha,c,\beta) \right) ~\geq~\liminf_{c\rightarrow \infty} g^*_{\abc}(\kappa) ~\geq~ \ln\left(\brute(\beta)\right),$$
	where the first inequality follows from \eqref{eq:amls_via_gstar}.
	Since $\ln$ is continuous the last inequality implies
	\begin{equation}
		\label{eq:liminf_amls}
		\liminf_{c\rightarrow \infty} \amlsbound(\alpha,c,\beta) ~\geq~ \brute(\beta),
	\end{equation}
	Also, by \Cref{lem:amls_vs_brute}, we have
	\begin{equation}
		\label{eq:limsup_amls}
		\limsup_{c\rightarrow \infty } \amlsbound(\alpha,c,\beta) ~\leq~ \limsup_{c\rightarrow \infty } \brute(\beta) ~=~ \brute(\beta).
	\end{equation}
	Combining \eqref{eq:liminf_amls} and \eqref{eq:limsup_amls} we get $\lim_{c\rightarrow \infty} \amlsbound(\alpha,c,\beta) =\brute(\beta)$.
\end{proof}

We can now proceed to the proof of \Cref{thm:amls_smaller_brute}.

\begin{proof}[Proof of \Cref{thm:amls_smaller_brute}]
	Let $\beta>1$ and let $\CL$ be a specification list.
	Pick an arbitrary element $(\alpha^*,c^*) \in \CL$ (recall that a specification list is always non-empty) and observe that $\alpha^*,c^*\geq 1$.
	Now, we have
	\begin{align*}
		\bestbound(\CL,\beta) ~&=~ \amlsbound(\CL,\beta)\\
		&=~ \exp\left( \max_{0\leq \kappa\leq \frac{1}{\beta}} \min_{(\alpha,c)\in \CL} g^*_{\abc}(\kappa))\right)\\
		&\leq~ \exp\left( \max_{0\leq \kappa\leq \frac{1}{\beta}}  g^*_{\alpha^*,c^*,\beta }(\kappa)\right)\\
		&=~ \amlsbound(\alpha^*,c^*,\beta)\\
		&<~ \brute(\beta),
	\end{align*}
	where the first equality follows from \Cref{cor:amls_is_the_best}, the second and third equalities follows from~\eqref{eq:amls_via_gstar}, and the last inequality follows from \Cref{lem:amls_vs_brute}.
	
	For the second part, let $\alpha \geq 1$.
	Then
	\[\lim_{c\rightarrow \infty} \bestbound(\alpha,c,\beta) ~=~ \lim_{c\rightarrow \infty} \amlsbound(\alpha,c,\beta) ~=~ \brute(\beta),\]
	where the first equality follows from \Cref{cor:amls_is_the_best}, and the second equality follows from \Cref{lem:amls_limit}.
\end{proof}

\section{Monotonicity Properties}
\label{sec:monotonicity_amls}
In this section we prove that $\amlsbound(\alpha,c,\beta)$ is strictly monotone in $\alpha$ in the interval $\left[1,\beta\right]$, and use this result to prove \Cref{lem:better_than_esa}.

\begin{lemma}
	\label{lem:monotone_in_alpha}
	For every $\beta\geq  \alpha'>\alpha \geq 1$ and every $c > 1$ it holds that $\amlsbound(\alpha,c,\beta) < \amlsbound(\alpha',c,\beta)$.
\end{lemma}

The proof of \Cref{lem:monotone_in_alpha} is given towards the end of this section. 
We first use  \Cref{lem:monotone_in_alpha} to prove \Cref{lem:better_than_esa}. 

\begin{proof}[Proof of \Cref{lem:better_than_esa}]
	Let $\beta>\alpha \geq 1$ and $c>1$. Then
	$$\bestbound(\alpha,c,\beta) ~=~ \amlsbound(\alpha,c,\beta) ~<~ \amlsbound(\beta,c,\beta) ~=~ \bestbound(\beta,c,\beta) ~=~ \esaamlsbound(\beta,c).$$
	The first and second equalities follow from \Cref{cor:amls_is_the_best}.
	The inequality holds by \Cref{lem:monotone_in_alpha}, and the last equality follows form \Cref{lem:coincide_with_esa}.
\end{proof}

In order to prove \Cref{lem:monotone_in_alpha} we give an alternative formula for $g^*_{\alpha,\beta,c}$ \eqref{eq:gstar_def} as a solution for a continuous optimization problem in two variables.
Though this alternative formula uses a continuous optimization problem, it is inspired by an interpretation of the discrete analysis of \Cref{algo:intermediary,algo:final}.
The algorithm samples a $t$-element set $X\subseteq U$, and the analysis focuses on samples which satisfy $\abs{X\cap \OPT} \geq y$, for a carefully selected $y$.
Subsequently, the $\alpha$-extension oracle is invoked with the query $(X, k-y)$.
The algorithm optimally selects $y=\left(1-\frac{\beta}{\alpha}\right) \cdot k +\frac{t}{\alpha}$.
The analysis we present in this section leaves $y$ as an additional parameter to be optimized.
Though this seems to only yield a more involved formula, this formula turns out to be useful to determine the behavior of $\amlsbound(\alpha,c,\beta)$ as $\alpha$ changes.

For every $1\leq \alpha\leq \beta$ we define
\begin{equation}
	\label{eq:X_def}
	X_{\alpha, \beta}(\kappa,\tau)  ~\coloneqq \left( 1 - \frac{\beta}{\alpha} \right)\cdot \kappa + \frac{\tau}{\alpha}.
\end{equation}
Similarly, for every $1\leq \alpha\leq \beta$ and $\kappa \in \left[ 0,\frac{1}{\beta}\right]$ we define a set
\begin{equation}
	\label{eq:D_def}
	D_{\alpha, \beta}(\kappa) ~\coloneqq \big\{ (\tau,y) \in \mathbb{R}^{2} \bigmid 0 \leq \tau \leq \beta \cdot \kappa,~ \max\{0,X_{\alpha, \beta}(\kappa,\tau)\} \leq y \leq \min\{\kappa,\tau\}\big\}.
\end{equation}
Finally, for every $1\leq \alpha\leq \beta$ and $\kappa\in \left[0,\frac{1}{\beta}\right]$ we define the function
\begin{equation}
	\label{eq:gtilk_def}
	\gtilk[c] (\tau,y) \coloneqq (\kappa-y)\cdot \ln(c) - \tau\cdot \entropy\left(\frac{y}{\tau}\right) - (1-\tau)\cdot \entropy\left( \frac{\kappa - y}{1 - \tau}\right) + \entropy(\kappa).
\end{equation}
It can be easily verified that
\begin{equation}
	\label{eq:gs_vs_gt}
	g^*_{\abc}(\kappa,\tau) = \gtilk[c]\big( \tau, X_{\alpha,\beta}\left(\kappa,\tau) \right)
\end{equation}
unless $\tau\in \{0,1\}$.

The next lemma provides the alternative formula for $g^*_{\alpha,\beta,c}$.

\begin{lemma}\label{lem:equiv_g}
	Let $\beta \geq \alpha \geq 1$,  $c> 1$ and $\kappa \in \left(0,\frac{1}{\beta}\right)$. Then
	\begin{equation}
		\label{eq:alternative_gstar}
		\gst[\abc](\kappa )~=~ \min_{(\tau,y) \in D_{\alpha, \beta}(k)} \gtilk[c](\tau, y).
	\end{equation}
	Moreover, for all $(\tau,y) \in D_{\alpha, \beta}(\kappa)$
	such that $\gst[\abc](\kappa) = \gtilk[c](\tau, y)$, we have $y = X_{\alpha,\beta}(\kappa,\tau)$.
\end{lemma}

\begin{proof}
	The proof uses an alternative representation of the set $D_{\alpha,\beta}(\kappa)$.
	We define
	\begin{equation}
		\label{eq:E_set_def}
		E_{\alpha, \beta}(\kappa) \coloneqq \left\{ (\tau,y) \in \mathbb{R}^{2} ~\middle|~ 0 \leq y \leq \kappa, ~y \leq \tau \leq \alpha \cdot \left( y - \left( 1 - \frac{\beta}{\alpha} \right) \cdot \kappa \right)  \right\}.
	\end{equation}

	\begin{claim}
		\label{claim:D_is_E}
		$E_{\alpha,\beta}(\kappa) = D_{\alpha,\beta}(\kappa)$.
	\end{claim}
	\begin{claimproof}
		Let $(\tau,y) \in D_{\alpha,\beta}(\kappa)$.
		Then
		\[0\leq \max\{0,X_{\alpha,\beta}(\kappa,\tau)\} \leq y\leq \min\{\kappa,\tau\}\leq \kappa.\]
		Furthermore, $y \leq \min\{\kappa,\tau\}\leq \tau$ and
		\[y\geq \max\{0,X_{\alpha,\beta}(\kappa,\tau)\} \geq X_{\alpha,\beta}(\kappa,\tau) = \left(1-\frac{\beta}{\alpha}\right)\kappa +\frac{\tau}{\alpha}\]
		and thus, $\tau \leq \alpha\left( y- \left(1-\frac{\beta}{\alpha}\right)\cdot \kappa \right)$.
		Overall, we have that $0\leq y\leq\kappa$ and $t\leq \tau \leq \alpha\left( y- \left(1-\frac{\beta}{\alpha}\right)\cdot \kappa \right)$.
		By \eqref{eq:E_set_def} we have $(\tau,y)\in E_{\alpha,\beta}(\kappa)$, and we conclude that
		\begin{equation}
			\label{eq:D_subset_E}
			D_{\alpha,\beta}(\kappa)\subseteq E_{\alpha,\beta}(\kappa).
		\end{equation}
		
		Similarly, let $(\tau',y') \in E_{\alpha,\beta}(\kappa)$.
		It follows that $\tau'\geq y' \geq 0$.
		Furthermore,
		\[\tau' ~\leq~ \alpha\left( y'- \left(1-\frac{\beta}{\alpha}\right) \cdot \kappa \right) ~\leq~ \alpha \cdot\left(\kappa- \left(1-\frac{\beta}{\alpha}\right) \cdot \kappa \right)  ~=~ \beta \cdot \kappa,\]
		where the second inequality follows from $y' \leq \kappa$.
		So together, we obtain that
		\begin{equation}
			\label{eq:tau_prime_cond}
			0\leq \tau'\leq \beta \cdot \kappa.
		\end{equation}
		
		By re-arranging the inequality $\tau' \leq \alpha \left(y'-\left( 1-\frac{\beta}{\alpha}\right)\cdot \kappa \right) $ we obtain 
		$$y'\geq \left( 1-\frac{\beta}{\alpha}\right) \cdot \kappa +\frac{\tau'}{\alpha} = X_{\alpha,\beta}(\kappa,\tau').$$
		As $(\tau',y')\in E_{\alpha,\beta}(\kappa)$, it also holds that $y'\geq 0$.
		So $y'\geq \max\{0,X_{\alpha,\beta}(\kappa,\tau')\}$.
		Finally, since $(\tau',y')\in E_{\alpha,\beta}(\kappa)$ we have $y'\leq \kappa$ and $y'\leq \tau'$, and hence $y'\leq \min\{\kappa,\tau'\}$.
		So overall
		\begin{equation}
			\label{eq:y_prime_cond}
			\max\{0,X_{\alpha,\beta}(\kappa,\tau')\} ~\leq~ y' ~\leq~ \min\{\kappa,\tau'\}.
		\end{equation}
		By \eqref{eq:tau_prime_cond} and \eqref{eq:y_prime_cond} it holds that $(\tau',y') \in D_{\alpha,\beta}(\kappa)$.
		Thus
		\begin{equation}
			\label{eq:E_subset_D}
			E_{\alpha,\beta}(\kappa)\subseteq D_{\alpha,\beta}(\kappa).
		\end{equation}
		By \eqref{eq:D_subset_E} and \eqref{eq:E_subset_D} we have $D_{\alpha,\beta}(\kappa) = E_{\alpha,\beta}(\kappa)$.
	\end{claimproof}

	Using \Cref{claim:D_is_E} we have
	\begin{equation}
		\label{eq:alternative_first_step}
		\min_{(\tau,y) \in D_{\alpha, \beta}(k)} \gtilk[c](\tau,y) ~=~ \min_{(\tau,y) \in E_{\alpha, \beta}(k)} \gtilk[c](\tau, y) ~=~ \min_{~0\leq y\leq \kappa~} \min_{~y\leq \tau \leq \alpha\cdot \left(y-\left(1-\frac{\beta}{\alpha}\right) \cdot \kappa \right)~} \gtilk[c](\tau, y).
	\end{equation}
	In order to simplify \eqref{eq:alternative_first_step} we use some analytical properties of $\gtilk[c](\tau,y)$.

	\begin{claim}
		\label{claim:gtilk_der}
		It holds that $\frac{\partial \gtilk[c](\tau,y)}{\partial \tau} = \ln \left(1- \frac{y}{\tau}\right) - \ln \left(1-\frac{\kappa-y}{1-\tau}\right)$.
	\end{claim}
	\begin{claimproof}
		For every $a\in \mathbb{R}$ define $q_{a} (x) = x\cdot \entropy\left( \frac{a}{x}\right)$. Using basic differentiation rules we have,
		\begin{align*}
			\frac{\partial q_a(x) }{\partial x} ~&=~ \entropy\left( \frac{a}{x}\right) +x\cdot \frac{-a}{x^2}\cdot \ln\left( \frac{1-\frac{a}{x}}{\left(\frac{a}{x}\right)}\right)\\
			&=~-\frac{a}{x} \cdot \ln \left( \frac{a}{x}\right)  - \left(1-\frac{a}{x}\right) \cdot \ln \left(1-\frac{a}{x}\right) -\frac{a}{x} \cdot \ln \left(1-\frac{a}{x}\right)  +\frac{a}{x} \cdot \ln \left( \frac{a}{x}\right) \\
			&=~ -\ln\left(1-\frac{a}{x}\right).
		\end{align*}
		Thus,
		\begin{align*}
			\frac{\partial \gtilk[c](\tau,y)}{\partial \tau} ~&=~ \frac{\partial}{\partial \tau} \left( (\kappa-y)\cdot \ln(c) - \tau\cdot \entropy\left( \frac{y}{\tau} \right)  - (1-\tau)\cdot \entropy\left( \frac{\kappa - y}{1 - \tau} \right)  + \entropy(\kappa) \right) \\
			&=~ -\frac{\partial q_y(\tau)}{\partial \tau} - \frac{\partial q_{\kappa-t} (1-\tau)}{\partial \tau}\\
			&=~ \ln\left( 1-\frac{y}{\tau}\right) - \ln \left(1-\frac{\kappa-y}{1-\tau} \right).
		\end{align*}
	\end{claimproof}

	We use \Cref{claim:gtilk_der} to show the following.
	\begin{claim}
		\label{claim:gtlik_convex}
		For every $0\leq y\leq \kappa$, it holds that $\gtilk[c](\tau,y)$ is strictly convex as a function of $\tau$, and has a minimum at $\tau =\frac{y}{\kappa}$.
	\end{claim}

	\begin{claimproof}
		For every $0\leq y\leq \kappa$ the expression $1-\frac{y}{\tau}$ is increasing with $\tau$ and the expression $1-\frac{\kappa-y}{1-\tau}$ is decreasing with $\tau$.
		Furthermore, $1-\frac{y}{\tau}$ is \emph{strictly} increasing, unless $y=0$, and in this case $1-\frac{\kappa-y}{1-\tau}$ is \emph{strictly} decreasing.
		It follows that $\frac{\partial \gtilk[c](\tau,y)}{\partial \tau} = \ln \left(  1- \frac{y}{\tau}\right)- \ln \left(1-\frac{\kappa-y}{1-\tau}\right)$ is strictly increasing with $\tau$ for every $0\leq y\leq \kappa$.
		We conclude that $\gtilk[c](\tau,y)$ is strictly convex as a function of $\tau$ (for a fixed $y\in [0,\kappa]$).
		Furthermore,
		\[\frac{\partial \gtilk[c](\tau,y)}{\partial \tau}\Bigg|_{\tau = \frac{y}{\kappa}} ~=~ \ln \left(  1- \frac{y}{\frac{y}{\kappa}}\right)- \ln \left(1-\frac{\kappa-y}{1-\frac{y}{\kappa}}\right)~=~ \ln(1-\kappa )- \ln(1-\kappa) ~=~ 0,\]
		so the minimum of $\gtilk[c](\tau,y)$, as a function of $\tau$ (for a fixed $y$), is at $\tau = \frac{y}{\kappa}$.
	\end{claimproof}

	For every $0\leq y\leq \kappa$,
	the value of $\gtilk[c](\tau,y)$ at its minimum is
	\begin{equation}
		\label{eq:gtilk_minimum}
		\begin{aligned}
			\gtilk[c]\left(\frac{y}{\kappa},y\right) ~&=~ (\kappa - y)\cdot\ln(c) - \frac{y}{\kappa}\cdot\entropy\left(\frac{y}{\frac{y}{\kappa}}\right)- \left( 1 - \frac{y}{\kappa} \right) \cdot \entropy\left( \frac{\kappa - y}{1 - \frac{y}{\kappa}} \right) + \entropy(\kappa) \\
			&=~ (\kappa - y)\cdot\ln(c) - \frac{y}{\kappa}\cdot\entropy(\kappa) - \left( 1 - \frac{y}{\kappa} \right) \cdot \entropy\left( \kappa \right) + \entropy(\kappa)\\
			&=~ (\kappa - y)\cdot\ln(c).
		\end{aligned}
	\end{equation}

	By the above, for every $0\leq y\leq \kappa$ such that $\frac{y}{\kappa} ~\leq ~ \alpha\cdot \left(y-\left(1-\frac{\beta}{\alpha}\right)\cdot \kappa \right)$ it holds that
	\begin{equation}
		\label{eq:minimum_in_range}
		\min_{~y\leq \tau \leq \alpha\cdot \left(y-\left(1-\frac{\beta}{\alpha}\right) \cdot \kappa \right)~ } \gtilk[c](\tau, y) ~=~ \gtilk[c]\left(\frac{y}{\kappa},y\right)  ~=~ \left( \kappa-y\right)\cdot \ln (c).
	\end{equation}
	The first equality holds since $\gtilk[c](\tau,y)$ is convex with a minimum at $\frac{y}{\kappa}\geq y$, as a function of $\tau$ (\Cref{claim:gtlik_convex}).
	The second equality follows from \eqref{eq:gtilk_minimum}.
	Also, observe that
	\begin{equation}
		\label{eq:y_greater_cond}
		\begin{aligned}
			&\frac{y}{\kappa} ~\leq~ \alpha\cdot \left(y-\left(1-\frac{\beta}{\alpha}\right)\cdot \kappa \right)\\
			\iff~&y\left( \frac{1}{\kappa} -\alpha\right) ~\leq~  -(\alpha -\beta)\cdot \kappa \\
			\iff~&y~\leq~ \frac{\beta -\alpha }{1-\alpha\kappa} \cdot \kappa^2 ~=~ M_{\alpha,\beta}(\kappa)\cdot \kappa,
		\end{aligned}
	\end{equation}
	where the second transition holds as $\frac{1}{\kappa} > \beta \geq \alpha$.
	Recall that $M_{\alpha,\beta}(\kappa)$ is defined in \eqref{eq:Mdef}.

	By \eqref{eq:minimum_in_range} and \eqref{eq:y_greater_cond} we have
	\begin{equation}
		\label{eq:alternative_elimination}
		\begin{aligned}
			&\min_{~0\leq y\leq M_{\alpha,\beta}(\kappa)\cdot \kappa~} \min_{~y\leq \tau \leq \alpha\cdot \left(y-\left(1-\frac{\beta}{\alpha}\right) \cdot \kappa \right)~ } \gtilk[c](\tau, y) \\
			=~& \min_{~0\leq y\leq M_{\alpha,\beta}(\kappa)\cdot \kappa~} (\kappa-y)\ln(c) \\
			=~& (\kappa - M_{\alpha,\beta}(\kappa)\cdot \kappa)\cdot \ln c \\
			=~& \gtilk[c]\left(\frac{M_{\alpha,\beta}(\kappa)\cdot \kappa}{\kappa} ,~M_{\alpha,\beta}(\kappa)\cdot \kappa\right) \\
			=~& \min_{~M_{\alpha,\beta}(\kappa)\cdot \kappa~\leq~ \tau ~\leq~ \alpha\cdot \left(M_{\alpha,\beta}(\kappa)\cdot \kappa-\left(1-\frac{\beta}{\alpha}\right) \cdot \kappa \right)~ } \gtilk[c]\Big(\tau,~ M_{\alpha,\beta}(\kappa)\cdot \kappa\Big).
		\end{aligned}
	\end{equation}

	It can be easily verified that $M_{\alpha,\beta}(\kappa) \leq 1$.
	Thus, we can use \eqref{eq:alternative_elimination} to change the range of $y$ in \eqref{eq:alternative_first_step} as follows:
	\begin{equation}
		\label{eq:alternative_second_step}
		\begin{aligned}
			\min_{(\tau,y) \in D_{\alpha, \beta}(k)} \gtilk[c](\tau, y)~=&~ \min_{~0\leq y\leq \kappa~} \min_{~y\leq \tau \leq \alpha\cdot \left(y-\left(1-\frac{\beta}{\alpha}\right) \cdot \kappa \right)~ } \gtilk[c](\tau, y)\\
			~=&~ \min_{~M_{\alpha,\beta}(\kappa)\cdot \kappa\leq y\leq \kappa~} \min_{~y\leq \tau \leq \alpha\cdot \left(y-\left(1-\frac{\beta}{\alpha}\right) \cdot \kappa \right)~ } \gtilk[c](\tau, y).
		\end{aligned}
	\end{equation}

	By \Cref{claim:gtlik_convex}, for every $0\leq y\leq \kappa$ such that $\frac{y}{\kappa }~\geq ~ \alpha\cdot \left(y-\left(1-\frac{\beta}{\alpha}\right)\cdot \kappa \right)$, it holds that
	\begin{equation}
		\label{eq:minimum_out_of_range}
		\min_{~y\leq \tau \leq \alpha\cdot \left(y-\left(1-\frac{\beta}{\alpha}\right) \cdot \kappa \right)~ } \gtilk[c](\tau, y) ~=~  	\gtilk[c]\left( \alpha\cdot \left(y-\left(1-\frac{\beta}{\alpha}\right) \cdot \kappa \right),~y\right) .
	\end{equation}
	Similarly to \eqref{eq:y_greater_cond} it holds that
	\begin{equation}
		\label{eq:y_geq_cond}
		\frac{y}{\kappa }~\geq ~ \alpha\cdot \left(y-\left(1-\frac{\beta}{\alpha}\right)\cdot \kappa \right) ~~\iff ~~y\geq M_{\alpha,\beta} (\kappa)\cdot \kappa.
	\end{equation}
	Thus, using \eqref{eq:minimum_out_of_range} and \eqref{eq:alternative_second_step}, we get
	\begin{equation}
		\label{eq:alternative_third_step}
		\begin{aligned}
			\min_{(\tau,y) \in D_{\alpha, \beta}(k)} \gtilk[c](\tau, y)
			~=~& \min_{~M_{\alpha,\beta}(\kappa)\cdot \kappa\leq y\leq \kappa~} \min_{~y\leq \tau \leq \alpha\cdot \left(y-\left(1-\frac{\beta}{\alpha}\right) \cdot \kappa \right)~ } \gtilk[c](\tau, y)\\
			=~& \min_{~M_{\alpha,\beta}(\kappa)\cdot \kappa\leq y\leq \kappa~} \gtilk[c]\left( \alpha\cdot \left(y-\left(1-\frac{\beta}{\alpha}\right) \cdot \kappa \right),~y\right) \\
			=~&\min_{\Mab(\kappa) \leq \tau \leq \beta\cdot\kappa} \gtilk[c](\tau, X_{\alpha,\beta}(\kappa,\tau))\\
			=~& \min_{\Mab(\kappa) \leq \tau \leq \beta\cdot\kappa} g_{\alpha,\beta,c}(\kappa,\tau) \\
			=~& \gst[\abc](\kappa).
		\end{aligned}
	\end{equation}
	The third equality follows from substitution $y$ with $\tau = \alpha \left( y- \left(1-\frac{\beta}{\alpha} \right) \cdot \kappa\right)$, which is equivalent to $y=X_{\alpha,\beta}(\kappa,\tau)$.
	The forth equality follows from $\gtilk(\tau,X_{\alpha,\beta}(\kappa,\tau)) = g_{\alpha,\beta,c}(\kappa,\tau)$.
	That last equality follows from \eqref{eq:gstar_def}.
	Observe that \eqref{eq:alternative_third_step} implies \eqref{eq:alternative_gstar}.

	Let $(\tau,y)\in D_{\alpha,\beta}(\kappa)=E_{\alpha,\beta}(\kappa)$ such that $g^*_{\alpha,\beta,c} (\kappa )  = \gtilk[c](\tau,y)$. To complete the proof of the lemma we are left to show that $y=X_{\alpha,\beta}(\kappa,\tau)$. By \eqref{eq:alternative_third_step} it follows that
	\begin{equation}
		\label{eq:show_cond}
	\gtilk[c](\tau,y) ~=~ \min_{ (\tau',y')\in D_{\alpha,\beta}(\kappa)} \gtilk[c](\tau',y')~=~\min_{ (\tau',y')\in E_{\alpha,\beta}(\kappa)} \gtilk[c](\tau',y').
	\end{equation}

	Assume towards contradiction that $y< M_{\alpha,\beta}(\kappa)\cdot \kappa$.
	Then, since $(y,\tau)\in E_{\alpha,\beta}(\kappa)$, we have
	\begin{equation}
		\label{eq:show_cond_y_small}
		\begin{aligned}
	\gtilk[c](\tau,y) ~\geq~&
	\min_{~y\leq \tau' \leq \alpha\cdot \left(y-\left(1-\frac{\beta}{\alpha}\right) \cdot \kappa \right)~ } \gtilk[c](\tau', y)  \\
	=~&(\kappa - y)\cdot \ln(c)\\
	>~&
	(\kappa - M_{\alpha,\beta}(\kappa)\cdot \kappa)\cdot \ln(c)  \\
	=~& \min_{~M_{\alpha,\beta}(\kappa)\cdot \kappa~\leq~ \tau' ~\leq~ \alpha\cdot \left(M_{\alpha,\beta}(\kappa)\cdot \kappa-\left(1-\frac{\beta}{\alpha}\right) \cdot \kappa \right)~ } \gtilk[c]\left(\tau',~ M_{\alpha,\beta}(\kappa)\cdot \kappa\right)\\
	\geq~& \min_{(\tau',t')\in E_{\alpha,\beta}(\kappa)}\gtilk[c]\left(\tau',~ y'\right).
		\end{aligned}
	\end{equation}
	The first inequality follows from \eqref{eq:show_cond} and the definition of $E_{\alpha,\beta}$ \eqref{eq:E_set_def}.
	The first equality follows from \eqref{eq:minimum_in_range} and \eqref{eq:y_greater_cond}.
	The second equality uses \eqref{eq:alternative_elimination}.
	The last inequality is also a consequence of the definition of $E_{\alpha,\beta}$ \eqref{eq:E_set_def}.
	Observe that \eqref{eq:show_cond_y_small} contradict \eqref{eq:show_cond}, so $y\geq M_{\alpha,\beta}(\kappa)\cdot \kappa$.

	By \eqref{eq:show_cond} it holds that
	\begin{equation}
		\label{eq:gtilk_greater_y}
		\gtilk[c](\tau,y) ~=~ \min_{~y\leq \tau' \leq \alpha\cdot \left(y-\left(1-\frac{\beta}{\alpha}\right) \cdot \kappa \right)~ } \gtilk[c](\tau', y)
	\end{equation}
	By \Cref{claim:gtlik_convex} we have that $\gtilk(\tau',y)$ is strictly convex as a function of $\tau'$ with a minimum at $\tau'=\frac{y}{\kappa}$ and by \eqref{eq:y_geq_cond} we have
	$\frac{y}{\kappa }\geq  \alpha\cdot \left(y-\left(1-\frac{\beta}{\alpha}\right)\cdot \kappa \right)$.
	Thus, the minimum of the RHS of \eqref{eq:gtilk_greater_y} is at $\tau'= \alpha\cdot \left(y-\left(1-\frac{\beta}{\alpha}\right)\cdot \kappa \right)$.
	Therefore, $\tau = \alpha\cdot \left(y-\left(1-\frac{\beta}{\alpha}\right)\cdot \kappa \right)$, and by rearranging the term we get $y= X_{\alpha,\beta}(\kappa,\tau)$.
\end{proof}

Using \Cref{lem:equiv_g} we can easily derive the following lemma.

\begin{lemma}
	\label{lem:pointwise_kappa}
	Let $\beta\geq \alpha' >\alpha \geq 1$, $c> 1$  and $\kappa \in \left(0,\frac{1}{\beta}\right)$.
	Then $g^*_{\alpha,\beta,c}(\kappa) < g^*_{\alpha',\beta,c} (\kappa)$.
\end{lemma}
\begin{proof}
	By \Cref{lem:equiv_g} we have $g^*_{\alpha',\beta,c}(\kappa)= \min_{(\tau,y)\in D_{\alpha',\beta}(\kappa)} \gtilk[c](\tau,y)$.
	Thus, there is $(\tau',y') \in D_{\alpha',\beta}$ such that $\gtilk[c](\tau',y')=g^*_{\alpha',\beta,c}(\kappa) $ and $y'= X_{\alpha',\beta}(\kappa,\tau')$.
	
	Assume towards contradiction that $ \tau'=\beta \kappa$. Then 
	$$
	g^*_{\alpha',\beta,c} (\kappa)  ~=~ \gtilk[c](\tau',y')~=~\gtilk[c](\beta\cdot \kappa, X_{\alpha',\beta}(\kappa,\beta\cdot \kappa )) ~=~g_{\alpha',\beta,c}(\kappa,\beta\cdot \kappa).
	$$
	where the last equality follows from \eqref{eq:gs_vs_gt}. However, by \Cref{lem:g_convex_by_tau} it holds that $g^*_{\alpha',\beta,c} (\kappa)  < g_{\alpha',\beta,c}(\kappa,\beta\cdot \kappa)$, contradicting the above. Thus $\tau'\neq \beta \kappa$.
	
	As $(\tau',y') \in D_{\alpha',\beta}$ it also  holds that $0\leq \tau'< \beta\cdot \kappa$ and $0\leq y'\leq \min\{\kappa,\tau'\}$. 
	Furthermore,
	\[
	\begin{aligned}
		y' ~&=~ X_{\alpha',\beta}(\kappa,\tau') \\
		&=~ \left(1-\frac{\beta}{\alpha'}\right) \cdot \kappa +\frac{\tau'}{\alpha'} \\
		&=~ 1 -\frac{1}{\alpha'} \left( \beta\cdot \kappa -\tau'\right)\\
		&>~ 1 -\frac{1}{\alpha} \left( \beta\cdot \kappa -\tau'\right)\\
		&=~ X_{\alpha,\beta}(\kappa,\tau'),
		\end{aligned}\]
	where the inequality holds since $\tau'<\beta\cdot \kappa$, and $ \alpha' >\alpha$.
	Thus $(\tau',y')\in D_{\alpha,\beta}(\kappa)$ and $y'\neq X_{\alpha,\beta}(\kappa,\tau')$.
	Using \Cref{lem:equiv_g} once more we get
	\[g^*_{\alpha',\beta,c} (\kappa) ~=~ \gtilk[c] (\tau',y') ~>~ \min_{(\tau,y)\in D_{\alpha,\beta}} \gtilk[c](\tau,t) ~=~g^*_{\alpha,\beta,c}(\kappa).\qedhere\]
\end{proof} 

\Cref{lem:monotone_in_alpha} essentially follows from \Cref{lem:pointwise_kappa}, though the proof itself involves some technical steps which exclude the possibility of a corner case.

\begin{proof}[Proof of \Cref{lem:monotone_in_alpha}]
	Define $\tilde{\kappa} = \frac{1}{2}\cdot \frac{1}{\beta}$.
	Since $\entropy(x)$ is concave, for every $M_{\alpha,\beta}(\tilde{\kappa})\leq \tau < \beta \tilde{\kappa}$, it holds that
	\begin{equation}
		\label{eq:mon_corner_first}
		\begin{aligned}
			g_{\alpha,\beta,c}(\tilde{\kappa},\tau) ~=~& \frac{\beta \tilde{\kappa} - \tau}{\alpha} \ln c -
				\tau\cdot \entropy\left(\gamma_{\alpha,\beta}(\tilde{\kappa},\tau)\right) -(1-\tau)\cdot \entropy\left(\delta_{\alpha,\beta} (\tilde{\kappa},\tau)\right) + \entropy\left(\tilde{\kappa}\right)\\
			\geq~& \frac{\beta \tilde{\kappa} - \tau}{\alpha} \ln c -
				\entropy\left(\tau \cdot \gamma_{\alpha,\beta}(\tilde{\kappa},\tau)+(1-\tau)\cdot\delta_{\alpha,\beta} (\tilde{\kappa},\tau)\right) + \entropy\left(\tilde{\kappa}\right)\\
	 		=~& \frac{\beta \tilde{\kappa} - \tau}{\alpha} \ln c -
				\entropy\left(\tilde{\kappa}\right) + \entropy\left(\tilde{\kappa}\right)\\
			=~& \frac{\beta \tilde{\kappa} - \tau}{\alpha} \ln c~>~0,
		\end{aligned}
	\end{equation}
	where the second equality holds as $\tau \cdot \gamma_{\alpha,\beta}(\tilde{\kappa},\tau)+(1-\tau)\cdot\delta_{\alpha,\beta} (\tilde{\kappa},\tau) = \tilde{\kappa}$ and the last inequality uses $c>1$.
	Furthermore, 
	\begin{equation}
		\label{eq:mon_corner_second}
		\begin{aligned}
			&g_{\alpha,\beta,c}(\tilde{\kappa},\beta \cdot \tilde{\kappa}) \\
			=~& \frac{\beta \tilde{\kappa} - \beta \cdot \tilde{\kappa}}{\alpha} \ln c -
				\beta \cdot \tilde{\kappa}\cdot \entropy\left(\gamma_{\alpha,\beta}(\tilde{\kappa},\beta \cdot \tilde{\kappa})\right) -(1-\beta \cdot \tilde{\kappa})\cdot \entropy\left(\delta_{\alpha,\beta} (\tilde{\kappa},\beta \cdot \tilde{\kappa})\right) + \entropy\left(\tilde{\kappa}\right)\\
			=~& -\beta \cdot \tilde{\kappa}\cdot \entropy\left(\frac{1}{\beta}\right) -(1-\beta \cdot \tilde{\kappa})\cdot \entropy\left(0\right) + \entropy\left(\tilde{\kappa}\right)\\
			>~& -\entropy\left(\beta \tilde{\kappa}\cdot \frac{1}{\beta} + \left(1-\beta\tilde{\kappa}\right)\cdot 0  \right) + \entropy\left(\tilde{\kappa}\right) ~=~0,
		\end{aligned}
	\end{equation}
	where the inequality holds as $\entropy$ is strictly concave.
	By \eqref{eq:mon_corner_first} and \eqref{eq:mon_corner_second} it follows that
	\[g^*_{\alpha,\beta,c} (\tilde{\kappa}) =\min_{M_{\alpha,\beta}(\tilde{\kappa})\leq \tau \leq \beta \tilde{\kappa}} g_{\alpha,\beta,c}(\tilde{\kappa},\tau) >0.\]
	
	Now, there exists $\kappa\in \left[ 0,\frac{1}{\beta}\right]$ such that
	\[g^*_{\alpha,\beta,c} (\kappa) ~=~ \max_{\kappa' \in \left[0,\frac{1}{\beta}\right]} g^*_{\alpha,\beta,c}(\kappa') ~>~ 0.\]
	It can be easily verified that $g^*_{\alpha,\beta,c}(0) = g^*_{\alpha,\beta,c}\left(\frac{1}{\beta}\right)= 0$ (select $\tau =0$ in the first and $\tau = \beta \cdot \kappa$ in the latter case).
	So $0< \kappa < \frac{1}{\beta}$.
	
	Thus, by \Cref{lem:pointwise_kappa}, we have
	\begin{align*}
		\ln\left(\amlsbound(\alpha,\beta,c)\right) ~=~& \max_{k'\in \left[0,\frac{1}{\beta}\right]} g^*_{\alpha,\beta,c}(\kappa')\\
		=~& g^*_{\alpha,\beta,c}(\kappa)\\
		<~& g^*_{\alpha',\beta,c}(\kappa)\\
		\leq~& \max_{k'\in \left[0,\frac{1}{\beta}\right]} g^*_{\alpha',\beta,c}(\kappa')\\
		=~& \ln\left(\amlsbound(\alpha',\beta,c)\right)
	\end{align*}
	The first and last equality follows from \eqref{eq:amls_via_gstar}.
	Thus, $\amlsbound(\alpha,\beta,c) < \amlsbound(\alpha',\beta,c)$.
\end{proof}

\section{Conclusion}
\label{sec:conclusion}
In this paper we studied how exponential-time approximation algorithms can be obtained from existing polynomial-time approximation algorithms, existing parameterized exact algorithms and existing parameterized approximation algorithms.
We provided a theoretical oracle model by which the above question can be rigorously studied and showed that the approximate monotone local search approach \cite{FominGLS19,EsmerKMNS22} attains optimal running times (up to polynomial factors). Furthermore, we provided the mathematical machinery to compute the running time of the resulting algorithms in practice.

While previous works on monotone local search \cite{FominGLS19,EsmerKMNS22,Lee21} only provided algorithmic results, in this work we use a restricted oracle model in which we are also able to show the optimality of our algoritms (which in particular implies the algorithms from \cite{FominGLS19,EsmerKMNS22} are optimal).
This way, we provide a complete answer of how to repurpose (any finite number of) parameterized approximation algorithms (which includes polynomial-time approximations and exact parameterized algorithms as special cases) for the design of exponential-time (approximation) algorithms.

Still, our work raises a number of follow-up questions.
First, we focused on allowing a finite number of extension oracles.
However, for problems such as \textsc{Vertex Cover}, there is a  parameterized $\alpha$-approximation algorithm for every $\alpha \geq 1$.
In the language of this work, this gives rise to an infinite number of extension oracles;  it would be interesting to properly formalize such a setting and extend our result to it.
Even if we only want to provide a finite number of extension oracles, it is already unclear how to choose these oracles in an optimal way.
In fact, we already encountered this problem in \Cref{sec:application_vc_3hs} where we adopted a simple discretization appraoch (which is likely not optimal) to choose a finite number of extension oracles.

The second question asks what happens with other types of oracles.
Indeed, in this work, we focused on repurposing parameterized approximation algorithms (which includes polynomial-time approximations and exact parameterized algorithms as special cases) for the design of exponential-time approximation algorithms.
Can we find other types of algorithms that can be repurposed in a similar way?
For example, is it possible to repurpose exact exponential-time algorithms (e.g., \textsc{Vertex Cover} can be solved in time $\OO^*(1.1996^n)$ \cite{XiaoN17}) in a meaningful way?
More generally, in \cite{GaspersL17} the authors show that monotone local search can also be used to convert a $c^k\cdot b^n\cdot n^{O(1)}$ time algorithm into a $\left( 1+b -\frac{1}{c}\right)^n \cdot n^{O(1)}$ time algorithm for the same problem.
What is the optimal way of repurposing such an algorithm in order to obtain an exponential $\beta$-approximation algorithm for any $\beta > 1$?
We remark that our lower-bound technique can be used to show the result of \cite{GaspersL17} already repurposes such algorithms in an optimal way in the exact setting.

Finally, we ask about weighted problems.
Similarly to our setting, one can define an extension oracle model for \emph{weighted} problems in which the objective is to find set $S \in \F$ of (approximately) minimum weight.
Using this model its possible to define a weighted variant of $\bestbound$.
Is this variant equal to the function $\bestbound$ defined in this paper?
Which algorithm attains this cost for weighted problems?

\bibliographystyle{plainurl}
\bibliography{refs}

\appendix

\section{Problem Definitions}
\label{sec:problem_definitions}
In this section, we give the problem definitions of all the problems discussed in the paper.

\medskip

\defproblem{{\sc Vertex Cover ({\sc VC})}}{An undirected graph $G$.}{Find a minimum set $S$ of vertices of $G$ such that $G-S$ has no edges.}

\defproblem{{\sc Partial Vertex Cover ({\sc VC})}}{An undirected graph $G$ and an integer $t \geq 0$.}{Find a minimum set $S$ of vertices of $G$ such that $G-S$ has at most $|E(G)| - t$ many edges.}

\defproblem{{\sc $d$-Hitting Set ({\sc $d$-HS})}}{A universe $U$ and set family $\mathcal{F} \subseteq \binom{U}{\leq d}$.}{Find a minimum set $S \subseteq U$ such that for each $F \in \mathcal{F}$, $S \cap F \neq \emptyset$.}

\defproblem{{\sc Feedback Vertex Set ({\sc FVS})}}{An undirected graph $G$.}{Find a minimum set $S$ of vertices of $G$ such that $G-S$ is an acyclic graph.}

\defproblem{{\sc Subset Feedback Vertex Set ({\sc Subset FVS})}}{An undirected graph $G$ and a set $T \subseteq V(G)$.}{Find a minimum set $S$ of vertices of $G$ such that $G-S$ has no cycle that contains at least one vertex of $T$.}

\defproblem{{\sc Tournament Feedback Vertex Set ({\sc TFVS})}}{A tournament graph $G$.}{Find a minimum set $S$ of vertices of $G$ such that $G-S$ is an acyclic tournament.}

\defproblem{{\sc Directed Feedback Vertex Set ({\sc DFVS})}}{A directed graph $G$.}{Find a minimum set $S$ of vertices of $G$ such that $G-S$ is a directed acyclic graph.}

\defproblem{{\sc Directed Subset Feedback Vertex Set ({\sc Subset DFVS})}}{A directed graph $G$ and a set $T \subseteq V(G)$.}{Find a minimum set $S$ of vertices of $G$ such that $G-S$ has no directed cycle that contains at least one vertex of $T$.}

\defproblem{{\sc Odd Cycle Transversal ({\sc OCT})}}{An undirected graph $G$.}{Find a minimum set $S$ of vertices of $G$ such that $G-S$ has no cycle of odd length.}

\defproblem{{\sc Directed Odd Cycle Transversal ({\sc DOCT})}}{A directed graph $G$.}{Find a minimum set $S$ of vertices of $G$ such that $G-S$ has no directed cycle of odd length.}

\defproblem{{\sc Multicut}}{An undirected graph $G$ and a set $\mathcal{P} \subseteq V(G) \times V(G)$.}{Find a minimum set $S$ of vertices of $G$ such that $G-S$ has no path from $u$ to $v$ for any $(u,v) \in \mathcal{P}$}

\defproblem{{\sc Edge Multicut on Trees}}{A tree $T$ and a set $\mathcal{P} \subseteq V(G) \times V(G)$.}{Find a minimum set $S$ of edges of $T$ such that $T-S$ has no path from $u$ to $v$ for any $(u,v) \in \mathcal{P}$}

\defproblem{{\sc $d$-Steiner Multicut}}{An undirected graph $G$, a family of at most $d$-sized sets $\mathcal{P} \subseteq \binom{V(G)}{d}$.}{Find a minimum set $S$ of vertices of $G$ such that for each $T \in \mathcal{P}$, there exists $u,v \in T$ such that $G-S$ either has no $u$ to $v$ path.}

\defproblem{{\sc Directed Symmetric Multicut}}{A directed graph $G$ and a set $\mathcal{P} \subseteq V(G) \times V(G)$.}{Find a minimum set $S$ of vertices of $G$ such that for each $(u,v) \in \mathcal{P}$, $G-S$ either has no $u$ to $v$ path, or no $v$ to $u$ path.}

\defproblem{{\sc Interval Vertex Deletion}}{An undirected graph $G$.}{Find a minimum set $S$ of vertices of $G$ such that $G-S$ is an interval graph.}

\defproblem{{\sc Proper Interval Vertex Deletion}}{An undirected graph $G$.}{Find a minimum set $S$ of vertices of $G$ such that $G-S$ is a proper interval graph.}

\medskip

For the next problems, we require some additional definitions.
A graph $G$ is \emph{cluster graph} if every connected component of $G$ is a complete graph.
We say $G$ is a \emph{block graph} if every $2$-connected component of $G$ is a complete graph.
A \emph{cograph} is a graph $G$ which does not contain $P_4$ (a path on $4$ vertices) is an induced subgraph.
Finally, a graph $G$ is a \emph{split graph} if the vertex set can be partitioned into two sets $V(G) = I \uplus C$ such that $I$ is an independent set and $C$ is a clique in $G$.

\medskip

\defproblem{{\sc Block Graph Vertex Deletion}}{An undirected graph $G$.}{Find a minimum set $S$ of vertices of $G$ such that $G-S$ is a block graph.}

\defproblem{{\sc Cluster Graph Vertex Deletion}}{An undirected graph $G$.}{Find a minimum set $S$ of vertices of $G$ such that $G-S$ is a cluster graph.}

\defproblem{{\sc Cograph Vertex Deletion}}{An undirected graph $G$.}{Find a minimum set $S$ of vertices of $G$ such that $G-S$ is a cograph.}

\defproblem{{\sc Split Vertex Deletion}}{An undirected graph $G$.}{Find a minimum set $S$ of vertices of $G$ such that $G-S$ is a split graph.}

\section{Running Times of Exponential Approximation Algorithms}
\label{sec:running_times}
We provide extensive data sets on the running times for the obtained exponential approximation algorothms for the problems listed in Section \ref{sec:application_main}.
More precisely, we provide data sets fr the problems \textsc{FVS}, \textsc{Tournament FVS}, \textsc{Subset FVS}, \textsc{$4$-Hitting Set}, \textsc{Odd Cycle Transversal}, \textsc{Interval Vertex Deletion}, \textsc{Proper Interval Vertex Deletion}, \textsc{Block Graph Vertex Deletion}, \textsc{Cluster Graph Vertex Deletion}, \textsc{Cograph Vertex Deletion}, \textsc{Split Vertex Deletion}, \textsc{Edge Multicut on Trees} and \textsc{Partial Vertex Cover}.
\Cref{table:runtimes-appendix-1,table:runtimes-appendix-2} contain the running times for selected approximation ratios,
and graphical visualizations can be found in \Cref{fig:runtimes-appendix-1,fig:runtimes-appendix-2,fig:runtimes-appendix-3}

\begin{table}[H]
 \small
 \centering
 {\sc Feedback Vertex Set}
 \medskip

 \begin{tabular}{c|c|c|c|c|c|c|c|c|c|}
	$(\alpha,c)$ & $1.1$ & $1.2$ & $1.3$ & $1.4$ & $1.5$ & $1.6$ & $1.7$ & $1.8$ & $1.9$\\
	\hline
	$(\beta,2.69998)$ & $1.465$ & $1.3861$ & $1.3331$ & $1.294$ & $1.2637$ & $1.2393$ & $1.2193$ & $1.2024$ & $1.188$\\
	\hline
	$(1.0,2.69998)$ & $1.4156$ & $1.3289$ & $1.2753$ & $1.2378$ & $1.2099$ & $1.1881$ & $1.1706$ & $1.1561$ & $1.144$\\
	\hline
	$(2.0,1.0)$ & $1.6588$ & $1.4847$ & $1.3657$ & $1.2768$ & $1.2072$ & $1.1507$ & $1.1037$ & $1.064$ & $1.0298$\\
	\hline
	combined & $1.4156$ & $1.3289$ & $1.2753$ & $1.2378$ & $1.2068$ & $1.1507$ & $1.1037$ & $1.064$ & $1.0298$\\
	\hline
\end{tabular}

 \medskip
 \medskip
 {\sc Tournament Feedback Vertex Set}
 \medskip

 \begin{tabular}{c|c|c|c|c|c|c|c|c|c|}
	$(\alpha,c)$ & $1.1$ & $1.2$ & $1.3$ & $1.4$ & $1.5$ & $1.6$ & $1.7$ & $1.8$ & $1.9$\\
	\hline
	$(\beta,1.618)$ & $1.2912$ & $1.2463$ & $1.2152$ & $1.1918$ & $1.1734$ & $1.1583$ & $1.1458$ & $1.1352$ & $1.126$\\
	\hline
	$(1.0,1.618)$ & $1.2348$ & $1.1837$ & $1.1531$ & $1.132$ & $1.1164$ & $1.1042$ & $1.0945$ & $1.0865$ & $1.0798$\\
	\hline
	$(2.0,1.0)$ & $1.6588$ & $1.4847$ & $1.3657$ & $1.2768$ & $1.2072$ & $1.1507$ & $1.1037$ & $1.064$ & $1.0298$\\
	\hline
	combined & $1.2348$ & $1.1837$ & $1.1531$ & $1.132$ & $1.1164$ & $1.1042$ & $1.0945$ & $1.064$ & $1.0298$\\
	\hline
\end{tabular}

 \medskip
 \medskip
 {\sc Subset Feedback Vertex Set}
 \medskip

 \begin{tabular}{c|c|c|c|c|c|c|c|c|c|c|}
	$(\alpha,c)$ & $1.1$ & $1.8$ & $2.5$ & $3.2$ & $3.9$ & $4.6$ & $5.3$ & $6.0$ & $6.7$ & $7.4$\\
	\hline
	$(\beta,4.0)$ & $1.5474$ & $1.2323$ & $1.1507$ & $1.1118$ & $1.0889$ & $1.0738$ & $1.0631$ & $1.0552$ & $1.049$ & $1.044$\\
	\hline
	$(1.0,4.0)$ & $1.5098$ & $1.1952$ & $1.1235$ & $1.0906$ & $1.0716$ & $1.0592$ & $1.0505$ & $1.044$ & $1.039$ & $1.035$\\
	\hline
	$(8.0,1.0)$ & $1.7153$ & $1.2891$ & $1.1787$ & $1.1225$ & $1.0871$ & $1.0623$ & $1.0436$ & $1.029$ & $1.0171$ & $1.0073$\\
	\hline
	combined & $1.5098$ & $1.1952$ & $1.1235$ & $1.0906$ & $1.0716$ & $1.0592$ & $1.0436$ & $1.029$ & $1.0171$ & $1.0073$\\
	\hline
\end{tabular}

 \medskip
 \medskip
 {\sc $4$-Hitting Set}
 \medskip

 \begin{tabular}{c|c|c|c|c|c|c|c|c|c|c|}
	$(\alpha,c)$ & $1.1$ & $1.4$ & $1.7$ & $2.0$ & $2.3$ & $2.6$ & $2.9$ & $3.2$ & $3.5$ & $3.8$\\
	\hline
	$(\beta,3.0755)$ & $1.4961$ & $1.3117$ & $1.2318$ & $1.1853$ & $1.1546$ & $1.1327$ & $1.1163$ & $1.1035$ & $1.0933$ & $1.0849$\\
	\hline
	$(1.0,3.0755)$ & $1.4506$ & $1.2592$ & $1.1861$ & $1.1459$ & $1.1202$ & $1.1023$ & $1.0891$ & $1.0789$ & $1.0708$ & $1.0643$\\
	\hline
	$(4.0,1.0)$ & $1.7147$ & $1.4208$ & $1.2882$ & $1.2072$ & $1.151$ & $1.1093$ & $1.0768$ & $1.0506$ & $1.029$ & $1.0107$\\
	\hline
	combined & $1.4506$ & $1.2592$ & $1.1861$ & $1.1459$ & $1.1202$ & $1.1023$ & $1.0768$ & $1.0506$ & $1.029$ & $1.0107$\\
	\hline
\end{tabular}

 \medskip
 \medskip
 {\sc Odd Cycle Transversal}
 \medskip

 \begin{tabular}{c|c|c|c|c|c|c|c|c|c|}
	$(\alpha,c)$ & $1.1$ & $1.2$ & $1.3$ & $1.4$ & $1.5$ & $1.6$ & $1.7$ & $1.8$ & $1.9$\\
	\hline
	$(\beta,2.3146)$ & $1.4223$ & $1.3521$ & $1.3046$ & $1.2695$ & $1.2421$ & $1.22$ & $1.2018$ & $1.1864$ & $1.1733$\\
	\hline
	$(1.0,2.3146)$ & $1.3689$ & $1.2908$ & $1.243$ & $1.2098$ & $1.185$ & $1.1658$ & $1.1503$ & $1.1376$ & $1.1269$\\
	\hline
\end{tabular}

 \medskip
 \medskip
 {\sc Interval Vertex Deletion}
 \medskip

 \begin{tabular}{c|c|c|c|c|c|c|c|c|c|c|}
	$(\alpha,c)$ & $1.1$ & $1.8$ & $2.5$ & $3.2$ & $3.9$ & $4.6$ & $5.3$ & $6.0$ & $6.7$ & $7.4$\\
	\hline
	$(\beta,8.0)$ & $1.6319$ & $1.262$ & $1.1688$ & $1.1248$ & $1.0991$ & $1.0822$ & $1.0702$ & $1.0613$ & $1.0544$ & $1.0489$\\
	\hline
	$(1.0,8.0)$ & $1.6111$ & $1.2401$ & $1.1526$ & $1.1121$ & $1.0886$ & $1.0733$ & $1.0625$ & $1.0545$ & $1.0483$ & $1.0434$\\
	\hline
	$(8.0,1.0)$ & $1.7153$ & $1.2891$ & $1.1787$ & $1.1225$ & $1.0871$ & $1.0623$ & $1.0436$ & $1.029$ & $1.0171$ & $1.0073$\\
	\hline
	combined & $1.6111$ & $1.2401$ & $1.1526$ & $1.1121$ & $1.0871$ & $1.0623$ & $1.0436$ & $1.029$ & $1.0171$ & $1.0073$\\
	\hline
\end{tabular}

 \caption{An entry at in row $(\alpha,c)$ and column $\beta$ is $\bestbound(\alpha,c,\beta)$.}
 \label{table:runtimes-appendix-1}
\end{table}

\begin{table}[H]
 \small
 \centering
 {\sc Proper Interval Vertex Deletion}
 \medskip

 \begin{tabular}{c|c|c|c|c|c|c|c|c|c|c|}
	$(\alpha,c)$ & $1.1$ & $1.6$ & $2.1$ & $2.6$ & $3.1$ & $3.6$ & $4.1$ & $4.6$ & $5.1$ & $5.6$\\
	\hline
	$(\beta,6.0)$ & $1.6038$ & $1.3001$ & $1.204$ & $1.1551$ & $1.1252$ & $1.105$ & $1.0905$ & $1.0795$ & $1.0708$ & $1.0639$\\
	\hline
	$(1.0,6.0)$ & $1.577$ & $1.2698$ & $1.1799$ & $1.1354$ & $1.1087$ & $1.0908$ & $1.078$ & $1.0684$ & $1.0609$ & $1.0549$\\
	\hline
	$(6.0,1.0)$ & $1.7153$ & $1.3436$ & $1.2215$ & $1.1546$ & $1.111$ & $1.0797$ & $1.0559$ & $1.0371$ & $1.0217$ & $1.0089$\\
	\hline
	combined & $1.577$ & $1.2698$ & $1.1799$ & $1.1354$ & $1.1087$ & $1.0797$ & $1.0559$ & $1.0371$ & $1.0217$ & $1.0089$\\
	\hline
\end{tabular}

 \medskip
 \medskip
 {\sc Block Graph Vertex Deletion}
 \medskip

 \begin{tabular}{c|c|c|c|c|c|c|c|c|c|c|}
	$(\alpha,c)$ & $1.1$ & $1.4$ & $1.7$ & $2.0$ & $2.3$ & $2.6$ & $2.9$ & $3.2$ & $3.5$ & $3.8$\\
	\hline
	$(\beta,4.0)$ & $1.5474$ & $1.3406$ & $1.2521$ & $1.201$ & $1.1674$ & $1.1435$ & $1.1257$ & $1.1118$ & $1.1007$ & $1.0916$\\
	\hline
	$(1.0,4.0)$ & $1.5098$ & $1.2961$ & $1.2131$ & $1.1672$ & $1.1379$ & $1.1174$ & $1.1022$ & $1.0906$ & $1.0813$ & $1.0738$\\
	\hline
	$(4.0,1.0)$ & $1.7147$ & $1.4208$ & $1.2882$ & $1.2072$ & $1.151$ & $1.1093$ & $1.0768$ & $1.0506$ & $1.029$ & $1.0107$\\
	\hline
	combined & $1.5098$ & $1.2961$ & $1.2131$ & $1.1672$ & $1.1379$ & $1.1093$ & $1.0768$ & $1.0506$ & $1.029$ & $1.0107$\\
	\hline
\end{tabular}

 \medskip
 \medskip
 {\sc Cluster Graph Vertex Deletion}
 \medskip

 \begin{tabular}{c|c|c|c|c|c|c|c|c|c|}
	$(\alpha,c)$ & $1.1$ & $1.2$ & $1.3$ & $1.4$ & $1.5$ & $1.6$ & $1.7$ & $1.8$ & $1.9$\\
	\hline
	$(\beta,1.9102)$ & $1.3584$ & $1.3007$ & $1.2614$ & $1.232$ & $1.209$ & $1.1904$ & $1.1749$ & $1.1619$ & $1.1507$\\
	\hline
	$(1.0,1.9102)$ & $1.3015$ & $1.2367$ & $1.1974$ & $1.1703$ & $1.1501$ & $1.1345$ & $1.1219$ & $1.1116$ & $1.1029$\\
	\hline
	$(2.0,1.0)$ & $1.6588$ & $1.4847$ & $1.3657$ & $1.2768$ & $1.2072$ & $1.1507$ & $1.1037$ & $1.064$ & $1.0298$\\
	\hline
	combined & $1.3015$ & $1.2367$ & $1.1974$ & $1.1703$ & $1.1501$ & $1.1345$ & $1.1037$ & $1.064$ & $1.0298$\\
	\hline
\end{tabular}

 \medskip
 \medskip
 {\sc Split Vertex Deletion}
 \medskip

 \begin{tabular}{c|c|c|c|c|c|c|c|c|c|}
	$(\alpha,c)$ & $1.1$ & $1.2$ & $1.3$ & $1.4$ & $1.5$ & $1.6$ & $1.7$ & $1.8$ & $1.9$\\
	\hline
	$(\beta,2.0)$ & $1.3749$ & $1.314$ & $1.2726$ & $1.2418$ & $1.2176$ & $1.1981$ & $1.1819$ & $1.1683$ & $1.1566$\\
	\hline
	$(1.0,2.0)$ & $1.3186$ & $1.2503$ & $1.2089$ & $1.1802$ & $1.1589$ & $1.1423$ & $1.129$ & $1.1181$ & $1.1089$\\
	\hline
	$(2.0001,1.0)$ & $1.6588$ & $1.4847$ & $1.3657$ & $1.2768$ & $1.2072$ & $1.1507$ & $1.1038$ & $1.064$ & $1.0298$\\
	\hline
	combined & $1.3186$ & $1.2503$ & $1.2089$ & $1.1802$ & $1.1589$ & $1.1423$ & $1.1038$ & $1.064$ & $1.0298$\\
	\hline
\end{tabular}

 \medskip
 \medskip
 {\sc Multicut on Trees}
 \medskip

 \begin{tabular}{c|c|c|c|c|c|c|c|c|c|}
	$(\alpha,c)$ & $1.1$ & $1.2$ & $1.3$ & $1.4$ & $1.5$ & $1.6$ & $1.7$ & $1.8$ & $1.9$\\
	\hline
	$(\beta,1.5538)$ & $1.2729$ & $1.2313$ & $1.2025$ & $1.1807$ & $1.1634$ & $1.1494$ & $1.1376$ & $1.1277$ & $1.1191$\\
	\hline
	$(1.0,1.5538)$ & $1.2173$ & $1.1699$ & $1.1416$ & $1.1221$ & $1.1076$ & $1.0964$ & $1.0874$ & $1.08$ & $1.0738$\\
	\hline
	$(2.0,1.0)$ & $1.6588$ & $1.4847$ & $1.3657$ & $1.2768$ & $1.2072$ & $1.1507$ & $1.1037$ & $1.064$ & $1.0298$\\
	\hline
	combined & $1.2173$ & $1.1699$ & $1.1416$ & $1.1221$ & $1.1076$ & $1.0964$ & $1.0874$ & $1.064$ & $1.0298$\\
	\hline
\end{tabular}

 \medskip
 \medskip
 {\sc Partial Vertex Cover}
 \medskip

 \begin{tabular}{c|c|c|c|c|c|c|c|c|c|}
	$(\alpha,c)$ & $1.1$ & $1.2$ & $1.3$ & $1.4$ & $1.5$ & $1.6$ & $1.7$ & $1.8$ & $1.9$\\
	\hline
	$(2.0,1.0)$ & $1.6588$ & $1.4847$ & $1.3657$ & $1.2768$ & $1.2072$ & $1.1507$ & $1.1037$ & $1.064$ & $1.0298$\\
	\hline
\end{tabular}

 \caption{An entry at in row $(\alpha,c)$ and column $\beta$ is $\bestbound(\alpha,c,\beta)$.}
 \label{table:runtimes-appendix-2}
\end{table}

\begin{figure}[H]
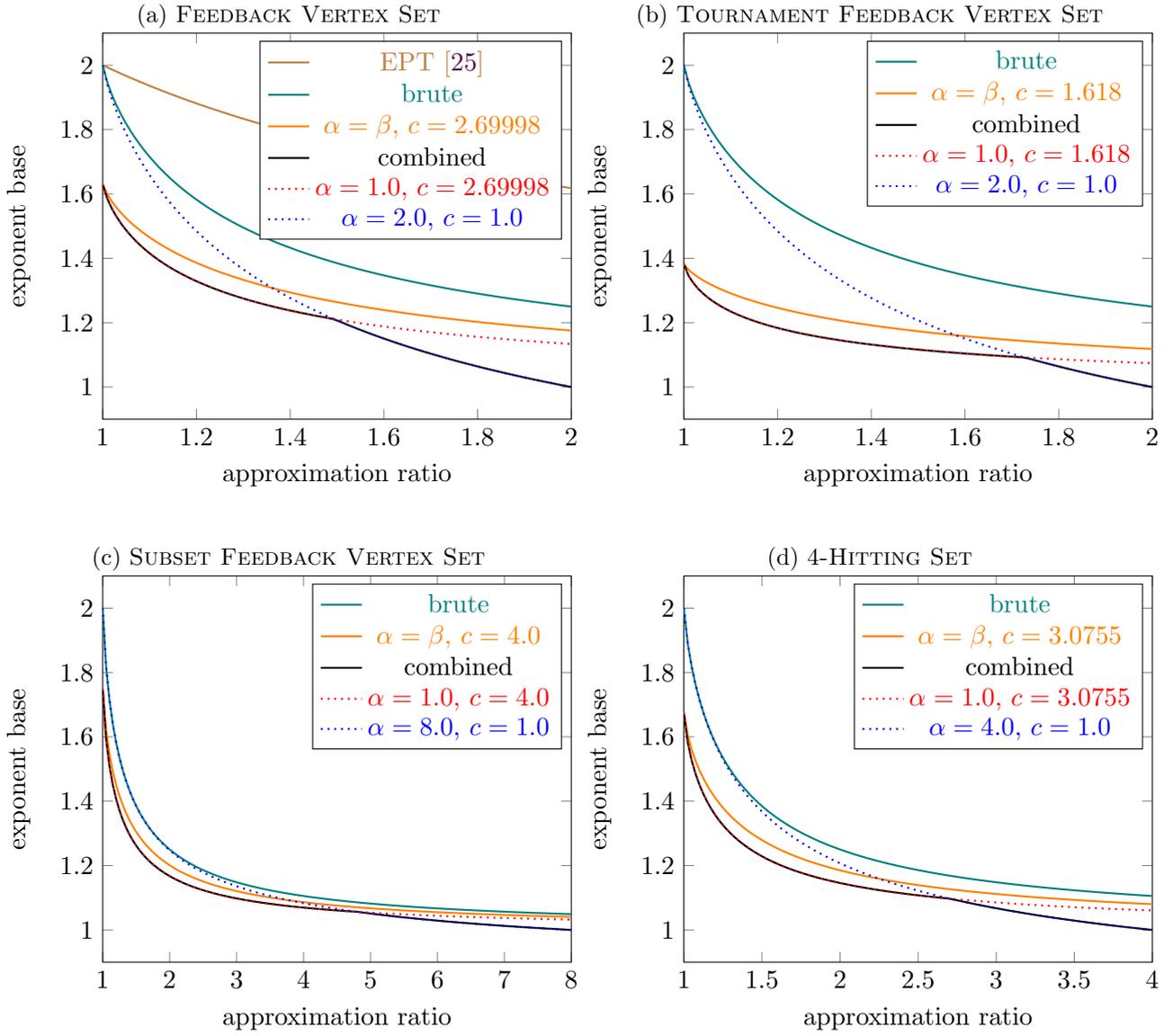

 \centering
 \begin{subfigure}{.5\textwidth}
  \centering
  \caption{{\sc Feedback Vertex Set}}
  \input{plots/figure_fvs.tex}
  \label{fig:fvs-results}
 \end{subfigure}%
 \begin{subfigure}{.5\textwidth}
  \centering
  \caption{{\sc Tournament Feedback Vertex Set}}
  \input{plots/figure_tfvs.tex}
  \label{fig:tfvs-results}
 \end{subfigure}%

 \begin{subfigure}{.5\textwidth}
  \centering
  \caption{{\sc Subset Feedback Vertex Set}}
  \input{plots/figure_subset_fvs.tex}
  \label{fig:subset-fvs-results}
 \end{subfigure}%
 \begin{subfigure}{.5\textwidth}
  \centering
  \caption{{\sc $4$-Hitting Set}}
  \input{plots/figure_4hs.tex}
  \label{fig:4hs-results}
 \end{subfigure}%

 \caption{Results for {\sc Feedback Vertex Set}, {\sc Tournament Feedback Vertex Set}, {\sc Subset Feedback Vertex Set} and {\sc $4$-Hitting Set}.
   A dot at $(\beta,d)$ means that the respective algorithm outputs an $\beta$-approximation in time $\OO^*(d^n)$.}
 \label{fig:runtimes-appendix-1}
\end{figure}

\begin{figure}[H]
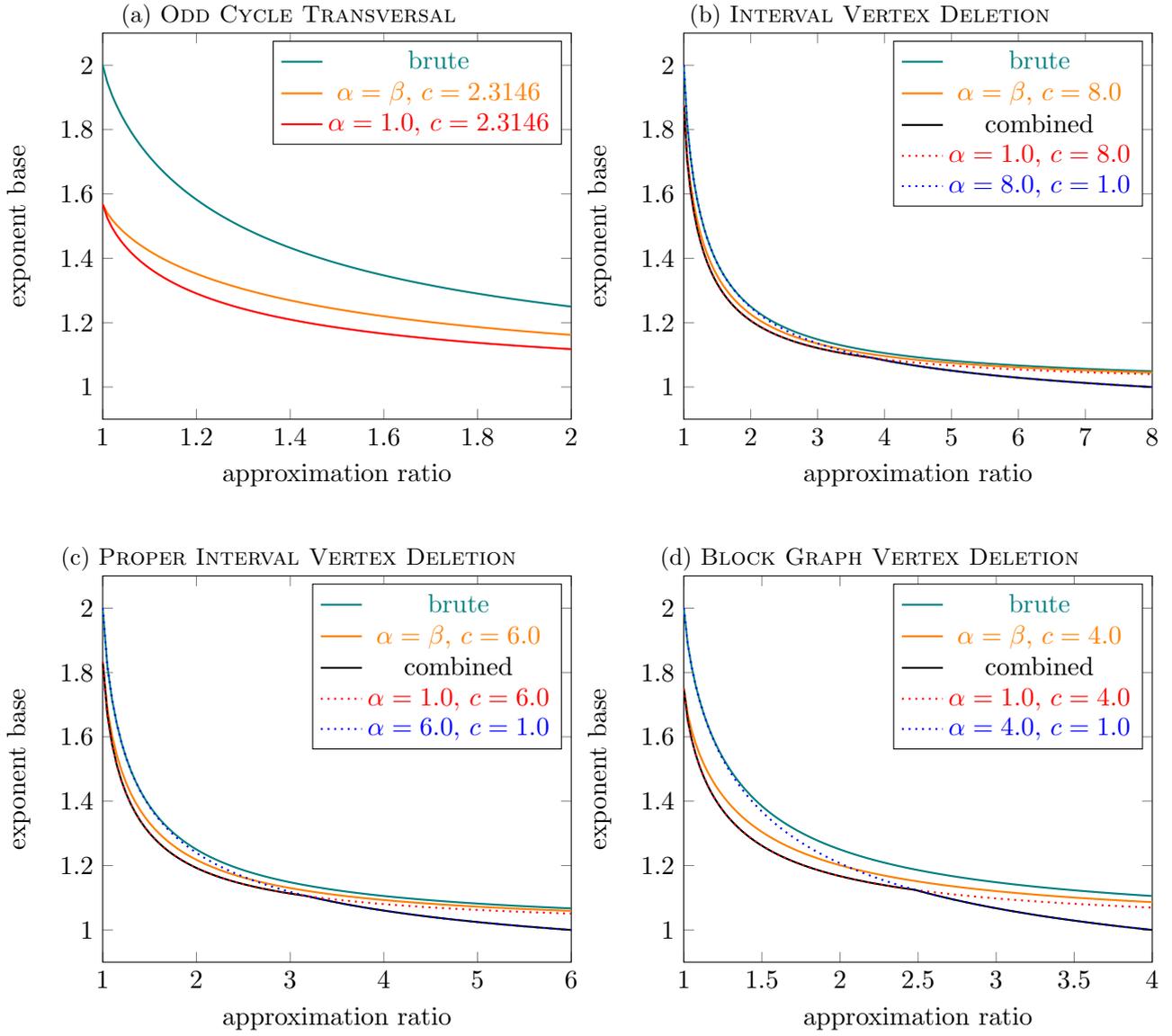

 \centering
 \begin{subfigure}{.5\textwidth}
  \centering
  \caption{{\sc Odd Cycle Transversal}}
  \begin{tikzpicture}[scale = 1.0]
	\begin{axis}[xmin = 1, xmax = 2, ymin = 0.9, ymax = 2.1, xlabel = {approximation ratio}, ylabel = {exponent base}]

	\addplot[teal, thick] coordinates {
		(1.0, 2.0)
		(1.01, 1.9454431345)
		(1.02, 1.9062508454)
		(1.03, 1.8731549013)
		(1.04, 1.8440491304)
		(1.05, 1.8178991111)
		(1.06, 1.794081097)
		(1.07, 1.7721758604)
		(1.08, 1.7518817491)
		(1.09, 1.7329712548)
		(1.1, 1.7152667656)
		(1.11, 1.698625911)
		(1.12, 1.6829321593)
		(1.13, 1.6680884977)
		(1.14, 1.6540130203)
		(1.15, 1.6406357531)
		(1.16, 1.6278963084)
		(1.17, 1.615742114)
		(1.18, 1.6041270506)
		(1.19, 1.5930103866)
		(1.2, 1.5823559323)
		(1.21, 1.5721313608)
		(1.22, 1.5623076557)
		(1.23, 1.552858658)
		(1.24, 1.5437606905)
		(1.25, 1.534992244)
		(1.26, 1.5265337131)
		(1.27, 1.5183671728)
		(1.28, 1.510476187)
		(1.29, 1.5028456449)
		(1.3, 1.4954616201)
		(1.31, 1.4883112476)
		(1.32, 1.4813826173)
		(1.33, 1.4746646809)
		(1.34, 1.4681471696)
		(1.35, 1.4618205222)
		(1.36, 1.4556758212)
		(1.37, 1.4497047356)
		(1.38, 1.443899471)
		(1.39, 1.4382527242)
		(1.4, 1.4327576428)
		(1.41, 1.427407789)
		(1.42, 1.4221971069)
		(1.43, 1.4171198929)
		(1.44, 1.4121707695)
		(1.45, 1.4073446605)
		(1.46, 1.4026367697)
		(1.47, 1.3980425604)
		(1.48, 1.3935577374)
		(1.49, 1.3891782307)
		(1.5, 1.3849001795)
		(1.51, 1.380719919)
		(1.52, 1.3766339674)
		(1.53, 1.3726390137)
		(1.54, 1.3687319076)
		(1.55, 1.3649096489)
		(1.56, 1.3611693784)
		(1.57, 1.3575083696)
		(1.58, 1.3539240206)
		(1.59, 1.3504138468)
		(1.6, 1.3469754742)
		(1.61, 1.343606633)
		(1.62, 1.3403051519)
		(1.63, 1.3370689522)
		(1.64, 1.3338960432)
		(1.65, 1.3307845174)
		(1.66, 1.3277325456)
		(1.67, 1.3247383734)
		(1.68, 1.3218003168)
		(1.69, 1.3189167589)
		(1.7, 1.3160861463)
		(1.71, 1.3133069861)
		(1.72, 1.3105778425)
		(1.73, 1.3078973347)
		(1.74, 1.3052641335)
		(1.75, 1.3026769593)
		(1.76, 1.3001345795)
		(1.77, 1.2976358065)
		(1.78, 1.2951794954)
		(1.79, 1.2927645423)
		(1.8, 1.2903898821)
		(1.81, 1.2880544871)
		(1.82, 1.2857573652)
		(1.83, 1.2834975581)
		(1.84, 1.2812741403)
		(1.85, 1.2790862173)
		(1.86, 1.2769329244)
		(1.87, 1.2748134255)
		(1.88, 1.2727269118)
		(1.89, 1.2706726008)
		(1.9, 1.2686497351)
		(1.91, 1.2666575813)
		(1.92, 1.2646954293)
		(1.93, 1.2627625911)
		(1.94, 1.2608584001)
		(1.95, 1.2589822102)
		(1.96, 1.2571333949)
		(1.97, 1.2553113469)
		(1.98, 1.2535154767)
		(1.99, 1.2517452127)
		(2.0, 1.25)
	};

	\addplot[orange, thick] coordinates {
		(1.0, 1.5679599067)
		(1.01, 1.5398203418)
		(1.02, 1.5199219491)
		(1.03, 1.5031469947)
		(1.04, 1.4883747622)
		(1.05, 1.4750673878)
		(1.06, 1.462906363)
		(1.07, 1.4516807077)
		(1.08, 1.4412403946)
		(1.09, 1.4314733531)
		(1.1, 1.4222927711)
		(1.11, 1.4136294993)
		(1.12, 1.4054272242)
		(1.13, 1.3976392524)
		(1.14, 1.3902262849)
		(1.15, 1.3831548271)
		(1.16, 1.3763960232)
		(1.17, 1.369924781)
		(1.18, 1.3637191039)
		(1.19, 1.3577595703)
		(1.2, 1.3520289231)
		(1.21, 1.3465117411)
		(1.22, 1.3411941738)
		(1.23, 1.3360637234)
		(1.24, 1.3311090656)
		(1.25, 1.3263199002)
		(1.26, 1.3216868255)
		(1.27, 1.3172012325)
		(1.28, 1.312855214)
		(1.29, 1.308641488)
		(1.3, 1.3045533309)
		(1.31, 1.3005845196)
		(1.32, 1.2967292817)
		(1.33, 1.2929822517)
		(1.34, 1.2893384324)
		(1.35, 1.2857931611)
		(1.36, 1.2823420798)
		(1.37, 1.2789811086)
		(1.38, 1.2757064217)
		(1.39, 1.272514427)
		(1.4, 1.2694017464)
		(1.41, 1.2663651995)
		(1.42, 1.2634017877)
		(1.43, 1.2605086808)
		(1.44, 1.2576832046)
		(1.45, 1.254922829)
		(1.46, 1.2522251584)
		(1.47, 1.2495879218)
		(1.48, 1.2470089644)
		(1.49, 1.2444862399)
		(1.5, 1.2420178033)
		(1.51, 1.2396018041)
		(1.52, 1.2372364805)
		(1.53, 1.2349201538)
		(1.54, 1.2326512232)
		(1.55, 1.2304281612)
		(1.56, 1.2282495091)
		(1.57, 1.2261138728)
		(1.58, 1.2240199195)
		(1.59, 1.2219663739)
		(1.6, 1.2199520149)
		(1.61, 1.2179756728)
		(1.62, 1.2160362265)
		(1.63, 1.2141326006)
		(1.64, 1.2122637633)
		(1.65, 1.2104287241)
		(1.66, 1.2086265312)
		(1.67, 1.2068562701)
		(1.68, 1.2051170616)
		(1.69, 1.2034080597)
		(1.7, 1.2017284504)
		(1.71, 1.2000774498)
		(1.72, 1.1984543029)
		(1.73, 1.1968582821)
		(1.74, 1.1952886859)
		(1.75, 1.193744838)
		(1.76, 1.1922260855)
		(1.77, 1.1907317987)
		(1.78, 1.1892613691)
		(1.79, 1.1878142093)
		(1.8, 1.1863897516)
		(1.81, 1.1849874474)
		(1.82, 1.1836067661)
		(1.83, 1.1822471946)
		(1.84, 1.1809082363)
		(1.85, 1.1795894109)
		(1.86, 1.1782902531)
		(1.87, 1.1770103124)
		(1.88, 1.1757491524)
		(1.89, 1.1745063502)
		(1.9, 1.1732814958)
		(1.91, 1.172074192)
		(1.92, 1.1708840532)
		(1.93, 1.1697107056)
		(1.94, 1.1685537864)
		(1.95, 1.1674129434)
		(1.96, 1.1662878349)
		(1.97, 1.1651781288)
		(1.98, 1.1640835028)
		(1.99, 1.1630036435)
		(2.0, 1.1619382467)
	};

	\addplot[red, thick] coordinates {
		(1.0, 1.5679599067)
		(1.01, 1.524293669)
		(1.02, 1.4952907253)
		(1.03, 1.4717896969)
		(1.04, 1.4517536905)
		(1.05, 1.4342018244)
		(1.06, 1.4185542344)
		(1.07, 1.4044287823)
		(1.08, 1.3915557712)
		(1.09, 1.379735594)
		(1.1, 1.3688152434)
		(1.11, 1.3586742179)
		(1.12, 1.349215548)
		(1.13, 1.3403598149)
		(1.14, 1.3320410104)
		(1.15, 1.3242035845)
		(1.16, 1.316800284)
		(1.17, 1.3097905364)
		(1.18, 1.3031392184)
		(1.19, 1.2968157012)
		(1.2, 1.2907931018)
		(1.21, 1.2850476855)
		(1.22, 1.2795583863)
		(1.23, 1.2743064162)
		(1.24, 1.2692749451)
		(1.25, 1.2644488351)
		(1.26, 1.2598144199)
		(1.27, 1.2553593189)
		(1.28, 1.2510722799)
		(1.29, 1.2469430464)
		(1.3, 1.2429622427)
		(1.31, 1.2391212763)
		(1.32, 1.2354122531)
		(1.33, 1.231827904)
		(1.34, 1.2283615206)
		(1.35, 1.2250068998)
		(1.36, 1.2217582941)
		(1.37, 1.2186103689)
		(1.38, 1.2155581637)
		(1.39, 1.2125970587)
		(1.4, 1.2097227444)
		(1.41, 1.2069311946)
		(1.42, 1.204218643)
		(1.43, 1.2015815607)
		(1.44, 1.199016638)
		(1.45, 1.1965207658)
		(1.46, 1.1940910209)
		(1.47, 1.191724651)
		(1.48, 1.1894190625)
		(1.49, 1.1871718082)
		(1.5, 1.184980577)
		(1.51, 1.1828431842)
		(1.52, 1.1807575624)
		(1.53, 1.1787217536)
		(1.54, 1.176733902)
		(1.55, 1.1747922467)
		(1.56, 1.1728951156)
		(1.57, 1.1710409202)
		(1.58, 1.1692281495)
		(1.59, 1.1674553656)
		(1.6, 1.1657211992)
		(1.61, 1.164024345)
		(1.62, 1.1623635585)
		(1.63, 1.1607376519)
		(1.64, 1.159145491)
		(1.65, 1.157585992)
		(1.66, 1.156058119)
		(1.67, 1.1545608806)
		(1.68, 1.1530933282)
		(1.69, 1.1516545532)
		(1.7, 1.1502436849)
		(1.71, 1.1488598889)
		(1.72, 1.1475023643)
		(1.73, 1.1461703431)
		(1.74, 1.1448630874)
		(1.75, 1.1435798888)
		(1.76, 1.1423200663)
		(1.77, 1.141082965)
		(1.78, 1.1398679554)
		(1.79, 1.1386744314)
		(1.8, 1.1375018097)
		(1.81, 1.1363495283)
		(1.82, 1.1352170461)
		(1.83, 1.1341038412)
		(1.84, 1.1330094107)
		(1.85, 1.1319332693)
		(1.86, 1.1308749489)
		(1.87, 1.1298339975)
		(1.88, 1.1288099788)
		(1.89, 1.1278024713)
		(1.9, 1.1268110677)
		(1.91, 1.1258353744)
		(1.92, 1.1248750109)
		(1.93, 1.123929609)
		(1.94, 1.1229988126)
		(1.95, 1.1220822771)
		(1.96, 1.1211796689)
		(1.97, 1.1202906649)
		(1.98, 1.1194149522)
		(1.99, 1.1185522277)
		(2.0, 1.1177021975)
	};

	\addlegendentry[no markers, teal]{brute}
	\addlegendentry[no markers, orange]{$\alpha = \beta$, $c = 2.3146$}
	\addlegendentry[no markers, red]{$\alpha = 1.0$, $c = 2.3146$}

	\end{axis}
\end{tikzpicture}
  \label{fig:oct-results}
 \end{subfigure}%
 \begin{subfigure}{.5\textwidth}
  \centering
  \caption{{\sc Interval Vertex Deletion}}
  \input{plots/figure_ivd.tex}
  \label{fig:ivd-results}
 \end{subfigure}%

 \begin{subfigure}{.5\textwidth}
  \centering
  \caption{{\sc Proper Interval Vertex Deletion}}
  \input{plots/figure_pivd.tex}
  \label{fig:pvid-results}
 \end{subfigure}%
 \begin{subfigure}{.5\textwidth}
  \centering
  \caption{{\sc Block Graph Vertex Deletion}}
  \input{plots/figure_bgvd.tex}
  \label{fig:bgvd-results}
 \end{subfigure}%
 \caption{Results for {\sc Odd Cycle Transversal}, {\sc Interval Vertex Deletion}, {\sc Proper Interval Vertex Deletion} and {\sc Block Graph Vertex Deletion}.
   A dot at $(\beta,d)$ means that the respective algorithm outputs an $\beta$-approximation in time $\OO^*(d^n)$.}
 \label{fig:runtimes-appendix-2}
\end{figure}

\begin{figure}[H]
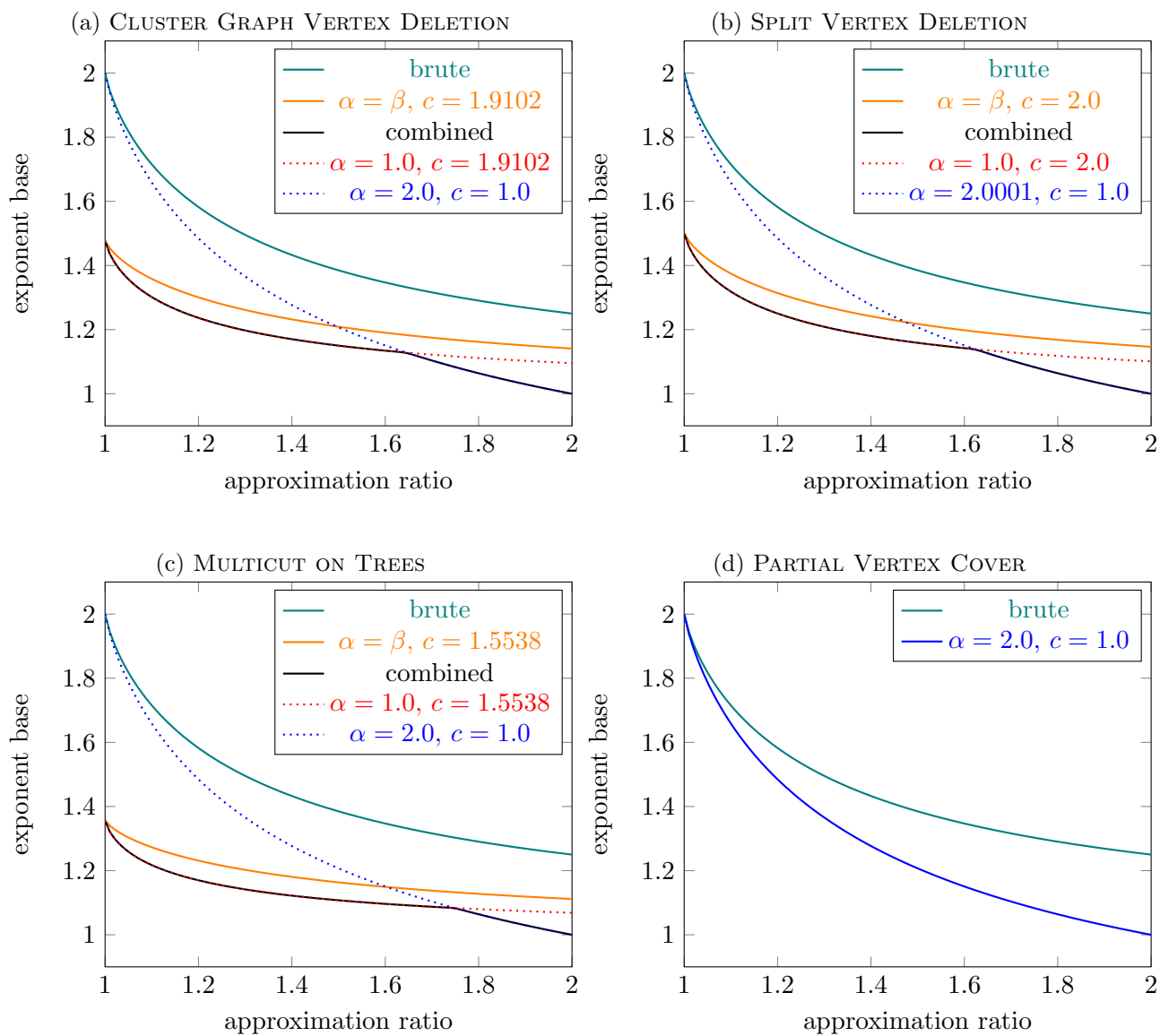

 \centering
 \begin{subfigure}{.5\textwidth}
  \centering
  \caption{{\sc Cluster Graph Vertex Deletion}}
  \input{plots/figure_cgvd.tex}
  \label{fig:cgvd-results}
 \end{subfigure}%
 \begin{subfigure}{.5\textwidth}
  \centering
  \caption{{\sc Split Vertex Deletion}}
  \input{plots/figure_svd.tex}
  \label{fig:svd-results}
 \end{subfigure}%

 \begin{subfigure}{.5\textwidth}
  \centering
  \caption{{\sc Multicut on Trees}}
  \input{plots/figure_mct.tex}
  \label{fig:mct-results}
 \end{subfigure}%
 \begin{subfigure}{.5\textwidth}
  \centering
  \caption{{\sc Partial Vertex Cover}}
  \begin{tikzpicture}[scale = 1.0]
	\begin{axis}[xmin = 1, xmax = 2, ymin = 0.9, ymax = 2.1, xlabel = {approximation ratio}, ylabel = {exponent base}]

	\addplot[teal, thick] coordinates {
		(1.0, 2.0)
		(1.01, 1.9454431345)
		(1.02, 1.9062508454)
		(1.03, 1.8731549013)
		(1.04, 1.8440491304)
		(1.05, 1.8178991111)
		(1.06, 1.794081097)
		(1.07, 1.7721758604)
		(1.08, 1.7518817491)
		(1.09, 1.7329712548)
		(1.1, 1.7152667656)
		(1.11, 1.698625911)
		(1.12, 1.6829321593)
		(1.13, 1.6680884977)
		(1.14, 1.6540130203)
		(1.15, 1.6406357531)
		(1.16, 1.6278963084)
		(1.17, 1.615742114)
		(1.18, 1.6041270506)
		(1.19, 1.5930103866)
		(1.2, 1.5823559323)
		(1.21, 1.5721313608)
		(1.22, 1.5623076557)
		(1.23, 1.552858658)
		(1.24, 1.5437606905)
		(1.25, 1.534992244)
		(1.26, 1.5265337131)
		(1.27, 1.5183671728)
		(1.28, 1.510476187)
		(1.29, 1.5028456449)
		(1.3, 1.4954616201)
		(1.31, 1.4883112476)
		(1.32, 1.4813826173)
		(1.33, 1.4746646809)
		(1.34, 1.4681471696)
		(1.35, 1.4618205222)
		(1.36, 1.4556758212)
		(1.37, 1.4497047356)
		(1.38, 1.443899471)
		(1.39, 1.4382527242)
		(1.4, 1.4327576428)
		(1.41, 1.427407789)
		(1.42, 1.4221971069)
		(1.43, 1.4171198929)
		(1.44, 1.4121707695)
		(1.45, 1.4073446605)
		(1.46, 1.4026367697)
		(1.47, 1.3980425604)
		(1.48, 1.3935577374)
		(1.49, 1.3891782307)
		(1.5, 1.3849001795)
		(1.51, 1.380719919)
		(1.52, 1.3766339674)
		(1.53, 1.3726390137)
		(1.54, 1.3687319076)
		(1.55, 1.3649096489)
		(1.56, 1.3611693784)
		(1.57, 1.3575083696)
		(1.58, 1.3539240206)
		(1.59, 1.3504138468)
		(1.6, 1.3469754742)
		(1.61, 1.343606633)
		(1.62, 1.3403051519)
		(1.63, 1.3370689522)
		(1.64, 1.3338960432)
		(1.65, 1.3307845174)
		(1.66, 1.3277325456)
		(1.67, 1.3247383734)
		(1.68, 1.3218003168)
		(1.69, 1.3189167589)
		(1.7, 1.3160861463)
		(1.71, 1.3133069861)
		(1.72, 1.3105778425)
		(1.73, 1.3078973347)
		(1.74, 1.3052641335)
		(1.75, 1.3026769593)
		(1.76, 1.3001345795)
		(1.77, 1.2976358065)
		(1.78, 1.2951794954)
		(1.79, 1.2927645423)
		(1.8, 1.2903898821)
		(1.81, 1.2880544871)
		(1.82, 1.2857573652)
		(1.83, 1.2834975581)
		(1.84, 1.2812741403)
		(1.85, 1.2790862173)
		(1.86, 1.2769329244)
		(1.87, 1.2748134255)
		(1.88, 1.2727269118)
		(1.89, 1.2706726008)
		(1.9, 1.2686497351)
		(1.91, 1.2666575813)
		(1.92, 1.2646954293)
		(1.93, 1.2627625911)
		(1.94, 1.2608584001)
		(1.95, 1.2589822102)
		(1.96, 1.2571333949)
		(1.97, 1.2553113469)
		(1.98, 1.2535154767)
		(1.99, 1.2517452127)
		(2.0, 1.25)
	};

	\addplot[blue, thick] coordinates {
		(1.0, 2.0)
		(1.01, 1.9387047513)
		(1.02, 1.8930929239)
		(1.03, 1.853849694)
		(1.04, 1.8188398378)
		(1.05, 1.7870069014)
		(1.06, 1.7577092525)
		(1.07, 1.7305126)
		(1.08, 1.705102307)
		(1.09, 1.6812394878)
		(1.1, 1.6587364406)
		(1.11, 1.6374417619)
		(1.12, 1.6172307723)
		(1.13, 1.5979990628)
		(1.14, 1.5796579792)
		(1.15, 1.5621313614)
		(1.16, 1.545353129)
		(1.17, 1.5292654531)
		(1.18, 1.5138173453)
		(1.19, 1.498963551)
		(1.2, 1.484663669)
		(1.21, 1.4708814415)
		(1.22, 1.457584176)
		(1.23, 1.4447422687)
		(1.24, 1.4323288084)
		(1.25, 1.4203192453)
		(1.26, 1.4086911107)
		(1.27, 1.3974237788)
		(1.28, 1.3864982635)
		(1.29, 1.3758970425)
		(1.3, 1.3656039065)
		(1.31, 1.3556038269)
		(1.32, 1.3458828408)
		(1.33, 1.3364279501)
		(1.34, 1.3272270326)
		(1.35, 1.3182687633)
		(1.36, 1.3095425448)
		(1.37, 1.3010384449)
		(1.38, 1.2927471418)
		(1.39, 1.2846598738)
		(1.4, 1.2767683951)
		(1.41, 1.2690649358)
		(1.42, 1.2615421652)
		(1.43, 1.2541931594)
		(1.44, 1.2470113718)
		(1.45, 1.2399906054)
		(1.46, 1.233124989)
		(1.47, 1.2264089541)
		(1.48, 1.2198372148)
		(1.49, 1.2134047491)
		(1.5, 1.2071067812)
		(1.51, 1.2009387662)
		(1.52, 1.1948963754)
		(1.53, 1.1889754825)
		(1.54, 1.1831721518)
		(1.55, 1.1774826264)
		(1.56, 1.1719033178)
		(1.57, 1.1664307958)
		(1.58, 1.1610617798)
		(1.59, 1.1557931299)
		(1.6, 1.1506218396)
		(1.61, 1.1455450278)
		(1.62, 1.1405599327)
		(1.63, 1.1356639047)
		(1.64, 1.1308544012)
		(1.65, 1.1261289803)
		(1.66, 1.1214852963)
		(1.67, 1.1169210943)
		(1.68, 1.112434206)
		(1.69, 1.1080225451)
		(1.7, 1.1036841036)
		(1.71, 1.0994169478)
		(1.72, 1.0952192149)
		(1.73, 1.0910891093)
		(1.74, 1.0870249001)
		(1.75, 1.0830249175)
		(1.76, 1.0790875504)
		(1.77, 1.0752112435)
		(1.78, 1.0713944949)
		(1.79, 1.067635854)
		(1.8, 1.0639339188)
		(1.81, 1.060287334)
		(1.82, 1.0566947891)
		(1.83, 1.0531550165)
		(1.84, 1.0496667894)
		(1.85, 1.0462289206)
		(1.86, 1.0428402603)
		(1.87, 1.0394996953)
		(1.88, 1.0362061466)
		(1.89, 1.0329585688)
		(1.9, 1.0297559486)
		(1.91, 1.0265973033)
		(1.92, 1.0234816797)
		(1.93, 1.0204081532)
		(1.94, 1.0173758263)
		(1.95, 1.0143838281)
		(1.96, 1.0114313128)
		(1.97, 1.008517459)
		(1.98, 1.0056414688)
		(1.99, 1.002802567)
		(2.0, 1.0)
	};

	\addlegendentry[no markers, teal]{brute}
	\addlegendentry[no markers, blue]{$\alpha = 2.0$, $c = 1.0$}

	\end{axis}
\end{tikzpicture}
  \label{fig:pvc-results}
 \end{subfigure}%
 \caption{Results for {\sc Cluster Graph Vertex Deletion}, {\sc Split Vertex Deletion}, {\sc Multicut on Trees} and {\sc Partial Vertex Cover}.
   A dot at $(\beta,d)$ means that the respective algorithm outputs an $\beta$-approximation in time $\OO^*(d^n)$.}
 \label{fig:runtimes-appendix-3}
\end{figure}

\end{document}